\documentclass[a4paper,11pt]{article}
\pdfoutput=1 

\usepackage{aas_macros}
\usepackage{jcappub}
\usepackage[T1]{fontenc}

\usepackage{amsmath}
\usepackage{amsthm}
\usepackage{lipsum}
\usepackage{bm}
\usepackage{cuted}
\usepackage{flushend}
\usepackage{comment}
\usepackage{cleveref}

\usepackage{orcidlink}
\usepackage{soul}
\usepackage{framed}
\usepackage{graphicx}
\usepackage{subcaption}
\usepackage{rotating}
\usepackage{float}
\usepackage{nicematrix}
\usepackage{colortbl}
\usepackage{cprotect}
\usepackage{scalerel}

\graphicspath{{./figures/}}

\DeclareMathOperator{\tr}{tr}

\renewcommand{\d}{{\mathrm{d}}}

\newcommand{\deltaD}[2]{\delta^{(#1)} (#2)}

\newcommand{\MyMatrix}{\vcenter{\hbox{$\displaystyle
  \begin{pmatrix}
        \partial_1 \left( \bm{v}_1 \cdot \nabla \lambda_1 \right)  &  \partial_1 \left( \bm{v}_2 \cdot \nabla ( \lambda_2 + \lambda_3 ) \right) & \partial_1 \left( \bm{v}_3 \cdot \nabla ( \lambda_2 + \lambda_3 ) \right)\\
        \partial_2 \left( \bm{v}_1 \cdot \nabla \lambda_1 \right)  &  \partial_2 \left( \bm{v}_2 \cdot \nabla ( \lambda_2 + \lambda_3 ) \right) & \partial_2 \left( \bm{v}_3 \cdot \nabla ( \lambda_2 + \lambda_3 ) \right)\\
        \partial_3 \left( \bm{v}_1 \cdot \nabla \lambda_1 \right)  &  \partial_3 \left( \bm{v}_2 \cdot \nabla ( \lambda_2 + \lambda_3 ) \right) & \partial_3 \left( \bm{v}_3 \cdot \nabla ( \lambda_2 + \lambda_3 ) \right)
    \end{pmatrix}$}}}

\newtheorem{theorem}{Theorem}

\definecolor{veryLightGray}{rgb}{0.95, 0.95, 0.95}

\title{A New Recipe for Caustic Pancakes: On the Reality of Walls in the Cosmic Web}

\author[a,b,1]{Benjamin Hertzsch\orcidlink{0009-0006-6231-8905},\note{Corresponding author.}}
\author[a]{Job Feldbrugge  \orcidlink{0000-0003-2414-8707}}
\author[a]{Maé Rodriguez\orcidlink{0000-0001-5417-9665},}
\author[b]{Rien van de Weygaert \orcidlink{0000-0001-8379-1263},}

\affiliation[a]{School of Physics and Astronomy, University of Edinburgh\\
Edinburgh, United Kingdom}
\affiliation[b]{Kapteyn Institute of Astronomy, University of Groningen\\
Groningen, Netherlands}

\emailAdd{benjamin.hertzsch@ed.ac.uk}

\abstract{The caustic skeleton model is a mathematically rigorous framework for studying the formation history of the emerging cosmic web from the caustics in the underlying dark matter flow. In a series of two papers, we use constrained $N$-body simulations to investigate the different cosmic web environments. For the current study, we focus on the cosmic walls. We derive the conditions of the centres of proto-walls and analyse their evolution with $N$-body simulations. Next, we investigate the statistical properties of Zel'dovich pancakes by studying the number density of the cosmic wall centres in scale space and, for the first time, we calculate the Lagrangian-space volume of cosmic walls. Finally, we infer the mean density and velocity fields and the distribution of haloes around cosmic walls with a suite of physically realistic dark-matter-only simulations. We compare the cosmic walls obtained with the caustic skeleton framework with previously proposed saddle point conditions on the primordial potential and density perturbation.}

\begin{document}
\maketitle
\flushbottom



\section{Introduction}

This study presents an elaborate treatment and discussion on the role, nature and infrastructure of walls within the cosmic web.
It entails a profound physical description based on the \textit{caustic skeleton} model of the cosmic web
\cite{ArnoldShandarinZeldovich1982, Feldbrugge+2018, FeldbruggeWeygaert2023}, in which features and components of the cosmic large-scale structure are
identified with (Lagrangian) phase-space singularities in the cosmic matter distribution. This analytical non-linear model for the
formation and hierarchical evolution of the cosmic web determines the mass elements in the primordial matter distribution that get
assembled in the walls of the cosmic web. It provides a complete specification of the complex geometric structure of the
cosmic web through a set of conditions on the eigenvalues and eigenvectors of the primordial tidal and deformation field.
By means of an extensive set of (non-linear) constrained simulations of structure formation \cite{FeldbruggeWeygaert2023}, we delve into
the detailed formation history, complex matter distribution in and around of the emerging walls in the matter distribution, as well
as the intricate interior structure and halo population of the walls in the cosmic web. 

\vskip 0.5truecm

\begin{figure*}
    \centering
    \includegraphics[width=\textwidth]{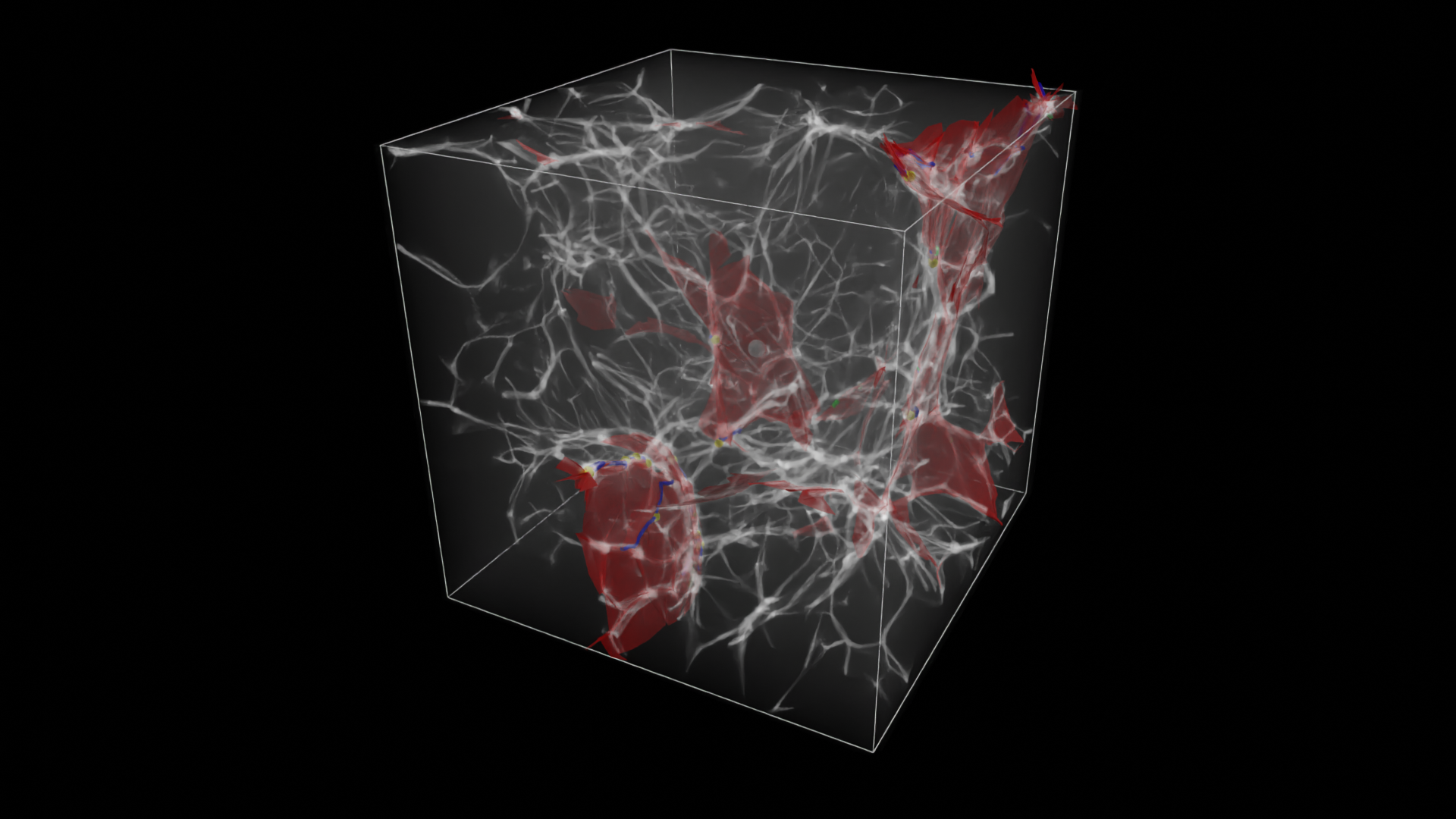}
    \cprotect\caption{High-resolution density field from a $256^3$-particle constrained simulation of cosmic wall formation from the caustic skeleton model. Shown is a volume of side length $50 \,h^{-1}\textrm{Mpc}$, with the large-scale cusp sheets, swallowtail and umbilic filaments and butterfly clusters in red, blue, green and yellow respectively. In the central white point we imposed the novel constraint presented in this paper, which hence resides at the centre of the simulated cosmic wall. A fly-through and rotating view animation is provided on the additional materials page at \verb|benhertzsch.github.io/papers/2025_Cosmic_Walls|.}
    \label{fig:sim_256_A3_3D}
\end{figure*}

\subsection*{Walls and the Cosmic Web}
On scales of a few to hundreds of megaparsecs, dark matter, gas, and galaxies are woven into the \textit{cosmic web}, a complex
spatial network consisting of superdense clusters, elongated filaments, and flattened walls that bound the near-empty void regions
\cite{Zeldovich1970, Einasto1977, BondKofmanPogosyan1996, WeygaertBond2008, Cautun+2014,EinastoM2025}. Representing the fundamental spatial
organisation of matter on these scales, its complexity is considerably enhanced by its intrinsic multiscale nature, featuring objects over a considerable range of spatial scales and densities. The filaments are the most visually outstanding features of the cosmic web, 
in which around $50\%$ of the mass and galaxies in the Universe reside \cite{Libeskind+2017, Ganeshaiah+2021}. On the other hand, almost $80\%$ of the
cosmic volume belongs to the interior of voids, see e.g. \cite{Cautun+2014, Ganeshaiah+2018}. By contrast, the walls in the
cosmic web are far less prominent, containing merely $\sim 24\%$ of the dark and baryonic matter, distributed over a volume
of about $18\%$ of the volume of the present-day Universe \cite{Cautun+2014, Libeskind+2017}. They form thin membranes
that separate the voids and are interspersed and surrounded by filaments.

In the observational reality, the existence and structure of the cosmic web has been revealed in most detail by maps of the nearby cosmos produced by large galaxy redshift surveys. Starting from first revelation of the web-like arrangement of galaxies by the CfA2 survey \citep[e.g.][]{LapparentGellerHuchra1986}, subsequent surveys such as 2dFGRS, the SDSS, the 2MASS and GAMA redshift surveys~\citep{Colless+2003, SDSS2004, Huchra+2012,GAMA2015} established the web-like arrangement of galaxies as a fundamental characteristic of cosmic structure. Maps of the galaxy distribution at larger cosmic depths, such as VIPERS~\citep{Vipers2013}, showed its existence over a sizeable fraction of cosmic time.

For understanding the process of cosmic structure formation, insight into the structure and evolution of walls is of pivotal
importance. Seemingly less visually outstanding than the filaments, it is the walls that play a defining role in the assembly of the cosmic web. They are the first genuine structures that emerge in the Universe, as was pointed out by Zel'dovich and collaborators
\citep{Zeldovich1970,Doroshkevich1970, ArnoldShandarinZeldovich1982, ShandarinZeldovich1989}. His \textit{pancake theory} stipulated how the contraction and collapse of overdensities first proceeds towards flattened wall-like morphologies, and subsequently towards the development of filaments and clumps, all embedded and mutually connected in a volume-filling network. Although the multiscale reality
of the current plausible $\Lambda$CDM cosmology is more multifaceted than the original pancake models, in nearly all viable
cosmologies cosmic walls constitute a major morphological element of the cosmic matter distribution. 

Notwithstanding their central role in the formation of structure, in the observational reality walls are not easy to recognise.
Their low surface density renders them more difficult to discern than the more massive and prominent filaments. When
seeking to identify them in the distribution of galaxies in (magnitude-limited) galaxy surveys, this is augmented by the fact
that walls are mainly populated by low-mass haloes \cite{Hahn+2007, Cautun+2014, AlonsoEardleyPeacock2015, MetukiLibeskindHoffman2016, Libeskind+2017}, within which, eventually, only faint galaxies are assembled \cite{Metuki+2015}. 
The sparsity of high-luminosity objects in the walls makes these notoriously difficult to outline in cosmic surveys, which may
explain why they frequently overlooked in studies of cosmic large-scale structure.

While the major share of walls in the web-like spatial arrangement of galaxies in the cosmic web may escape attention, galaxy
surveys have revealed the most striking representatives in the Local Universe. The earliest recognised example of a large wall-like
structure in the Local Universe is that of the Local Supercluster, of which the Virgo Cluster is the central mass concentration. de
Vaucouleurs (1958) \cite{Vaucouleurs1958} recognised the flattened arrangement of nearby galaxies on the sky, culminating in the definition of the Supergalactic Coordinate System. As galaxy redshift surveys started to map the spatial galaxy distribution in the Local Universe to further distances, two major wall-like structure have been identified, namely the Perseus-Pisces supercluster \cite{GiovanelliHaynes1986} and the
CfA2 Great Wall. The latter is also known as Coma Great Wall and was detected by the second CfA Redshift Survey
\cite{GellerHuchra1989}. Perhaps the best illustration of the structure, interior and boundary of a Local Universe wall is the Perseus-Pisces (PP) wall at a distance of around $75\,h^{-1}\textrm{Mpc}$, of which one of its boundary ridges is the massive PP chain,
the prime and canonical example of a dominant cosmic filament. The 21cm survey by Giovanelli \& Haynes \cite{GiovanelliHaynes1986}
produced an impressive face-on map of the surface density of galaxies detected via their 21cm HI line emission, nicely
outlining the PP wall with a size of around $50 \times 100 \,h^{-1}\textrm{Mpc}$, surrounded by dense filamentary ridges. The map reveals the
intimate structural affinity of walls with the filaments defining its boundary, as well as its interior substructure. For a recent up-to-date map, meticulously illustrating the intricate dark matter distribution implied by the 2M++ galaxy survey,
see Hidding et al. (2016) \cite{HiddingWeygaertShandarin2016}. 

The most impressive and outstanding examples out to a few hundred Megaparsec are the Sloan Great Wall, discovered in the
Sloan Digital Sky Survey (SDSS) \cite{Gott+2005, Einasto+2011} and, most recently, the SDSS-III's Baryon Oscillation Spectroscopic
Survey (BOSS) identified the BOSS Great Wall \cite{Lietzen+2016, Einasto+2017}; for an excellent recent review of its supercluster
complexes see \cite{Einasto2025}. These superstructures are not so much single, more or less smooth, flattened objects but rather
flattened assemblies whose interior is marked by a highly complex interior structure of filaments and clumps, interspersed by
low-density regions. Face-on maps of the implied interior density provide a good impression of this, see e.g. the maps in the
PhD thesis of Platen (2009) \cite{Platen2009} and \cref{fig:SDSS_maps} of the present study.

There are indications for the existence
of even larger flattened assemblies in the galaxy distribution, stretching out to many hundreds of Megaparsec. Tully (1986)
\cite{Tully1986} pointed out that extending along the plane of the Local Supercluster, one could recognise a flattened
distribution of galaxies and clusters out to even $0.1c \approx 360 \,h^{-1}\textrm{Mpc}$, with a similar flattened configuration pointed out at various
occasions by Peebles and discussed in detail in his recent study \cite{Peebles2023}. In fact, the extension of flattened
configurations may be a direct reflection of the geometric properties of the cosmic web, yielding assemblies of walls
that are geometrically aligned over vast distances. The recent identification of a range of such superstructures in the
distribution of X-ray clusters of galaxies, most notably the Quipu complex \cite{Boehringer2025}, may provide a telling
illustration of such geometrically outlined features in the cosmic matter distribution.  

Current Stage-IV surveys such as the Dark Energy Spectroscopic Instrument (DESI) \cite{DESI2016} or Euclid \cite{Euclid2024}
are observing our Universe to an unprecedented level of detail and are expected to uncover ever more walls in our close and
intermediate cosmic neighbourhood. These and the other structural elements of the present-day, highly non-Gaussian cosmic web contain
substantially more cosmological information than the non-linear power spectrum alone, see e.g.
\cite{Bermejo+2024, SunseriBayerLiu2025}. To fully profit from the wealth of incoming data, a better theoretical understanding of
the identity and formation of the cosmic web is essential. One relevant issue is that we do not have a complete picture of
how current and upcoming galaxy surveys provide a biased view of the underlying multiscale web-like structure of the
cosmic (dark) matter distribution. As important is the absence of a physically motivated, profound and unequivocal agreement
on what defines a filament or a wall. A variety of methods and techniques have been developed and
forwarded to identify and classify web-like structures in cosmological computer simulations and galaxy surveys (see
\cite{Libeskind+2017} for an overview and comparison of currently well-known methods). Having allowed considerable progress
with respect to the properties and evolution of the cosmic web, most techniques have a largely heuristic character and
are limited to the specific cosmological context of the $N$-body simulations to which they are applied. A versatile
and general theoretical understanding of cosmic web formation still needs to be established. 
  
\subsection*{Cosmic Web and Tidal Forces}  
In the state-of-the-art paradigm of late-time cosmology, the present-day cosmic web originated from the gravitationally driven
growth of primordial Gaussian density fluctuations \cite{Peebles1980,ShandarinZeldovich1989}. Understanding from which primordial
configurations in the initial Gaussian random field these structures and objects have emerged, and how different global cosmological
conditions including the nature of dark matter, dark energy, and the mass of neutrinos have impacted their structure and
evolution, is of key importance. According to the prevailing view, they originated from quantum noise during the inflationary era.

A vast array of computer simulations, along with detailed maps of the current galaxy distribution, have indicated that the subsequent
formation and evolution of structure proceeds via a salient spatial web-like pattern, consisting of a complex, intricate and
multiscale network. It marks the dynamical transition at which we see the emergence, presence and proliferation of distinct and
recognizable non-linear cosmic structure, assembled in a complex spatial pattern. It means that web-like structures represent a key
stage on the path towards the subsequent fully non-linear development, and the condensation of the rich variety of (astrophysical)
objects populating our current universe. It is the direct causal link between the primordial matter distribution and the emerging
components of the cosmic web which singles them out as ideal probes of the cosmic structure formation process. At the current epoch, we find that the transition from primordial to web-like structure is occurring at Megaparsec scales, rendering the cosmic web observed
in cosmic surveys as the fossils of cosmic structure formation.

The cosmic web has emerged out of a long phase of linear evolution of the primordial density and velocity perturbations before turning into a more advanced non-linear stage involving contraction and collapse of the growing mass inhomogeneities. At the heart of the emergence of the web-like structure of the mildly non-linear mass distribution is the anisotropy of the gravitational force field generated by
the inhomogeneous mass distribution. The gravitational tidal force field effects the deformation of the matter distribution and the
contraction and collapse of matter in increasingly anisotropic, flattened and/or elongated, features. This had first been recognised, and accurately described in the mildly non-linear stage, by the Zel'dovich formalism \citep{Zeldovich1970}. Hence, by implication, the key to modelling and understanding the structure of the cosmic web is the structure of the gravitational tidal force field. Its
seminal role in shaping the anisotropic wall-like and filamentary structures in the cosmic web has been recognised for in a range of studies \cite{Zeldovich1970, BondKofmanPogosyan1996, Hahn+2007, WeygaertBond2008b, Feldbrugge+2018,  ParanjapeHahnSheth2018, Paranjape2021}, recently followed
by the detailed study by Kugel \& van de Weygaert (2024) \cite{KugelWeygaert2024} on the connection between the tidal force field and the various morphological components of the cosmic web. Interesting new insights on this have been obtained in the context of the \textit{Caustic Skeleton} theory of the cosmic web, the subject of the present study. Feldbrugge et al. (2018) \cite{Feldbrugge+2018} demonstrated that a full understanding of the cosmic web structure is obtained through the spatial characteristics of the \textit{eigenvalue} and \textit{eigenvector} fields of the cosmic tidal force field. Moreover, underlining this is the realization that the embryonic outline of the cosmic web, in particular its filamentary and wall-like network, can already be seen in the primordial tidal eigenvalue field
\cite{Wilding2022, FeldbruggeYanWeygaert2023, FeldbruggeWeygaert2023, FeldbruggeWeygaert2024, RamWeygaertFeldbrugge2025}.

\subsection*{Phase Space Dynamics and Walls in the Caustic Skeleton}
In recent years, the dominant position of the tidal force field in laying out the spatial structure of the cosmic web has led to
the formulation of an analytical theory for the fully non-linear evolution of the cosmic web, the \textit{caustic skeleton} formalism, by
Feldbrugge and collaborators (2018) \cite{Feldbrugge+2018}. It entails an elaboration of the early work by Arnol'd et al. (1982)
\cite{ArnoldShandarinZeldovich1982}, and the recent two-dimensional analysis by Hidding et al. (2013) \cite{Hidding+2013}, in which
the morphological elements of the cosmic web are identified with the formation of singularities in the mass
distribution. The caustic skeleton is based on the realisation that a physically motivated model for the structure and evolution of the cosmic web
should assess the evolving cosmic matter distribution in the full six-dimensional phase-space. 

The build-up of structure in the universe goes along with the migration of mass \cite{Peebles1980}. The structure, evolution and
fate of the corresponding mass flows is therefore a key aspect for understanding how the primordial mass distribution gets
gradually folded into the richly structured pattern of the cosmic web. To some extent, one may argue that insight into the
cosmic mass flows is as, if not more, essential for understanding the emergence of structure in the Universe than that of
the mass distribution itself. In fact, while the \textit{Zel'dovich pancake model} and the surprising accuracy and versatility of
the \textit{Zel'dovich formalism} is widely known, less known is the fact that Zel'dovich already recognised that the  
formation of structures goes along with \textit{shell-crossing} of the cosmic flows, and hence with 
corresponding emergence of multistream regions. A full phase-space description allows one to take this realisation into account.

By appreciating that the phase-space structure of the evolving cosmic mass distribution is directly related to the
identification of the structural components of the cosmic web, several structure identification formalisms have been developed \cite{Shandarin2011, AbelHahnKaehler2012, Neyrinck2012, Shandarin2011, Feldbrugge2024}. Given their solid physical
underpinning, they may possibly represent the most profound classification method of the cosmic web as are currently available, see also \cite{Libeskind+2017}. When the initial conditions are known, these methods yield an identification of the matter streams building
up the hierarchical structure \cite{Sheth2004, ShethWeygaert2006, Shen+2006, AragonCalvo+2010} in our Universe. This allows for the
definition of objective physical criteria for what constitutes the various structural elements of the cosmic web.

The phase-space evolution of the cosmic matter distribution, in particular when restricting it to the pressureless dark matter distribution,
can be seen as the folding of an initially smooth and featureless dark matter sheet. Following shell-crossing, we
see the formation of corresponding multistream regions. Dependent on the (evolving) geometric complexity of these regions, a
hierarchy of emerging non-linear structures can be identified. In the cosmological context, their geometry
can be directly related to the nature, identity and morphology of the emerging web-like features. \textit{Catastrophe theory} \cite{Thom1972, Zeeman1977, Arnold1975, Saunders1980, Arnold1992, ArnoldGuseinZadeVarchenko2012}
is the full mathematical framework that defines and classifies these stable folding configurations.

The phase-space-based \textit{caustic skeleton} description of the evolving web-like pattern in the cosmic matter
distribution is codified in terms of a complete set of caustic conditions \cite{Feldbrugge+2018}. The key insight and
central aspect is that these caustic conditions, for the full set of $A$, $D$ (and in higher-dimensional settings also $E$) Lagrangian singularities\footnote{We remark here that caustic skeleton, as an application of catastrophe theory (see in particular \cite{Arnold1975}), is a mathematically general formalism that holds for arbitrary space dimensions. In the physical case of the three-dimensional Universe, the $A$ and $D$ catastrophes are stable, constituting an exhaustive list of seven elementary catastrophes in accordance with Thom's morphogenetic definition \cite{Thom1972}.}, were shown
to be fully specified by the \textit{eigenvalues} and \textit{eigenvectors} of the
\textit{primordial tidal field},  i.e. of the Hessian of the gravitational potential field. This opened the path towards the
practical calculation and implementation of the caustic conditions to any viable three-dimensional cosmological matter
field, and its extension and elaboration into a
fully fledged analytical model for cosmic web formation. 

Within the context of the caustic skeleton theory, the
walls are the first objects to condense out of the primordial mass distribution:
In the language of catastrophe theory, the flattened, ellipsoidal overdensities undergo a cusp catastrophe and the incoming mass streams fold into a \textit{three-streaming} configuration. In the cosmological context, the resulting multistream regions are known as the 
\textit{Zel'dovich pancakes} \cite{ArnoldShandarinZeldovich1982, ShandarinZeldovich1989}, see \cref{fig:caustics,fig:cusp_sketch}. It is important to realise that a pancake or wall is
not simply a flat overdensity: its oblate nature is inherently inside-out, as it is created by matter streams folding onto each other.
The flattened overdensities that we know as the present-day cosmic walls are the result of gravitational phase-mixing in these
multistream regions. In other words, the Zel'dovich pancakes\footnote{It is important to remark here that the early theory of a pancaking Universe from structure-less multistream regions is obsolete. Our investigation is not confined to the limitations of the early Russian-school models and we use the term ``pancake'' in its modern form: The caustic skeleton gives a central stage to the multistream regions and cusp sheets as the lowest-order catastrophes, but is fully sensitive to their scale-space nature and their seeding of higher collapse, yielding the intricate web-like substructure in the physically realistic cosmic walls.} are the progenitors of the cosmic walls, the \textit{proto-walls}.
The caustic skeleton formalism \cite{ArnoldShandarinZeldovich1982, Hidding+2013, Feldbrugge+2018} recognises this process and formalises the definition of cosmic walls. Their growing network is crucial to the build-up of the entirety of the cosmic web, as filaments and clusters form from the hierarchy of higher caustics within those pancake-shaped three-stream regions \cite{ArnoldShandarinZeldovich1982, Hidding+2013, Feldbrugge+2018}.

The subject of the present study is a detailed study of the structure of walls in the context of the \textit{Caustic Skeleton}
theory of the cosmic web. It involves a fundamentally expanded vision on how \textit{pancakes} form and evolve, and
what their resulting structure and substructure is. While the original Zel'dovich pancake theory assumed pancake walls to smooth structural entities, it was limited to cosmologies
in which structure evolved at one scale or cosmologies whose primordial field with power over a (highly) restricted range of scales.
In $\Lambda$CDM cosmology, and any currently viable cosmological model, structure emerged at a wide range of scales and evolves in a
hierarchical process in which small structures condense first, and subsequently merge into ever larger entities. The \textit{caustic
  skeleton} model allows to take account of this, as it takes account of the multiscale nature of the primordial tidal field.
Hence, it enables a detailed study of the evolution of walls in the cosmic web, taking into account their hierarchical
build-up and assess the resulting complex interior structure, in particular that of filamentary tendrils, clumps and haloes. 

In summary, this work entails a detailed investigation on a rigorous mathematical footing of the dynamics that underlie
the formation and build-up of the cosmic walls. In doing so, we complete and extend the Zel'dovich pancake model as part of the more
general caustic skeleton formalism, and embed it into the modern view on the physical cosmic web.

\begin{figure*}
    \centering
    \includegraphics[width=\textwidth]{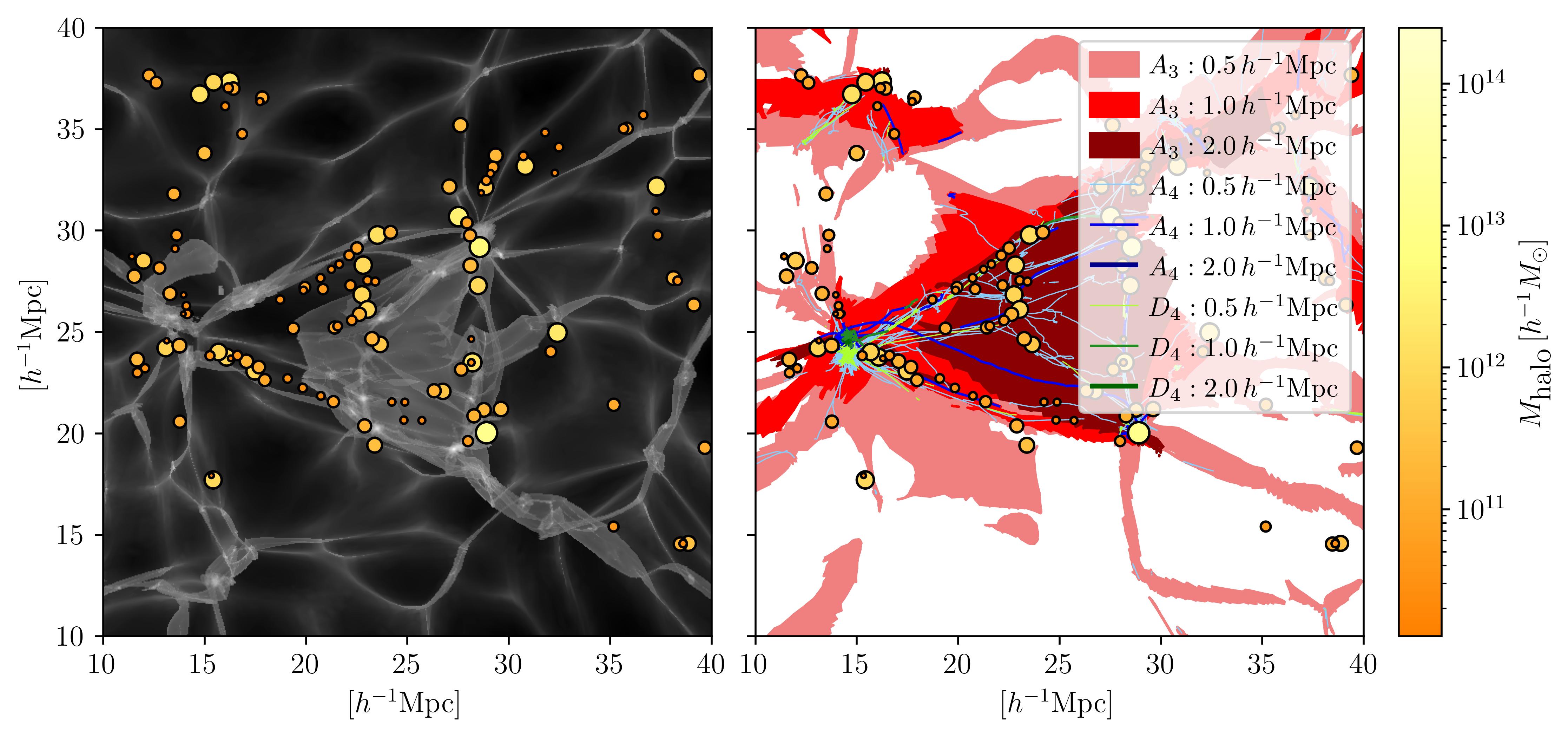}
    \caption{Haloes identified in a high-resolution constrained simulation of cosmic wall formation. The left panel shows the density field sliced along the face of the wall, along with the haloes identified in a thin volume (thickness $\epsilon = 1.5 \,h^{-1}\textrm{Mpc}$) around the wall, projected into the same slice. The right panel shows the caustic skeleton evaluated at different length scales $\sigma$ (see \cref{subsec:theory-scale_space}), with the $A_3$ cusp walls in red, the $A_4$ swallowtail filaments in blue and the $D_4$ umbilic filaments in green respectively.}
    \label{fig:haloes_introduction}
\end{figure*}

\subsection*{Key results}

\paragraph{Caustic skeleton constraint}  Within the caustic skeleton framework, the cosmic walls are associated to the cusp caustic and algebraically defined by the manifold of points in Lagrangian space for which the directional derivative of the first eigenvalue field in the direction of the corresponding eigenvector field vanishes, i.e.,  $\bm{v}_1 \cdot \nabla \lambda_1=0$, with the eigenvalue and eigenvector fields $\lambda_i$ and $\bm{v}_i$. Extending the 2D study of \cite{FeldbruggeWeygaert2024} to the 3D case, we identify on these manifolds the \textit{wall centres} as the points at which the cusp sheet is maximally expanding. This is given by the vanishing of two additional directional derivatives, i.e., explicitly, $\bm{v}_2 \cdot \nabla (\lambda_2 + \lambda_3)=\bm{v}_3 \cdot \nabla (\lambda_2 + \lambda_3)=0$. Using non-linear constrained Gaussian random field theory \cite{FeldbruggeWeygaert2023}, we sample Gaussian initial conditions with a combination of a Hamiltonian Monte Carlo \cite{Duane+1987} sampler and the Bertschinger-Hoffman-Ribak algorithm \cite{Bertschinger1987, HoffmanRibak1991, WeygaertBertschinger1996} and implement these into specialised constrained simulations of cosmic wall formation. In doing so, we provide a rigorous and reproducible recipe for the physically realistic Zel'dovich pancake, enabling direct applications in future studies of galaxy formation within cosmic walls.

\paragraph{Dark matter structure of Cosmic Walls} We evolve the constrained initial conditions in dark-matter-only $N$-body simulations and analyse the morphology of the resulting cosmic walls in the present-day Universe (see \cref{fig:sim_256_A3_3D}). Our proposed wall centre constraint successfully sets up the cosmic walls as three-streaming regions traced out by the cusp sheet. The physically realistic walls are interspersed by filaments and clusters, and connected to the global cosmic web. This reflects the inherent connectivity of the caustic network beyond the multiplicity of clusters and filaments \cite{CodisPogosyanPichon2018}, and underlines the essential role of the walls for the formation of filaments and clusters \cite{ArnoldShandarinZeldovich1982, Feldbrugge+2018}. The simulated walls are moderately overdense with a density contrast of about $\rho / \bar{\rho}\approx 2\textrm{--}10$, which is consistent with the existing literature \cite{Forero-Romero+2009, AragonCalvo+2010,  ShandarinSalmanHeitmann2012, Hoffman+2012, Cautun+2014, Libeskind+2017}, but augments the same by the more rigorous criterion of shell-crossing into a planar three-streaming configuration. In addition to the density field, we study the mean velocity field around the simulated walls and find clear evidence for a dipolar pattern of mass being transported from the cosmic voids into the growing wall, in accordance with the paradigm of mass transport in the cosmic web \cite{Cautun+2014}.

\paragraph{Cosmic Wall Formation times} To simulate physically realistic walls, we evaluate the characteristic formation time of the walls as a function of their length scale. To this end, we employ Rice's formula and calculate the formation statistics from the primordial fields. Imposing different formation times and length scales in the constrained simulations, we study the influence of the parameters on the resulting dark matter density field that characterises the walls. We find that walls that formed earlier have accumulated more mass and are thicker than walls that formed more recently. Further, their spatial extent is in clear proportionality to the length scale at which the caustic constraint is imposed. With this paper, we also present the generalised Rice's formula that enables the evaluation of arbitrary geometric expectation values beyond the number density of point-like constraints. Using this formula, we investigate for the first time the Lagrangian-space area density of the cosmic walls as a function of cosmic time for different wall length scales. Our calculation constitutes the first quantitative analysis of the wall network build-up in Lagrangian space, and forms an important step towards a full theoretical understanding of the formation of the entirety of the cosmic web in scale space.

\paragraph{Cosmic Wall Halo Population} In a suite of high-resolution simulations, we investigate the embedding of haloes in the cosmic walls. We find that their distribution is highly non-uniform and characterised by filamentary patterns within the cusp sheet (see \cref{fig:haloes_introduction}). The caustic skeleton reveals that these are traced by the small-scale swallowtail and umbilic caustic forming within the large-scale wall sheet. The filamentary distributions of haloes in our constraint simulations directly reflects the observational reality of the galaxy population in the physical walls in our Local Universe, notably the Pisces-Perseus Supercluster \cite{GiovanelliHaynes1986}, the Coma Wall \cite{GellerHuchra1989}, the Sloan Great Wall \cite{Gott+2005, Einasto+2011} and the BOSS Great Wall \cite{Lietzen+2016, Einasto+2017}. Their intricate substructure \cite{Platen2009, Platen+2011, Einasto+2011, Einasto+2017, EinastoM2025} may be explained by the scale-space nature of the dark matter flow singularities. In agreement with previous studies \cite{Hahn+2007, Cautun+2014, AlonsoEardleyPeacock2015, MetukiLibeskindHoffman2016, Libeskind+2017}, we find that the wall haloes are less massive than those residing in the denser filaments and clusters. In combination, our results establish the caustic skeleton as the first theoretical foundation for the elusive nature of the cosmic walls in observational cosmology.

\paragraph{Alternative Wall Constraints} For a fair comparison with the caustics-based analysis, we run suites of constrained simulations for the conventional saddle point constraints in the primordial potential perturbation $\phi$ and density perturbation $\delta$. We find that neither of these is as successful at simulating cosmic walls as our proposed wall centre constraint based on the singularities in the dark matter flow. Concretely, while the $\phi$ saddle points induce the required planar geometry of the Eulerian density field, the overdensities do not generally shell-cross into multistreaming cosmic walls. Whereas for the $\delta$ saddle points, while the constraint results in a highly multistreaming configuration, the morphology of the Eulerian density field typically resembles an extended cluster rather than a sheet-like wall. In light of these shortcomings, our results demonstrate that the proposed Zel'dovich pancake recipe significantly improves upon conventional Lagrangian-space saddle-point analyses.

\begin{table}[H]
    \centering
    \begin{tabular}{c c}
    \hline
    symbol & meaning \\
    \hline
    $\bm{q}$ & Lagrangian coordinates \\
    $\bm{x}$ & Eulerian coordinates \\
    $\bm{s}_t(\bm{q})$ & displacement field \\
    $\nabla_{\bm{q}} \bm{s}_t(\bm{q})$ & deformation tensor \\
    \rowcolor{veryLightGray}
    $\sigma$ & smoothing scale \\
    \rowcolor{veryLightGray}
    $b(t)$ & linear growing mode \\
    \rowcolor{veryLightGray}
    $b_c$ & constraint growing mode \\
    $\phi(\bm{q})$ & primordial potential perturbation \\
    $\Psi(\bm{q})$ & primordial displacement potential $\Psi \propto \phi$ \\
    $\delta(\bm{q})$ & primordial density perturbation $\delta = \nabla^2 \Psi$ \\
    $\rho(\bm{x})$ & dark matter density at current time \\
    \rowcolor{veryLightGray}
    $M_{\textrm{halo}}$ & halo mass \\
    \rowcolor{veryLightGray}
    $\rho(N_{\textrm{halo}})$ & halo number density \\
    \rowcolor{veryLightGray}
    $\rho(M_{\textrm{halo}})$ & halo mass density \\
    $P(k)$ & power spectrum of  $\Psi$ \\
    $\sigma_i^2$ & generalised moments of $P(k)$ \\
    $t_{ij \ldots k}$ & Cartesian field derivatives $\partial_i \partial_j \ldots \partial_k \Psi$ \\
    $T_{ij \ldots k}$ & eigenframe field derivatives $\partial_i \partial_j \ldots \partial_k \Psi$ \\
    $\Sigma_{a b}$ & covariance matrix of $a$th- and $b$th-order derivs. $t_{ij \ldots k}$ \\
    \rowcolor{veryLightGray}
    $(\alpha, \beta, \gamma)$ & Euler angles in $ZYZ$-convention \\
    \rowcolor{veryLightGray}
    $R^{(a, \ldots, b)}(\alpha, \beta, \gamma)$ & rotation matrix of $a$th- through $b$th-order derivs. $t_{ij \ldots k}$\\
    $\mu_i(\bm{q}, t),\, \bm{v}_i(\bm{q}, t)$ & eigenvalues/-vectors of $\nabla_{\bm{q}} \bm{s}_t(\bm{q})$ \\
    $\lambda_i(\bm{q}),\, \bm{v}_i(\bm{q})$ & eigenvalues/-vectors of $\mathcal{H} \Psi(\bm{q})$ \\
    \rowcolor{veryLightGray}
    $A_2$ & fold caustic \\
    \rowcolor{veryLightGray}
    $A_3$ & cusp caustic \\
    \rowcolor{veryLightGray}
    $A_4$ & swallowtail caustic \\
    \rowcolor{veryLightGray}
    $A_5$ & butterfly caustic \\
    \rowcolor{veryLightGray}
    $D_4$ & umbilic caustic \\
    \rowcolor{veryLightGray}
    $D_5$ & parabolic umbilic caustic \\
    $\phi^{(+--)}$ & saddle point of type $(+--)$ in $\phi$ \\
    $\delta^{(++-)}$ & saddle point of type $(++-)$ in $\delta$ \\
    \rowcolor{veryLightGray}
    $\delta^{(n)}_{D}(\cdot)$ & $n$-dimensional Dirac-$\delta$ function \\
    \rowcolor{veryLightGray}
    $\mathcal{H}(\cdot)$ & Hessian $\mathcal{H}(\cdot) = [\partial \partial (\cdot)]_{ij}$ \\
    \hline \\
    \end{tabular}
    \caption{Table of symbols used throughout this article.}\label{tabel:symbols}
\end{table}

\section*{Outline}
In \cref{sec:theory}, we review the caustic skeleton framework, focusing on the formation of multi-stream regions and the Zel'dovich approximation. Next, in \cref{sec:A3_centre_constraint}, we derive a new set of conditions for the identification of progenitors of cosmic walls in terms of eigenvalue and eigenvector fields of the primordial deformation tensor. We  translate these conditions into a constraint on the local geometry of the primordial fields in \cref{sec:eigenframe}. Subsequently, in \cref{sec:recipe}, we develop a method to generate constrained Gaussian initial conditions subject to the caustic constraints using a Hamiltonian Monte Carlo scheme and the Bertschinger-Hoffman-Ribak algorithm. Using this implementation, in sections \ref{sec:sims} and \ref{sec:haloes}, we run suites of constrained $N$-body simulations and study the morphology of the dark matter fields and embedded halo distributions in the walls of the present-day cosmic web. In \cref{sec:sims}, we also consider the formation time of the cosmic walls and investigate the influence of the physical constraint parameters on the resulting wall morphology. A comparison with simulations from conventional saddle point constraints is given in \cref{sec:alternatives}. Here, we analyse a set of $N$-body simulations constrained on saddle points in the primordial gravitational potential and the primordial density perturbation respectively. We compare the resulting structures with the caustic walls obtained in \cref{sec:sims} and discuss the shortcomings of the conventional saddle point constraints. Finally, the results of the present article are summarised in \cref{sec:conclusion} and an outlook to future work is given.


\section{The caustic skeleton of cosmological structure formation}\label{sec:theory}

The cosmic web is an intricate and multiscale geometric pattern that emerges through non-linear gravitational collapse from the primordial conditions. In this section, we discuss the inherently \textit{multistreaming} nature of the cosmic web and review the build-up of its geometric backbone in the \textit{caustic skeleton} formalism.

\begin{figure*}
    \includegraphics[width=\textwidth]{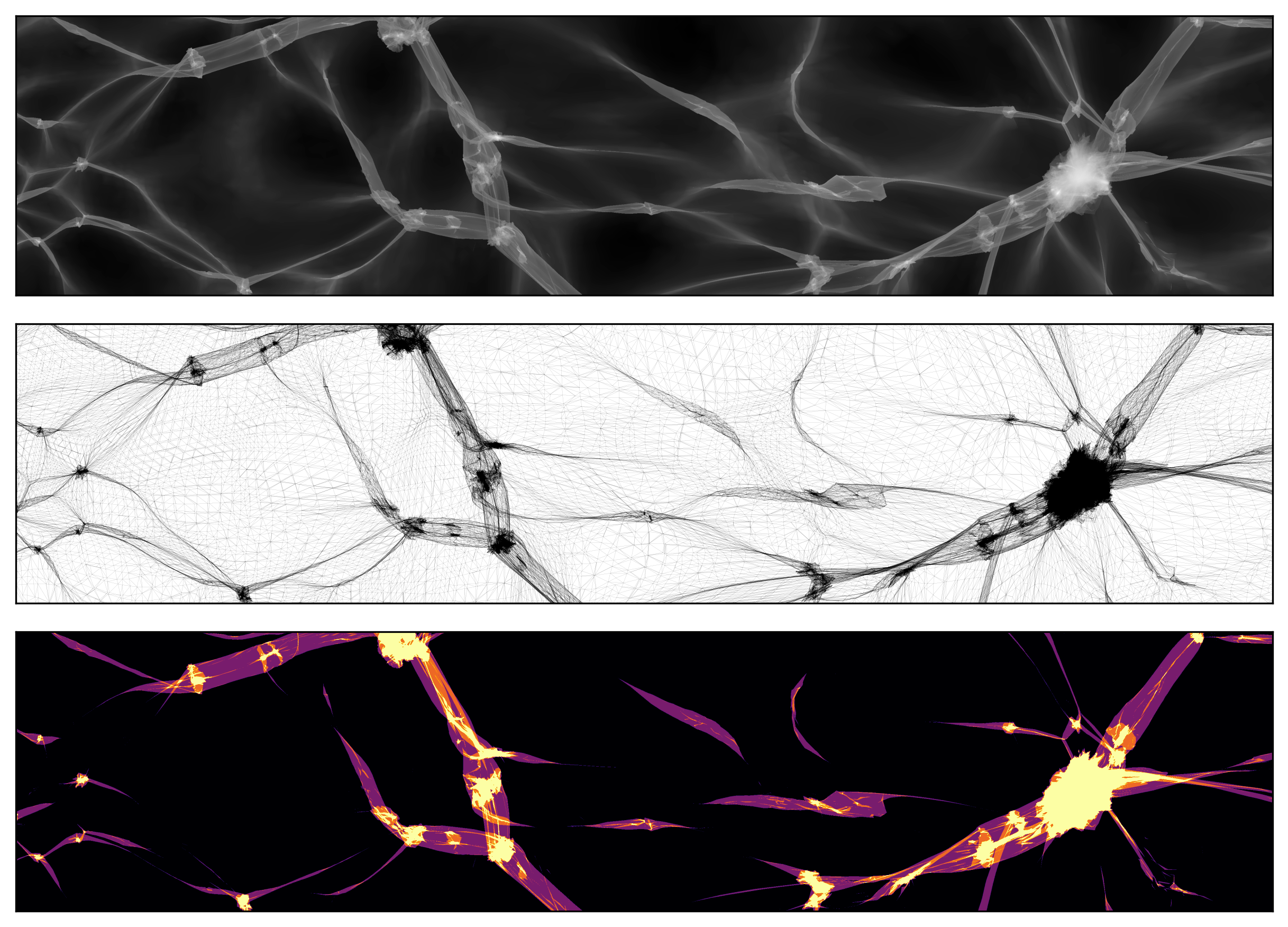}
    \caption{Different views of the cosmic web from slices of width $70\,h^{-1
    }\textrm{Mpc} $ and height $20\,h^{-1
    }\textrm{Mpc}$ slices through a $256^3$-particle $N$-body simulation in a box of side length $100\,h^{-1
    }\textrm{Mpc}$. The upper panel shows the density field, the middle shows a slice of through the folding particle mesh and the lower panel shows corresponding the number of streams coming into the Eulerian positions. The upper and lower panels were evaluated using the PS-DTFE method.}
    \label{fig:sim_256_triple_plot}
\end{figure*}

\subsection{The multistreaming cosmic web}\label{subsec:theory-cosmic_web}

Understanding how the cosmic web emerges and evolves over cosmic time one of the major theoretical challenges of late-time cosmology. The intricate, multiscale pattern of cosmic voids, walls, filaments and clusters forms as the solution to a set of simple equations of motion (EoMs) through non-linear collapse from the primordial random field conditions. In an FLRW-$\Lambda$CDM Universe in the Newtonian limit, the EoMs in comoving coordinates are given by the Poisson equation and Newton's law \cite{BertschingerGelb1991, Peebles1994},
\begin{align}
    \nabla^2_{\bm{x}} \Phi &= 4 \pi \bar{\rho}(t) \delta(t)  \\
    \ddot{\bm{x}} + 2 H(t) \dot{\bm{x}} &= -\frac{1}{a(t)^2} \nabla_{\bm{x}} \Phi \,,
\end{align}
with the comoving coordinates $\bm{x}$, the Hubble function $H(t) = \dot{a}(t)/a(t)$, the scale factor $a(t)$, the gravitational potential $\Phi$ and the density contrast $\delta(t) = \frac{\rho(t) - \bar{\rho}(t)}{\bar{\rho}(t)}$ over the mean cosmic density $\bar{\rho}(t)$.

$N$-body simulations are a highly useful tool for reproducing the plethora of structures observed on the sky in targeted numerical experiments. This is achieved by the discretised, numerical approximation to the solutions to the governing EoMs for a set of particles that trace the continuous matter fields, see e.g. \cite{BertschingerGelb1991}. Hence, while cosmological surveys observe the physical cosmic web in its current state, $N$-body codes allow for structure formation experiments under variation of input parameters and, notably, the identification of the initial conditions out of which particular late-time structures form. Nonetheless, despite the wide use of $N$-body simulations in numerical applications, structure formation is yet to be understood on a rigorous theoretical footing. Our present study is built on the realisation that this is only possible by considering the phase-space evolution of the dark matter fluid and assessing the full hierarchy of features arising from the singularities in its flow.

Cosmological structure formation is suitably captured in the Lagrangian fluid formalism, where a particle evolves from an initial (\textit{Lagrangian}) position $\bm{q}$ to a final (\textit{Eulerian}) position $\bm{x}_t(\bm{q})$ through a displacement field $\bm{s}_t(\bm{q})$ as
\begin{equation}
    \bm{x}_t(\bm{q}) = \bm{q} + \bm{s}_t(\bm{q})\,.
    \label{eq:Lagrangian_formalism}
\end{equation}
By the conservation of mass, the density $\rho$ at a position $\bm{x}$ follows from the expansion/contraction of a congruence of paths
\begin{equation}
    \begin{split}
    \rho \left( \bm{x}, t \right)
    &= \sum_{\bm{q} \in \bm{x}_t^{-1}(\bm{x})} \frac{\rho(\bm{q}) }{|\det \nabla_{\bm{q}} \bm{x}_t(\bm{q})|} \\
    &= \sum_{\bm{q} \in \bm{x}_t^{-1}(\bm{x})} \frac{\rho(\bm{q}) }{|\det (I +\nabla_{\bm{q}} \bm{s}_t(\bm{q}))|} \\
    &= \sum_{\bm{q} \in \bm{x}_t^{-1}(\bm{x})} \frac{\rho(\bm{q}) }{|1 + \mu_1(\bm{q},t)| |1 + \mu_2(\bm{q},t)| |1 + \mu_3(\bm{q},t)|} \,
    \end{split}
    \label{eq:Eulerian_density}
\end{equation}
with the deformation tensor $\nabla_{\bm{q}} \bm{s}_t$ and the associated space- and time-dependent eigenvalues $\mu_i(\bm{q}, t)$ and eigenvectors $\bm{v}_i(\bm{q}, t)$ defined by the characteristic equation
\begin{equation}
    (\nabla_{\bm{q}} \bm{s}_t) \bm{v}_i(\bm{q}, t) = \mu_i(\bm{q}, t) \bm{v}_i(\bm{q}, t)  \,. 
\end{equation}

\Cref{eq:Eulerian_density} encapsulates a crucial concept: Several Lagrangian coordinates $\bm{q}$ may stream into the same Eulerian coordinate $\bm{x}$ and contribute to the observed density field. This is well illustrated in \cref{fig:sim_256_triple_plot}, where the overlapping streams of the particle mesh in the middle panel contribute to the overdensities of the collapsed structures observed in the upper panel. Regions with several incoming streams are referred to as \textit{multistream regions}, and emerge when the particle trajectories cross. As the structures emerge from single-streaming initial conditions, it can be observed that the foldings of the particle mesh always result in an odd number of incoming streams at any Eulerian coordinate. The single-streaming volume has not undergone shell-crossing and thus corresponds to the uncollapsed cosmic voids. The emerging cosmic web, on the other hand, is an inherently multistreaming structure, as is illustrated by the number of streams field in \cref{fig:sim_256_triple_plot}. Throughout this article, we will further elaborate on the central role of multistreaming to the characteristic geometry of the cosmic web elements, and the cosmic walls in particular.

\begin{figure*}
    \includegraphics[width=\textwidth]{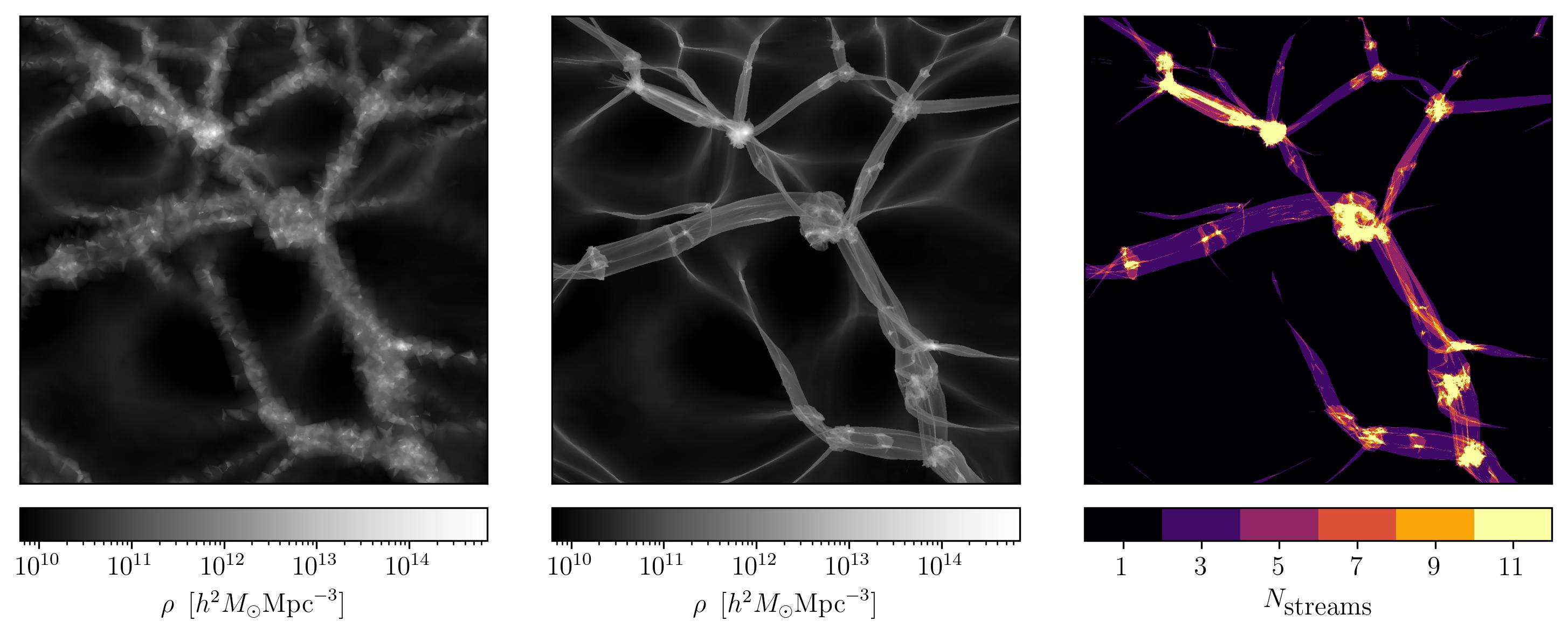}
    \caption{Comparison of the density estimates for zoomed region of the simulation of \cref{fig:sim_256_triple_plot} with the ordinary DTFE method (left panel) and the PS-DTFE method (middle panel). The corresponding number of streams evaluated from the PS-DTFE method is displayed in the right panel.}
    \label{fig:DTFE_comparison}
\end{figure*}

Numerical investigations of structure formation do not only depend on high-resolution $N$-body simulations, but also on accurate methods to infer the smooth, continuous matter fields traced by the discretised simulation particles. A commonly used method for the estimation of continuous density and velocity fields from the particle positions and velocities is the \textit{Delaunay tessellation field estimator} (DTFE) \cite{BernardeauWeygaert1996, SchaapWeygaert2000, Schaap2007}, which uses the Delaunay tessellations of the Eulerian particle coordinates and velocities to obtain smooth and adaptive field estimates. However, the standard DTFE method is ignorant to the multistreaming nature of the cosmic web imprinted in the phase space of Lagrangian and Eulerian coordinates. In this work, we therefore use the recently proposed \textit{phase-space Delaunay tessellation field estimator} (PS-DTFE) \cite{Feldbrugge2024} method, which uses a tessellation of the Lagrangian coordinates and subsequent evolution into Eulerian space to smoothly and adaptively evaluate \cref{eq:Eulerian_density}. \textit{Phase space} here refers to the combination of Lagrangian and Eulerian information through the respective particle coordinates. The PS-DTFE method extends previously proposed phase-space field estimators \cite{Shandarin2011, ShandarinSalmanHeitmann2012, AbelHahnKaehler2012, Hahn+2015} and preserves the adaptive and mass-conserving nature of the DTFE method while providing a significantly better density estimate in the multistream regions. Fig. \ref{fig:DTFE_comparison} exemplarily shows how the cosmic elements are resolved much more finely by the PS-DTFE than the DTFE field estimate, as the latter is blurry in the multistream regions that make up the overdense cosmic web elements. A stable and efficient implementation of the PS-DTFE method is provided by the publicly available \verb|julia| package \verb|PhaseSpaceDTFE.jl|\cprotect\footnote{Hosted at 
\verb|github.com/jfeldbrugge/PhaseSpaceDTFE.jl| and available for installation from the \verb|julia| package manager.} \cite{FeldbruggeHertzsch2025}, which we use for all field evaluations throughout this article.

\subsection{Caustic skeleton theory \& cosmological collapse}
\label{subsec:theory-caustic_skeleton}

We now come back to the Eulerian density, \cref{eq:Eulerian_density}, and discuss a second crucial concept encapsulated in the equation. We observe that the density diverges whenever one of terms $1 + \mu_i(\bm{q},t)$ vanishes. Geometrically, these \textit{shell-crossing} events occur when the particle sheet folds over and the qualitative nature of the flow changes into a higher-streaming configuration. This can be clearly observed it
\Cref{fig:sim_256_triple_plot}, where the overdensities in the upper panel correspond to the particle mesh foldings in the middle panel. In mathematical terms, the shell-crossing events are known as \textit{catastrophes} or \textit{caustics}.

\subsubsection{The role of caustics in large-scale structure formation}
\label{subsubsec:theory-caustics}

\begin{figure*}
    \centering
    \begin{subfigure}[b]{0.3\textwidth}
        \includegraphics[width=\textwidth]{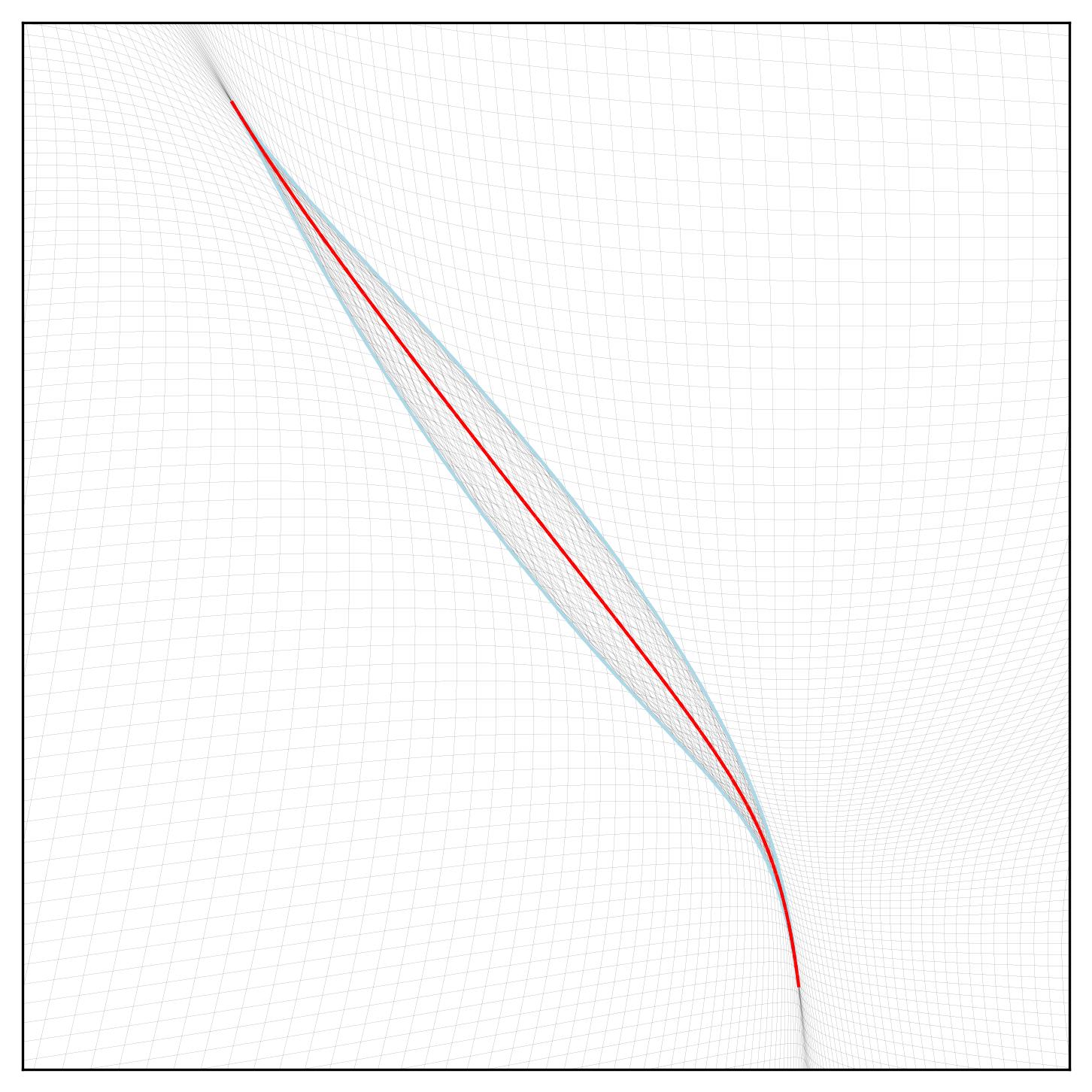}
        \caption{$A_3$ cusp caustic}
    \end{subfigure}
    \hfill
    \begin{subfigure}[b]{0.3\textwidth}
        \includegraphics[width=\textwidth]{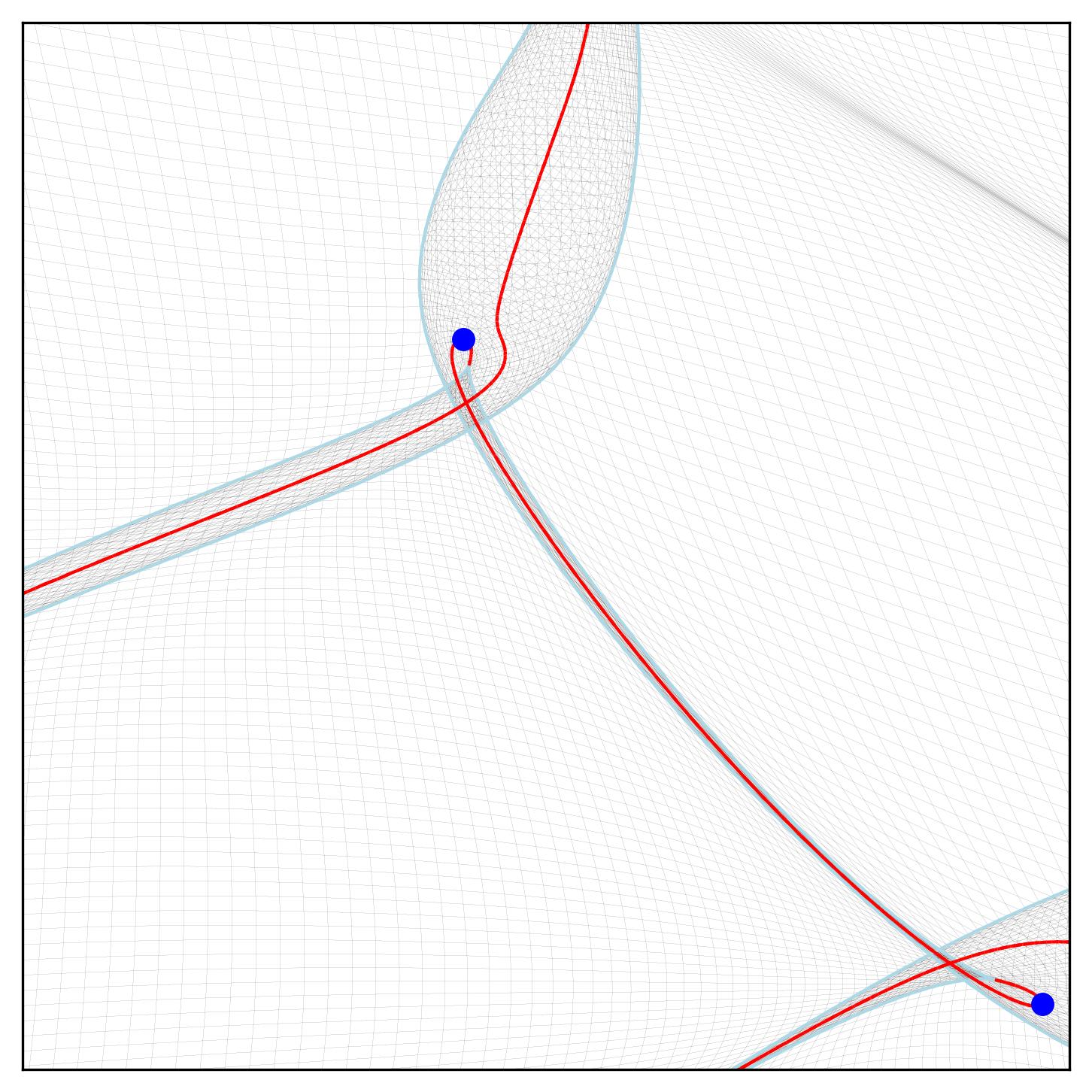}
        \caption{$A_4$ swallowtail caustic}
    \end{subfigure}
    \hfill
    \begin{subfigure}[b]{0.3\textwidth}
        \includegraphics[width=\textwidth]{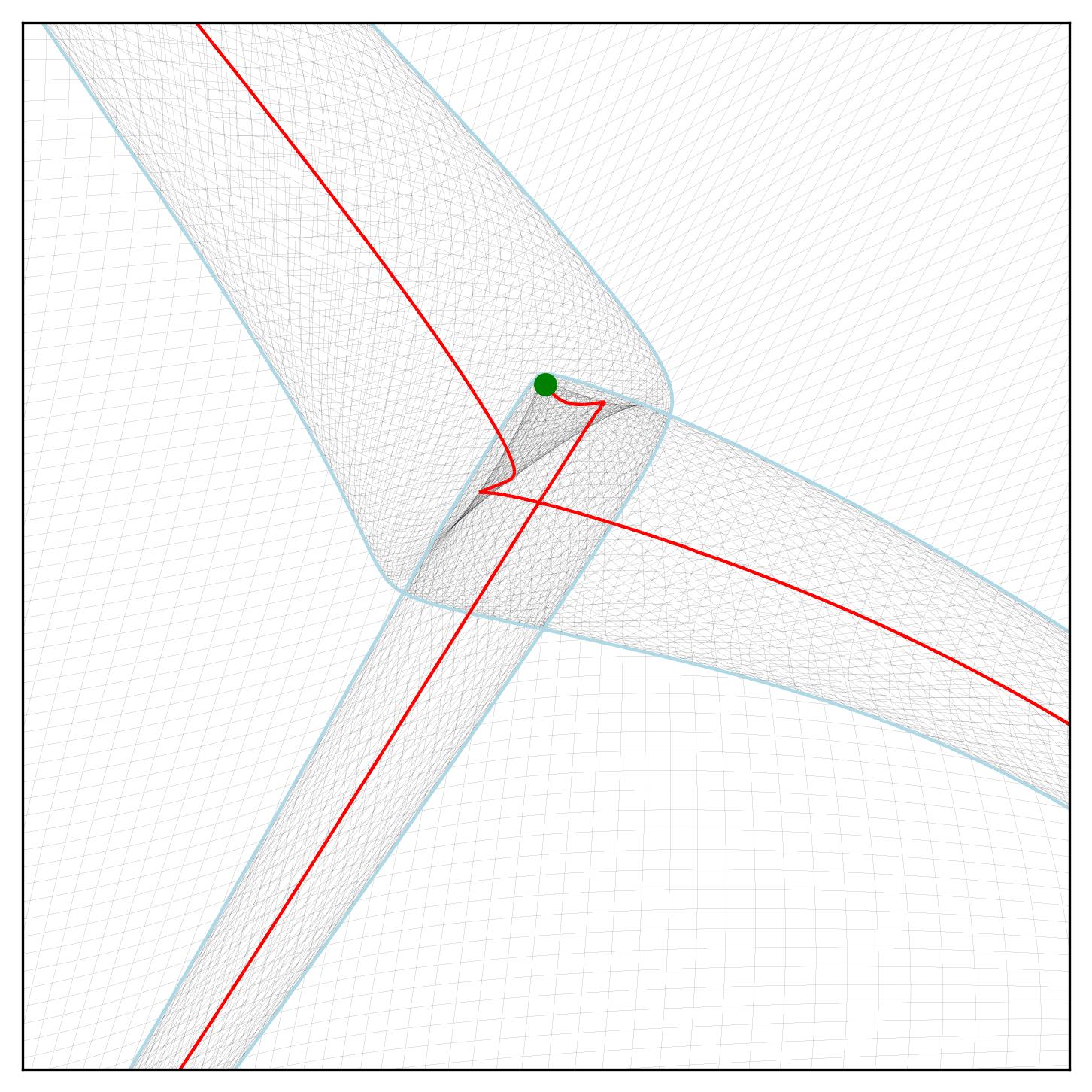}
        \caption{$D_4$ umbilic caustic}
    \end{subfigure}
    \cprotect\caption{Caustics making up the structural elements of the cosmic web, illustrated using a 2D mesh simulation of cosmic structure formation. The left panel shows a multistream region (Zel'dovich pancake) bisected by the cusp line, corresponding to a wall in the 3D cosmic web. The middle panel shows two swallowtail caustics emerging from a cusp sheet folding onto itself. The right panel shows the umbilic caustic, corresponding to a filamentary junction of three incoming walls. For an animation of the Zel'dovich pancake formation in a 2D mesh simulation, see \verb|benhertzsch.github.io/papers/2025_Cosmic_walls/multistream.mp4|.}
    \label{fig:caustics}
\end{figure*}

\textit{Caustic skeleton theory} \cite{ArnoldShandarinZeldovich1982, Arnold1986, Feldbrugge+2018, FeldbruggeWeygaert2023} classifies the shell-crossing events using catastrophe theory \cite{Thom1972, Arnold1992, ArnoldGuseinZadeVarchenko2012} and associates the different caustics to the structural elements of the cosmic web (multistream regions, walls, filaments and clusters) by differentiating their unique formation processes. The description of the cosmic web in terms of the caustics of the underlying dark matter flow was first brought forward in Arnol'd et al. (1982) \cite{ArnoldShandarinZeldovich1982} for structure formation in a two-dimensional universe. Feldbrugge et al. (2018) \cite{Feldbrugge+2018} extended this study to the three-dimensional case and provided a mathematically rigorous and parameter-free classification of the collapse processes of the dark matter sheet occurring in our Universe.

The key result of  \cite{Feldbrugge+2018} is an exhaustive list of seven \textit{caustic conditions} which identify the cosmic web elements through non-linear algebraic conditions of the eigenvalue and eigenvector fields  $\mu_i(\bm{q},t)$ and $\bm{v}_i(\bm{q},t)$ of the deformation tensor $\nabla_{\bm{q}} \bm{s}_t$. The key role of the eigenvalue fields has been recognised as early as \cite{Zeldovich1970}; since then, the community has traditionally assumed a naive association of the sheet-, line- and point-like cosmic web elements to $1, 2$ and $3$ terms respectively vanishing in the denominator of \cref{eq:Eulerian_density}. In particular, by relating the deformation tensor $\nabla_{\bm{q}} \bm{s}_t$ to the tidal shear tensor $\mathcal{H} \phi$ of the primordial potential perturbation $\phi$, \cite{Hahn+2007} first proposed the following explicit identification of the emerging structure in terms of the eigenvalues $\lambda_i$ of $\mathcal{H} \phi$:
\begin{equation*}
    \begin{split}
        \mathrm{cluster} \quad & \lambda_1 > \lambda_2 > \lambda_3 >  0  \\
        \mathrm{filament} \quad & \lambda_1 > \lambda_2 > 0 > \lambda_3 \\
        \mathrm{wall} \quad & \lambda_1 > 0 > \lambda_2 > \lambda_3 \\
        \mathrm{void} \quad &  0 > \lambda_1 > \lambda_2 > \lambda_3
    \end{split}
\end{equation*}
This classification (known as the \textit{T-web formalism}) and variations thereof have since been applied in numerous analyses of the cosmic web \cite{AragonCalvo+2007, Forero-Romero+2009, Hoffman+2012, Libeskind+2017, AycoberryBarthelemyCodis2024}. However, structure formation research has so far been largely oblivious of the role of the corresponding eigenvector fields, which are equally central to the geometric analysis of the dark matter flow. In \cref{app:caustic_skeleton-general}, we summarise the caustic conditions along with their geometric correspondence in the cosmic web, as described in detail in \cite{Feldbrugge+2018}. For more detailed visualisations of the associated structures in a two- and three-dimensional universe, we refer to \cite{Hidding+2013, FeldbruggeWeygaert2023} and \cite{Feldbrugge+2018} respectively.

\begin{figure}
    \centering
    \includegraphics[width=0.8\columnwidth]{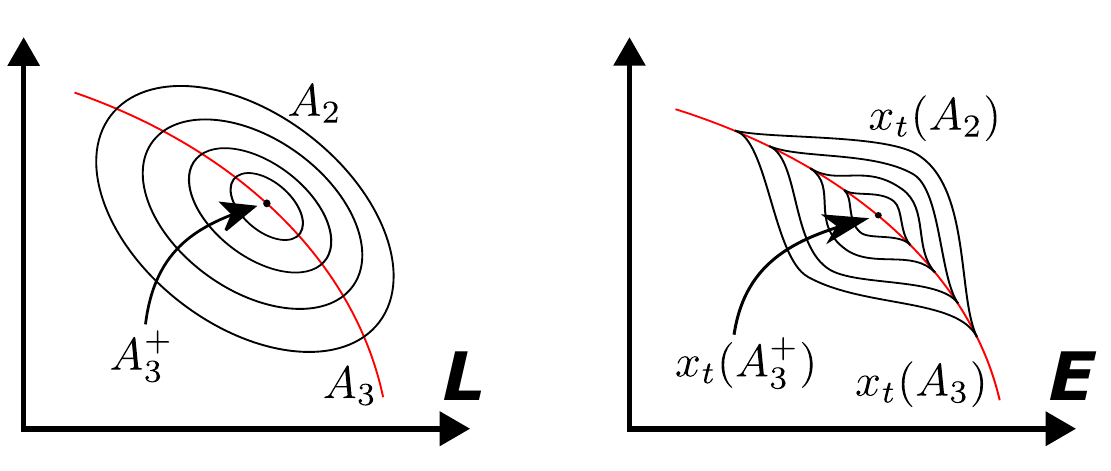}
    \cprotect\caption{Structure the emerging cusp sheet in Lagrangian (left) and Eulerian (right) space. In Lagrangian space, the multistream region is an ellipsoidal that concentrically emanates from the $A_3^+$ point. In Eulerian space, the edge points (ring in 3D) exhibit the characteristic cusp geometry and trace out the cusp sheet (red) over time.  For three-dimensional view of the Zel'dovich pancake, see \verb|benhertzsch.github.io/papers/2025_Cosmic_walls/pancake.mp4|.}
    \label{fig:cusp_sketch}
\end{figure}

The caustic skeleton model is not a theory of isolated geometric objects, but rather a formalism describing the formation history of the entirety of the emerging, evolving and merging cosmic web. The caustics discussed \cref{app:caustic_skeleton-general} do not exist in isolation, but are intricately connected to each other through the various collapse phenomena. Firstly, it can be observed that any multistream region is necessarily bisected by a cusp sheet. A multistream region emerges in a local maximum of the first eigenvalue field, $\nabla \lambda_1(\bm{q}, t) = \bm{0}$. Being a critical point of $\lambda_1(\bm{q}, t) $, this so-called $A_3^+$ point \cite{Feldbrugge+2018} identically fulfills the cusp condition, \cref{eq:A3}, and thus forms the origin of the emerging cusp sheet (see \cref{fig:cusp_sketch}). Similarly, the merger of two pancakes occurs in the $A_3^-$ saddle point \cite{Feldbrugge+2018} of $\lambda_1(\bm{q}, t)$. This point again fulfills \cref{eq:A3} and therefore constitutes the merger of through two cusp sheets. \Cref{fig:eigenvalue_visualisation} visualises the first eigenvalue field along with the $A_3^{\pm}$ points for caustics of the Zel'dovich approximation, as will be discussed in the the following section.

The two filament families are then related to the emerging network of cusp walls. Firstly, the swallowtail filament occurs when the cusp sheet folds onto itself and thus appears either as an overdense ``kink'' within the wall sheet or an elongated structure at its boundary. Similarly, the umbilic filament unfolds into three outgoing cusp sheets, and so exists only in connection with three cosmic walls. The swallowtail and butterfly clusters then form when the respective filaments fold onto themselves, and thus cannot exist in isolation from the extended cosmic web elements. We will investigate the two filament families in more detail in the follow-up article to the present study.

From the preceding discussion, it is clear that the entirety of the cosmic web needs to be considered as a network of merging caustics, the lowest of which is given by the cusp sheet tracing out a three-stream region corresponding to a cosmic wall. The elements of the cosmic web form in the same order as in the Zel'dovich theory \cite{Zeldovich1970}, namely that the formation of walls precedes that of filaments, which in turn precedes that of clusters. Note that this in contrast to the seminal cosmic web study by Bond et al. (1996) \cite{BondKofmanPogosyan1996}, which investigated the weaving of the overdense network from the peak-patch picture of the ZA-mapped density field. Here, the authors concluded that the first objects to form are the highly dense clusters, followed by filaments and walls forming latest and playing only a minor role in the current-time web. However, the more sophisticated, dynamical picture of the morphogenesis of singularities reveals that this conclusion is incorrect on geometric grounds: Highly-streaming filaments and clusters can only be formed out of lower-streaming configurations. The three-streaming cusp sheets therefore take a central stage in the build-up of the entirety of the overdense filaments and clusters. In this context, the caustic skeleton formalism also significantly enhances the traditional notion of the connectivity of the cosmic web from the morphology of the primordial density perturbation $\delta$ \cite{CodisPogosyanPichon2018}. In particular, \cite{CodisPogosyanPichon2018} defines connectivity through the multiplicity of filaments and clusters while being fully agnostic to the existence of walls within the cosmic web. Clearly, this does not reflect the crucial role that the walls take in embedding the filaments, superdense clusters and eventually the observable galaxies.

The caustic skeleton formalism unambiguously reveals that the cosmic walls, consisting of  sheet-like three-stream regions,  are indeed real structures in the physical cosmic web. This motivates our study of the phenomenology of cusp sheets in our Universe, which we shall develop for the remainder of this article.

\subsubsection{Evaluating caustics in the Zel'dovich approximation}
\label{subsubsec:theory-caustics_ZA}

\begin{table}
	\centering
	\begin{tabular}{lccr}
		\hline
		name & symbol & caustic condition & cosmic web element\\
		\hline
		fold & $A_2$ & $\lambda_1 = b_c^{-1}$ & multistream region\\
        cusp & $A_3$ & $\bm{v}_1 \cdot \nabla \lambda_1 = 0$ & wall\\
        swallowtail & $A_4$ & $\bm{v}_1 \cdot \nabla \left( \bm{v}_1 \cdot \nabla \lambda_1 \right) = 0$ &  filament\\
        butterfly & $A_5$ & $ \bm{v}_1 \cdot \nabla\left(\bm{v}_1 \bm{\cdot} \nabla \left( \bm{v}_1 \cdot \nabla \lambda_1 \right) \right) = 0$ &  cluster \\
        umbilic & $D_4$ & $\lambda_1 = \lambda_2 = b_c^{-1}$ & filament\\
        parabolic umbilic & $D_5$ & $ \bm{v}_1 \cdot \nabla \left(\lambda_1 - \lambda_2\right) = \bm{v}_2 \cdot \nabla \left(\lambda_1 - \lambda_2\right)  = 0$ & cluster \\
		\hline
	\end{tabular}
    \caption{Summary of the caustic conditions in the Zel'dovich approximation with the correspondence to the structural elements of the cosmic web in the three-dimensional Universe.}
    \label{tab:caustic_skeleton}
\end{table}

Up to this point, the caustic skeleton been presented in mathematical generality. We made no assumptions about the deformation tensor $\nabla_{\bm{q}}\bm{s}_t(\bm{q})$ and the associated eigenvalue and eigenvector fields. Written in their general form, the caustic conditions enable the mathematically rigorous and unambiguous phase-space identification of the emerging large-scale structure from the eigenfields.

However, physical studies of the cosmic web often aim at relating the late-time structures to the primordial conditions out of which they emerge. This is achieved through analytical models such as Lagrangian Perturbation Theory (LPT) \cite{Bouchet+1995, Catelan1995, RampfHahn2021}, which approximate the non-perturbative $N$-body dynamics up to the mildly non-linear evolution. In relation to the caustic skeleton, the seminal \textit{Zel'dovich approximation} (ZA) \cite{Zeldovich1970} allows one to translate the caustic conditions into algebraic conditions on the configuration of the primordial potential perturbation $\phi(\bm{q})$, thus enabling the identification of the approximate Lagrangian-space caustics only from the primordial fields, without the need for computationally expensive $N$-body simulations.

The ZA is the first-order LPT solution, and thus represents the simplest perturbative Lagrangian model of large-scale structure formation. In the ZA, the displacement field is given by the gradient map
\begin{equation}
    \bm{x}(\bm{q},t) = \bm{q} - b(t) \nabla \Psi(\bm{q}) \,.
    \label{eq:ZA}
\end{equation}
Here, the \textit{primordial displacement potential} $\Psi(\bm{q})$ is related to the primordial potential perturbation $\phi$ by the proportionality \cite{Hidding+2013, FeldbruggeWeygaert2024}
\begin{equation}
    \Psi(\bm{q}) = \frac{2}{3 H_0^2 \Omega_m} \phi(\bm{q}) = \frac{2}{3 b(t) a(t)^2 H_0^2 \Omega_m} \phi_{\textrm{lin}}(\bm{q}) \,,
    \label{eq:displacement_potential}
\end{equation}
where in the second equality we used the linearly extrapolated potential perturbation $\phi_{\textrm{lin}}$ at time $t$.

The temporal dependence is absorbed by the linear growing mode  $b(t)$, which gives the solution for the linear growth of perturbations through the defining equation \cite{Peebles1994}
\begin{equation}
    \ddot{b}(t) + 2 H(t) \dot{b}(t) - 4 \pi G \rho_0 b(t)=0 \,,
\end{equation}
with the boundary conditions $b(0)=0$ and $b(t_0)=1$, the Hubble function $H(t)$,  the primordial mean density $\rho_0$ and the gravitational constant  $G$. Throughout this article, we will not work in physical time $t$, but express the temporal dependence of the caustic skeleton directly in terms of the growing mode $b(t)$.

In the ZA, the deformation tensor $M(\bm{q}, t)= \nabla_{\bm{q}}\bm{s}_t$ becomes
\begin{equation}
    M(\bm{q}, t) = - b(t) \mathcal{H}\Psi(\bm{q}) 
\end{equation}
and is thus proportional to the scaled tidal shear tensor $\mathcal{H} \Psi(\bm{q})$ of the primordial displacement potential $\Psi$. The relevant fields now become the eigenvalue and eigenvector fields $\lambda_i(\bm{q})$ and $\bm{v}_i(\bm{q})$ of $\mathcal{H} \Psi(\bm{q})$, which are functions of space only. In the ZA, the caustic conditions of \cref{subsec:theory-caustic_skeleton} therefore simplify to a set of algebraic conditions \cite{Feldbrugge+2018} of the space-dependent fields $\lambda_i(\bm{q})$ and $\bm{v}_i(\bm{q})$ and the growing mode $b(t)$, which we list in \cref{app:caustic_skeleton-ZA} and summarise in \cref{tab:caustic_skeleton}.

\begin{figure*}
    \centering
    \begin{subfigure}[t]{0.32\textwidth}
        \includegraphics[width=\textwidth]{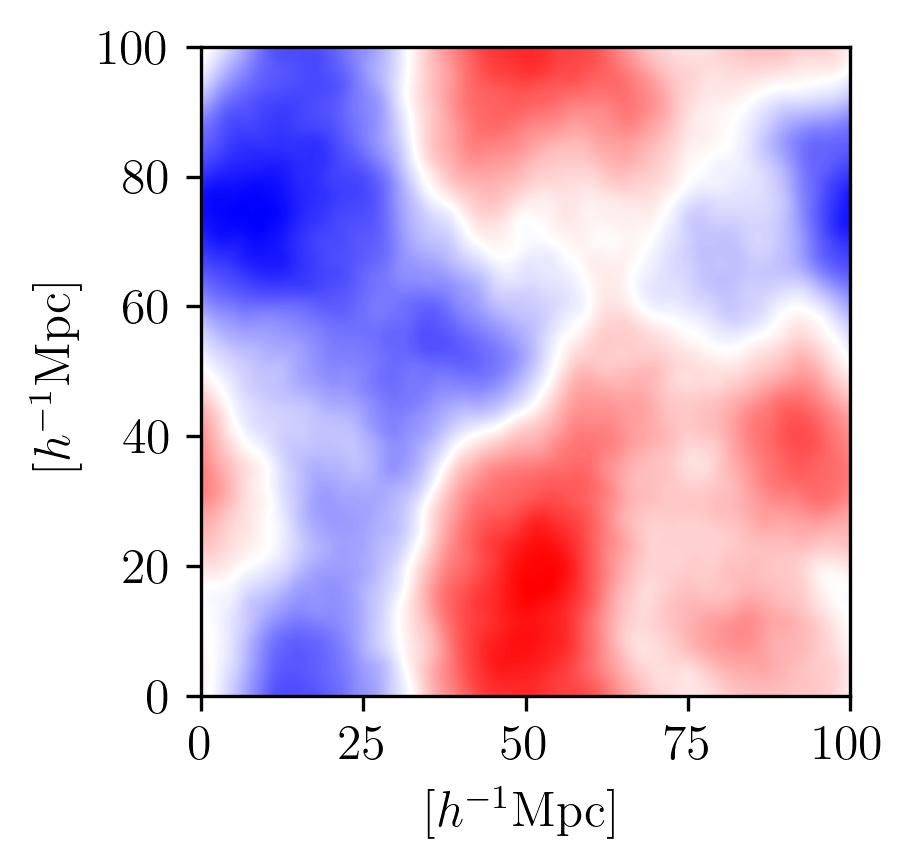}
        \caption{slice through $\Psi$}
    \end{subfigure}
    \hfill
    \begin{subfigure}[t]{0.32\textwidth}
        \includegraphics[width=\textwidth]{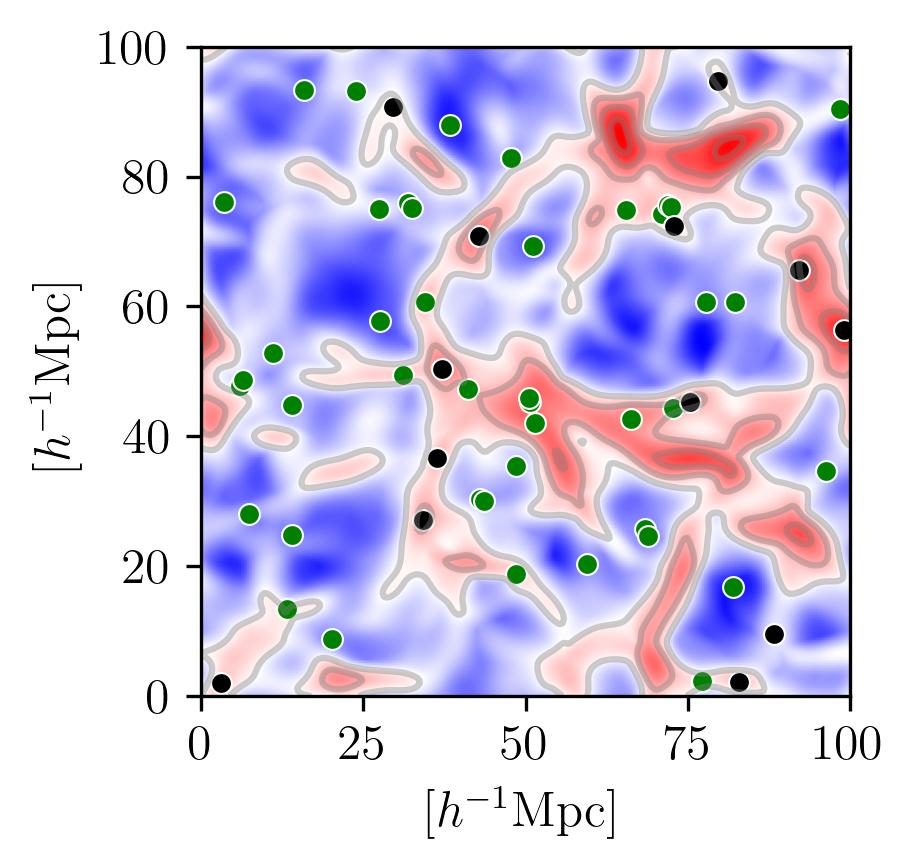}
        \caption{slice through $\lambda_1$}
    \end{subfigure}
    \hfill
    \begin{subfigure}[t]{0.31\textwidth}
        \includegraphics[width=\textwidth]{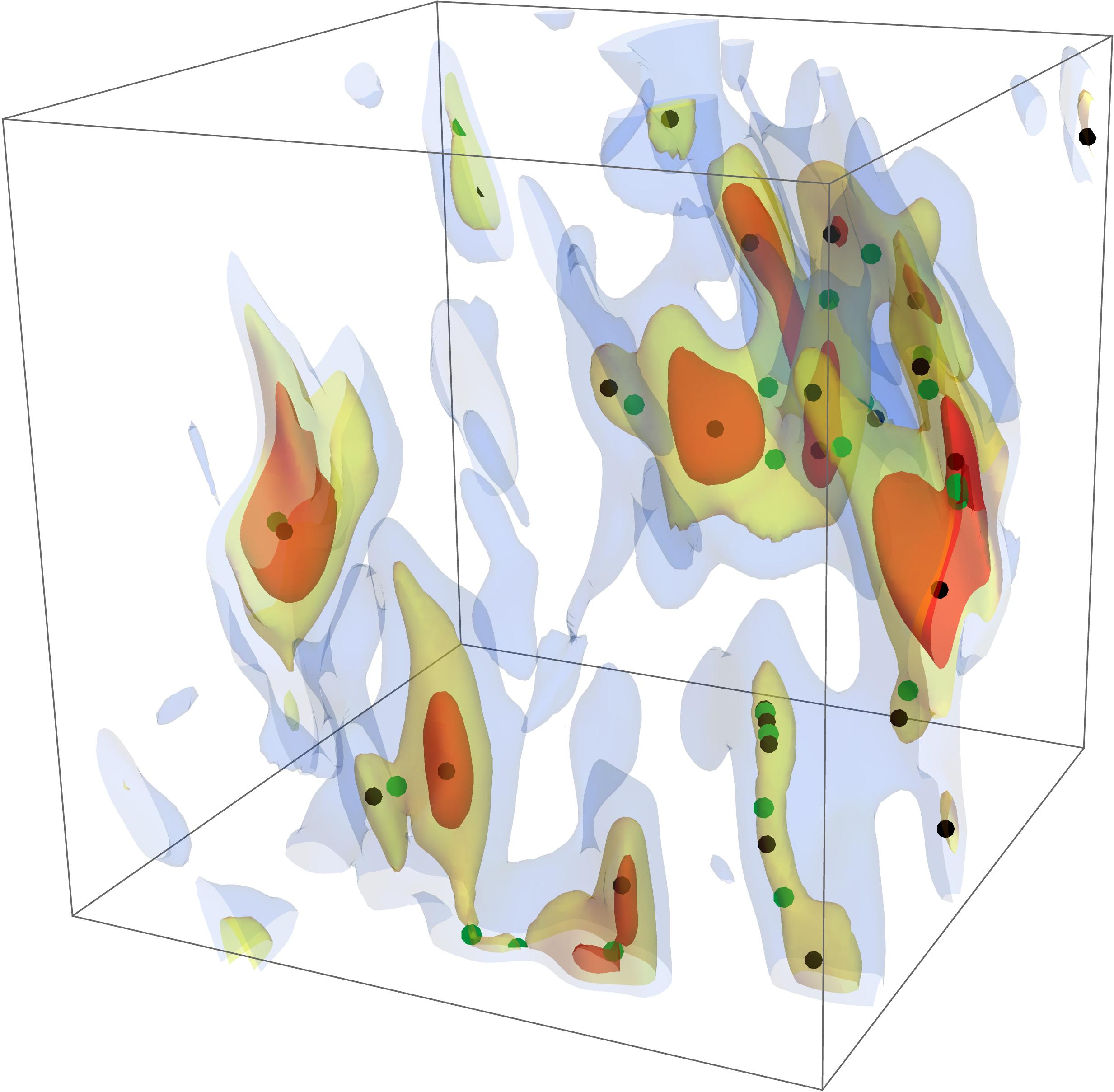}
        \caption{isocontours of $\lambda_1$}
    \end{subfigure}
    \caption{Primordial fields making up the embryonic cosmic web. The left panel shows a slice through a random realisation of the primordial displacement potential $\Psi$, smoothed at $\sigma = 1.0 \,h^{-1}\textrm{Mpc}$ (see \cref{subsec:theory-scale_space}). The middle panel shows the corresponding slice through first eigenvalue field $\lambda_1$. The right panel shows a 3D zoomed view of the isocontours of the first eigenvalue field, demonstrating the growth of multistream regions by the fold condition, \cref{eq:A2}. The black and green dots are the local maxima and saddle points of $\lambda_1$, in which the multistream regions emerge and merge respectively.}
    \label{fig:eigenvalue_visualisation}
\end{figure*}

 It is well known that the ZA fails when shell-crossing occurs, as secondary collapse is ignored and particles keep following to their linear trajectories, making shell-crossed structures grow to infinity. Despite this failure, the Lagrangian-space caustics identified from the ZA are a suitable approximation to the Lagrangian-space caustics from first collapse appearing in the non-perturbative $N$-body dynamics. The reasoning here is that the non-linear $N$-body caustics occur only when the linear theory predicts a configuration sufficiently close to shell-crossing. In other words, the ZA gives the dominant contribution to the large-scale collapse, with the non-linear dynamics affecting only negligible deviations of the linearly predicted collapse time and Lagrangian coordinates. Phase mixing from secondary collapse does not significantly affect the spinal outline of the cosmic web, but merely manifests itself in higher overdensities of the structures traced out by the caustics calculated from the linear theory.
 
 Having identified the ZA caustics in Lagrangian space, these may be mapped to Eulerian space by any suitable method to trace the cosmic web in its final configuration. In this work, we choose to map the ZA caustics forward with the non-perturbative $N$-body displacement field $\bm{s}_t$. The upper panel of \Cref{fig:sim_256_caustics} illustrates how aptly the evolved ZA caustics trace out the overdense structure of the cosmic web, with the filaments and walls being in stunning agreement with the caustics characterising the various particle mesh foldings (see lower panel).

\subsection{Cosmic web in scale space}
\label{subsec:theory-scale_space}

\begin{figure*}[htp]            \includegraphics[width=0.95\textwidth]{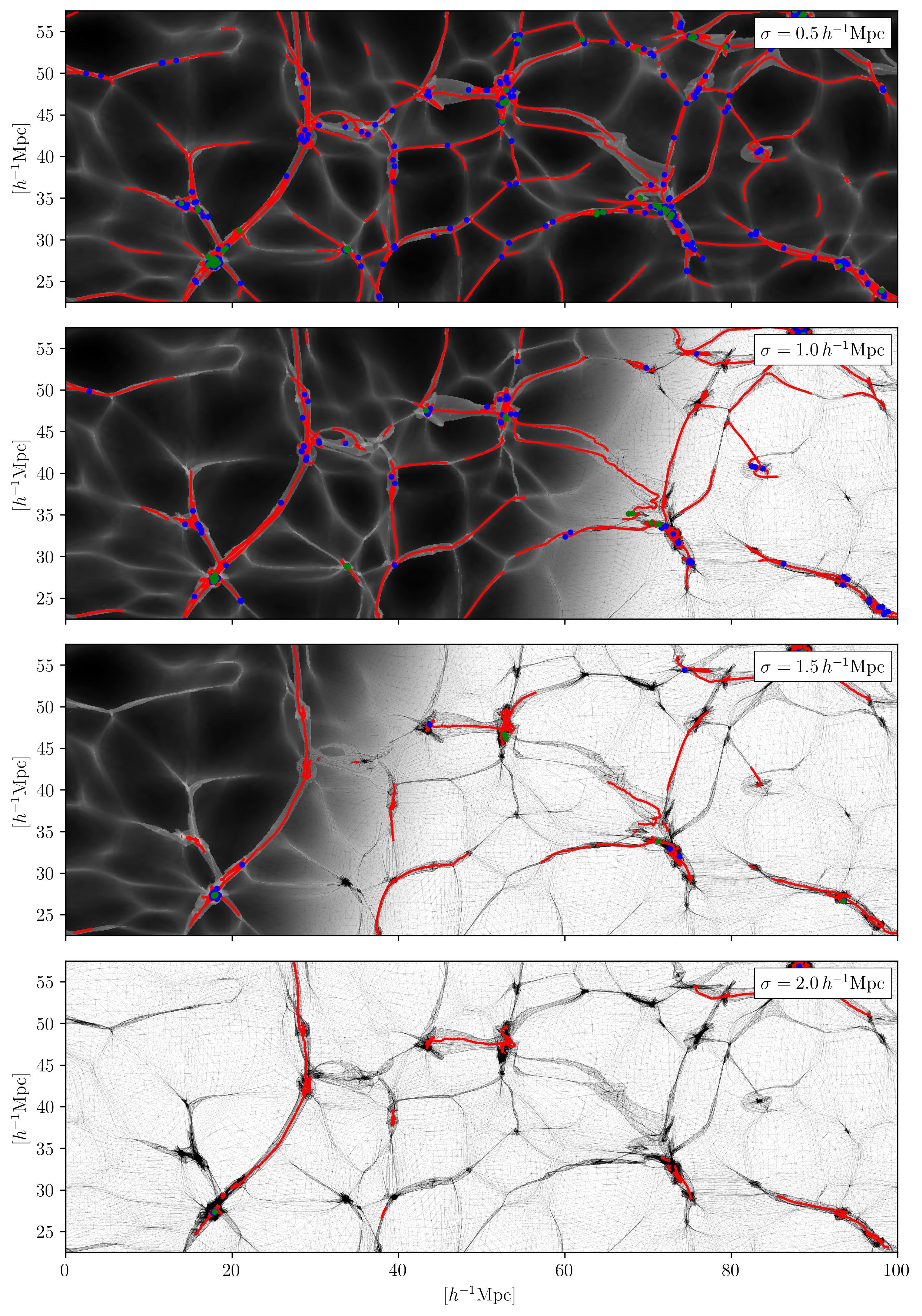}
    \caption{Slices through the density field smoothly blending over into the particle mesh of a $256^3$-particles $N$-body simulation. Superimposed are slices of the caustic skeleton evaluated at different smoothing scales $\sigma$, with the the red lines corresponding to slices through the cusp walls and the blue and green dots being slices through the swallowtail and umbilic filaments respectively.}
    \label{fig:sim_256_caustics}
\end{figure*}

The caustic conditions discussed in the previous sections provide a parameter-free description of the emerging cosmological structure from the primordial potential perturbation $\phi$. Observations of the cosmic microwave background \cite{Planck2018} strongly suggest that $\phi$ is aptly described by a random field with near-scale fluctuation spectrum. Due to its near-scale invariant nature, the primordial initial conditions seed collapse phenomena at different length scales. This leads to the concept of \textit{scale space}, motivating the study of the cosmic web not only over time, but also over physical length scales.

The relevance of scales for the hierarchical nature of structure formation has been appreciated early on, notably in the seminal works of Zel'dovich (1970) \cite{Zeldovich1970}, Doroshkevich (1970) \cite{Doroshkevich1970} and Doroshkevich et al. (1978) \cite{DoroshkevichShandarinSaar1978}. The famous Press-Schechter formalism \cite{PressSchechter1974} pioneered a quantitative treatment of the mass-dependence of cosmological density clustering from excursion sets of the Eulerian density field. This was further solidified in Bond et al. (1991) \cite{Bond+1991} and extended to anistotropic collapse in Sheth et al. (2001) \cite{ShethMoTormen2001}. In the groundbreaking work of Bond et al. (1996) \cite{BondKofmanPogosyan1996}, these notions were formalised to the emergence of a multiscale cosmic web over a range of cosmic time scales, with the relation to peak patches in the Lagrangian space being first investigated in \cite{BondMyers1996}. Ever since these studies, the scale-space treatment of the hierarchical cosmic web has become the established standard in the community, with numerous modern studies investigating the multiscale properties of cosmic web and its influence on the embedded galaxies \cite{AragonCalvo+2010, Subhonenko+2011, Jaber+2023, AragonCalvo2024}. From the Lagrangian perspective, the scale-space nature of the primordial fields has been studied primarily with regard to the critical points and the Morse-Smale complex \cite{Cadiou+2020}. In terms of cosmic web identification, the multiscale morphology filter (MMF) \verb|NEXUS(+)| \cite{AragonCalvo+2007, Cautun+2012} considers the scale-space morphology of e.g. the Eulerian density field to identify the web-like structure from the classification over a range of length scales. This also highlights the practical meaning of scale dependence for cosmological surveys, as these are not equally sensitive to all length scales, halo masses and galaxy luminosities. Instead, sureys typically aim at identifying the most dominant (largest-scale) features of the cosmic web. A scale-space treatment of the observable data effectively filters out the  small-scale structure to retrieve the large-scale web. This is of particular relevance in the context of scale-dependent halo bias \cite{Paranjape+2013} over the dark matter density field, though we leave these issues aside for our present study.

In this work, we are interested in the scale-space cosmic web as identified through the caustic skeleton formalism from the morphology of the primordial fields. To study the structure emerging at a physical length scale $\sigma$, it is customary to smooth the initial conditions at the scale $\sigma$ to obtain smoothed initial conditions. Starting from the primordial displacement potential $\Psi$, this is done through a convolution with the window function $W_{\sigma}(\bm{q})$,
\begin{equation}
    \Psi_{\sigma}(\bm{q}) = \int \d \bm{q}^{\prime} \Psi(\bm{q} - \bm{q}^{\prime} ) W_{\sigma}(\bm{q}^{\prime} ) \,.
    \label{eq:smoothing_real_space}
\end{equation}
One common choice for the window function is the spherical top-hat function \cite{Peebles1994}, which is uniform inside a sphere of radius $\sigma$ and vanishing elsewhere,
\begin{equation}
    W_{\sigma}(\bm{q}) = \frac{3}{4 \pi \sigma^3} \Theta(r - \|\bm{q}\|) \,.
    \label{eq:window_function_top_hat}
\end{equation}
In this work, we choose to work instead with the Gaussian kernel  \cite{Peebles1994} smoothed on scale $\sigma$,
\begin{equation}
    W_{\sigma}(\bm{q}) = \frac{1}{2 \pi \sigma^2} e^{-\frac{\bm{q}^2}{2 \sigma^2}} \,.
    \label{eq:window_function}
\end{equation}
The Gaussian kernel is an alternative popular choice, and more suitable in this work for the numerical evaluation of the integrals in \cref{sec:recipe}. We do not expect the results of this article to be qualitatively different from those that would be obtained with the top-hat function.

\Cref{eq:smoothing_real_space} is a non-local operation that averages out the field's local information over that of nearby points. The physical reasoning for this approach is that while the large-scale structure is expected to influence the embedded small-scale structure, the small-scale dynamics are assumed not to significantly back-react on the large-scale structure. This assumption was checked in an approximate calculation in \cite{Peebles1980} and systematically verified using $N$-body simulations in Little et al. (1991) \cite{LittleWeinbergPark1991}\footnote{While not commonly appreciated as such, Little et al. (1991) \cite{LittleWeinbergPark1991} actually lays the numerical foundation for any investigation of structure formation from $N$-body simulations. By demonstrating that the backreaction of small scales on the large scales is negligible, the authors demonstrated that the cosmic web may be effectively simulated in finite-resolution simulations that are blind to the below-grid-resolution noise. This is in opposition to e.g. turbulence simulations. These results also demonstrate that the bottom-up model of structure formation cannot explain the cosmic web if the effective treatment is to hold. Instead, the well-established bottom-up clustering of haloes has to the married to the top-down model of the large-scale web assembly, together forming the modern model of hierarchical large-scale structure formation.}. The negligible backreaction of the small scales on the large scales allows for an effective treatment of the cosmic web at different length scales: The smoothed potential $\Psi_{\sigma}(\bm{q})$ contains the information about the structure forming at a length scale $\sigma$, and the entirety of the cosmic web is imprinted in the initial conditions smoothed over a range of cosmological length scales.

The caustic skeleton is an intrinsically multiscale formalism due to the near-scale invariant nature of the primordial fluctuation spectrum \cite{Feldbrugge+2018, FeldbruggeYanWeygaert2023, FeldbruggeWeygaert2023}. The network of caustics may be understood as a fractal-like structure that mirrors the near self-similarity of the cosmic web \cite{Einasto+2020, Einasto2025} over a range of cosmological scales. The cosmic web elements at a length scale $\sigma$ are identified by evaluating the caustic conditions on the smoothed initial conditions $\Psi_{\sigma}(\bm{q})$. \Cref{fig:sim_256_caustics} clearly illustrates how the large-scale features of the cosmic web are captured by the caustics evaluated at a large smoothing scale ($\sigma = 2.0 \, h^{-1}\textrm{Mpc}$). Moreover, the upper panels reveal how the caustic skeleton evaluated on smaller scales captures increasingly finer details of the cosmic web. In the limit of vanishing smoothing, $\sigma \rightarrow 0$, the caustics recover sensitivity to all shell-crossing that can be numerically resolved on the grid of the simulation. More phsyically, the caustic formalism in principle may be applied down to the scale at which the collisionless fluid dark matter assumption breaks down. This is a truly remarkable result: With the calculation of the caustic skeleton from the ZA, one could have expected that the caustics only offer a useful description at the largest scales. However, we clearly find that the caustic skeleton offers us more than we bargained for, in that the entirety of the cosmic web is identified from the caustic conditions evaluated on the primordial field $\Psi_{\sigma}$ at small $\sigma$. In forthcoming work \cite{FeldbruggeHertzschWeygaert2025}, we will further illustrate this aspect of the caustic skeleton formalism, as well as time-dependence of the caustic network making up the connected cosmic web.

Beyond the identification of the cosmic web in scale space, the caustic skeleton also offers an unprecedented analytical view on the formation times of the different cosmic web environments on varying length scales from the statistics of the primordial fields \cite{FeldbruggeYanWeygaert2023, FeldbruggeWeygaert2024}. Within the present article, in \cref{subsec:sims-formation_time}, we will place particular emphasis on the formation times of the Zel'dovich pancakes on different length scales. Doing so, we provide a new quantitative understanding of the hierarchical formation of the cosmic walls. Note at this point that the Zel'dovich theory timeline \cite{Zeldovich1970} of first wall, then filament, then cluster formation is to be understood in scale-space. We will further elaborate on this timeline in \cref{subsec:sims-formation_time} and in the follow-up paper on cosmic filament formation. The constrained simulations of \cref{subsec:sims-fields} will reveal how the scale-space caustic constraints are reflected in the dark matter fields arising from the simulations of cosmic wall formation. Finally, \cref{sec:haloes} will directly demonstrate the observational impact of the caustic network, as we show, for the first time, how the scale-space caustics are traced by haloes of different mass scales. 


\section{A new condition for cosmic walls}
\label{sec:A3_centre_constraint}

The preceding section discussed in detail how the cosmic web emerges as a multistreaming structure under gravitational attraction and how its spinal outline is traced by the caustics of the dark matter flow occurring on different length scales. In this article, we are interested in the properties of cosmic walls, and want to use these notions to set up specialised simulations of their formation process. To this end, we extend the caustic skeleton formalism and derive a novel constraint that allows for an unprecedented analytical treatment of cosmic walls and their implementation into physically realistic constrained simulations.

In principle, one can construct initial conditions that obey the caustic conditions of \cref{tab:caustic_skeleton}. The relevant constraint for the formation of a cosmic wall from a cusp sheet emerging at shell-crossing time $b_c$ is
\begin{equation}
    \lambda_1 =\frac{1}{b_c} \qquad \bm{v}_1 \cdot \nabla \lambda_1 =0 \,.
    \label{eq:cusp_constraint}
\end{equation}
If one naively imposes this constraint into the initial conditions at some Lagrangian coordinate $\bm{q}$, the Eulerian coordinate $\bm{x}_t(\bm{q})$ will indeed be part of a cusp sheet forming in the simulation. However, there is a caveat that makes this approach impractical to systematic constraint simulations: Firstly, the non-linear collapse of clusters and filaments compresses the Lagrangian volume into a few compact Eulerian structures. Consequently, a generic point on the Lagrangian cusp sheet is generally biased towards the higher collapse of the filaments and clusters bounding a cosmic wall, and does therefore not generally close to the center of the Eulerian cusp sheet. Moreover, the ZA generally fails in the highly non-linear cluster regions. If the point $\bm{q}$ is mapped into a cluster, the imposed cusp constraint from the ZA may therefore fail and not actually impose a cusp-like geometry in the non-linear $N$-body dynamics. These considerations make it unfeasible to infer physical observables of cosmic walls from simulations of the naive cusp constraint \cref{eq:cusp_constraint}.

Analogous arguments were first brought forward in \cite{FeldbruggeWeygaert2024}, which studied 2D simulations of cosmic filaments. Here, the Eulerian bias towards higher collapse was avoided by imposing a  novel constraint for the approximate centre of a cosmic filament. Following the same arguments, we now derive a condition for the formation of a \textit{cosmic wall centre}.

One candidate for the centre of the wall is the $A_3^{+}$ point \cite{Hidding+2013, Feldbrugge+2018}, the maximum of the first eigenvalue field restricted to the cusp sheet, which we discussed in \cref{subsec:theory-caustic_skeleton}. Under the ZA, the $A_3^{+}$ point is given by
\begin{equation}
    A_3^{+} :\quad \nabla \lambda_1 = 0 \quad \mathcal{H}\lambda_1 \,\,\mathrm{neg. \,def.} \quad \mathrm{lying \, in} \quad A_3 \,.
    \label{eq:A3p_condition} 
\end{equation}
Being a local maximum, this cusp point is the first in the multistream region to shell-cross, and thus forms the centre of emergence of the cusp sheet. However, when mapped into Eulerian space, the $A_3^{+}$ point may still be biased towards the non-linear collapse of the bounding filaments and clusters. This is because \cref{eq:A3p_condition} is sensitive only to the first eigenvalue field $\lambda_1$, which merely specifies the contraction rate along the collapse direction $\bm{v}_1$, but is agnostic to the dynamics within the cusp sheet.

\begin{figure*}
    \centering
    \includegraphics[width=0.65\textwidth]{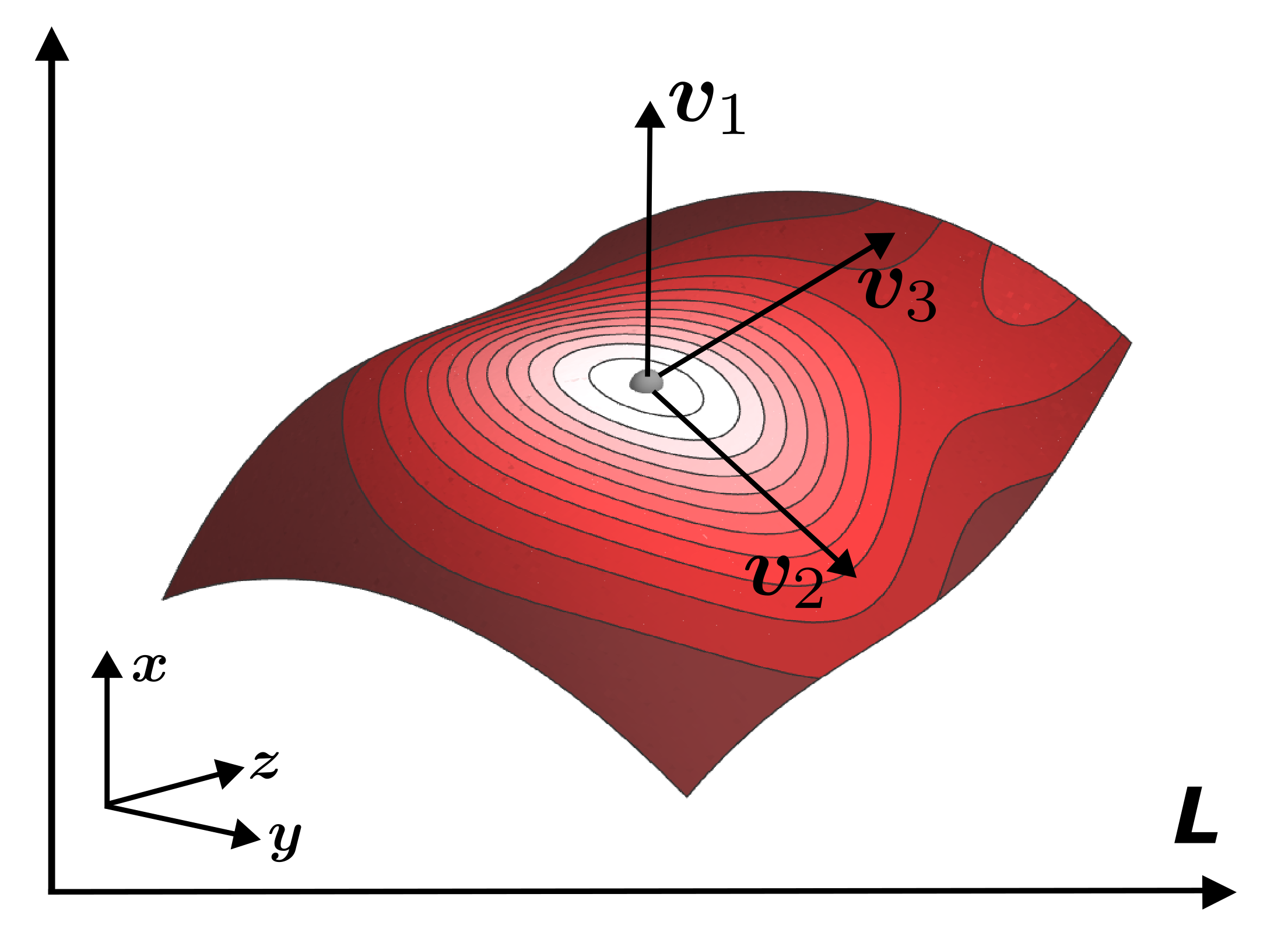}
    \caption{Sketch of the Lagrangian cusp sheet with isocontours of the contraction rate $\kappa$ and the $A_3$ centre point shown by the grey sphere. The isocontours are shown in different hues of red, with white corresponding to maximal expansion. The first eigenvector $\bm{v}_1 \approx \bm{\hat{x}}$ is normal to the cusp sheet spanned by $\bm{v}_2$ and $\bm{v}_3$ in the constraint point. The isocontours are anisotropic, as $\kappa = \lambda_2 +\lambda_3$ is a random field sliced by the $A_3$ surface. Note that the $A_3$ centre point need not be a critical point (in the curvature sense) of the cusp sheet.}
    \label{fig:stretch_sketch}
\end{figure*}

To obtain a point that lies well within the emerging wall object, we therefore propose to consider the stretch of the cusp plane in Lagrangian space and identify the wall centre as the sheet's maximally stretching point. The orientation of the cusp sheet may in principle be identified directly from the eigenfields. However, to significantly simplify the calculation, we note that the normal vector $\bm{n}_c$ of the cusp sheet is on average aligned with the first eigenvector,
\begin{equation}
    \bm{n}_c \approx \bm{v}_1 \,.
\end{equation}
We will quantitatively validate this claim in \cref{subsec:recipe-orientation}. Upon assuming the mean alignment, $\bm{n}_c = \bm{v}_1$, the cusp sheet is spanned by the $\bm{v}_2 \bm{v}_3$-plane. We then consider a unit area on the Lagrangian cusp sheet as
\begin{equation}
    A_L = \| \bm{v}_2 \times \bm{v}_3 \| \,.
\end{equation}
To maximise the stretch of the cusp sheet, we wish to maximise the magnitude of the area when evaluated on the Eulerian-evolved vectors. Analogously to the derivation presented in \cite{FeldbruggeWeygaert2024}, we evolve the vectors with the gradient matrix $\nabla \bm{x}_t$, so that the area in Eulerian space is given by
\begin{equation}
    A_E = \| \left((\nabla \bm{x}_t) \bm{v}_2 \right) \times \left( (\nabla \bm{x}_t)  \bm{v}_3 \right) \| \,.
\end{equation}
Using the ZA, we write $\nabla\bm{x}_t =I -b(t)\mathcal{H}\Psi$, and by the definition $(\mathcal{H} \Psi) \bm{v}_i = \lambda_i  \bm{v}_i $, the eigenvectors $\bm{v}_{2}, \bm{v}_{3}$ are stretched into
\begin{equation}
    (\nabla \bm{x}_t) (\bm{v}_i) = (1 - b(t) \lambda_i)  \bm{v}_i \,.
\end{equation}
The Eulerian area therefore becomes
\begin{equation}
    \begin{split}
        A_E &= \| (1 - b(t) \lambda_2)  \bm{v}_2 \times (1 - b(t) \lambda_3)  \bm{v}_3   \| \\
        & = 1- b(t) (\lambda_2 + \lambda_3) + \mathcal{O}(b(t)^2) \,,
    \end{split}
\end{equation}
where we have used the normalisation $\|\bm{v}_2 \times \bm{v}_3\| = 1$. Truncating the expansion at linear order in the growing mode $b(t)$, we find that the contraction rate in the cusp sheet is given by
\begin{equation}
    \kappa = \lambda_2 + \lambda_3\,,
    \label{eq:kappa}
\end{equation}
where positive eigenvalues $\lambda_{2},\lambda_{2} $ correspond to contraction and negative $\lambda_{2},\lambda_{3}$ correspond to expansion. The maximal stretch point is obtained by minimising $\kappa$ in the $\bm{v}_2\bm{v}_3$-plane. This is achieved by requiring the $\bm{v}_{2}$- and $\bm{v}_{3}$-directional derivatives of $\kappa$ to vanish,
\begin{equation}
    \bm{v}_2 \cdot \nabla \left(\lambda_2 + \lambda_3 \right) = 0 \qquad
    \bm{v}_3 \cdot \nabla \left(\lambda_2 + \lambda_3 \right) = 0 \,.
    \label{eq:kappa_derivs_vanish}
\end{equation}
To assert that the critical point is a local minimum, we evaluate the second-order directional derivative matrix
\begin{equation}
    M = \begin{pmatrix}
        \bm{v}_2^\mathrm{T} \mathcal{H}(\lambda_2 + \lambda_3) \bm{v}_2 & \bm{v}_2^\mathrm{T} \mathcal{H}(\lambda_2 + \lambda_3) \bm{v}_3 \\
        \bm{v}_3^\mathrm{T} \mathcal{H}(\lambda_2 + \lambda_3) \bm{v}_2 & \bm{v}_3^\mathrm{T} \mathcal{H}(\lambda_2 + \lambda_3) \bm{v}_3
    \end{pmatrix}
    \label{eq:l2_l3_stretch_M}
\end{equation}
and require that $M$ be positive definite, i.e. that both eigenvalues of $M$ be positive.
To obtain a cosmic wall realisation that is expanding rather than contracting towards secondary collapse in the ZA, we further propose to restrict the constraint to negative second eigenvalues, $\lambda_2 < 0$. In practice, this filtering is a weak constraint that does not qualitatively change the discussion of this article. We relax the assumption $\lambda_2 < 0$ in \cref{app:contracting_walls} and briefly discuss the implications for simulations of contracting walls.

To summarise, we propose the following \textit{comic wall centre condition} (or \textit{$A_3$ centre condition}) to be studied for the remainder of this article:
\begin{framed}
\begin{equation}
    \begin{split}
    &\lambda_1 = \frac{1}{b_c} \qquad \lambda_2 < 0 \qquad M\,\text{ pos. def.}\\
    &\bm{v}_1 \cdot \nabla \lambda_1 = 0 \qquad \bm{v}_2 \cdot \nabla \left(\lambda_2 + \lambda_3 \right) = 0 \qquad
    \bm{v}_3 \cdot \nabla \left(\lambda_2 + \lambda_3 \right) = 0
    \end{split}
    \label{eq:A3_centre_constraint}
\end{equation}
\end{framed}

We note for completeness that the proposed conditions assert that the wall centre lives far away from cosmic filaments. By requiring $\lambda_2 < 0$, it is evident that the $D_4$ umbilic condition of \cref{tab:caustic_skeleton} is not fulfilled. Furthermore, it can be noted that the $A_4$ swallowtwail condition can be written as $\bm{v}_1 \cdot \nabla \left( \bm{v}_1 \cdot \nabla \lambda_1 \right) = \bm{v}_1 \cdot \bm{n}_c = 0 $ with the normal vector $\bm{n}_c$ of the cusp sheet. From the average alignment $\bm{n}_c \approx \bm{v}_1$, it is seen that the $A_4$ condition is suppressed in the $A_3$ centre point.

At this point, we briefly refer to the discussion of \cref{app:caustic_skeleton-ZA}. Here, it is found that the naive cusp constraint, \cref{eq:cusp_constraint}, constitutes a single constraint equation and thus algebraically defines a two-dimensional object. With the two additional constraints of \cref{eq:kappa_derivs_vanish}, the $A_3$ centre condition now represents three constraint equations, which geometrically define a point-like structure, as is required for consistency of the preceding derivation.

The proposed wall centre constraint has two physical parameters: Firstly, the growing mode $b_c$ at which collapse is imposed, and secondly, the smoothing scale $\sigma$ corresponding to the length scale at which the caustic is measured. \Cref{subsubsec:sims-stochastic_geometry-N} will study the statistics of \cref{eq:A3_centre_constraint} in this two-dimensional parameter space $(\sigma, b_c)$.


\section{Cosmic walls from primordial field derivatives}
\label{sec:eigenframe}

We intend to impose the cosmic wall centre condition derived above into the primordial fields to run specialised simulations of cosmic wall formation. Unfortunately, there is no direct way to set up constraint fields obeying a particular configuration of the eigenvalue and eigenvector fields, and a further step of preparation is needed. We therefore now investigate the translation of the eigenfield conditions into constraints on the local geometry of the primordial displacement potential $\Psi$ by means of the \textit{eigenframe transformation}. 

\subsection{The eigenframe transformation}
\label{subsec:eigenframe}

 It is clear that the wall centre condition, \cref{eq:A3_centre_constraint}, expressed in terms of the eigenfields, constrains the local geometry of $\Psi(\bm{q})$ in the vicinity of the constraint point $\bm{q}$. By expanding $\Psi$ in a Taylor series, one may express the eigenfield configuration through the coefficients of this expansion. We define the \textit{Cartesian field derivatives} of $\Psi$ as
\begin{equation}
     t_{ij \ldots k}(\bm{q}) = \frac{\partial}{\partial x_i} \frac{\partial}{\partial x_j} \ldots  \frac{\partial}{\partial x_k} \Psi(\bm{q})
     \label{eq:field_derivatives}
\end{equation}
and aim at translating \cref{eq:A3_centre_constraint} into a condition on these variables $t_{ij \ldots k}$.

The key quantity for the cosmic wall condition is the scaled tidal tensor $\mathcal{H}\Psi$, which is given in terms of the second-order field derivatives as
\begin{equation}
    \mathcal{H}\Psi = [ \partial_i \partial_j \Psi ]_{ij} = \begin{pmatrix}
        t_{11} & t_{12} & t_{13} \\
        t_{12} & t_{22} & t_{23} \\
        t_{13} & t_{23} & t_{33} \\
    \end{pmatrix}
    \label{eq:Hphi_cartesian}
\end{equation}
and we have used the commutativity $\partial_x \partial_y =  \partial_y \partial_x$ etc.. The eigenvalues $\lambda_i$ are given by the roots of the characteristic cubic polynomial
\begin{equation}
    \begin{split}
    \chi(\lambda) &= -\lambda^3 + \tr(\mathcal{H} \Psi) \lambda^2 - \frac{1}{2} \left( \left( \tr \mathcal{H} \Psi \right)^2 - \tr ((\mathcal{H} \Psi)^2) \right) \lambda + \det \mathcal{H} \Psi \\
    & = \left( \lambda - \lambda_1 \right)\left( \lambda - \lambda_2 \right)\left( \lambda - \lambda_3 \right) \,,
    \end{split}
    \label{eq:characteristic_polynomial}
\end{equation}
where in the second line we factorise $\chi(\lambda)$ by the eigenvalue solutions $\lambda_{1,2,3}$. In the first line, the polynomial is expanded in terms of the principal invariants of $\mathcal{H}\Psi$. These are given by
\begin{equation}
    \begin{split}
    \tr(\mathcal{H} \Psi) &= t_{11}+t_{22}+t_{33} \\
    \frac{1}{2} \left( \left( \tr \mathcal{H} \Psi \right)^2 - \tr ((\mathcal{H} \Psi)^2) \right) &= t_{11}t_{22} + t_{11}t_{33} + t_{22}t_{33} - t_{12}^2 -  t_{13}^2 - t_{23}^2 \\
    \det \mathcal{H} \Psi &= t_{11}t_{22}t_{33} + 2 t_{12}t_{13}t_{23} - t_{11}t_{23}^2 - t_{22} t_{13}^2 - t_{33} t_{12}^2 \,.
    \end{split}
    \label{eq:invariants_cartesian}
\end{equation}
With the solutions $\lambda_i$ being the cubic roots of \cref{eq:characteristic_polynomial}, it is apparent that the eigenvalue fields $\lambda_i(\bm{q})$ are highly non-linear functions of the second-order field derivatives $t_{ij}(\bm{q})$. The associated eigenvectors $\bm{v}_i(\bm{q})$ are obtained from the $\lambda_i(\bm{q})$. The proposed wall condition, \cref{eq:A3_centre_constraint}, can be directly evaluated in terms of these eigenfields and their derivatives. However, the nonlinearity of \cref{eq:characteristic_polynomial} makes it difficult to relate these constraints to the local geometry of $\Psi$ expressed in a Taylor series of the Cartesian field derivatives $t_{ij \ldots k}$. Moreover, the ordering $\lambda_1 \geq \lambda_2\geq \lambda_3$ of the solutions $\lambda_i$ is not manifestly apparent, but needs to be manually imposed. 

When calculating properties of the eigenvalue fields $\lambda_i$, it is therefore convenient the rotate the coordinates into a frame in which the tidal tensor $\mathcal{H}\Psi$ is diagonal. The rotation of the second-order derivatives was first used in Doroshkevich (1970) \cite{Doroshkevich1970} for the derivation of the famous \textit{Doroshkevich formula}  \cite{Doroshkevich1970, Bardeen+1985}, which we will discuss in \cref{subsec:recipe-grf_theory}. In Feldbrugge \& van de Weygaert (2023) \cite{FeldbruggeWeygaert2023}, the rotation was generalised to the \textit{eigenframe transformation} of the eigenvalue and eigenvector fields and their derivatives in the two-dimensional case. The results were further used in 2D cosmic web studies in \cite{FeldbruggeYanWeygaert2023, FeldbruggeWeygaert2024}. Following up on these works, we now derive the eigenframe transformation for the three-dimensional case.

Geometrically, the rotation of the coordinate system is achieved through application of a standard 3D rotation matrix $R(\alpha, \beta, \gamma)$ with Euler angles $(\alpha, \beta, \gamma )$; see \cref{app:derivs_rotation}. The Euler angles of the spatial rotation are chosen such that $\mathcal{H}\Psi$ transforms with the usual matrix transformation law (cf. \cref{app:derivs_rotation}) into the diagonal form
\begin{equation}
    \begin{pmatrix}
        \lambda_1 & 0 & 0 \\
        0 & \lambda_2 & 0  \\
        0 & 0 & \lambda_3
    \end{pmatrix} = R(\alpha, \beta, \gamma) ^{\mathrm{T}} \begin{pmatrix}
        t_{11} & t_{12} & t_{13} \\
        t_{21} & t_{22} & t_{23} \\
        t_{31} & t_{32} & t_{33} \\
    \end{pmatrix} R(\alpha, \beta, \gamma) \,.
    \label{eq:eigendecomposition}
\end{equation}
In the rotated frame, the eigenvectors coincide with the coordinate axes,
\begin{equation}
    \bm{v}_1 =
    \begin{pmatrix}
        1 \\ 0 \\ 0
    \end{pmatrix} \qquad
    \bm{v}_2 =
    \begin{pmatrix}
        0 \\ 1 \\ 0
    \end{pmatrix} \qquad
    \bm{v}_3 =
    \begin{pmatrix}
        0 \\ 0 \\ 1
    \end{pmatrix} \,.
    \label{eq:eigenframe_eigenvectors}
\end{equation}
In the diagonal tidal tensor, the diagonal entries correspond to the eigenvalues $\lambda_i$, while the off-diagonal terms vanish. Expressing the rotated second-order derivatives as $T_{ij}$, we define the \textit{eigenframe} \cite{FeldbruggeWeygaert2023} to be given by the second-order derivatives obeying the \textit{eigenframe condition} 
\begin{equation}
    (T_{11},\, T_{12}, \,T_{13},\, T_{22},\, T_{23},\, T_{33}) = (\lambda_1,\, 0,\, 0,\, \lambda_2,\, 0,\, \lambda_3)
    \label{eq:eigenframe_condition}
\end{equation}
Here, the explicit ordering $\lambda_1 \geq \lambda_2 \geq \lambda_3$ is chosen for convenience of the following discussion.

From \cref{eq:eigenframe_condition} and \cref{eq:eigenframe_eigenvectors}, it is evident that the eigenvalue and eigenvector fields are linear functions in the eigenframe variables $T_{ij}$. However, the caustic conditions of \cref{subsec:theory-caustic_skeleton} and \cref{sec:A3_centre_constraint} generally involve derivatives of the eigenfields, and care must be taken when evaluating these in the eigenframe. As the eigenvalues are expressed in terms of the second-order derivatives $T_{ij}$, derivatives of the eigenfields correspond to higher-order field derivatives, which can be considered a generalisation of the eigenvalues. In analogy to the Cartesian field derivatives, \cref{eq:field_derivatives}, we denote the \textit{eigenframe derivatives} as $T_{ij \ldots k}$ (with capital $T$).

The non-linear relations between the Cartesian derivatives $t_{ij\ldots k}$ and the eigenframe derivatives $T_{ij\ldots k}$ are derived from the non-linear equalities under the defining rotation, \cref{eq:eigendecomposition}. In principle, the relations may be directly derived from the trigonometric expressions of the rotation of the higher-order derivatives. However, \cite{FeldbruggeYanWeygaert2023} finds that this geometric approach is not feasible in practice, as the resulting system of equations cannot be solved in reasonable computation time. Instead,  \cite{FeldbruggeYanWeygaert2023} proposes an algebraic approach to the construction of the \textit{eigenframe solution}, and discusses the details for the 2D case. We now repeat the analysis for the 3D case and construct the eigenframe solution relevant to this article. As the proposed wall centre constraint, \cref{eq:A3_centre_constraint}, depends on up to second-order derivatives $\partial_i \partial_j \lambda_i$ of the eigenvalue fields, we shall investigate the eigenframe solution in second- up to fourth-order variables $T_{ijkl}$. We show in \cref{sec:sims} that these are the necessary derivatives to infer the statistics of the Zel'dovich pancake formation time, as well as its orientation in constraint simulation.

\begin{framed}
To derive the transformation $t_{ij \ldots k} \rightarrow T_{ij\ldots k}$ under the diagonalisation  \cref{eq:eigendecomposition}, one starts by considering the invariants of $\mathcal{H} \Psi$ in the eigenframe :
\begin{equation}
    \begin{split}
    \tr(\mathcal{H} \Psi) &= \lambda_{1}+\lambda_{2}+\lambda_{3}\\
    \frac{1}{2} \left( \left( \tr \mathcal{H} \Psi \right)^2 - \tr ((\mathcal{H} \Psi)^2) \right) &= \lambda_1 \lambda_2 + \lambda_1 \lambda_3 + \lambda_2 \lambda_3 \\
    \det \mathcal{H} \Psi &= \lambda_1 \lambda_2 \lambda_3 \,.\\
    \end{split}
    \label{eq:invariants_eigenframe}
\end{equation}
By rotational invariance, the quantities of \cref{eq:invariants_eigenframe} are equal to those of \cref{eq:invariants_cartesian} expressed in the Cartesian frame. That is,
\begin{align}
    \lambda_{1}+\lambda_{2}+\lambda_{3} &= t_{11}+t_{22}+t_{33} \label{eq:equality_1} \\
    \lambda_1 \lambda_2 + \lambda_1 \lambda_3 + \lambda_2 \lambda_3 &= t_{11}t_{22} + t_{11}t_{33} + t_{22}t_{33} - t_{12}^2 -  t_{13}^2 - t_{23}^2 \label{eq:equality_2} \\
    \lambda_1 \lambda_2 \lambda_3 &= t_{11}t_{22}t_{33} + 2 t_{12}t_{13}t_{23} - t_{11}t_{23}^2 - t_{22} t_{13}^2 - t_{33} t_{12}^2\,.
    \label{eq:equality_3}
\end{align}
The eigenframe solution for the eigenvalues is obtained by applying the set of first through third-order differential operators
\begin{equation}
    \{ \partial_x, \partial_y, \partial_z, \partial_x^2, \partial_{x}\partial_y, \partial_y^2, \ldots ,\partial_{z}^3\}
\end{equation}
on the left and right hand side of \cref{eq:equality_1}, \cref{eq:equality_2} and \cref{eq:equality_3} respectively to obtain the system of equations
\begin{equation}
    \begin{split}
        \partial_x \left(\lambda_{1}+\lambda_{2}+\lambda_{3} \right) &= \partial_x \left(  t_{11}+t_{22}+t_{33} \right) \\
        \partial_x \left(\lambda_1 \lambda_2 + \lambda_1 \lambda_3 + \lambda_2 \lambda_3  \right) &= \partial_x \left(  t_{11}t_{22} + t_{11}t_{33} + t_{22}t_{33} - t_{12}^2 -  t_{13}^2 - t_{23}^2 \right) \\
        \partial_x \left(\lambda_1 \lambda_2 \lambda_3 \right) &= \partial_x \left(  t_{11}t_{22}t_{33} + 2 t_{12}t_{13}t_{23} - t_{11}t_{23}^2 - t_{22} t_{13}^2 - t_{33} t_{12}^2 \right) \\
        \partial_y \left(\lambda_{1}+\lambda_{2}+\lambda_{3} \right) &= \partial_y \left(  t_{11}+t_{22}+t_{33} \right) \\
        \ldots \\
        \partial_z^3 \left( \lambda_1 \lambda_2 \lambda_3 \right) &=  \partial_z^3 \left( t_{11}t_{22}t_{33} + 2 t_{12}t_{13}t_{23} - t_{11}t_{23}^2 - t_{22} t_{13}^2 - t_{33} t_{12}^2 \right) \,.
    \end{split}
    \label{eq:eigenframe_system_1}
\end{equation}
We now substitute the eigenframe conditions, \cref{eq:eigenframe_condition}, and solve for the unique solution of this system of equations. This solution gives derivatives of the eigenvalues, $\partial_i \ldots \partial_j \lambda_k$, in terms of the eigenframe variables $T_{ij \ldots k}$. Noting the commutativity $\partial_i \partial_j \ldots \partial_k = \partial_j \partial_i \ldots \partial_k$, we express the solution in the set of independent derivatives $T_{ij \ldots k}$ only.

As stressed in the preceding sections, the caustic conditions depend not only on the eigenvalue fields but also on the eigenvector fields. To express derivatives of the latter in the eigenframe variables, we now derive their corresponding solution by relating the eigenvectors to the eigenvalues using the characteristic equation \begin{equation}
    (\mathcal{H} \Psi) \bm{v}_i = \lambda_i \bm{v}_i \,.
    \label{eq:characteristic_equation}
\end{equation}
We apply the first and second-order derivatives  $\{\partial_x, \ldots, \partial_z^2\}$ to obtain the system
\begin{equation}
    \begin{split}
        \partial_x \left( (\mathcal{H} \Psi) \bm{v}_1  \right) &= \partial_x \left( \lambda_1 \bm{v}_1  \right) \\
        \partial_x \left( (\mathcal{H} \Psi) \bm{v}_2  \right) &= \partial_x \left( \lambda_2 \bm{v}_2  \right) \\ 
        \partial_x \left( (\mathcal{H} \Psi) \bm{v}_3  \right) &= \partial_x \left( \lambda_3 \bm{v}_3  \right) \\
        \ldots \\
        \partial_z^3 \left( (\mathcal{H} \Psi) \bm{v}_3  \right) &= \partial_z^3 \left( \lambda_3 \bm{v}_3  \right)
        \,,
    \end{split}
    \label{eq:eigenframe_system_2}
\end{equation}
where on the left hand side we express $\mathcal{H}\Psi$ in terms of the field derivatives. By solving this system of equations subject to \cref{eq:eigenframe_condition} and \cref{eq:eigenframe_eigenvectors}, we wish to obtain the expressions for the derivatives of the eigenvectors in the eigenframe variables.

It turns out, however, that the solution is not unique. The reason is that the eigenvectors may be arbitrarily scaled and will still fulfil the characteristic equation,  \cref{eq:characteristic_equation}. To obtain the unique solution, we therefore also consider the orthonormality conditions
\begin{equation}
    \bm{v}_i \cdot \bm{v}_j = \delta_{ij}
\end{equation}
and again apply the first through second-order derivatives $\{\partial_x, \ldots, \partial_z^2\}$ to obtain
\begin{equation}
    \begin{split}
        \partial_x \left( \bm{v}_1 \cdot \bm{v}_1\right) &= 0 \\
        \partial_x \left( \bm{v}_1 \cdot \bm{v}_2\right) &= 0 \\
        \ldots \\
        \partial_z^2 \left( \bm{v}_3 \cdot \bm{v}_3\right) &= 0 \,.
    \end{split}
    \label{eq:eigenframe_system_3}
\end{equation}
Finally, the eigenframe solution for the derivatives of the eigenvalues and eigenvectors is obtained by solving the full system of equations \cref{eq:eigenframe_system_1}, \cref{eq:eigenframe_system_2} and \cref{eq:eigenframe_system_3} subject to the eigenframe conditions \cref{eq:eigenframe_condition} and \cref{eq:eigenframe_eigenvectors}. 

The construction above gives the eigenframe solution up to the second-order derivatives of the eigenvalue fields, $\partial_i \partial_j \lambda_k$, or equivalently in up to the fourth-order field derivatives $T_{ijkl}$. The full list of relevant relations up to the chosen order is given in \cref{app:eigenframe}. This concludes the derivation of the eigenframe solution.
\end{framed}

Any differential equation on the eigenfields may be evaluated in terms of the field derivatives $T_{ij \ldots k}$.
Exemplarily, for the first eigenvalue $\lambda_1$, we find the following the $x$-derivatives:
\begin{equation*}
    \lambda_1 = T_{11} \quad \partial_x \lambda_1 = T_{111} \quad \partial_x^2 \lambda_1 = T_{1111} + \frac{2 T_{112}^2}{T_{11}-T_{22}} + \frac{T_{113}^3}{T_{11}-T_{33}}
\end{equation*}
The first equation is the definition of the eigenframe variable $T_{11}$. In the second equation, the expression for the first derivative $\partial_x \lambda_1$ is the same as would be expected for the Cartesian variable. However, it would be fallacy to assume all expressions to behave as trivially, as the eigenframe expressions generally pick up non-linear corrections. This is seen here, for instance, in the second derivative $\partial_x^2 \lambda_1$ not being equal to $T_{1111}$. Care must therefore be taken when evaluating general functions of the eigenvalue and eigenvector fields in the diagonalised frame.

For the remainder of this paper, the caustic constraints and related expressions will be evaluated in the eigenframe variables (denoted by capital $T_{ij\ldots k})$. The solution listed in \cref{app:eigenframe} allows the reader to confirm the expressions used in the following sections.

\subsection{Cosmic wall centres in the eigenframe}
\label{subsec:cosmic_walls_eigenframe}

With the eigenframe solution derived in the preceding section, we can now translate the cosmic wall centre condition, \cref{eq:A3_centre_constraint}, from a configuration of the eigenvalue fields into a constraint on the local geometry of the primordial displacement potential $\Psi$ expressed in the field derivatives $T_{ij\ldots k}$. Using the solution listed in \cref{app:eigenframe}, the reader can confirm that of the extremal stretch condition in the cusp sheet is given by
\begin{equation}
    \begin{split}
    &T_{11} = b_c^{-1} \qquad T_{12}=0 \qquad  T_{13}=0\qquad  T_{23}=0\\
    &\qquad T_{111}=0 \qquad T_{222}+T_{233}=0  \qquad  T_{223} + T_{333}=0
    \end{split} \,
    \label{eq:A3_constraint_eigenframe_}
\end{equation}
with the ordering $T_{11}  \geq T_{22} \geq T_{33}$ of the eigenvalues. The second-order derivative matrix \cref{eq:l2_l3_stretch_M} of the stretch in the cusp sheet becomes
\begin{equation}
    M = \begin{pmatrix}
        \frac{2 T_{122}^2}{T_{22}-T_{11}}+\frac{2 T_{123}^2}{T_{33}-T_{11}}+T_{2222}+T_{2233} & \frac{2 T_{122} T_{123}}{T_{22}-T_{11}}+\frac{2 T_{123}T_{133}}{T_{33}-T_{11}}+T_{2223}+T_{2333} \\
        \frac{2 T_{122} T_{123}}{T_{22}-T_{11}}+\frac{2 T_{123}T_{133}}{T_{33}-T_{11}}+T_{2223}+T_{2333} &
        \frac{2 T_{123}^2}{T_{22}-T_{11}}+\frac{2 T_{133}^2}{T_{33}-T_{11}}+T_{2233}+T_{3333}
    \end{pmatrix}
    \label{eq:l2_l3_strech_M_eigenframe}
\end{equation}
and we again restrict the sampling to positive definite $M$ and negative second eigenvalue, $T_{22} < 0$. The cosmic wall centre constraint in the eigenframe therefore reads
\begin{framed}
\begin{equation}
    \begin{split}
    \mathcal{C} = \biggl( &T_{11} = b_c^{-1} \qquad T_{12}=0 \qquad T_{13}=0 \qquad T_{23}=0 \\
    &\qquad T_{111}=0 \qquad T_{222}+T_{233}=0 \qquad T_{223} + T_{333}=0\\
    & \qquad T_{22} < 0 \qquad M\, \mathrm{\, pos. \, def.} \biggr)
    \end{split}
    \label{eq:A3_constraint_eigenframe}
\end{equation}
\end{framed}

\section{A recipe for pancakes}
\label{sec:recipe}

In this section, we will explain in detail how \cref{eq:A3_constraint_eigenframe} can be implemented into random primordial conditions to set up constrained simulations of cosmic walls.

\subsection{Elements of Gaussian random field theory}
\label{subsec:recipe-grf_theory}

Observations of the cosmic microwave background (CMB) \cite{Planck2018} strongly suggest that the  primordial potential perturbation $\phi$ can be accurately described as a homogeneous and isotropic Gaussian random field (GRF). A GRF $f$ is completely characterised by its power spectrum $P(\bm{k})$, defined as
\begin{equation}
    \langle \tilde{f} (\bm{k}) \tilde{f}^* (\bm{k'})\rangle = (2\pi)^3 P(\bm{k}) \delta_D^{(3)}(\bm{k'}-\bm{k}) \,,
    \label{eq:power_spectrum}
\end{equation}
with the Fourier expansion $\hat{f}(\bm{k})$ of the field given by
\begin{equation}
    \hat{f}(\bm{k}) = \int \d \bm{x} f(\bm{x}) e^{i \bm{k} \cdot \bm{x}} \qquad f(\bm{x}) = \int \frac{\d \bm{k}}{(2\pi)^3} \hat{f}(\bm{k}) e^{-i \bm{k} \cdot \bm{x}} \,.
\end{equation}
In the homogeneous and isotropic case, the power spectrum depends only on the magnitude of $\bm{k}$, that is $P(\bm{k}) = P(k)$. Throughout this article, $P(k)$ will denote the power spectrum of the primordial displacement potential $\Psi \propto \phi$. We take $\phi$ to be described by a GRF with Eisenstein-Hu power spectrum \cite{EisensteinHu1998} in a spatially-flat $\Lambda$CDM cosmology with parameters
\begin{equation}
    \{\Omega_0,\, \Omega_b,\, h,\, \sigma_8 \}  = \{0.308,\, 0.0482,\, 0.672,\,  0.9\}\,
    \label{eq:cosmology}
\end{equation}
with the total matter and baryonic matter densities $\Omega_0$ and $\Omega_b$, the Hubble parameter $h$ and the normalisation constant $\sigma_8$ of the linear-theory power spectrum \cite{Peebles1994}. For details on the standard algorithm to generate random field realisations obeying a particular power spectrum, we refer to \cite{FeldbruggeWeygaert2023}.

As was discussed in \cref{subsec:theory-scale_space}, we wish to study the scale-space caustic skeleton by smoothing the real-space field $\Psi(\bm{q})$ on varying scales $\sigma$. Plugging \cref{eq:smoothing_real_space} into \cref{eq:power_spectrum}, one finds that the power spectrum of to the smoothed displacement potential $\Psi_{\sigma}(\bm{q})$ is given by
\begin{equation}
    P_{\sigma}(k) = P(k) e^{-\sigma^2 k^2 } \,.
    \label{eq:smoothed_power_spectrum}
\end{equation}
All the scale-space statistical information of $\Psi$ is encapsulated in the smoothed power spectrum $P_{\sigma}(k)$ over a range of cosmologically relevant scales $\sigma$.

In \cref{sec:eigenframe}, we discussed in detail how the cosmic wall centre condition may be translated into a constraint on the Taylor expansion of the field $\Psi$ in the vicinity of the constraint point $\bm{q}$, which we express in terms of the field derivatives
\begin{equation}
     t_{ij \ldots k}(\bm{q}) = \frac{\partial}{\partial x_i} \frac{\partial}{\partial x_j} \ldots  \frac{\partial}{\partial x_k} \Psi(\bm{q}) \,.
\end{equation}
To study the statistical properties of the field derivatives $\bm{t}(\bm{q})$ in the random primordial conditions, we now discuss the relevant results of Gaussian random field theory. As we are only interested in the derivatives at the constraint point $\bm{q}$, we omit the argument in the following and evaluate all properties at the same point.

The field derivatives $\bm{t}$ of the GRF $\Psi$ form a multivariate normal distribution of zero mean. The covariance matrix of $\bm{t}$ can be readily evaluated from the field's power spectrum $P(k)$. To this end, it is useful to define the generalised moments
\begin{equation}
    \begin{split}
    \sigma_i^2 
    &= \frac{1}{(2\pi)^3} \int \d \bm{k} \|\bm{k}\|^{2i} P_{\sigma}(\bm{k})\\ 
    &= \frac{1}{2\pi^2} \int \d k k^{2i +2} P_{\sigma}(k) \,,
    \end{split}
    \label{eq:generalised_moments}
\end{equation}
where the field's isotropy was used in the second line to perform the integration in spherical coordinates. The generalised moments here are evaluated on a smoothed power spectrum $P_{\sigma}(k)$ with some scale $\sigma$, so that the scale-space statistical information of $\Psi$ is absorbed into these generalised moments. However, we will omit an explicit $\sigma$-subscript in the variables $\sigma_i$ for the sake of clarity. 

The covariance between two field derivatives is calculated in Fourier space by expanding $\Psi(\bm{q})$ in the Fourier components $\tilde{\Psi}(\bm{k})$ and evaluating the defining expectation value. Exemplarily, the covariance between the field and its second-order $x$-derivative is given by
\begin{equation}
    \begin{split}
    \left< \partial_x^2 \Psi(\bm{q}) \Psi^*(\bm{q})  \right>
    &= \left< \frac{1}{(2\pi)^6} \int \d \bm{k}\partial_x^2 \left( \tilde{\Psi}(\bm{k}) e^{-i \bm{k} \cdot \bm{q}} \right) \int \d \bm{k^{\prime}} \tilde{\Psi}(\bm{k^{\prime}}) e^{i \bm{k^{\prime}} \cdot \bm{q}} \right> \\
    &= \frac{1}{(2\pi)^3} \int \d \bm{k} (- i k \cos \theta )^2  \left< \tilde{\Psi}(\bm{k})  \tilde{\Psi}^{*}(\bm{k^{\prime}}) \right>
    \delta_D^{(3)}( \bm{k^{\prime}} -\bm{k}) \\
    &= \frac{1}{(2\pi)^3} \iiint \d \phi \, \d \theta \, \d k \cos^2\theta \sin^2\theta \cos^2\phi k^4 P_{\sigma}(k) \\
    &= -\frac{\sigma_1^2}{3}
    \end{split}
\end{equation}
and we used \cref{eq:generalised_moments} for the definition of $\sigma_1$ in the last line.

Generalising this calculation, one finds that the covariance between two arbitrary field derivatives at $\bm{q}$ is given by
\begin{equation}
    \begin{split}
        &\left< \partial_x^{(l_1)} \partial_y^{(m_1)}  \partial_z^{(n_1)} f(\bm{q}) \partial_x^{(l_2)} \partial_y^{(m_2)}  \partial_z^{(n_2)}  f^*(\bm{q})  \right> = \\
        &\quad \frac{1}{2}(-1)^{l_1 + m_1 + n_1} i^{l_1 + m_1 + n_1 + l_2 + m_2 + n_2} \sigma_{(l_1 + m_1 + n_1 + l_2 + m_2 + n_2)/2}^2\times
        \\
        &\quad 
        \int_{0}^{2\pi} \d \phi \int_{0}^{\pi} \d \theta (\sin \theta)^{1+l_1 + m_1 + l_2 + m_2} (\cos \phi)^{l_1 + l_2} (\sin \phi)^{m_1 + m_2}  (\cos \theta)^{n_1 + n_2}  \,.
    \end{split}
    \label{eq:deriv_covariance}
\end{equation}

Note that an immediate consequence of \cref{eq:deriv_covariance} is that all covariances with any odd pairing $l_1 +l_2$, $m_1 +m_2$ or $n_1 +n_2$ at the same point $\bm{q}$ identically vanish.

In \cref{app:eigenframe}, it was discussed that the cosmic wall centre constraints lives in the space of the second up to fourth-order field derivatives. We denote and order these as
\begin{equation}
    \begin{split}
    \bm{t}^{(2)} &= \{t_{11}, t_{12}, t_{13}, t_{22}, t_{23}, t_{33} \} \\
    \bm{t}^{(3)} &= \{t_{111}, t_{112}, t_{113}, t_{122}, t_{123}, t_{133}, t_{222}, t_{223}, t_{233}, t_{333} \} \\
    \bm{t}^{(4)} &= \{ t_{1111}, t_{1112}, t_{1113}, t_{1122}, \ldots t_{3333} \}
    \end{split}
\end{equation}
with the full statistic being the second, third and fourth-order derivatives
\begin{equation}
    \bm{t} = \{t^{(2)}, t^{(3)}, t^{(4)}\}=\{ t_{11}, t_{12}, \ldots, t_{3333} \} \,.
    \label{eq:t_statistic}
\end{equation}
$\bm{t}$ is a $31$-dimensional random vector that is distributed with a multivariate normal distribution of zero mean with the covariance matrix $\Sigma$ having the block structure
\begin{equation}
    \Sigma =  \Sigma^{(2,3,4)} =  \begin{pmatrix}
        \Sigma_{22} & 0 & \Sigma_{24} \\
        0 & \Sigma_{33} & 0 \\
        \Sigma_{24}^{\mathrm{T}} & 0 & \Sigma_{44}
    \end{pmatrix} \,.
    \label{eq:cov_2t4}
\end{equation}
The probability density is given by the multivariate normal distribution
\begin{equation}
    p(\bm{t}) \d \bm{t} = \frac{1}{\sqrt{(2\pi)^{31} |\det \Sigma|}} \exp \left( - \frac{1}{2} \bm{t}^{\mathrm{T}} \Sigma \bm{t} \right)  \d \bm{t} 
    \label{eq:distribution_Cartesian}
\end{equation}
with the covariance matrix $\Sigma$ given by \cref{eq:cov_2t4}.

Having discussed the statistics of the field derivatives $t_{ij\ldots k}$, it is now instructive to investigate the distribution of the second-order derivatives $t_{ij}$ and the corresponding eigenvalue fields $\lambda_i$. The latter is given by the famous \textit{Doroshkevich formula} \cite{Doroshkevich1970}, which we shall now rederive in brevity; for more details, we refer the reader to Bardeen et al. (1986) \cite{Bardeen+1985}. The probability density function of the second-order derivatives $t_{ij}$ is given by
\begin{equation}
    p(\bm{t}^{(2)}) = \frac{1}{(2\pi)^{6/2} (\det \Sigma_{22})^{1/2}} \exp \left(-\frac{1}{2} (\bm{t}^{(2)})^{\mathrm{T}} \Sigma_{22} \bm{t}^{(2)} \right)
    \label{eq:distribution_derivs_t2}
\end{equation}
To derive the distribution of the corresponding eigenvalues, as in \cref{sec:eigenframe}, one starts again with the eigendecomposition, \cref{eq:eigendecomposition}. This transformation corresponds to a change of variables $\{t_{11}, t_{12},t_{13},t_{22},t_{23},t_{33}\} \rightarrow (\lambda_1, \lambda_2, \lambda_3, \alpha, \beta, \gamma)$ with the Euler angles $(\alpha, \beta, \gamma)$. The integration measure changes to \cite{Bardeen+1985}
\begin{equation}
    \prod_{i, j \geq i} \d t_{ij} = |\lambda_1 - \lambda_2||\lambda_1 - \lambda_3||\lambda_2 - \lambda_3|\d \lambda_1 \d \lambda_2 \d \lambda_3 \d \alpha \wedge \d \beta \wedge \d \gamma \,,
\end{equation}
Here, the terms $|\lambda_i -\lambda_j|$ arise as the Jacobian determinant of the rotation by the Euler matrix. The probability measure in the eigenvalues is obtained by integrating out the Euler angles $(\alpha, \beta, \gamma)$ to account for all possible orientations of the diagonalised frame. Appendix B of \cite{Bardeen+1985} discusses in detail that this integration yields the volume $2\pi^2$ of the unit three-sphere. The Jacobian of the eigendecomposition \cref{eq:eigendecomposition} is therefore given by
\begin{equation}
    J = 2 \pi^2 |\lambda_1 - \lambda_2||\lambda_1 - \lambda_3||\lambda_2 - \lambda_3|
    \label{eq:Jacobian}
\end{equation}
The entries of the rotated deformation tensor, $t_{11} = \lambda_1, t_{12} = 0$ etc., are then inserted into the exponent of \cref{eq:distribution_derivs_t2}, which becomes
\begin{equation}
    - \frac{1}{2} (\bm{t}^{(2)})^{\mathrm{T}} \Sigma_{22} \bm{t}^{(2)} =
    - \frac{1}{10 \sigma_2^2}\left( 3\lambda_1^2+3\lambda_2^2+3\lambda_3^2 + 2 \lambda_1 \lambda_2  + 2 \lambda_1 \lambda_3  + 2 \lambda_2 \lambda_3   \right)
\end{equation}
Finally, the Doroshkevich formula for the distribution of the eigenvalue fields is given by \cite{Doroshkevich1970, Bardeen+1985}
\begin{equation}
    \begin{split}
    p(\lambda_i) \d \lambda_i &= \frac{1}{(2\pi)^{6/2} (\det \Sigma_{22})^{1/2}}  2 \pi^2 |\lambda_1 - \lambda_2||\lambda_1 - \lambda_3||\lambda_2 - \lambda_3| \\
    &\qquad \times \exp \left(  - \frac{1}{10 \sigma_2^2}\left( 3\lambda_1^2+3\lambda_2^2+3\lambda_3^2 + 2 \lambda_1 \lambda_2  + 2 \lambda_1 \lambda_3  + 2 \lambda_2 \lambda_3   \right) \right) \d \lambda_i
    \end{split}
    \label{eq:Doroshkevich_formula}
\end{equation}
As expected, the distribution of the eigenvalues is non-Gaussian. This is consistent with the discussion of \cref{sec:theory}, where \cref{fig:eigenvalue_visualisation} revealed how the non-linear nature of the eigenvalue fields as the cubic roots of $\mathcal{H}\Psi$ manifests itself in a filamentary pattern, thus giving rise to the progenitors of the multistream regions that delineate the cosmic web.

The eigenframe transformation of \cref{sec:eigenframe} generalises the rotation \cref{eq:eigendecomposition} to the higher-order field derivatives. In addition the second-order field derivatives, we now to study the third and fourth-order eigenframe derivatives $T_{ij\ldots k}$ and derive the distribution of the 31-dimensional statistic
\begin{equation}
    \bm{T} = \{ T_{11}, T_{12}, \ldots, T_{3333} \}
    \label{eq:T_statistic}
\end{equation}
Under the rotation of \cref{eq:eigendecomposition} with the Euler angles $(\alpha, \beta, \gamma)$, the higher derivatives transform non-trivially with the generalised rotation matrices presented in \cref{app:derivs_rotation}. From these matrices, one can evaluate the Jacobian of the higher-derivatives rotation, with both $t_{ijk} \rightarrow T_{ijk}$ and $t_{ijkl} \rightarrow T_{ijkl}$ and contributing a unit factor.
The conceptual difference of the higher-deriatives rotation to \cref{eq:eigendecomposition} is that the eigenframe condition, \cref{eq:eigenframe_condition}, constrains the off-diagonal variables $T_{ij}$ to vanish. This algebraic constraint necessitated the change of variables to the Euler angles $(\alpha, \beta, \gamma)$, and thus induces the Jacobian terms $|\lambda_i - \lambda_j|$ above. As there are no algebraic constraints on the higher derivatives, the transformation here simply becomes a rotation in a suitable representation, which can contribute but a factor of unity.

The full Jacobian of the eigenframe transformation is the therefore  same as \cref{eq:Jacobian} derived above for the Doroshkevich formula. Inserting this Jacobian and the eigenframe variables $T_{ij \ldots k}$ into \cref{eq:distribution_Cartesian}, the distribution of the second- through fourth-order derivatives $\bm{T}$ in the eigenframe is given by
\begin{equation}
    p(\bm{T}) \d \bm{T} = \frac{1}{\sqrt{(2\pi)^{31} |\det \Sigma|}}(T_{11}-T_{22})(T_{11}-T_{33})(T_{22}-T_{33}) \exp \left(-\frac{1}{2} \bm{T}^{\mathrm{T}} \Sigma \bm{T} \right) \d \bm{T}
    \label{eq:T_distribution}
\end{equation}
with the ordering $T_{11} \geq T_{22} \geq T_{33}$ and the covariance matrix \cref{eq:cov_2t4}.

\subsection{Constraint sampling — local geometry}
\label{subsec:recipe-sampling}

Having discussed the statistics of the field derivatives $T_{ij\ldots k}$, we discuss now the implementation of the $A_3$ centre constraint, \cref{eq:A3_constraint_eigenframe}, into a primordial displacement potential $\Psi$. To this end, we follow \cite{FeldbruggeWeygaert2023}, which prescribes a two-step procedure for the generation of constrained random initial conditions.

\begin{figure}
    \centering
    \includegraphics[width=0.45\columnwidth]{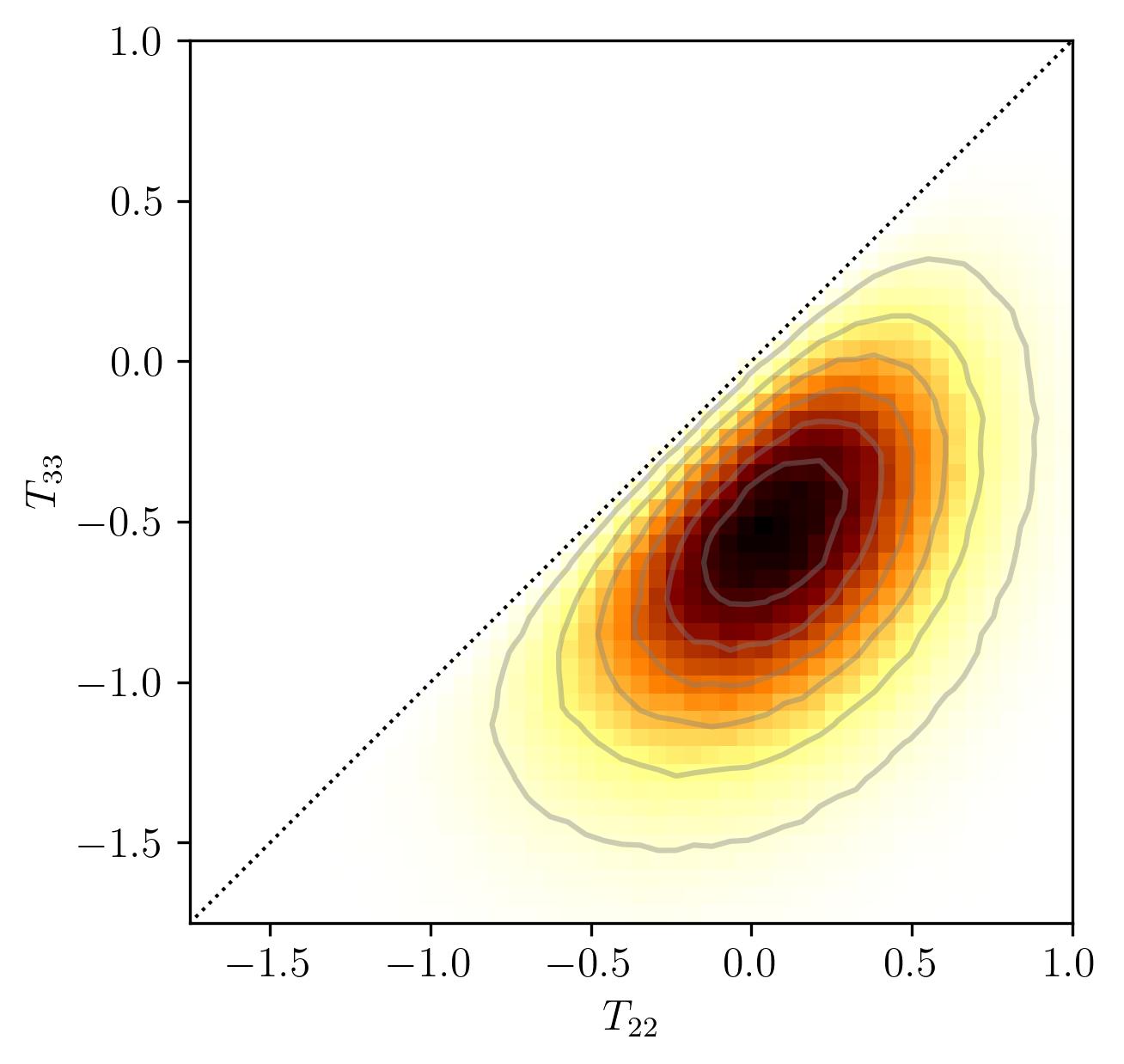}
    \caption{Histogram of the marginal distribution of the second and third eigenvalue field, $\lambda_2 = T_{22}$ and  $\lambda_3 = T_{33}$, evaluated over a sampled chain of wall centre constraint realisations, \cref{eq:A3_constraint_eigenframe}, for $(\sigma, b_c) = (2.0, \,0.8)$.}
    \label{fig:sample_histogram}
\end{figure}

The first step concerns the morphology of the wall from the field derivatives and defines the local geometry of $\Psi$ around the constraint point $\bm{q}$. The cosmic wall centre condition derived above gives a number of unconstrained variables $(T_{22}, T_{33}, T_{112}, \ldots, T_{3333})$ in the space of second- through fourth-order field derivatives. In the terminology of \cite{FeldbruggeWeygaert2023}, \cref{eq:A3_constraint_eigenframe} defines the constraint manifold $\mathcal{C}$, with the unconstrained variables corresponding to an infinite number of possible wall morphologies. A single \textit{constraint realisation} $\bm{C} \in \mathcal{C}$ is obtained by drawing a sample of the constraint statistic 
\begin{equation}
    \begin{split}
        \bm{Y} = (&b_c^{-1},\,0,\,0,\,T_{22},\,0,\,0, \\
        &0,\,T_{112},\, T_{113},\, T_{122},\, T_{123},\, T_{133},\, T_{222},\, T_{223},\, -T_{223},\, -T_{233}, \\
        &T_{1111},\, T_{1112},\,T_{1113},\,\ldots,\, T_{3333})
    \end{split}
    \label{eq:constrained_statistic}
\end{equation}
from the distribution given by \cref{eq:T_distribution}. Note that the parameters $(\sigma, b_c)$ are explicitly here through the fold constraint level $T_{11} = b_c^{-1}$ in \cref{eq:constrained_statistic} and the generalised moments $\sigma_i$ evaluated at $\sigma$, which enter the covariance matrix, \cref{eq:cov_2t4}. The sampling from the non-Gaussian distribution may be performed by any suitable method such as rejection sampling or Hamiltonian Monte Carlo (HMC) \cite{Duane+1987}. We choose the latter and obtain a chain of samples $\{\bm{C}_i\}$ by a Hamiltonian walk in the potential $ U(\bm{Y})$ given by the negative log-likelihood of the constraint distribution,
\begin{equation}
    \begin{split}
        U(\bm{Y}) &= - \ln \mathcal{L}(\bm{Y})\\
        &=\frac{1}{2}\bm{Y}^{\mathrm{T}} \Sigma \bm{Y} - \ln(T_{11}-T_{22}) - \ln(T_{11}-T_{33}) - \ln(T_{22}-T_{33})\\
        &=\frac{1}{2}\bm{Y}^{\mathrm{T}} \Sigma \bm{Y} - \ln(b_c^{-1}-T_{22}) - \ln(b_c^{-1}-T_{33}) - \ln(T_{22}-T_{33}) \,,
    \end{split}
    \label{eq:sampling_potential}
\end{equation}
where $\Sigma$ is given by \cref{eq:cov_2t4} evaluated for some smoothing scale $\sigma$ and we have omitted an irrelevant constant.
We use a stable and efficient implementation of the No-U-Turn-Sampling (NUTS) HMC algorithm \cite{HoffmanGelman2014} provided by the \verb|AdvancedHMC| library \cite{Xu+2020} in the \verb|julia| language. With this setup, a chain of order $10^5$ constraint realisations is obtained within less than a minute of computation time on a modern CPU.

Having obtained a chain of samples $\{\bm{C}_i\}$, the additional conditions $T_{22} < 0$ and $M$, \cref{eq:l2_l3_strech_M_eigenframe}, being positive definite are straightforwardly enforced by filtering out those samples in $\{\bm{C}_i\}$ that violate these conditions.

\subsection{Constraint implementation — global geometry}
\label{subsec:recipe-HR}

In the second step of the constraint implementation, one extends the local geometry of the wall to the global geometry of the simulation box. To this end, one constructs \textit{constrained field realisation} $\Psi_c$ by imposing a sampled constraint realisation $\bm{C}$ into a random realisation of the primordial displacement potential $\Psi$.

Having completed the sampling step of the previous section, the chosen sample $\bm{C}$ now defines the Taylor expansion of $\Psi$ around the constraint point $\bm{q}$. $\bm{C}$ therefore constitutes a linear constraint, assigning 31 numerical values to the second- through fourth-order field derivatives $\bm{t}$. We impose the constraint into $\Psi$ by means of the well-known Bertschinger-Hoffman-Ribak (BHR) algorithm \cite{Bertschinger1987, HoffmanRibak1991, WeygaertBertschinger1996}, which exploits the linearity of the constraint to construct the mean field $\bar{\Psi}$ satisfying  $\bm{C}$ and add to this a residual field $\delta \Psi$ fulfilling the required Gaussian random statistics \cite{Bertschinger1987}; for details, see also \cite{FeldbruggeWeygaert2023}.

We reiterate here that the initial conditions are set up in a two-step process. Catastrophe theory determines the conditions on the local potential underlying the wall formation. The random morphology of the wall progenitor is obtained by randomly sampling a constraint realisation obeying these conditions. The constraint is then extended to the simulation box by imposing it into a random primordial potential realisation. While we stress that this is a fully Lagrangian treatment, our procedure yields a realistic and random wall that is woven into a random global cosmic web under the evolution into Eulerian space.

The exemplary constraint simulation shown in \cref{fig:sim_256_A3_3D} demonstrates that the resulting constrained potential $\Psi_c$ forms the initial conditions for the formation of a cosmic wall, with the caustic skeleton evaluated on the smoothed field $\Psi_{c, \sigma}$ yielding a cusp sheet centred on the constraint point $\bm{q}$.

\subsection{Constraint orientation}
\label{subsec:recipe-orientation}

\begin{figure}
    \centering
    \includegraphics[width=0.45\columnwidth]{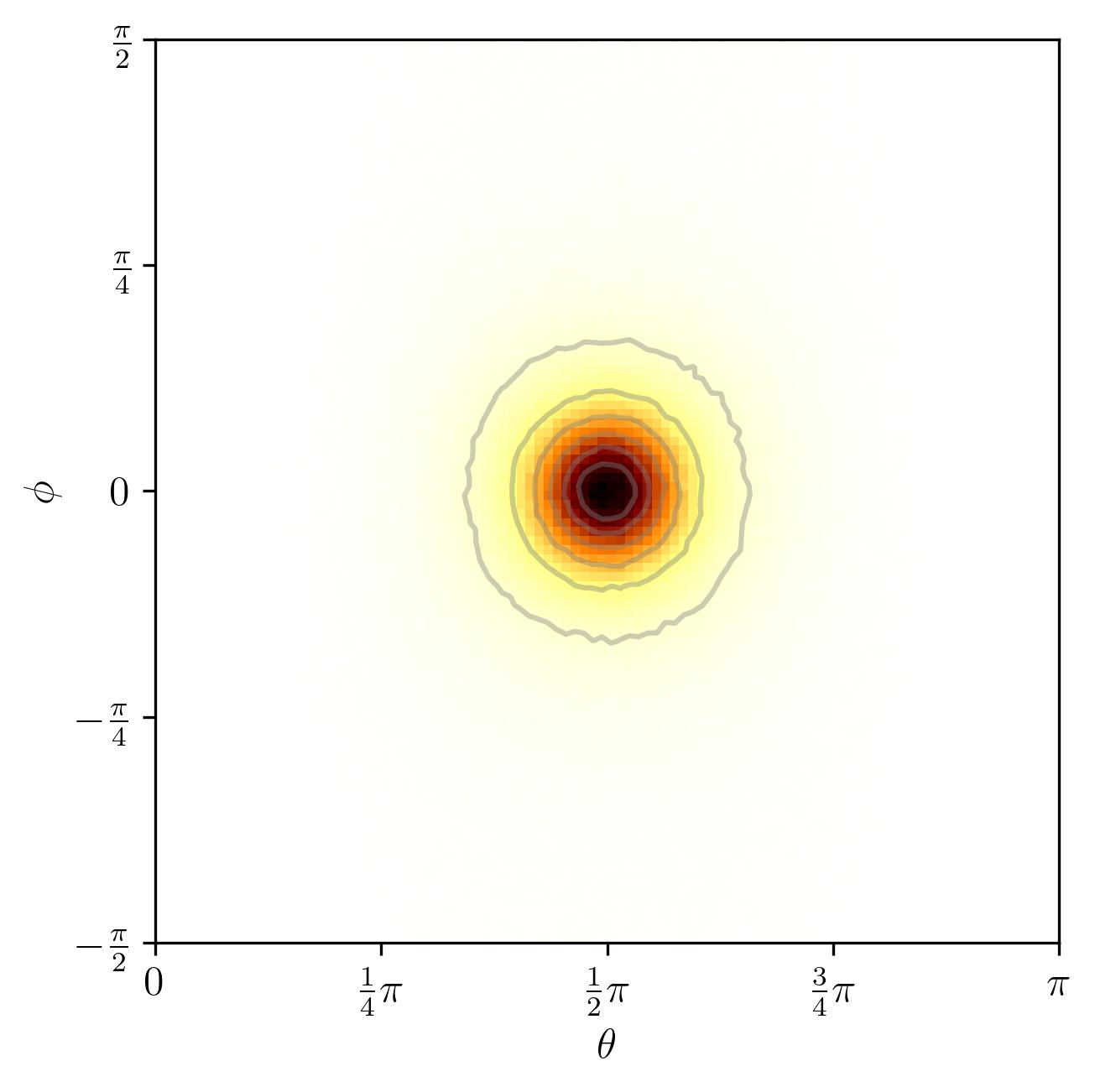}
    \caption{Histogram of the angular distribution of $\bm{n}_c$ from the wall centre constraint sampled for $(\sigma, b_c) = (2.0, \,0.8)$, with the usual convention for azimuthal and polar angles $\theta$ and $\phi$. The distribution is sharply centered about the mean orientation along the $x$-axis.}
    \label{fig:nc_histogram}
\end{figure}

The results of the preceding sections can be directly applied to run specialised $N$-body simulations of cosmic wall formation. To reliably infer mean fields through field stackings, however, it is essential that the simulated cosmic walls be accurately oriented. Conveniently, the constraint algebra enables a rigorous procedure for the orientation of the emerging walls, first in Lagrangian space and subsequently in Eulerian space, as we shall now explain.

\paragraph*{Lagrangian orientation} In Lagrangian space, the orientation of a constraint realisation $\bm{T}$ is obtained from the algebra of the field derivatives in the eigenframe by evaluating the gradients of the relevant caustic conditions. Concretely, the normal vectors $\bm{n}_c$, $\bm{n}_f$ to the cusp sheet and the fold line at the instant $b_c$ of shell-crossing are given by
\begin{align}
    \bm{n}_c &= \nabla \left(\bm{v}_1 \cdot \nabla \lambda_1\right) = \begin{pmatrix}
        \frac{3 T_{112}^2}{T_{11}-T_{22}}+\frac{3
   T_{113}^2}{T_{11}-T_{33}}+T_{1111} \\
        \frac{3 T_{112} T_{122}}{T_{11}-T_{22}}+\frac{3
   T_{113} T_{123}}{T_{11}-T_{33}}+T_{1112} \\
        \frac{3 T_{112}
   T_{123}}{T_{11}-T_{22}}+\frac{3 T_{113} T_{133}}{T_{11}-T_{33}}+T_{1113}
    \end{pmatrix}
    \label{eq:n_cusp} \\
    \bm{n}_f &= \nabla \lambda_1 = (T_{111}, T_{112}, T_{113})^{\rm{T}} \,.
    \label{eq:n_fold}
\end{align}
The vector $\bm{n}_c$ fixes the normal orientation of the cusp sheet, while  $\bm{n}_f$ defines the direction of shell-crossing at $b_c$. By specifying the direction of growth of the cusp sheet, $\bm{n}_f$ breaks the rotational symmetry of the cusp constraint in the wall's tangent plane. Both $\bm{n}_c$ and $\bm{n}_f$ are evaluated on the sampled eigenframe variables $\bm{C}$, and we take the vectors to be normalised in the following. From the eigenframe definition, \cref{eq:eigenframe_condition}, and the imposed constraint, \cref{eq:A3_constraint_eigenframe}, $\bm{n}_c$ is expected to be distributed around the mean alignment with the $x$-axis, $\bm{n}_c \approx \bm{\hat{x}}$. \Cref{fig:nc_histogram} confirms this claim from a chain of sampled constraint realisations.

From the condition $T_{111} = 0$ in \cref{eq:A3_constraint_eigenframe}, one finds that $\bm{n}_f$ is of the general form $\bm{n}_f = (0, u, v)$. While $\bm{n}_f$ is by construction normal to the fold line, it is therefore not in general orthogonal to $\bm{n}_c$, of which no components are identically vanishing. To obtain an orthonormal basis system, we therefore introduce a tangent vector to the cusp sheet, $\bm{t}_c$, by projecting $\bm{n}_f$ in the subspace normal to $\bm{n}_c$, i.e.
\begin{equation}
    \bm{t}_c \equiv \bm{n}_c - \text{proj}_{\bm{n}_c} \bm{n}_f \,.
\end{equation}
As this vector is by construction tangent to the cusp sheet, we take it to be one of of the two tangent vectors spanning the Lagrangian wall plane. The basis of orientation vectors is completed by introducing a second tangent vector $\bm{t}_f$ given by
\begin{equation}
    \bm{t}_f \equiv \bm{n}_c \times \bm{t}_c \,.
\end{equation}
From the construction of the normal and tangent vectors, it is easily seen that  $\bm{t}_f$ is the tangent vector to the fold line, with $\{\bm{t}_c, \bm{t}_f\} = \text{span}(A_3)$ spanning the cusp sheet.

\begin{figure}
    \centering
    \begin{subfigure}[b]{0.35\textwidth}
        \includegraphics[width=\textwidth]{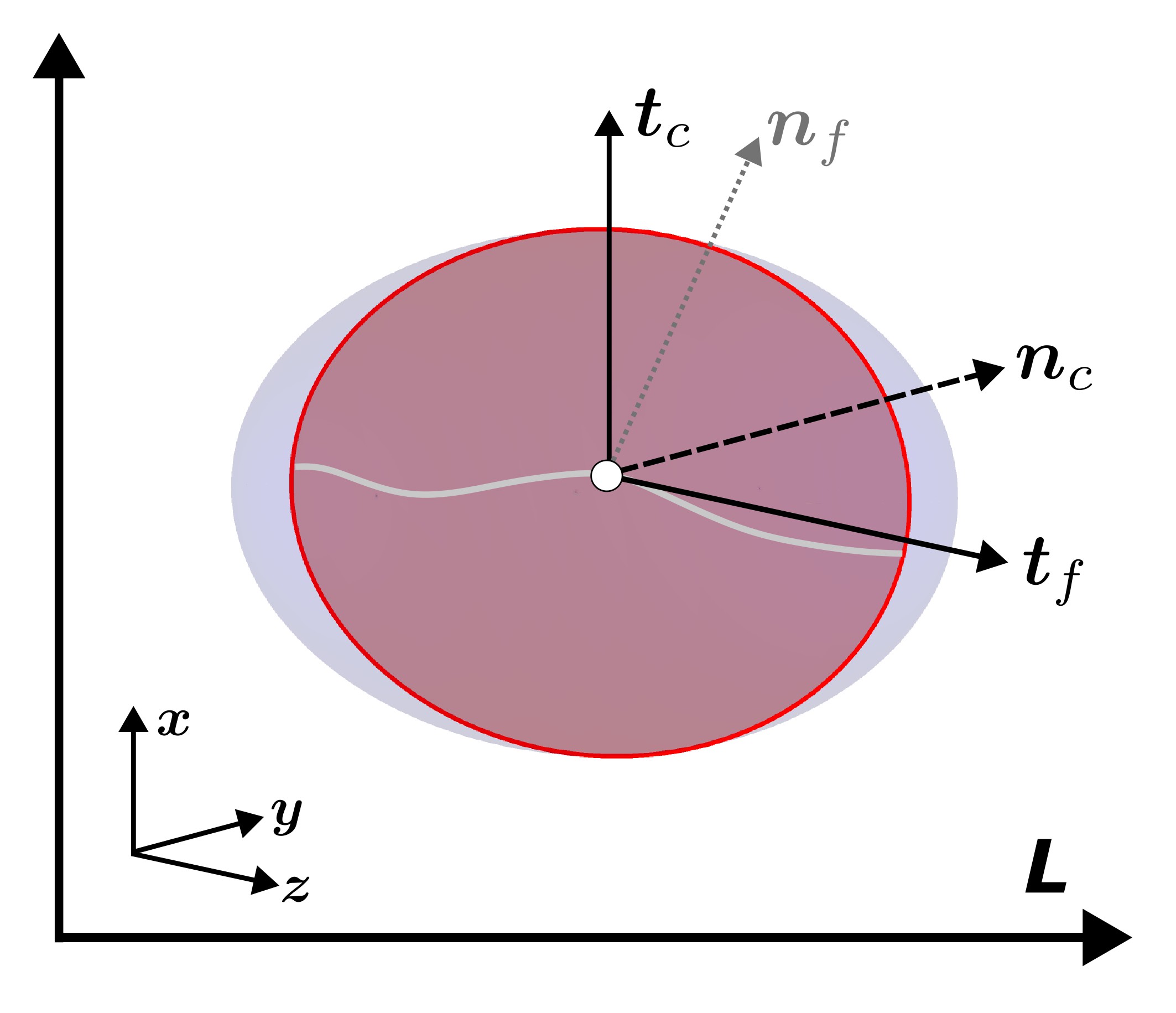}
        \caption{Lagrangian space}
    \end{subfigure}
    \hspace{0.1\textwidth}
    \begin{subfigure}[b]{0.35\textwidth}
    \includegraphics[width=\textwidth]{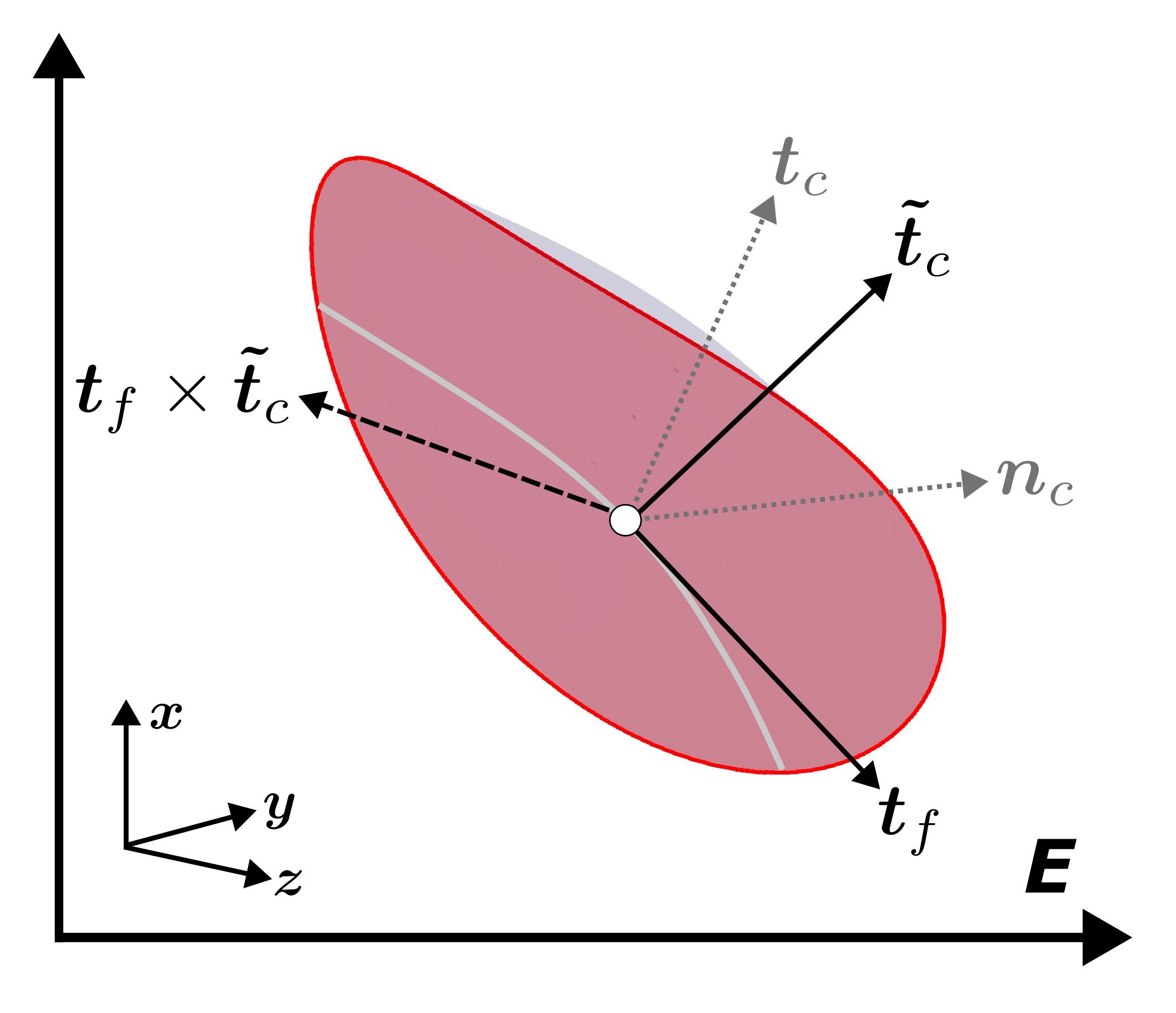}
        \caption{Eulerian space}
    \end{subfigure}
    \caption{Sketch of the cusp sheet orientation in Lagrangian space (left panel) and after evolution into Eulerian space (right panel).}
    \label{fig:orientation_sketches}
\end{figure}

The right-handed set of vectors
$\{\bm{n}_c, \bm{t}_c, \bm{t}_f\}$ is evaluated on the sampled constraint realisation $\bm{C}$. We then use the algorithm of presented in \cref{app:euler_angles} to calculate the Euler angles $(\alpha, \beta, \gamma)$ that rotate the system $\{\bm{n}_c, \bm{t}_c, \bm{t}_f\}$ into a desired orientation, which we choose to be the Cartesian basis $\{\bm{\hat{x}}, \bm{\hat{y}}, \bm{\hat{z}}\} = \{(1,0,0)^{\mathrm{T}}, (0,1,0)^{\mathrm{T}}, (0,0,1)^{\mathrm{T}}\}$. Under this rotation of the coordinate system, the field derivatives $T_{ij\ldots k}$ transform with the generalized rotation matrices presented in \cref{app:derivs_rotation}; see also \cite{FeldbruggeWeygaert2023}. The constraint $\bm{C}$ is therefore oriented by rotating it with the transformation matrix $R^{(2,3,4)}(\alpha, \beta, \gamma)$ for the second- through fourth-order derivatives,
\begin{equation}
     \bm{\tilde{C}} = R^{(2,3,4)}(\alpha, \beta, \gamma) \bm{C}
     \label{eq:Lagrangian_c_rotation} \,.
\end{equation}
Upon imposing the rotated constraint $\bm{\tilde{C}}$ into $\Psi$, the resulting constrained potential $\Psi_c$ yields a cusp sheet with the desired orientation  $\{\bm{n}_c, \bm{t}_c, \bm{t}_f\} = \{\bm{\hat{x}}, \bm{\hat{y}}, \bm{\hat{z}}\}$ centred on $\bm{q}$.

For completeness, note that the rotation of \cref{eq:Lagrangian_c_rotation} changes the eigenframe of the constraint, as the eigenvectors of $\bm{\tilde{C}}$ are not the the same as those of $\bm{C}$ for which the eigenframe was constructed. The expressions \cref{eq:n_cusp} and \cref{eq:n_fold} therefore no longer hold for the rotated constraint $\bm{\tilde{C}}$.

\paragraph*{Eulerian orientation} Having fixed the Lagrangian orientation to $\{\bm{n}_{c}, \bm{t}_{c}, \bm{t}_{f}\} = \{\bm{\hat{x}}, \bm{\hat{y}}, \bm{\hat{z}}\}$, the evolution of the vector basis into Eulerian space now gives the orientation of the cosmic wall in the Eulerian coordinate configuration. A Lagrangian vector $\bm{v}_L$ is evolved into Eulerian space by
\begin{equation}
    \bm{v}_E = \nabla_{\bm{q}}\, \bm{x}_{\sigma, t}(\bm{q}) \cdot \bm{v}_L = \left( I + \nabla_{\bm{q}}\, \bm{s}_{\sigma, t}(\bm{q}) \right)\cdot \bm{v}_L \,,
\end{equation}
where $I$ is the identity matrix and the subscript $\sigma$ denotes smoothing of the coordinates and displacement field. The reason for choosing the smoothed deformation tensor $\nabla_{\bm{q}}\, \bm{s}_{\sigma, t}(\bm{q})$ is that the small-scale non-linear dynamics generally induce creases on the cusp sheet. The unsmoothed displacement field $\bm{s}_{t}(\bm{q})$ is sensitive to this small-scale structure, and so would yield the orientation of the constraint point along the small-scale creases. But the wall is constructed on the physical length scale $\sigma$, on which it forms a smooth sheet-like object. This large-scale orientation is captured by the smoothed displacement field $\bm{s}_{\sigma, t}(\bm{q})$, which is hence the relevant quantity to use in our procedure. The same argument was brought forward in \cite{FeldbruggeWeygaert2024} for the 2D case.

To infer the Eulerian orientation, it is necessary to evolve the Lagrangian tangent rather than normal vectors, as the normalised Eulerian-evolved vectors $\{\bm{n}_{c, E}, \bm{t}_{c,E}, \bm{t}_{f,E}\}$ do not generally form an orthonormal basis.\footnote{In a generic Hamiltonian flow (including $N$-body simulations), the gradient matrix $\nabla_{\bm{q}} \bm{x}_t(\bm{q})$ is a symmetric, but not necessarily orthogonal matrix. Under action of a symmetrix matrix $\bm{A}$, the orthogonality of two orthogonal vectors $\bm{w}_1, \bm{w}_2$ with $\bm{w}_1 \cdot \bm{w}_2 =0$ is not generally preserved:
\begin{equation}
    \bm{w}_1^{\prime} \cdot \bm{w}_2^{\prime} = (A \bm{w}_1)^{\rm{T}} (A \bm{w}_1) \neq \bm{w}_1 \cdot \bm{w}_2
\end{equation}}
In particular, $\bm{n}_{c, E}$ is not generally normal to the Eulerian cusp sheet, and $\bm{t}_{c,E} \cdot \bm{t}_{f,E} \neq 0$ in general. However, the evolved vector $\bm{t}_{f, E}$ is still by construction tangent to the Eulerian auxiliary fold line, $A_{2, E}$, and $\{\bm{t}_{c, E}, \bm{t}_{f, E}\} = \text{span}(A_{3, E})$ still span the tangent plane of the Eulerian cusp sheet. To obtain an orthonormal basis of the cusp sheet, we therefore proceed as in the Lagrangian case and introduce a new basis vector $\bm{\tilde{t}}_{c, E}$ by rejecting $\bm{t}_{c, E}$ from $\bm{t}_{f, E}$,
\begin{equation}
    \bm{\tilde{t}}_{c, E} \equiv \bm{t}_{c, E} - \text{proj}_{\bm{t}_{f, E}} \bm{t}_{c, E} \,.
\end{equation}
The normal of the Eulerian cusp sheet is then obtained as
\begin{equation}
    \bm{\tilde{n}}_{c, E} = \bm{t}_{f, E} \times \bm{\tilde{t}}_{c, E} \,.
\end{equation}
The geometry of these orientation vectors is visualised in the lower panel of \cref{fig:orientation_sketches}.

The system $\{\bm{\tilde{n}}_{c, E}, \bm{\tilde{t}}_{c, E}, \bm{f}_{f, E}\}$ is now used to orient the Eulerian wall object. Again, applying the algorithm of \cref{app:euler_angles}, we calculate the Euler angles $(\alpha, \beta, \gamma)$ that rotate the vector basis into the desired orientation, which we again choose as $\{(1,0,0)^{\mathrm{T}}, (0,1,0)^{\mathrm{T}}, (0,0,1)^{\mathrm{T}}\}$. Using these Euler angles, the wall is oriented by rotating the Eulerian space coordinates $\bm{x}$ about the constraint point $\bm{x}_t(\bm{q})$ with the standard Euler matrix $R(\alpha, \beta, \gamma)$. Schematically,
\begin{equation}
    \bm{\tilde{x}} = R(\alpha, \beta, \gamma) \bm{x} \,.
\end{equation}
The rotated coordinates $\bm{\tilde{x}}$ are then translated such that the Eulerian constraint point $\bm{\tilde{x}}_t(\bm{q})$ be at the centre of the simulation box. Finally, the density and velocity fields, halo positions and other observables are evaluated in the transformed coordinates. However, we will for simplicity omit the tilde symbol and refer to the transformed coordinates simply as ``Eulerian coordinates'', with the meaning being clear from the context.

\subsection{Summary}
\label{subsec:recipe-summary}

\begin{figure*}
    \includegraphics[width=\textwidth]{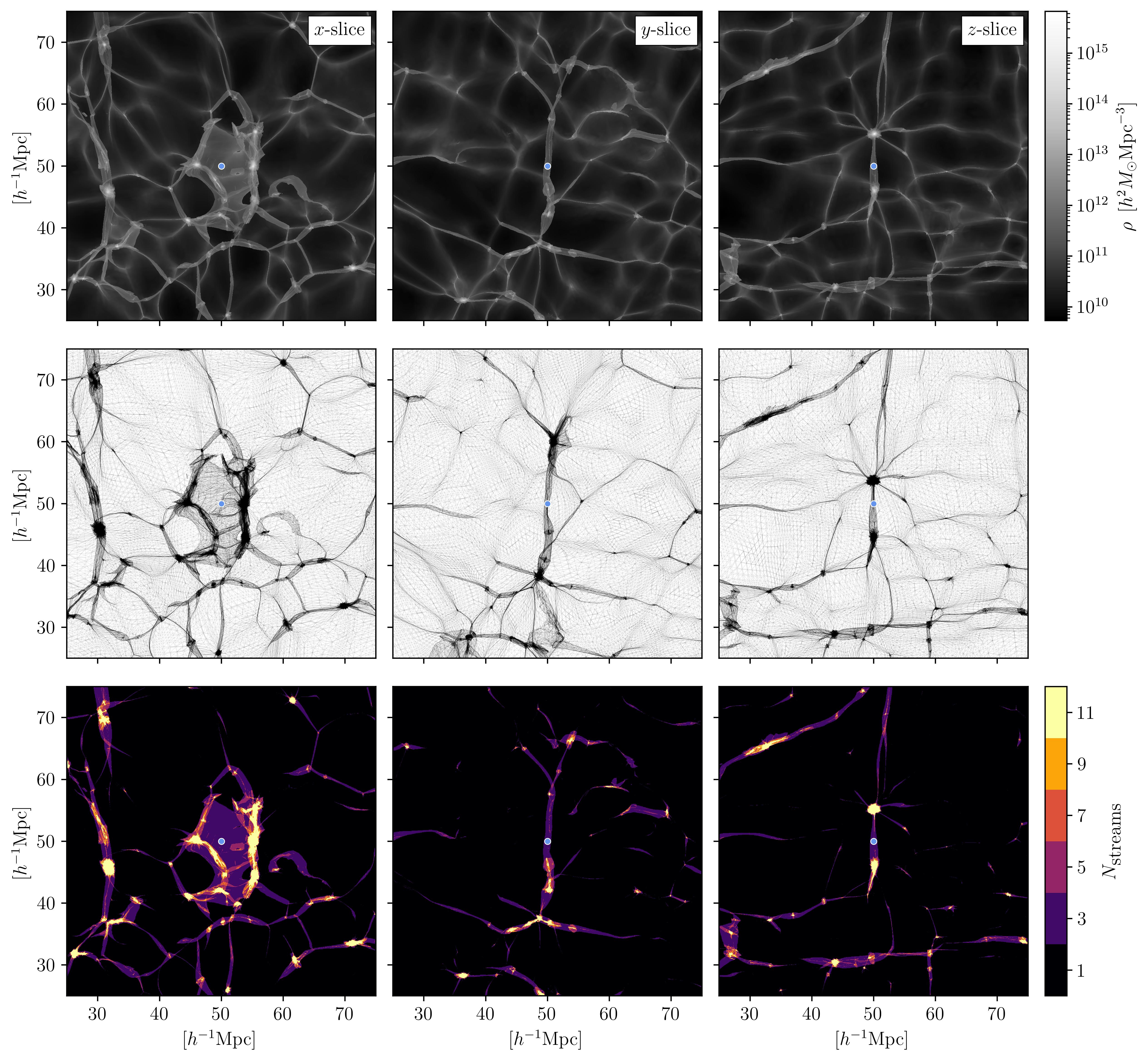}
    \caption{Slices through a high-resolution constraint simulation of cosmic wall formation, shown in 3D in \cref{fig:sim_256_A3_3D}. The constraint realisation was sampled for parameters $(\sigma, b_c) = (2.0\, h^{-1}\textrm{Mpc}, \,0.8)$. The left, middle and right panel correspond to the $x$-, $y$- and $z$-slices through the constraint point, slicing the wall along its face and through its major and minor axis respectively. The upper, middle and lower panel row show the density field, the folding particle mesh and the number of streams field respectively, with the constraint point displayed by the solid blue dot.}
    \label{fig:sim_256_A3}
\end{figure*}

We now summarise the previous sections into a rigorous and reproducible \textit{recipe for pancakes}. For a chosen pair of physical parameters $(\sigma, b_c)$, a cosmic wall centre constraint simulation is run as follows:
\begin{enumerate}
    \item Sample a large number of $A_3$ constraint realisations for parameters $(\sigma, b_c)$ using HMC.
    \item Generate an unconstrained primordial displacement potential $\Psi$.
    \item Orient a sampled constraint realisation $\bm{C} = (T_{11}, \ldots, T_{3333})$ in Lagrangian space. Impose the rotated constraint into $\Psi$ using the BHR algorithm to obtain the constrained potential  $\Psi_c$.
    \item Set $\Psi_c$ as the initial conditions for an $N$-body simulation and run the simulation to the current cosmological time.
    \item From the simulation output, determine the wall orientation Eulerian space. Evaluate the physical observables (density and velocity fields, halo positions, ...) in the transformed coordinates.
\end{enumerate}
When running a suite of simulations for fixed $(\sigma, b_c)$, step (i) is only done once and the different constraint realisations are drawn from the sampled chain $\{ \bm{C}_i\}$. 

Finally, \cref{fig:sim_256_A3} illustrates how the procedure detailed throughout this section results in a physically realistic cosmic wall centred around the constraint point. The $x$-, $y$- and $z$-slices clearly show a sheet-like overdensity that is accurately oriented in Eulerian space. The corresponding number of streams in bulk of the wall plane is three, in accordance with the cusp sheet construction. The particle mesh folding reveals that the overdense region indeed corresponds to an overlap of incoming mass stream. The simulated wall does not form a structureless object, but it is interspersed by filaments and embedded in the global cosmic web, thus reflecting the connectivity of the caustic network that we discussed in \cref{subsec:theory-caustic_skeleton}. Evidently, \cref{fig:sim_256_A3} confirms the validity of our procedure, which we shall now apply to systematic constrained simulation suites of cosmic wall formation.


\section{Constraint simulations of cosmic walls}
\label{sec:sims}
Using the recipe for pancakes, we study the properties of cosmic walls with a series of 3D dark-matter-only constrained simulations. We first evaluate the typical formation time of the walls for varying length scales. Subsequently, we study the constrained realisations and estimate the mean density and velocity fields.

\subsection{When do cosmic walls form?}
\label{subsec:sims-formation_time}

The caustic skeleton does not only provide a rigorous geometric understanding of the multistreaming structure, but also offers valuable insights into the statistical properties of the emerging cosmic web elements from the primordial conditions.

\subsubsection{Stochastic geometry}
\label{subsubsec:sims-stochastic_geometry}

Stochastic geometry is a branch of mathematics that studies the properties and behaviour of random geometric structures. Ever since the seminal study of Bardeen et al. (1986) \cite{Bardeen+1985} on the peaks in the primordial density field, the statistical properties and stochastic geometry \cite{Adler1981, AdlerTaylor2009} of random fields have been integral to theoretical cosmology. A crucial tool for these investigations is the famous Rice's formula \cite{Rice1945, Longuet-Higgins1957, Bardeen+1985, Adler1981, AdlerTaylor2009}, expressing the number density $\mathcal{N}$ of points satisfying the condition $\bm{f}(\bm{x}) = (f_1(\bm{x}), f_2(\bm{x}), f_3(\bm{x})) = \bm{0}$ as the expectation value
\begin{equation}
    \mathcal{N} = \left\langle  |\det \nabla \bm{f}| \delta^{(3)}_D(\bm{f})\right\rangle \,,
    \label{eq:Rice_formula}
\end{equation}
where the $\delta^{(n)}_D$ denotes the $n$-dimensional Dirac-$\delta$ function. 

In this article, we are not only interested in the density of points, but extend Rice's formula to the hypervolume density of ($3-k$)-manifolds defined by $k$ conditions on the primordial density field. Hypervolume densities have been featured in earlier studies \cite{Longuet-Higgins1957, NovikovColombiDore2006, Sousbie+2008, Pogosyan+2009} focusing on the length densities of curves in two- and three-dimensional embedding space. The derivation of the extended Rice's formula was however of an ad-hoc nature. We here give a general derivation of the extended Rice's formula for the hypervolume density of $(n-k)$-dimensional manifolds in $n$-dimensional embedding space. For a more detailed treatment of this extension of Rice's formula, we refer to the Azaïs \& Wschebor (2009) \cite{AzaïsWschebor2009}.

\begin{theorem}
    The hypervolume density $\mathcal{V}$ of a $k$-dimensional constraint $\bm{f}=\bm{0}$ in $n$-dimensional space is given by
    \begin{equation}
        \mathcal{V} = \left\langle |J_k \bm{f}| \delta_D^{(k)}(\bm{f})\right\rangle\,,
        \label{eq:generalised_Rice_formula}
    \end{equation}
    where the generalized Jacobian $|J_k \bm{f}|$ is defined as $|J_k \bm{f}(\bm{x})| = \left(\det J J^T \right)^{1/2}$ with $J = \nabla \bm{f} (\bm{x}).$
        
\end{theorem}

\begin{proof}
    Consider a vector-valued function $\bm{f}: \Omega \to \mathbb{R}^k$ defined on an $n$-dimensional compact manifold $\Omega \subset \mathbb{R}^n$. The set of points 
    \begin{align}
        \bm{f}^{-1}(\bm{0}) = \{\bm{x} \in \Omega\, |\, \bm{f}(\bm{x}) = 0\}
    \end{align} 
    form an $(n-k)$-dimensional variety with the hypervolume
    \begin{align}
        V_{\bm{f}} = \int_{\bm{f}^{-1}(\bm{0})} \mathrm{d}\bm{\sigma}
    \end{align}
    in terms of the induced metric $\mathrm{d}\bm{\sigma}$. Now, using the generalised coarea formula, 
    \begin{align}
        \int_\Omega g(\bm{x}) |J_k \bm{f}(\bm{x})| \mathrm{d}\bm{x} = \int_{\mathbb{R}^k}\left(\int_{\bm{f}^{-1}(\bm{t})} g(\bm{x}) \mathrm{d} \bm{\sigma}\right)\mathrm{d}\bm{t}\,,
    \end{align}
    for some scalar function $g:\Omega \to \mathbb{R}$ \cite{Federer1959}, the hypervolume $V_{\bm{f}}$ assumes the form
    \begin{align}
        V_{\bm{f}} = \int_\Omega |J_k \bm{f}(\bm{x})| \delta_D^{(k)}(\bm{f}(\bm{x})) \mathrm{d}\bm{x}\,,
    \end{align}
    upon the substitution $g(\bm{x}) = \delta_D^{(k)}(\bm{f}(\bm{x}))$.
    When the function $\bm{f}$ is a random field, the expectation value of the hypervolume is
    \begin{align}
        \langle V_{\bm{f}} \rangle &= \left\langle\int_\Omega |J_k \bm{f}(\bm{x})| \delta_D^{(k)}(\bm{f}(\bm{x})) \mathrm{d}\bm{x}\right\rangle\\
        &= \left\langle |J_k \bm{f}| \delta_D^{(k)}(\bm{f})\right\rangle V(\Omega)\,,
    \end{align}
    with the volume of the base space $V(\Omega)$. This proves the generalized Rice formula for hypervolume densities, as $\mathcal{V} = \langle V_{\bm{f}}\rangle / V(\Omega)$.
\end{proof}

In three-dimensional space, we find the volume densities:
\begin{enumerate}
    \item As the Jacobian $|J_3 \bm{f}| = |\det(\nabla f_1, \nabla f_2, \nabla f_3)|$, we recover Rice's formula for the point density defined by three conditions $\bm{f}(\bm{x}) = (f_1(\bm{x}), f_2(\bm{x}), f_3(\bm{x}))=\bm{0}$ 
    \begin{align}
        \mathcal{N} = \left\langle |\det \nabla \bm{f}| \delta_D^{(3)}(\bm{f})\right\rangle\,.
        \label{eq:sg_formula_point}
    \end{align} 
    \item As the Jacobian $|J_2 \bm{f}| = \|\nabla f_1 \times \nabla f_2\|$, the curve length density defined by two conditions $\bm{f}(\bm{x}) = (f_1(\bm{x}), f_2(\bm{x}))=\bm{0}$ is given by 
    \begin{align}
        \mathcal{L} = \left\langle \| \nabla f_1 \times \nabla f_2\| \delta_D^{(2)}(\bm{f})\right\rangle\,.
        \label{eq:sg_formula_line}
    \end{align}
    This equation was previously used in \cite{Sousbie+2008, Pogosyan+2009}. 
    \item As $|J_1 f| = \|\nabla f\|$, the area density defined by a single condition $f(\bm{x}) = 0$ is given by
    \begin{align}
        \mathcal{A} = \left\langle \|\nabla f\| \delta_D^{(1)}(f)\right\rangle\,.
        \label{eq:sg_formula_surface}
    \end{align}
    This formula has not yet been presented in earlier studies on the cosmic web.
\end{enumerate}

\subsubsection{Number density of wall centres}
\label{subsubsec:sims-stochastic_geometry-N}

Rice's formula is well known in cosmology and has been extensively to used to study critical points of (Gaussian) random fields, in particular with regard to maxima, minima and saddle points of the primordial potential perturbation $\phi$ and corresponding density perturbation $\delta$. With the understanding of the cosmic wall centre expressed as a special point on the cusp sheet, Rice's formula now allows us to derive the number density of the proposed $A_3$ centre constraint. We are interested in this quantity as a function of the two parameters $(\sigma, b_c)$ and find that the number density $\mathcal{N}_{\sigma}(b_c)$ of wall centres at scale $\sigma$ and emerging at growing mode $b_c$ is given by

\begin{equation}
    \begin{split}
    \mathcal{N}_{\sigma}(b_c) = \frac{1}{b_{c}^2}
    &\scaleleftright[1.75ex]{<}
{\left| \det \MyMatrix \right|  \vphantom{\begin{bmatrix}1\\1\\1\\1\end{bmatrix}} }
{}  \\[2pt]
    &  \quad\quad \quad \times \,\,\,  \scaleleftright[1.75ex]{.}
    {\begin{array}{c}
       \delta_D^{(1)}(\lambda_1 - b_c^{-1}) \,\,\, \delta_D^{(1)}(\bm{v}_2 \cdot \nabla (\lambda_2 + \lambda_3)) \,\,\,  \delta_D^{(1)}(\bm{v}_3 \cdot \nabla (\lambda_2 + \lambda_3)) \\[4pt]
       \quad  \Theta(-\lambda_2) \,\,\,  \Theta(\lambda_2-\lambda_3) \,\,\, \Theta(M \textrm{ pos. def.})
     \end{array} \vphantom{\begin{bmatrix}1\\1\\1\\1\end{bmatrix}}}
    {>} \,.
    \end{split}
    \label{eq:A3_Rice_formula}
\end{equation}

The functional determinant $\det \nabla \bm{c}$ weighs the configurations of the directional derivatives of the eigenvalue fields satisfying the caustic conditions that are imposed by the Dirac-$\delta$ functions. The prefactor $1/b_{c}^2$ arises as the Jacobian from the coordinate transformation $\lambda_1  \rightarrow b_c = \lambda_1^{-1}$. 

Due to the absolute value of the determinant, \cref{eq:A3_Rice_formula} cannot be calculated analytically. We therefore evaluate the expectation value numerically on a large number of constraint samples. To this end, we translate \cref{eq:A3_Rice_formula} into an expression of the eigenframe field derivatives $T_{ij\ldots k}$ of second through fourth order. To account for the possible orientations of the diagonalised frame, as discussed in \cref{subsec:eigenframe}, the functional determinant is multiplied by the Jacobian \cref{eq:Jacobian}.

Following \cite{FeldbruggeYanWeygaert2023}, the number density can then be evaluated by writing the constraint statistic, \cref{eq:constrained_statistic}, as $\bm{Y} = (\bm{Y}_1, \bm{Y}_2)$ with the free statistic $\bm{Y}_1$ and the conditioned statistic $\bm{Y}_2$. The latter is given by
\begin{equation}
    \bm{Y}_2 = (T_{11}, T_{12}, T_{13}, T_{23}, T_{111}, T_{222}, T_{223}, T_{233}, T_{333})
\end{equation}
and takes the constraint value $\bm{y}_2$:
\begin{equation}
    T_{11} = b_c^{-1},\, T_{12}=0,\, T_{13}=0,\, T_{23}=0, T_{111}=0,\, T_{222}+T_{233}=0,\, T_{223} + T_{333}=0 \,.
    \label{eq:constraint_value}
\end{equation}
The free statistic $\bm{Y}_1 = \bm{Y} \backslash \bm{Y}_2$ is the complement set of variables. Using the definition of conditional probabilities, the probability of the statistic $p_{\bm{Y}}(\bm{Y})$ is given by $p_{\bm{Y}}(\bm{Y}) = p_{\bm{Y}_1|\bm{Y}_2}(\bm{Y}_1|\bm{Y}_2) p_{\bm{Y}_2}(\bm{Y}_2) $.  Inserting this expression and the Jacobian $J$, \cref{eq:Jacobian}, into Rice's formula gives
\begin{equation}
    \langle J | \det(\ldots ) | \Theta(\ldots )  \delta_D(\bm{Y}_2 - \bm{y}_2) \rangle_{\bm{Y}} = p_{\bm{Y}_2}(\bm{y}_2) \langle J | \det(\ldots ) | \Theta(\ldots ) \,\big|\, \bm{Y}_2 = \bm{y}_2 \rangle_{\bm{Y}_1 | \bm{Y}_2} \,,
\end{equation}
where on the left side we have taken the expectation value over the unconditioned variables and on the right side we have evaluated the Rice's formula on the conditioned variables. We do this numerically on a large number of constraint realisations obtained as HMC samples of the constraint statistic given by \cref{eq:constrained_statistic}, see \cref{subsec:recipe-sampling}.

From the covariance matrix \cref{eq:deriv_covariance}, the reader can confirm that the constraint probability $p_{\bm{Y_2}}(\bm{y}_2)$ is given by
\begin{align*}
    p_{\bm{Y}_2}(\bm{y}_2) &= \int \mathrm{d} T_{11}\ldots \mathrm{d}  T_{3333} \, p_{\bm{Y_2}} (T_{11}, \ldots, T_{3333}) \, \delta^{(1)}_D(T_{12})\delta^{(1)}_D(T_{13})\delta^{(1)}_D(T_{23})  \delta^{(1)}_D(T_{111}) \\
    &\qquad\qquad\qquad \times \delta^{(1)}_D(T_{222}+T_{233})\delta^{(1)}_D(T_{223}+T_{333}) \\
    &= \sqrt{\frac{21}{2}} \frac{2625}{64 \pi^{7/2} \sigma_2^4 \sigma_3^3} e^{-\frac{5}{2 \sigma_2^2 b_c^2}} \,.
\end{align*}
Note that the constraint probability depends on both the value of the growing mode $b_c$ and the smoothing scale $\sigma$ at which the generalised moments $\sigma_i$ are evaluated.

\begin{figure}
    \centering
    \includegraphics[width=0.6\columnwidth]{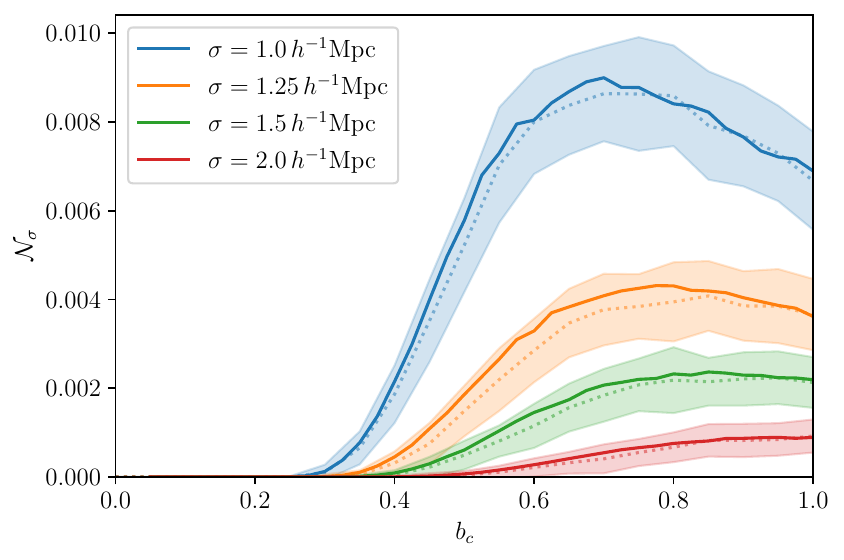}
    \caption{Number density of the $A_3$ centres as a function of $b_c$ for vayring $\sigma$, calulated from Rice's formula (solid lines), and evaluated from numerical measurements of random field realisations (dotted lines with $1\sigma$-bands). For the latter, the sample mean and variance were evaluated for 50, 75, 100 and 200 field realisations of a $100 \,h^{-1}\textrm{Mpc}$ box with $128^3, \,64^3,\, 64^3$ and $64^3$ particles for $\sigma = 1.0, \, 1.25,\, 1.5,\, 2.0 \,h^{-1}\textrm{Mpc}$ respectively.}
    \label{fig:A3_number_density}
\end{figure}

\Cref{fig:A3_number_density} shows the number density of the wall centre constraint as a function of the growing mode for a set of different smoothing scales  $\sigma = \{1.0,\, 1.25,\, 1.5,\, 2.0\}$. It is apparent that the Rice's formula evaluation is in excellent agreement with the measurement of $A_3$ centre points on random field realisations. The small systematic shift to higher Rice's formula values at earlier times may be attributed to missing detections of $A_3$ centre points from the numerical evaluation of the caustic skeleton on a finite-resolution grid. We have confirmed that higher grid resolutions ameliorate this effect.

The number densities are not constant, but, for each smoothing scale, exhibit a characteristic peak in formation time. These peaks corresponds to the growing modes $b_c$ at which most $A_3$ centre points are produced in random fields. Our calculations demonstrate that smaller-scale cusp caustics form earlier (significant number density at $b_c \approx 0.5$ for $\sigma=1.0\,h^{-1}\textrm{Mpc}$), whereas larger-scale cusp caustics form later (significant number density at $b_c \approx 0.8$ for $\sigma=2.0\,h^{-1}\textrm{Mpc}$). This is consistent with the hierarchical assembly of the cosmic web --- in particular the network of cosmic walls --- through the merger of smaller-scale elements. The characteristic formation peaks correspond to physically realistic parameter configurations $(\sigma, b_c)$  that we will impose into the constrained simulations in \cref{subsec:sims-fields}. Note that the amplitudes of the number density curves for different $\sigma$-values do not represent relative likelihoods of the constraint smoothing scale. The smaller values for large $\sigma$  simply reflect the fact that fewer large-scale walls fit into a unit cosmological volume.

The early formation peak of small-scale walls is consistent the Zel'dovich model \cite{Zeldovich1970} of pancakes forming as the first structural elements of the cosmic web. In the follow-up article on cosmic filament formation, we will calculate the number densities for the analogous centre points of the swallowtail and umbilic filaments. Doing so, we will quantitatively investigate the relative formation times of the walls and filaments making up the scale-space cosmic web.

\subsubsection{Area density of cosmic walls}
\label{subsubsec:sims-stochastic_geometry-A}

While the number density of the $A_3$ centre points is the key quantity to identify realistic configurations $(\sigma, b_c)$ for constraint simulations, the hierarchical build-up of the wall network is imprinted in the growth of wall area over cosmic time and length scales. For the first time, we now calculate the Lagrangian-space area density $\mathcal{A}_{\sigma}(b_c)$ of the cosmic walls produced at time $b_c$ for a fixed smoothing scale $\sigma$.

By identifying the walls as cusp sheets, and using \cref{eq:sg_formula_surface}, we find that the wall area density is given by
\begin{equation}
    \mathcal{A}_{\sigma}(b_c)= \left\langle \| \nabla \left( \bm{v}_1 \cdot \nabla \lambda_1 \right) \, \| \, \delta^{(1)}_D(\bm{v}_1 \cdot \nabla \lambda_1) \Theta \left( \lambda_1 - b_c^{-1} \right)\right\rangle \,,
\end{equation}
which we compute the from the statistics of the smoothed displacement potential $\Psi_{\sigma}$. As before, we evaluate the expression in the eigenframe and numerically calculate the expectation value on samples of the second up to fourth-order field derivatives.

\Cref{fig:A3_area} shows the area densities $\mathcal{A}_{\sigma}(b_c)$ from Rice's formula and numerical measurements on random field realisations. Clearly, the Rice's formula evaluations are consistent within $1 \sigma$ of the sample variance of the random field measurements. Again, we the systematic shift towards earlier peaks in the Rice's formula can be attributed to the numerical limitations of the cusp sheet identifications on the finite-resolution grid. As the cusp sheets are calculated as isocontours on the primordial fields, the emerging sheets can only be identified once their sizes exceed the grid resolution. Similarly, their area is overestimated by the numerical triangulation once the smooth sheets have grown to considerable sizes. We have confirmed that both issues are ameliorated by finer grid resolutions, at the cost of longer computation times.

\begin{figure}
    \centering
    \includegraphics[width=0.6\columnwidth]{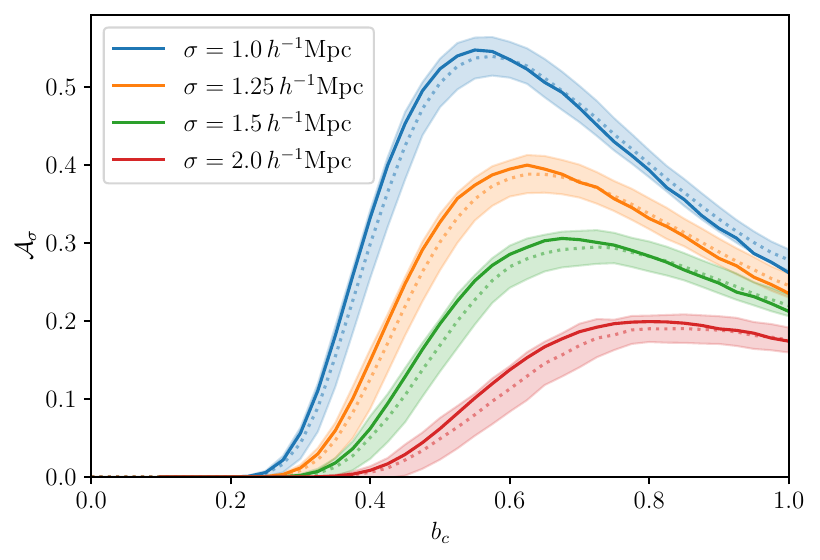}
    \caption{Area density produced by the Lagrangian $A_3$ sheets a function of $b_c$ for vayring $\sigma$, calulated from Rice's formula (solid lines), and evaluated from numerical measurements of random field realisations (dotted lines with $1\sigma$-bands). For the latter, the sample mean and variance were evaluated for 50, 75, 100 and 200 field realisations of a $100 \,h^{-1}\textrm{Mpc}$ box with $128^3, \,64^3,\, 64^3$ and $64^3$ particles for $\sigma = 1.0, \, 1.25,\, 1.5,\, 2.0 \,h^{-1}\textrm{Mpc}$ respectively.}
    \label{fig:A3_area}
\end{figure}

Again, we find that the area densities for different smoothing scales $\sigma$ each have a characteristic peak growing mode $b_c$. As expected, the curves are qualitatively similar to the number densities shown in \cref{fig:A3_number_density}. We find that the area accumulated by small-scale cusp sheets peaks at early times ($b_c = 0.5$ for  $\sigma=1.0\,h^{-1}\textrm{Mpc}$), whereas large-scale cusp sheets form at late times ($b_c = 0.7$ for  $\sigma=2.0\,h^{-1}\textrm{Mpc}$). Wall formation is negligible on cosmological scales at earlier times than $b_c \lesssim 0.2$. To our knowledge, this constitutes the first analytical quantification of the hierarchical assembly of the wall network from the primordial conditions. In the follow-up paper, we will analogously calculate the length density of the filaments defined the swallowtail and umbilic caustic conditions. Moreover, in future work, we plan to extend our calculation to the corresponding area and length densities in Eulerian space.

\subsection{Dark matter fields from cosmic wall simulations}
\label{subsec:sims-fields}

We implement the pancake recipe to run a suite of dark-matter-only constrained simulations of cosmic wall formation. To this end, we use the \verb|Gadget-4| \cite{SpringelPakmorZier+2021} code and set up simulations with $64^3$ particles in a box of size $100 \,h^{-1}\textrm{Mpc}$. We evolve the particles to the current cosmological time and infer the resulting density and velocity fields. We focus first on a single and physically realistic configuration $(\sigma, b_c) = (2.0 \,h^{-1}\textrm{Mpc},\, 0.8)$, which corresponds approximately to the peak of the number density from Rice's formula evaluated at the chosen length scale. To simulate a large number of random walls at this parameter point, we run a suite of 100 constraint simulations from randomly drawn constraint realisations imposed into primordial field realisations according to \cref{subsec:recipe-sampling} and \cref{subsec:recipe-HR}.

\subsubsection{Random wall realisations}

In \cref{fig:A3_realisations}, we show the density fields for three exemplary wall realisations of our suite of 100 simulations; the remaining realisations are available from the additional materials webpage \verb|benhertzsch.github.io/papers/2025_Cosmic_Walls|. As in \cref{fig:sim_256_A3}, the density fields are shown  in the $x$-, $y$- and $z$-slice, corresponding to the cut along the face of wall and the two directions through the wall's plane. Being physically realistic and random objects, the simulated walls are solid structures in the three-dimensional density field, and generally curve into or out of the sliced $x$-plane. This is to be considered when evaluating the PS-DTFE density field slices, as these are more sensitive to the small-scale curving than traditional DTFE estimates. However, at the chosen resolution of $64^3$ particles in a box of size $100 \,h^{-1}\textrm{Mpc}$, the curving of the wall does not significantly impact the density slices, and the walls appear as extended objects in the $x$-slice and line-like objects in the $y$- and $z$-slices. Clearly, this is consistent with the identification of accurately oriented sheet-like objects in the three-dimensional space, as was shown in \cref{fig:sim_256_A3_3D} with the constraint point being approximately at the visual centre of the cusp sheet. Moreover, it is apparent that the extent of the wall in the $y$-slice is generally larger than than in the $z$-slice (as will become clearer in the subsequent section on median fields). This is a consequence of the chosen ordering $\lambda_2 > \lambda_3$ and the according orientation of the oval-like wall sheet.

It is evident that the walls are overdense structures in the cosmic background. The typical density values are about $\rho \approx 10^{12} \,h^2 M_{\odot} \textrm{Mpc}^{-3}$ in the wall planes. The simulated objects are therefore about $10$ times as massive as the cosmic mean density of $\bar{\rho} \approx 8.3 \cdot 10^{10}\,h^2 M_{\odot} \textrm{Mpc}^{-3}$, which is somewhat larger than but of the same order of magnitude as the density contrast of $2\textrm{--}3$ reported in the literature on cosmic walls \cite{Forero-Romero+2009, AragonCalvo+2010,  ShandarinSalmanHeitmann2012, Hoffman+2012, Cautun+2014, Libeskind+2017}. Note that these studies are based traditional cosmic web identification methods and --- with the exception of \cite{ShandarinSalmanHeitmann2012} --- do not employ phase-space field estimators for the respective density field evaluations. Considering the sensitivity of the PS-DTFE method to the density divergences from small-scale caustics within the large-scale walls, we conclude that our results are consistent with the literature. It would be interesting to repeat the preceding analyses with the PS-DTFE method: We expect that by taking into account their multistreaming nature, the reported density contrasts of cosmic walls would be shifted to slightly larger values.

While we omit the corresponding number of streams fields here, we have checked that the walls are indeed multistreaming structures.  We find that the wall planes are mostly made of three-streaming volume, as was shown for a single high-resolution simulation in \cref{fig:sim_256_A3}. This confirms the successful simulation of a cusp sheet resulting in a multistreaming, sheet-like overdensity.

The interior structure of the simulated walls demonstrates that the objects do not exist in isolation of the other cosmic web elements. Instead, small-scale creases within the cusp sheet lead to the formation of filaments and clusters of different length scales, in accordance with the caustic skeleton theory we discussed in \cref{subsubsec:theory-caustics}. In particular, while the cusp sheets can in principle grow as three-streaming regions into the cosmological background, we observe that the majority of walls are bounded by massive filaments and clusters that emerge out of the cusp sheet into a higher-streaming configuration and hence accumulate mass. This is particularly clear in the second field realisation of \cref{fig:A3_realisations}, and is similarly seen in the high-resolution simulation of \cref{fig:A3_realisations}. In both cases, line-like caustics transverse the plane of the wall, and manifest themselves as highly overdense filaments ($\rho \approx 10^{13} \textrm{--} 10^{14} \,h^2 M_{\odot} \textrm{Mpc}^{-3} $) that can be identified as line-like structures in the $x$-slice. Within these filaments, massive node-like clusters can form, with density values reaching up to about $\rho \approx 10^{15}\, h^2 M_{\odot} \textrm{Mpc}^{-3} $. These structures naturally arise in our constraint simulations of wall formation because the different caustics form a structural network whose connectivity is intrinsic to the emerging cosmic web. This confirms the discussion of \cref{subsubsec:theory-caustics}: The walls are indeed real structures in the cosmic web, and their multistreaming volume forms the embedding for the higher collapse into filaments and clusters.

\begin{figure*}
    \includegraphics[width=\textwidth]{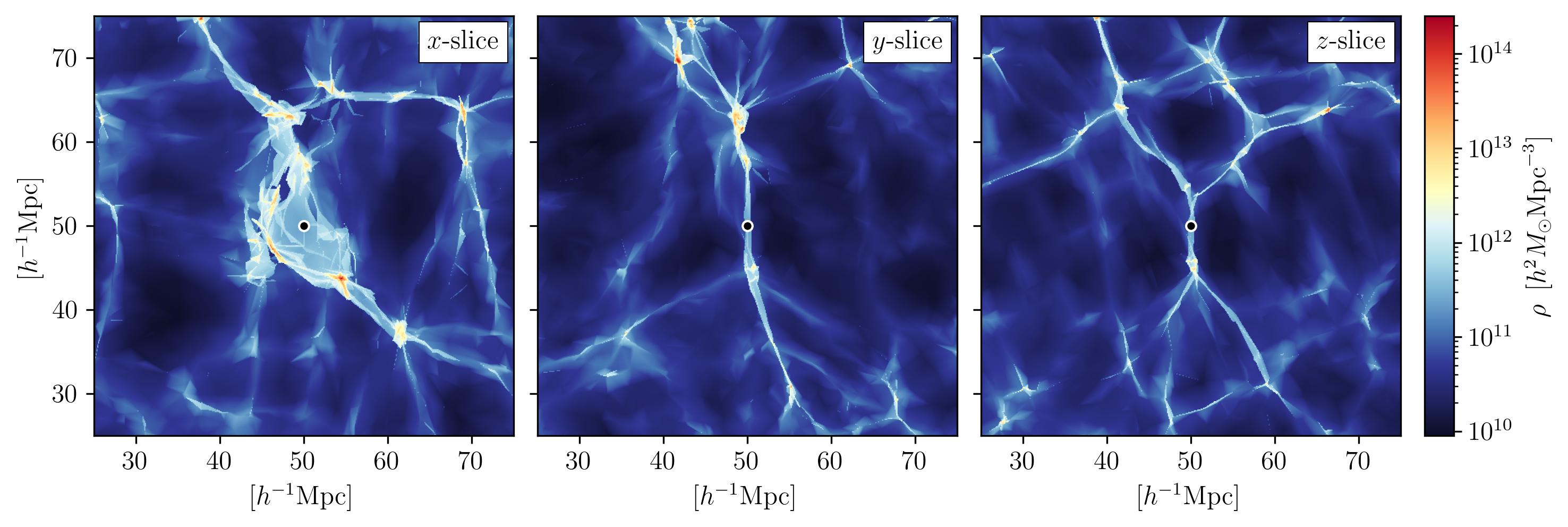}
    \includegraphics[width=\textwidth]{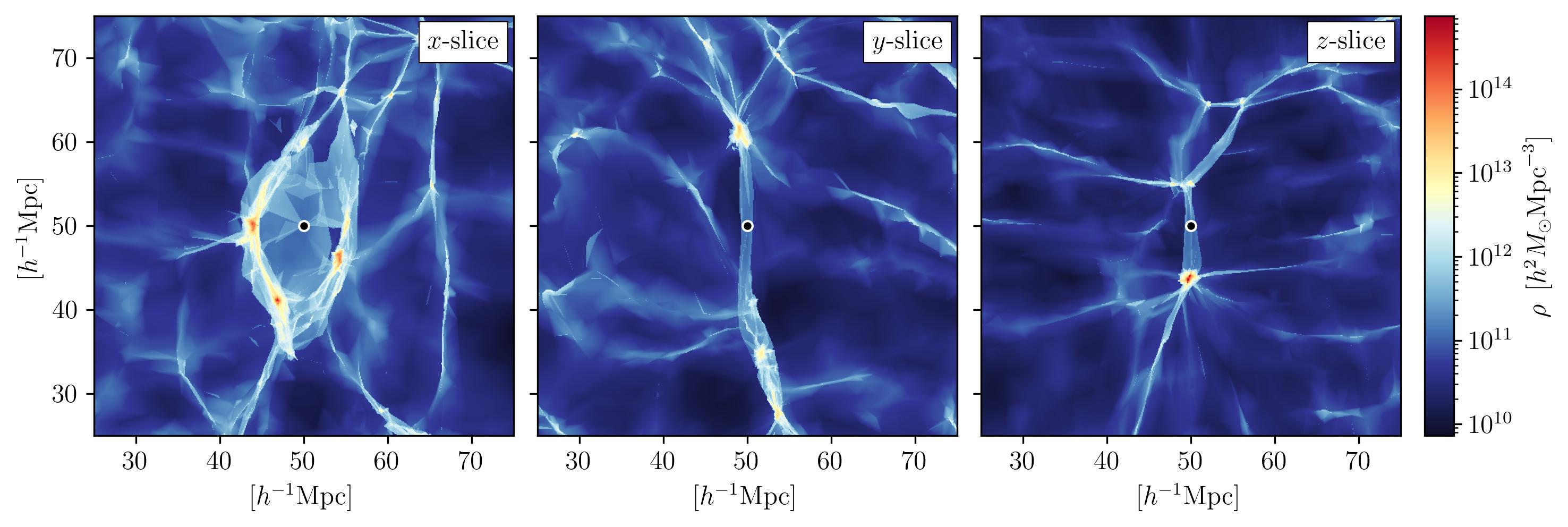}
    \includegraphics[width=\textwidth]{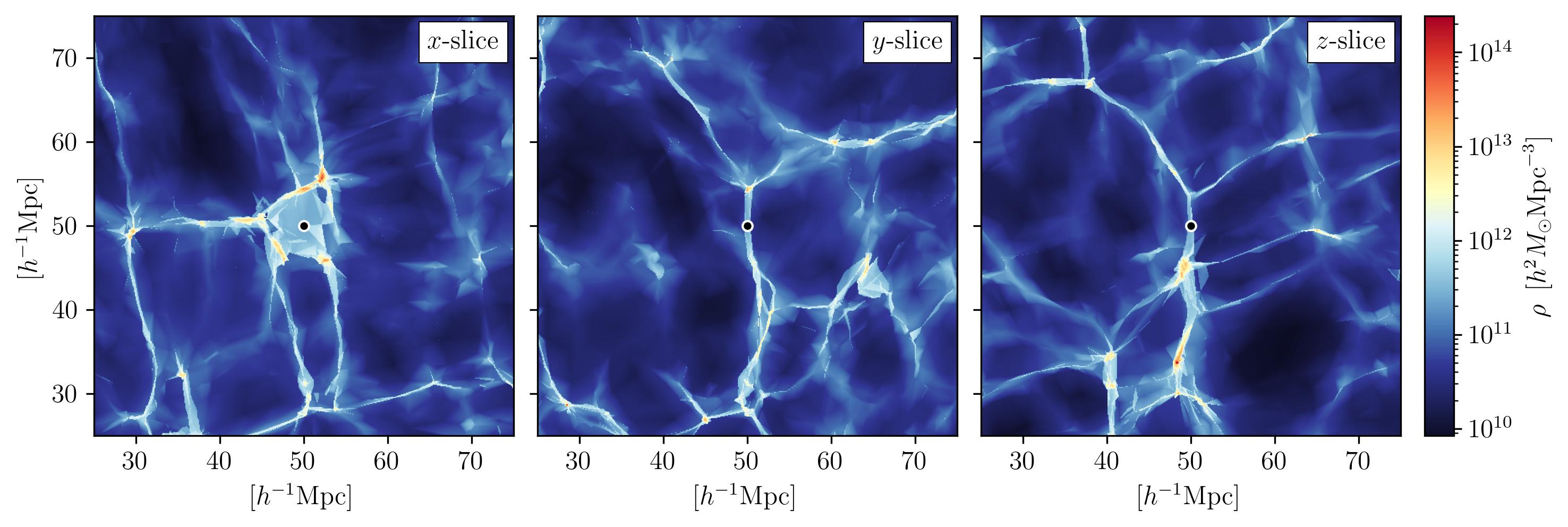}
    \caption{Exemplary density field realisations of the wall centre constraint for $(\sigma, b_c) = (2.0\, h^{-1}\textrm{Mpc}, \,0.8)$. The slices in the left, middle and right columns are as in \cref{fig:sim_256_A3}.}
    \label{fig:A3_realisations}
\end{figure*}

\begin{figure*}
    \includegraphics[width=\textwidth]{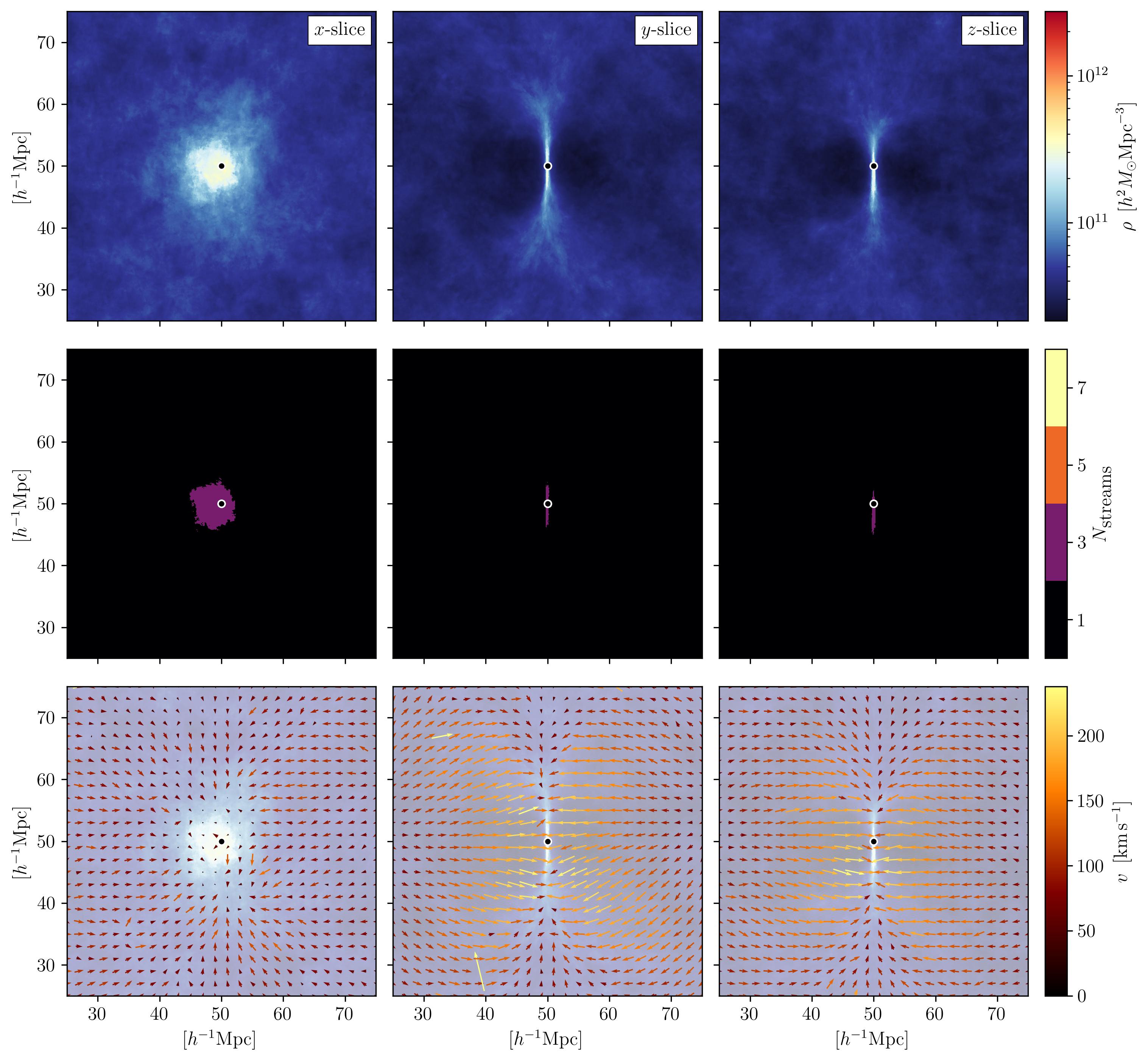}
    \caption{Median or mean fields for the $A_3$ wall centre constraint with $(\sigma, b_c) = (2.0\, h^{-1}\textrm{Mpc}, \,0.8)$. The upper row shows the median density field, the middle show shows the median number of streams field and the lower row shows the average velocity projected into the slicing plane and superimposed on the median density.}
    \label{fig:A3_constraint_mean_field}
\end{figure*}

\subsubsection{Median and mean fields}

We now study the characteristic median or mean fields inferred from the the random wall realisations in our simulation suite. The upper panel of \cref{fig:A3_constraint_mean_field} shows the median density fields, which we obtain by evaluating the median from the density fields for all random wall realisations. As was already remarked in \cite{FeldbruggeWeygaert2024}, the reason for taking the median is that naive average would be sensitive to  the density divergences at the caustics. The median is therefore a more appropriate measure of the characteristic density distribution. We find that the characteristic density distribution is in accordance with the imposed wall geometry, as the constraint results in a planar overdensity that separates two underdense voids. The characteristic extent for the constraint imposed at $\sigma = 2.0 \,h^{-1}\textrm{Mpc}$ is about $10 \,h^{-1}\textrm{Mpc}$. This illustrates that the proposed wall centre constraint influences the local environment beyond the smoothing scale, which is in contrast to traditional constraints in the primordial density perturbation that we will discuss in \cref{subsec:alternatives-density}. The median fields demonstrate clearly that the wall's extent in the $y$-slice is larger than in the $z$-slice. As remarked earlier, this reflects the ordering $\lambda_2>\lambda_3$ of the eigenvalue fields and the accurate orientation of the resulting oval-like cosmic walls. We find that the median density in the wall plane is about $\bar{\rho}_{\textrm{wall}} \approx 1.5 \cdot 10^{11} \,h^2 M_{\odot} \textrm{Mpc}^{-3}$, which is in quantitative agreement with the random realisations shown in \cref{fig:A3_realisations}. With the mean cosmic density being about $\bar{\rho} \approx 8.3 \cdot 10^{10} \,h^2 M_{\odot} \textrm{Mpc}^{-3}$, the averaged wall density is therefore about 2 times that of cosmic background, which is consistent with the literature on the density of cosmic walls in the $\Lambda$CDM cosmic web \cite{Forero-Romero+2009, AragonCalvo+2010,  ShandarinSalmanHeitmann2012, Hoffman+2012, Cautun+2014, Libeskind+2017}.

The corresponding median number of streams is shown in the second row of \cref{fig:A3_constraint_mean_field}. It is evident that the mean cosmic wall is a three-stream region, with the characteristic extent being again about $10 \, h^{-1}\textrm{Mpc}$, in agreement with the median density field. The three-streaming region is planar and thus consistent with the construction of a cusp sheet in the (on average) single-streaming cosmic background. Note that more overdense, smaller-scale structures within the wall, as observed in \cref{fig:sim_256_A3} and \cref{fig:A3_realisations}, generally correspond to higher-streaming configurations. These are not captured by the median number of streams field, which confirms our construction of random realisations of typical (i.e. three-streaming) cusp sheets. We will come back to the median number of streams field and discuss analogous arguments in \cref{sec:alternatives}.

The lower panel of \cref{fig:A3_constraint_mean_field} shows the mean velocity field, which we obtain by taking the mean of the PS-DTFE stream-mass weighted velocity fields for each simulation in our suite. Note that while the PS-DTFE velocity estimate for a single simulation is trivial in the single-streaming regions, the situation is more complicated in the multistreaming regions, where several streams with different velocities coincide. By evaluating the stream-mass weighted mean of these stream velocities, we obtain a physically realistic estimate for the overall dark matter velocity in the multistreaming regions. Surprisingly, we find that the mean velocity field inferred from all the simulations in our suite is in remarkable agreement with the imposed planar wall geometry. One could have expected that turbulence from phase mixing in the cosmic web makes the velocity fields much more sensitive to shell-crossing, thus spoiling any significant structure in the mean fields. Nevertheless, the mean velocity exhibits a clear geometric pattern, with turbulence contributing only to minor noise and a few outliers from evaluations close to caustics. In the $x$-slice of \cref{fig:A3_constraint_mean_field}, particles move with moderate velocities of about $50 \,\textrm{km}\,\textrm{s}^{-1}$ concentrically towards the wall. The central overdensity acts as a sink absorbing matter from the cosmic background. Within the shell-crossed walls, the multistreaming particles move more randomly due to phase mixing and the contributions of various matter streams. This manifests itself in random and close to vanishing mean velocity vectors in the left panel of the figure. The consistency of the velocity field with the imposed wall geometry is most prominently seen in the $y$- and $z$-slices. Here, the flow exhibits a clear dipolar pattern, with the matter being transported from the voids into the planar overdensity. The flow is particularly oriented towards the cusp points out of which the multistream region emerges. The mean velocity normal to the wall is higher than within the wall's plane (i.e. in the $x$-slice), with the void matter being transported towards the wall at up to $200 \,\textrm{km}\,\textrm{s}^{-1}$, consistent with previous studies on transport in the cosmic web, e.g. \cite{Cautun+2014}. In accordance with the oval-like shape of the wall seen in the mean density field, the high-velocity region in the $y$-slice reaches further than in the $z$-slice. However, in all slices, it is surprising to see that while the magnitudes of the velocities falls off away from the wall, the mean cosmic flow is still oriented in the dipolar pattern towards the central region. This observation demonstrates that cosmic walls influence their tidal environment well beyond the extent of the overdense multistream region.

\begin{sidewaysfigure}
    \centering
    \includegraphics[width=\textwidth]{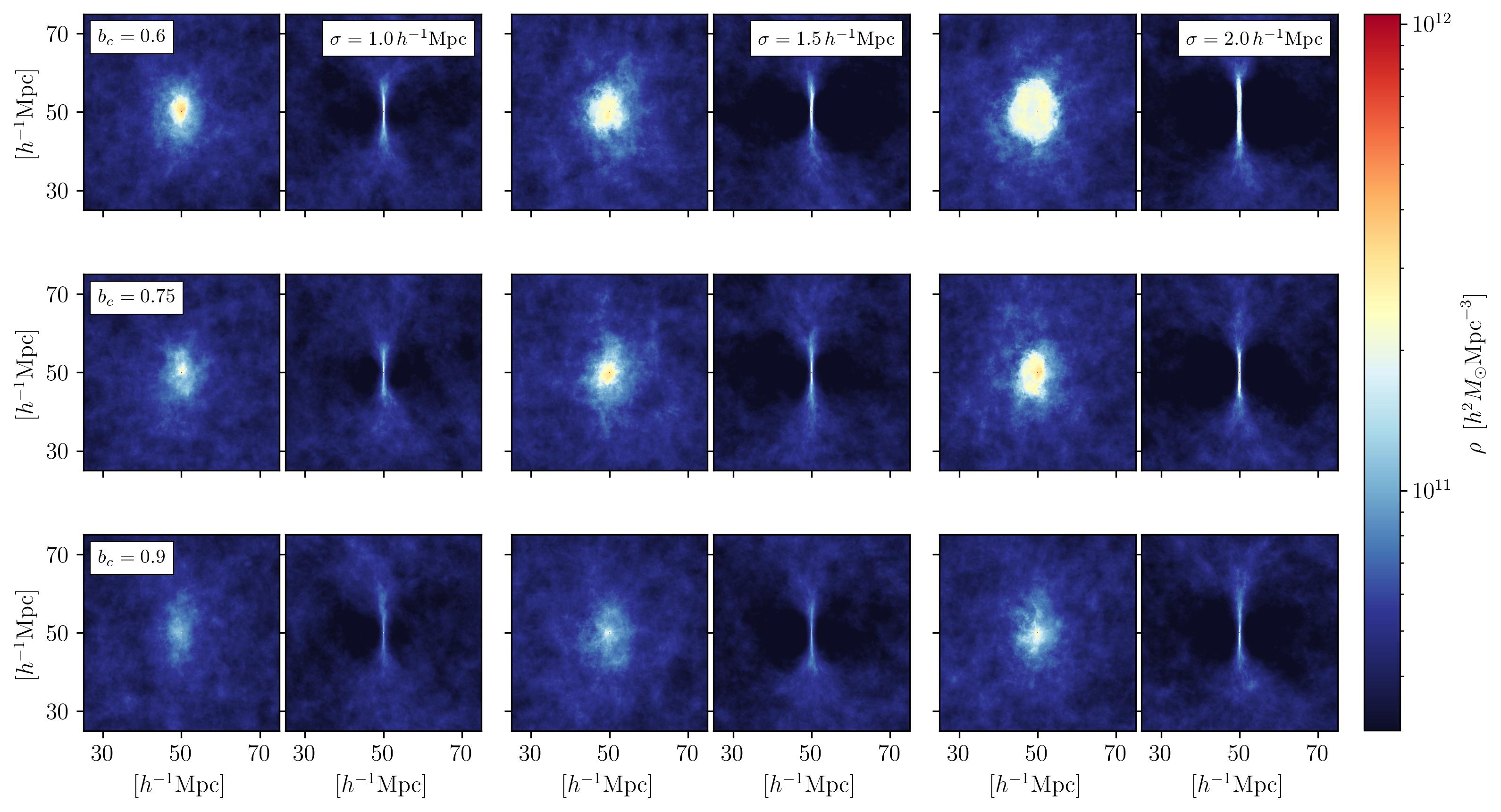}
    \caption{Median density fields for the $A_3$ wall centre constraint for varying parameters $(\sigma, b_c)$. The upper, middle and lower row correspond to $b_c = 0.6, \,0.75,\, 0.9$ respectively and the left, middle and right column correspond to $\sigma=1.0,\, 1.5,\, 2.0 \, h^{-1}\textrm{Mpc}$ respectively.}
    \label{fig:A3_constraint_parameter_space}
\end{sidewaysfigure}

\subsubsection{Walls in parameter space $(\sigma, b_c)$}

Having discussed the mean fields of a suite of constrained wall simulations for parameters $(\sigma, b_c) = (2.0\,h^{-1}\textrm{Mpc}, \,0.8)$, we conclude this section with a discussion of walls in the parameter space $(\sigma, b_c)$, corresponding to the length scale and formation time of the wall. To this end, we vary the smoothing scale $\sigma = \{1.0\,h^{-1}\textrm{Mpc}, \,1.5\,h^{-1}\textrm{Mpc},\, 2.0\,h^{-1}\textrm{Mpc}\}$ and the growing mode $b_c=\{ 0.6, \, 0.75, \, 0.9\}$. For each combination $(\sigma, b_c)$, we run a suite of 100 simulations and evaluate the median density field, the $x$- and $y$-slices of which are plotted in the panels of \cref{fig:A3_constraint_parameter_space}. From the Rice's formula calculations of \cref{subsec:sims-formation_time}, the configurations on the diagonal from the upper left to the lower right correspond approximately to the most likely parameter configurations, in the sense of peaking number densities. We make the field realisations for these available on the additional materials webpage.

The mean fields in \cref{fig:A3_constraint_parameter_space} reveal that the wall morphology is indeed dependent on the constraint parameters. At fixed $b_c$, constraints imposed on a larger smoothing scale $\sigma$ result in more extended overdensities. Similarly, at fixed $\sigma$, later formation times $b_c$ result in less dense objects. This is because shortly after the shell-crossing event creating the cusp sheet, the wall is only moderately overdense. With more time elapsing after the wall formation, matter inflow and phase mixing lead to higher overdensities within the planar structure. This is particular pronounced in the upper left panels with parameters $(\sigma, b_c)=(1.0 \,h^{-1}\textrm{Mpc},\, 0.6)$. With the constraint imposed at a small smoothing scale, matter clusters on smaller scales, with phase mixing resulting in high overdensities. The walls formed at $(\sigma, b_c)=(2.0 \,h^{-1}\textrm{Mpc},\, 0.6)$ in turn are not as overdense, but significantly thicker, which is again understood by the inflow of matter onto the planar structure. However, despite these consistent observations, the partial degeneracy between $\sigma$ and $b_c$ makes it challenging to relate the parameter space of inferred wall density fields to traditional cosmic web identification methods such as \verb|DisPerSE| \cite{Sousbie2011, Sousbie2011b} or \verb|NEXUS(+)| \cite{AragonCalvo+2007, Cautun+2012}. Being morphology filters, these do not take into account the formation times of the cosmic web elements, but only classify their observed density fields on different length scales. We leave an application of such methods to our constrained simulations to future studies.


\section{Haloes embedded in cosmic walls}
\label{sec:haloes}

In the preceding section, we have studied cosmic walls from an unprecedented analytical point of view using constrained simulations. Our results on their formation times and characteristic dark matter density and velocity fields significantly enhance our understanding the build-up of the large-scale wall network. Yet, in their physical reality, the cosmic walls are notoriously difficult to observe in cosmic surveys. The reason for this is that the low surface density of cosmic walls (as opposed to the highly dense filaments and ultra-dense clusters) induces the formation of low-luminosity galaxies, which are challenging to detect over the luminous background of the cosmic web. Indeed, numerous numerical studies \cite{Hahn+2007, Cautun+2014, Metuki+2015, AlonsoEardleyPeacock2015, MetukiLibeskindHoffman2016, Libeskind+2017} found that the haloes residing within cosmic walls cover a lower mass range than those residing in the more overdense filaments in clusters, and are hence populated by fainter galaxies \cite{Metuki+2015}. Using the caustic skeleton formalism, we now revisit these analyses and systematically address the physical reality of cosmic walls from the distribution and masses of their embedded dark matter haloes. Using a suite of high-resolution constrained simulations, for the first time, we relate the properties of the haloes to the the substructure of the caustic skeleton within the large-scale cosmic web. As we shall see, the formalism yields a theoretical foundation for the previous numerical observations, and encourages future studies of galaxy formation in the Zel'dovich pancakes.

We run a suite of 20 high-resolution constrained simulations with $256^3$ particles in a box of $50 \, h^{-1} \textrm{Mpc}$ side length each.
With this resolution, each simulation particle has a mass of $6.4 \cdot 10^8 \, h^{-1} M_{\odot}$. We identify the haloes with \verb|Gadget-4|'s \verb|Subfind| algorithm \cite{SpringelPakmorZier+2021}, which extends the traditional friends-of-friends (FoF) method by identifying gravitationally bound (sub-)haloes within the FoF groups. This choice is appropriate for our study, as we are primarily interested in gravitationally bound small-mass objects, rather than extended matter clumps. Consisting of a few tens to order $10^6$ bound particles, the identified haloes cover a mass range of about $10^{10}\textrm{--}10^{14} \,h^{-1
} M_{\odot}$, which is well within the intermediate to upper range of relevant halo masses for astrophysical studies.

\subsection{Filamentary patterns of haloes in walls}

\begin{figure*}
    \centering
    \includegraphics[width=0.9\textwidth]{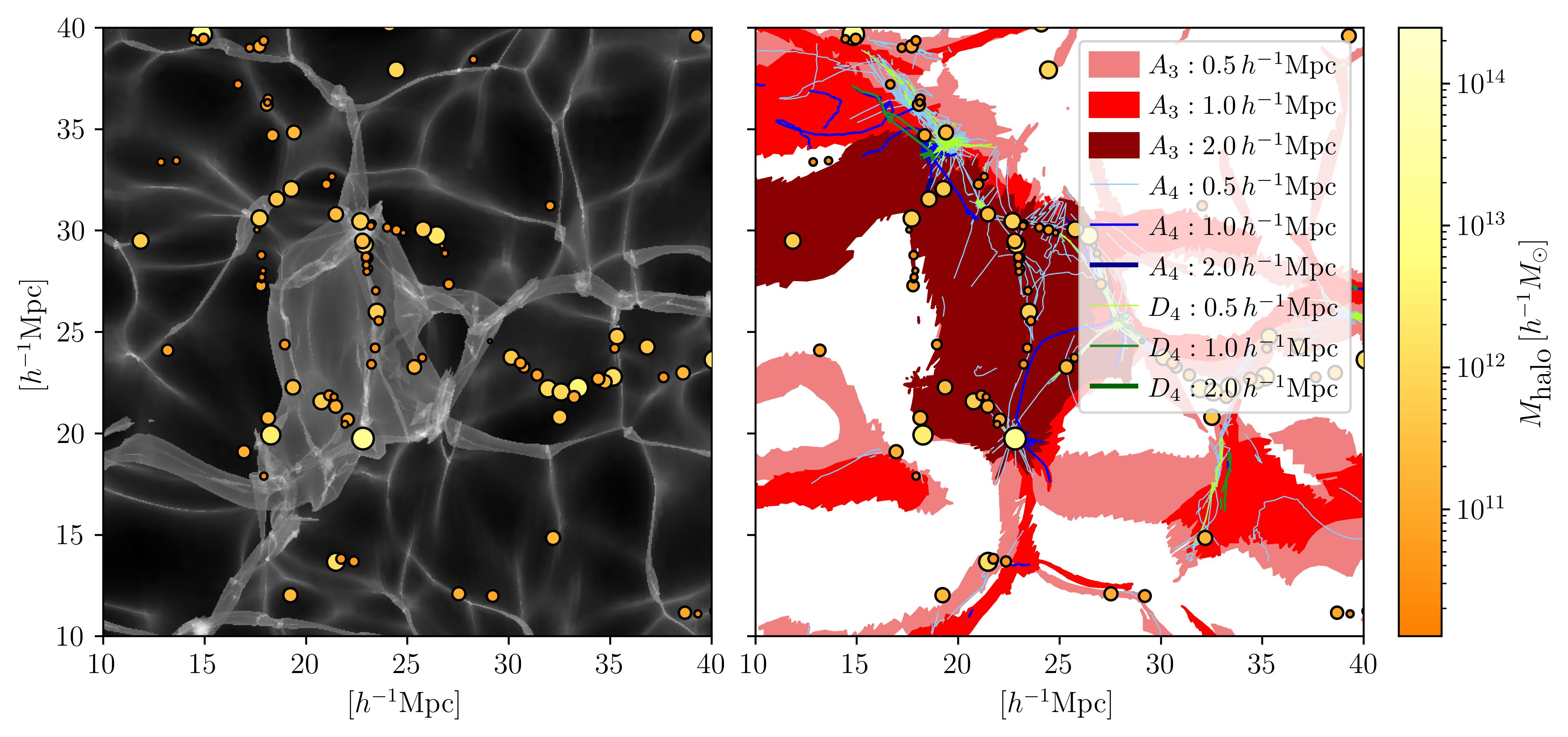}
    \includegraphics[width=0.9\textwidth]{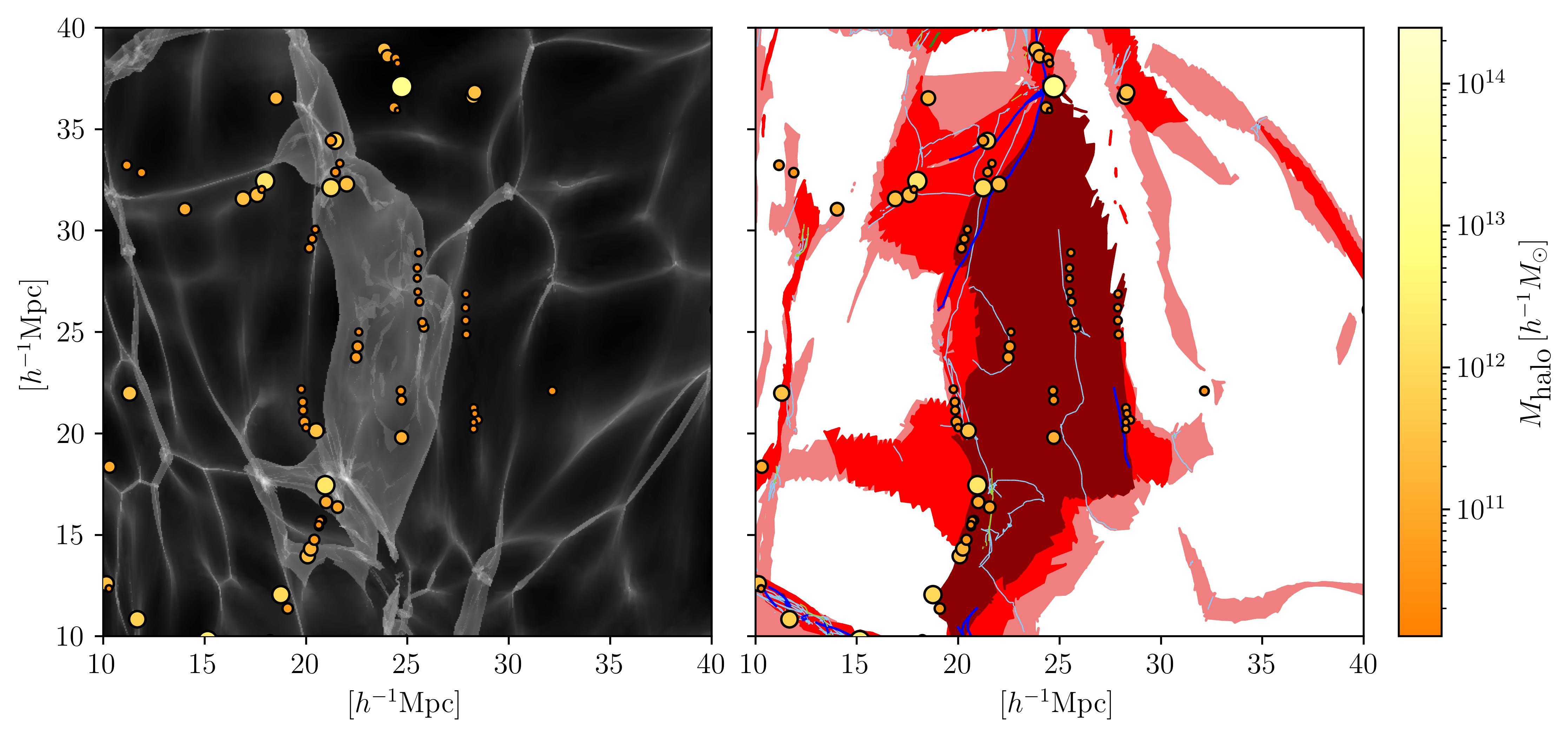}
    \includegraphics[width=0.9\textwidth]{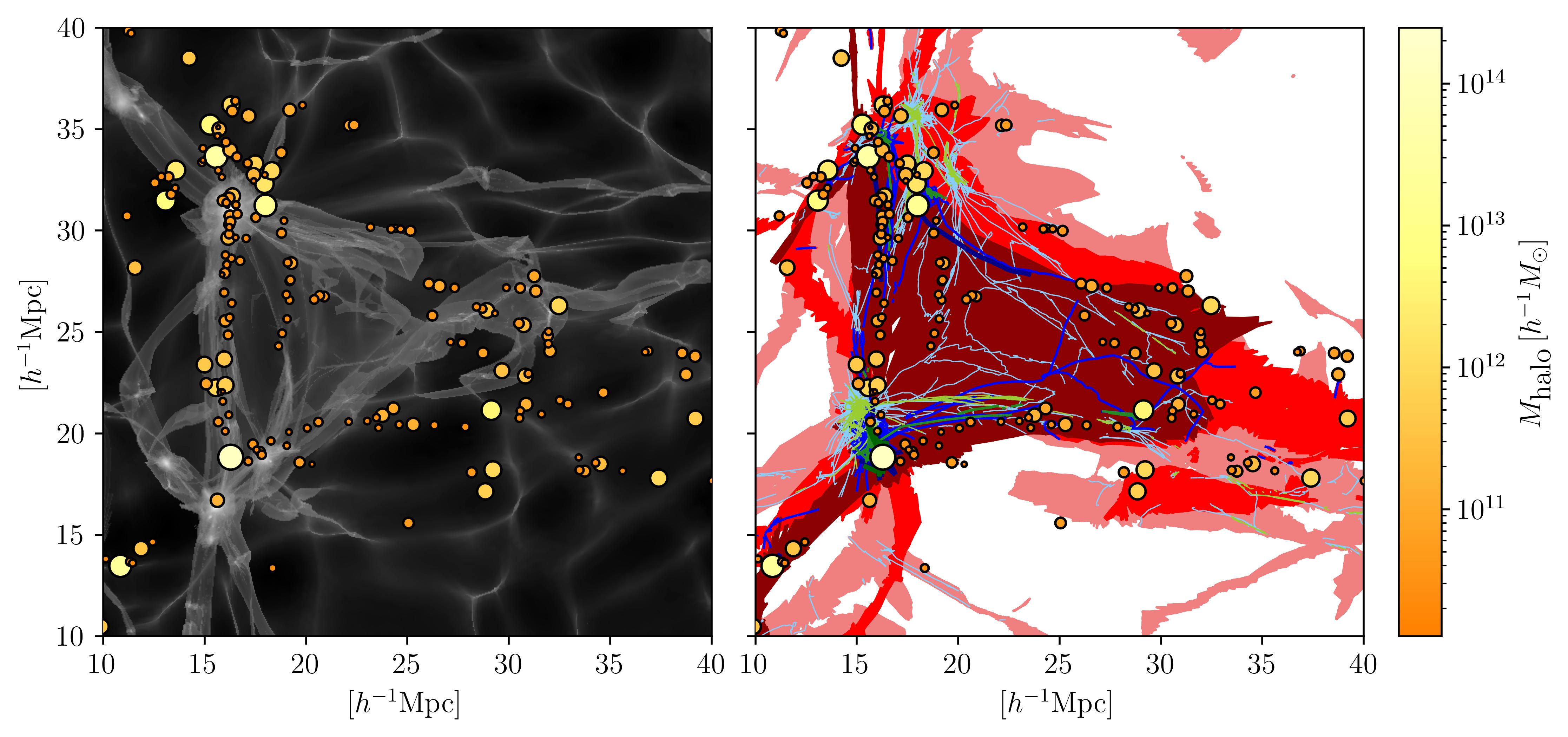}
    \caption{Exemplary realisations of high-resolution constrained simulations of cosmic walls along with their embedded dark matter haloes (same as \cref{fig:haloes_introduction}). The left panel shows the density field sliced along the face of the wall ($x$-slice), superimposed by the haloes identified in a thin volume ( (thickness $\epsilon = 1.5 \,h^{-1}\textrm{Mpc}$) around the wall and projected into the same slice. The caustic skeleton, projected from the same thin volume, is shown in the right panel, with the $A_3$ cusp walls in red, the $A_4$ swallowtail filaments in blue and the $D_4$ umbilic filaments in green respectively. The caustics are evaluated at different length scales $\sigma$.}
    \label{fig:haloes}
\end{figure*}

 \Cref{fig:haloes} displays three exemplary realisations of the cosmic walls with their embedded dark matter haloes. The density fields in the left panels show the face-on slices through the walls. Note that due to the sharpness of the PS-DTFE density field, the infinitesimal density slice is highly sensitive to the wall curving into or out of the plane. The perforated appearance of the walls is an artifact of this slicing, and the solid (yet curved) structure of the wall would be apparent in a three-dimensional visualisation. The right-hand panels show the corresponding caustic skeleton for different smoothing scales $\sigma$. The caustics are evaluated over a thin volume around the $x$-slice and projected into the slicing plane. This way, the cusp sheet appears a plane-like structure and the filaments traversing the wall (i.e. being approximately parallel to the slice) appear as line-like structures. Note that this visualisation is different from the caustic slices shown in \cref{sec:theory}. The solid cusp sheets confirm that the simulated walls are in fact solid objects, with the density field from the infinitesimal slices being in sufficient agreement.

Both the left- and right-hand panels of \cref{fig:haloes} also display the haloes identified in the same thin volume and projected down into the slicing plane. As expected, we find that the distribution of haloes is not random within the simulation volume, but highly dependent on the imposed wall geometry. Firstly, it is apparent that the projected halo positions within the wall are in approximate agreement with the density field, keeping in mind the infinitesimal density slicing. However, the haloes are not evenly distributed over the wall area, but rather follow filamentary patterns of varying extents. Numerous haloes are found near the large-scale filaments and superdense clusters that bound the sheet-like overdensity. Within the wall plane, the haloes appear in groups that extend over elongated regions. This distribution is explained by the smaller-scale caustics, which reveal that the cusp sheets are not structureless, but characterised by numerous small-scale creases. These correspond to the line-like swallowtail and umbilic caustics evaluated at small smoothing scales, e.g. $\sigma = 0.5 \,h^{-1}\textrm{Mpc}$. This observation is naturally understood in the caustic skeleton formalism, which predicts that the near scale-invariant primordial fluctuation spectrum seeds small-scale structures within the extended cosmological objects, as was illustrated in \cref{fig:sim_256_caustics}. The intrinsic connectivity of the scale-space caustic network is physically manifested by the emergence of overdense filaments and embedded haloes within the moderately overdense wall sheets. As we discussed in \cref{subsubsec:theory-caustics}, this observation enhances former notions of connectivity that are based solely on the multiplicity of filaments and clusters, notably \cite{CodisPogosyanPichon2018}.

We note here that the caustics at $\sigma=0.5  \,h^{-1}\textrm{Mpc}$ are numerically evaluated at a resolution close to the simulation grid size. The zigzag-like structure observed for some of the small-scale filaments is an artifact of this resolution limit. In future work, we plan to repeat our analysis with even higher-resolution simulations to calculate the small-scale caustics more smoothly, and to identify the lower mass haloes embedded in the same. Despite the numerical limitations, our constrained simulations clearly reveal, for the first time, the mathematical foundation of the filamentary halo patterns within cosmic walls.

\begin{figure*}
    \centering
     \centering
    \begin{subfigure}[b]{\textwidth}
        \includegraphics[width=\textwidth, clip, trim={0 0cm 0 0.5cm}]{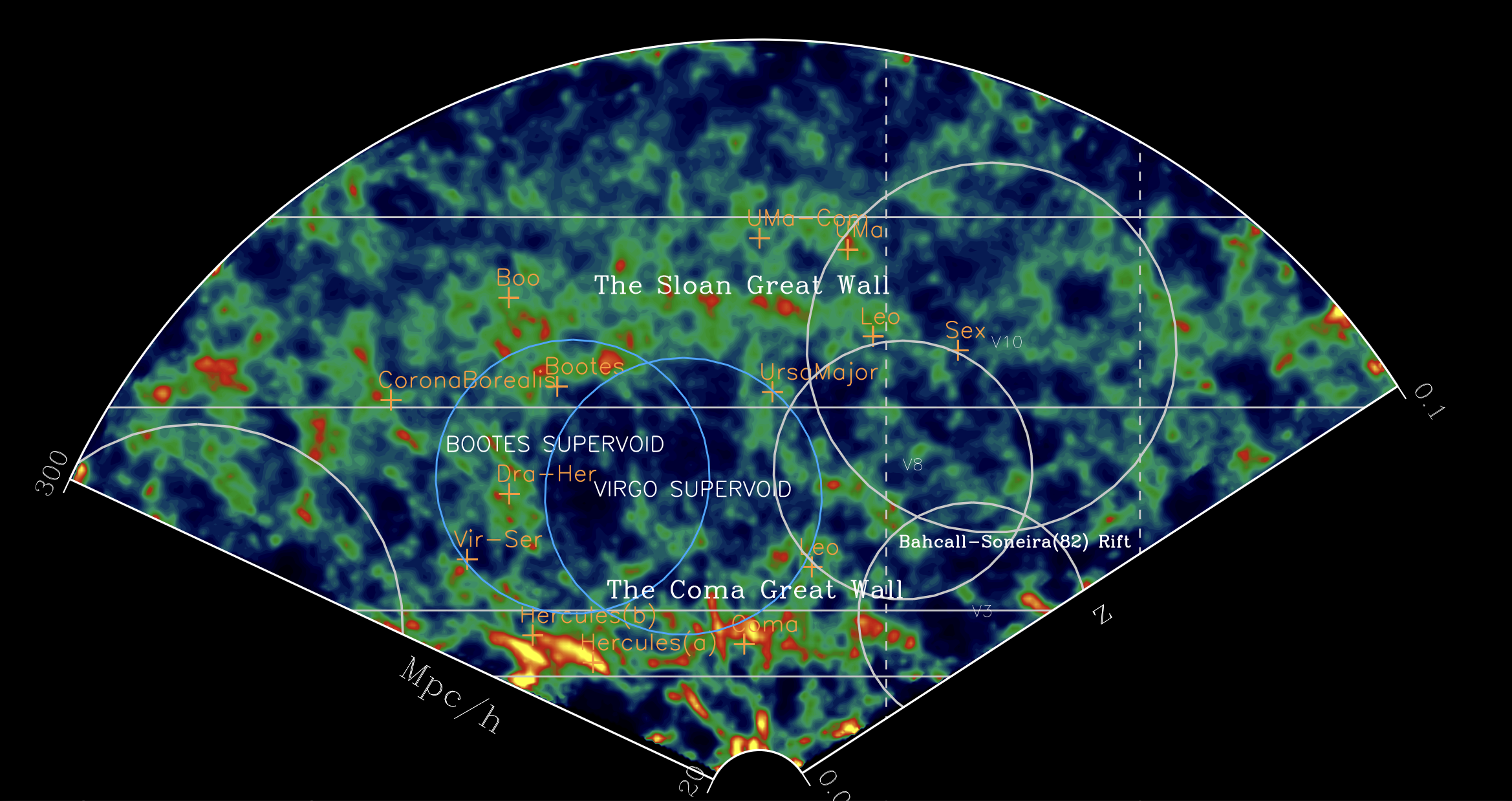}
    \end{subfigure}\vspace{0.5em}
    \begin{subfigure}[b]{\textwidth}
        \includegraphics[width=\textwidth, clip, trim={0 2cm 0 2cm}]{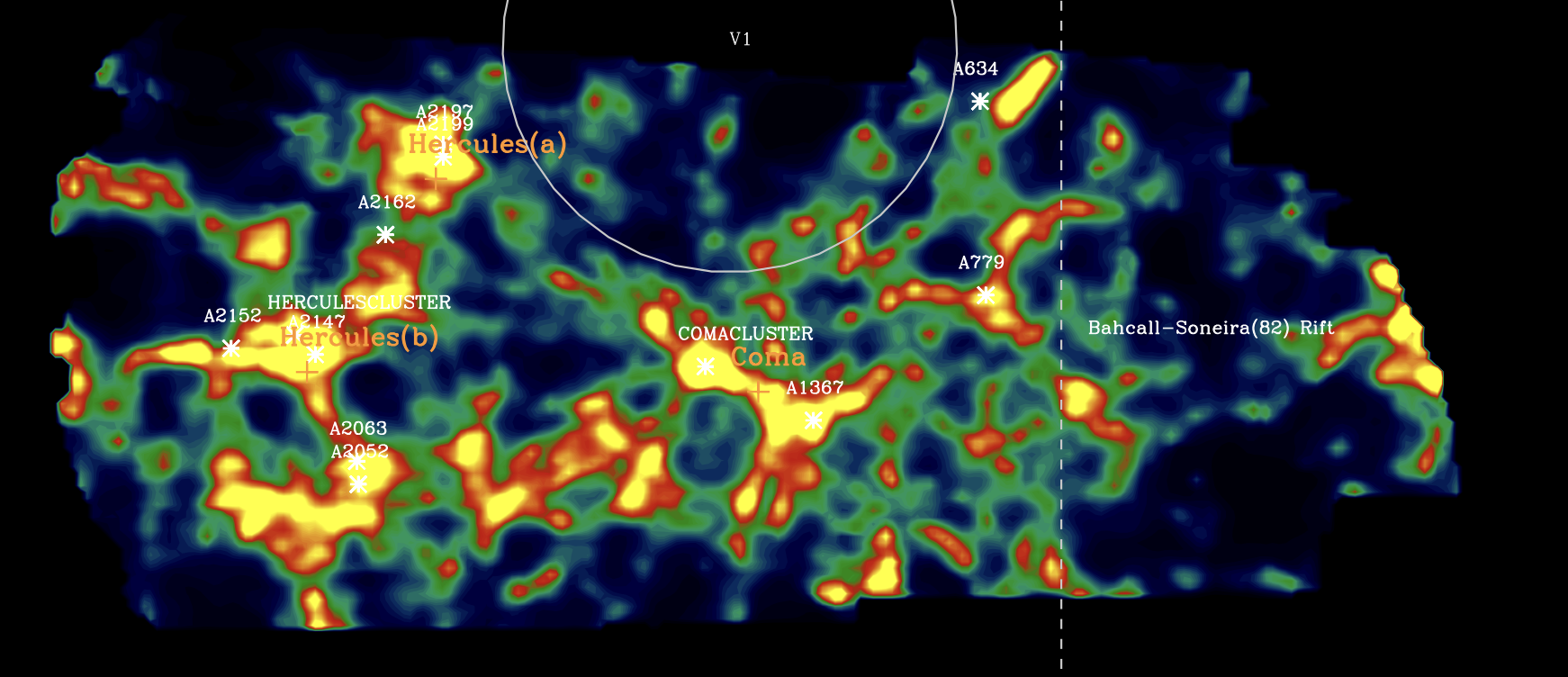}
    \end{subfigure}\vspace{0.5em}
    \begin{subfigure}[b]{\textwidth}
        \includegraphics[width=\textwidth, clip, trim={0 1.6cm 0 5cm}]{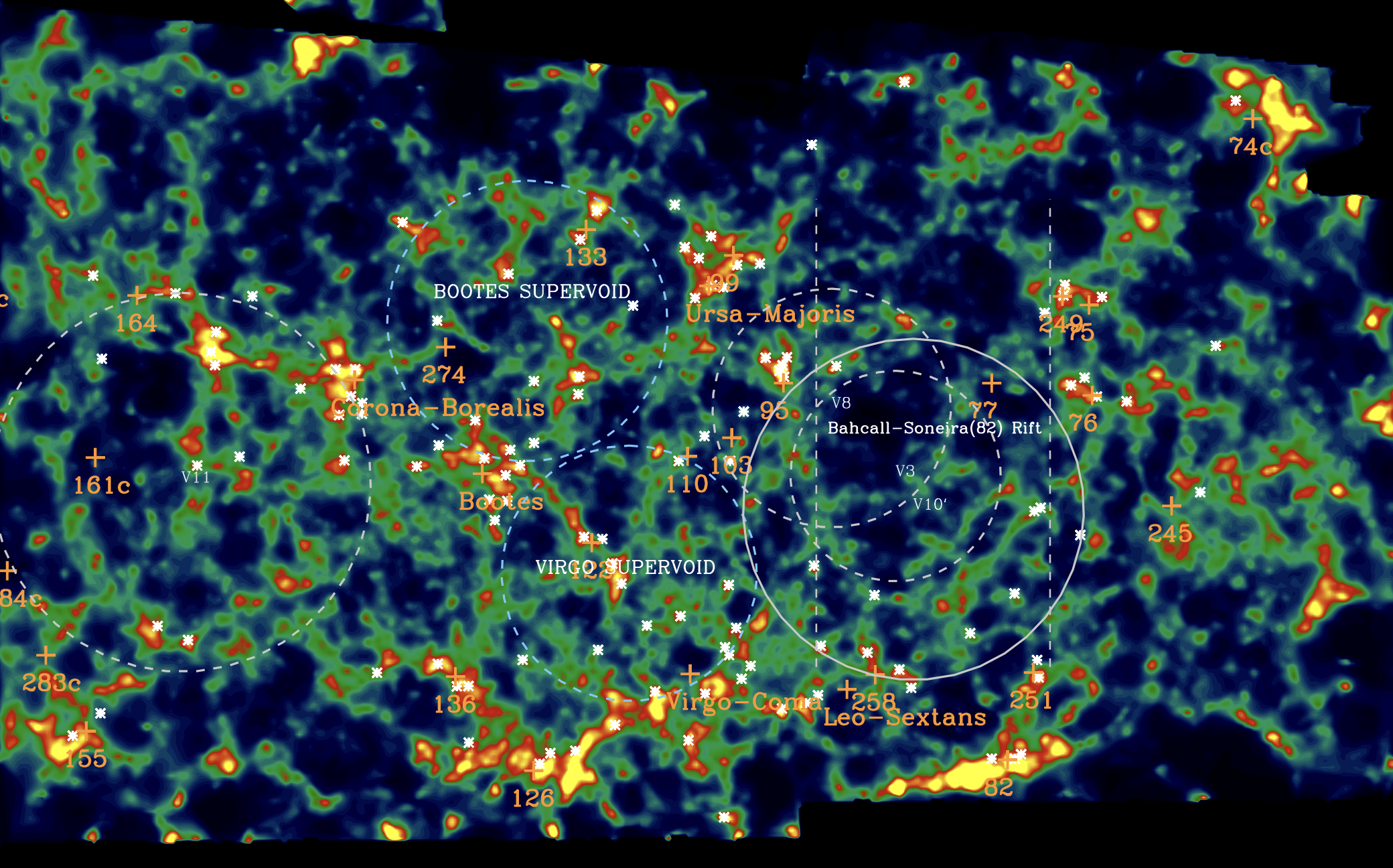}
    \end{subfigure}
    \caption{Galaxy density maps obtained from DTFE evaluations of the Sloan Digital Sky Survey (SDSS) data  \cite{Platen2009}. Shown are a top-view of the Great Walls and intermediate voids (upper panel), face-on view of the Coma Wall (middle panel) and face-on view of the Sloan Great Wall (lower panel). The positions and and sizes of the major cosmic voids (Boötes, Virgo, V3, V8, V11) in the line-of-sight foreground of the SGW (i.e. in between the Coma Wall and the SGW) are indicated by the white dotted circles. Image credit to Erwin Platen (2009) \cite{Platen2009}, created as follow-up work to the analysis published in \cite{Platen+2011}.}
    \label{fig:SDSS_maps}
\end{figure*}

Our findings are not only relevant to understanding halo distributions from $N$-body simulations, but they directly reflect the physical reality of the walls that are being identified in ongoing and upcoming cosmological surveys. To date, these are, notably, the Pisces-Perseus Supercluster \cite{GiovanelliHaynes1986}, the Coma Wall \cite{GellerHuchra1989}, the Sloan Great Wall (SGW) \cite{Gott+2005} and the recently identified BOSS Great Wall (BGW) \cite{Lietzen+2016}.\footnote{The identification of these large-scale objects as ``walls'' stems from flattened morphology of the galaxy distributions on a gigantic spatial extent of several hundreds of megaparsec. However, as surveys only observe the Eulerian configuration of galaxies, this classification has not yet been analysed in a rigorous phase space treatment. It is currently unclear whether the ``Great Walls'' actual form dynamical walls in the terminology of the present article. Instead, it is likely that individual ``superclusters'' \cite{Einasto+2011, Einasto+2017} within the Great Walls --- forming flattened distributions of galaxies on scales of tens to about a hundred megaparsec --- constitute sheet-like multistreaming regions in the sense that we are interested in here. This concerns, in particular, the substructure of the SGW \cite{Platen2009, Einasto+2011, EinastoM2025}, as well as the A, B, C and D superclusters in the BGW \cite{Lietzen+2016, Einasto+2017,  EinastoM2025}.}
The substructure \cite{Platen2009, Platen+2011, Einasto+2011, Lietzen+2016, Einasto+2017, EinastoM2025} of these gigantic objects is characterised by an intricate web-like geometry that is remarkably similar to the halo distributions that we find in our constraint simulations. Taken from the analysis of \cite{Platen2009}, the middle panel of \cref{fig:SDSS_maps} demonstrates that this can be observed for the entirety of the Coma Wall. With an extent of about $60 \times 170\,h^{-1}\textrm{Mpc}$, the Coma Wall is somewhat larger than, but comparable to, the objects we have simulated above. The SGW is significantly larger, with an extent of order $400 \,h^{-1}\textrm{Mpc}$ \cite{WeygaertBond2008}. Yet, the lower panel of \cref{fig:SDSS_maps} shows that similar features are prominent within the substructures ("superclusters") of the SGW, see also \cite{Platen2009, Einasto+2011, EinastoM2025}. Moreover, for BGW, the C and D superclusters \cite{Einasto+2017} in particular are of comparable extent ($\approx 50 \,h^{-1} \textrm{Mpc}$) to the walls we studied above, and are distinguished by an intricate filamentary pattern of embedded galaxies (see fig. 1 of \cite{Einasto+2017}) that is of striking resemblance to  \cref{fig:haloes_introduction} or the lower panel of \cref{fig:haloes}. These observations naturally lend themselves to a treatment in the caustic skeleton formalism, which may shed light on the distribution of galaxies from the scale-space singularities underlying the formation of the extended large-scale structure objects. In future work, we plan to make use of the recently developed Local Universe reconstruction techniques \cite{JascheWandelt2013, JascheLavaux2019, Valade+2022, Valade+2024, McAlpine+2025} to probe the caustic skeleton underlying the assembly of nearby walls in our cosmic neighbourhood. By identifying the dominant dynamical singularities, we plan to investigate the distribution and properties of the observed galaxies in the local cosmic web. Besides the Coma Wall and numerous structures in the SGW, we expect the results of the present article to be directly applicable to the C and D BGW superclusters: We speculate at this point that these objects may be formed by cusp sheets at scale $\sigma \approx 2.0 \,h^{-1} \textrm{Mpc}$, with the galaxy distribution traced by the swallowtail and umbilic caustics at $\sigma \approx 0.5 \textrm{--
} 1.0 \,h^{-1} \textrm{Mpc}$. We will unambiguously probe this hypothesis and, more generally, the reality of the caustic skeleton formalism in Local Universe reconstruction simulations in future work.

\subsection{Cosmic wall halo mass function}

Not only the positions of the haloes are sensitive to the caustic skeleton, but also their masses are dependent on the cosmic web environments they are embedded in. Within the scope of our simulations, \cref{fig:haloes} reveals that the most massive haloes ($M_{\textrm{halo}}\sim 10^{13}\text{--}10^{14}\,h^{-1
} M_{\odot}$) typically reside in the large-scale filaments and superdense clusters that bound the cosmic walls. The haloes located along the small-scale foldings within the cosmic wall are typically of lower mass ($M_{\textrm{halo}}\sim 10^{11}\text{--}10^{12}\,h^{-1} M_{\odot}$). This observation is consistent with the results of numerous previous studies \cite{Hahn+2007, Cautun+2014, Metuki+2015, AlonsoEardleyPeacock2015, MetukiLibeskindHoffman2016, Libeskind+2017}, which used different cosmic web tracing methods to study cosmic environments and found that haloes forming in cosmic walls are less massive than those forming in filaments in clusters. One such tracing method is the multi-scale morphology filter (MMF) \verb|NEXUS(+)| algorithm \cite{AragonCalvo+2007, Cautun+2012}, which identifies cosmic web objects based on the dominant signature of the eigenvalues of the Hessian of a tracer field evaluated over a range of smoothing scales. Using this method, \cite{Cautun+2014} found that the mass of function of haloes in cosmic walls is shifted to significantly lower values than the mass functions for the haloes embedded in filaments and clusters respectively.

\begin{figure*}
    \centering
    \includegraphics[width=0.85\textwidth]{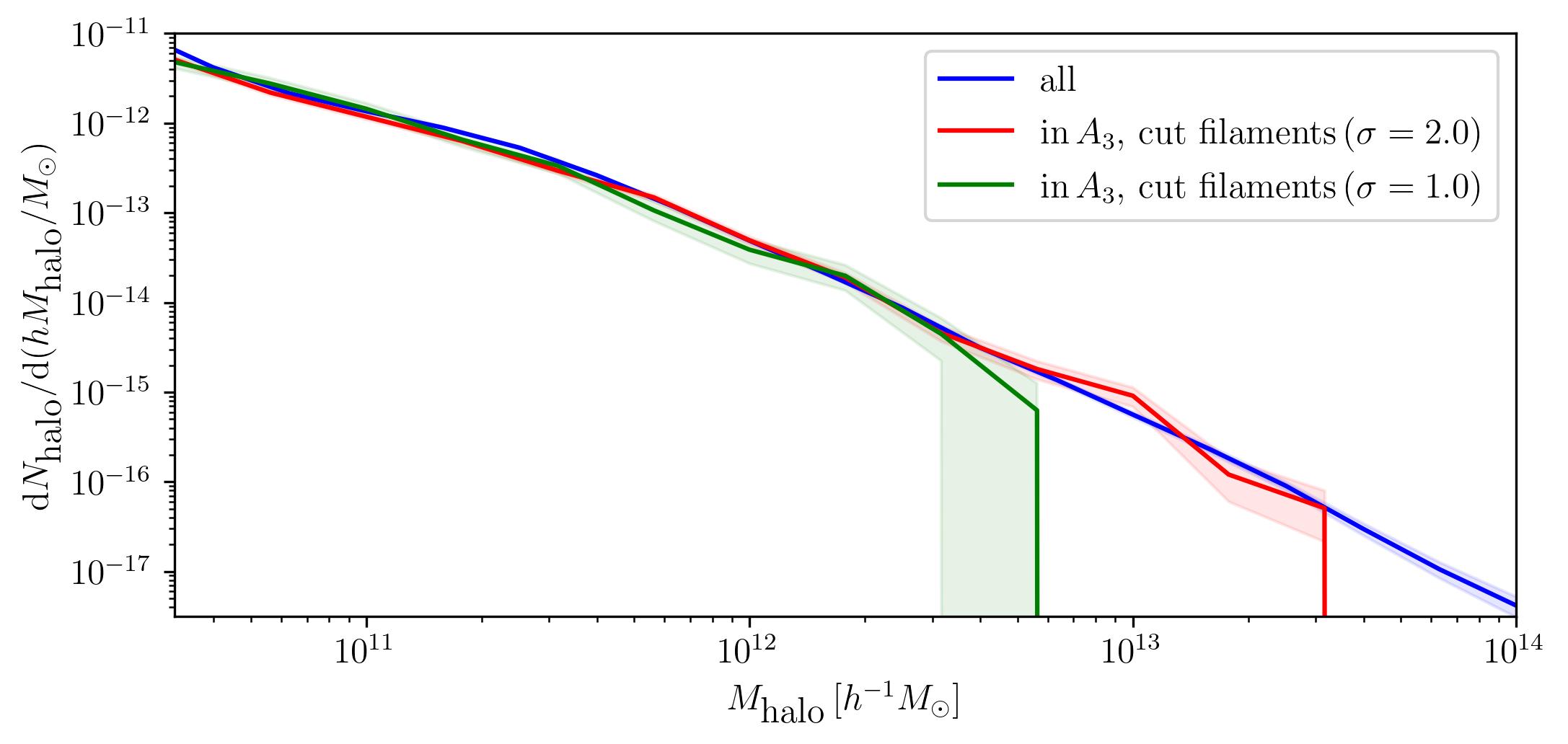}
    \caption{Normalised mass function of the wall haloes, obtained from the suite of 20 high-resolution constrained simulations. The blue curve shows the mass function of all haloes that reside in the large-scale ($\sigma=2.0\, h^{-1}\textrm{Mpc}$) cusp sheet in central region of the simulation box. The red and green curve show the mass function after removing those haloes that lie near any filamentary (swallowtail or umbilic) caustics at $\sigma=2.0\, h^{-1}\textrm{Mpc}$ and $\sigma=1.0\, h^{-1}\textrm{Mpc}$ respectively. Using the numerically computed caustic skeleton (see \cref{fig:haloes}), we assess whether a halo resides within $\epsilon = 1.0\, h^{-1}\textrm{Mpc}$ of any of the determined cusp sheet simplices, and not within $\epsilon = 1.0\, h^{-1}\textrm{Mpc}$ of any of the filamentary caustic simplices.  The 1$\sigma$-bands are estimated assuming Poissonian halo counts, such that $\sigma \sim 1/\sqrt{N_{\rm{halo}}}$.}
    \label{fig:halo_mass_function}
\end{figure*}

We now repeat this analysis with our constrained simulations and measure the mass function of haloes in the simulated cosmic walls. However, while traditional identification methods classify cosmological volumes as belonging to either voids, walls, filaments or clusters, the caustic skeleton model suggests that this notion is too simplistic to describe the connected scale-space cosmic web. In particular, the filaments, made of either swallowtail or umbilic caustics, are always connected to the cosmic walls formed by the cusp sheets. The classification of a halo as belonging either to a wall or a filament therefore becomes ambiguous. Considering the fractal-like structure of smaller-scale caustics within the cusp sheet, the caustic formalism suggests that it is not meaningful to consider ``wall-embedded'' as opposed to ``filament-embedded'' haloes, but rather to identify those haloes that lie within the cusp sheets, but sufficiently far away from any swallowtail and umbilic caustics, as evaluated on different smoothing scales. In \cref{fig:halo_mass_function}, we compare the mass function of all haloes in our simulations with those obtained from the haloes in the wall by cutting away those haloes lying near increasingly smaller-scale filaments. Despite the moderate number of data points, it is evident that the exclusion of haloes near the filaments (and therefore also clusters) of scale $\sigma = 2.0 \,h^{-1}\textrm{Mpc}$ removes the peak masses $M_{\textrm{halo}}\gtrsim 10^{13.5} \,h^{-1
} M_{\odot}$ from the mass function. The mass function is further shifted to the left when the filaments at  $\sigma = 1.0 \,h^{-1}\textrm{Mpc}$ are excluded, thus leaving only the small to moderate mass haloes $M_{\textrm{halo}}\lesssim 10^{12.5} \,h^{-1} M_{\odot}$. Clearly, these results are consistent with previous studies \cite{Hahn+2007, Cautun+2014, Metuki+2015, AlonsoEardleyPeacock2015, MetukiLibeskindHoffman2016, Libeskind+2017} in that we find the more massive haloes to lie near the large-scale filaments and clusters, whereas the masses of the wall haloes residing near only small-scale foldings are lower. In future work, it would be interesting to repeat this analysis with even higher resolutions to resolve in more detail the lowest mass haloes and the small-scale filamentary patterns, which may shed more light on how the near scale-invariant primordial perturbations seed the fractal-like web geometry \cite{Einasto+2020, Einasto2025} observed today.

\subsection{Halo densities near walls}

\begin{figure}
    \centering
    \begin{subfigure}[b]{\textwidth}
        \includegraphics[width=\textwidth]{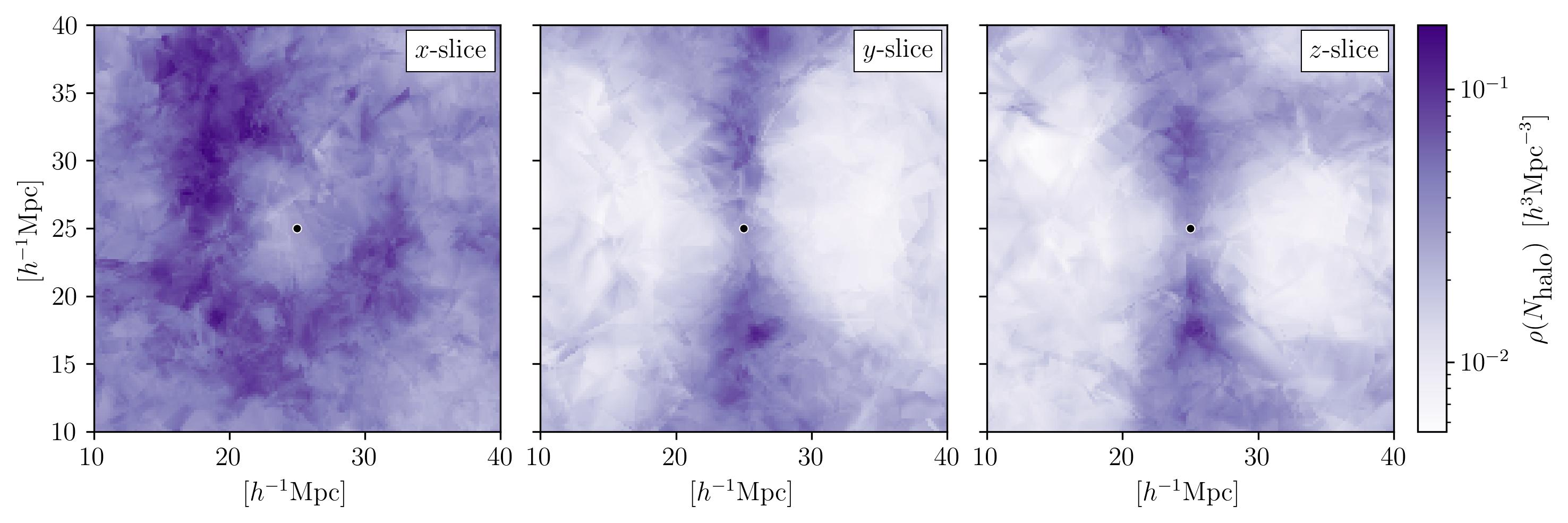}
        \caption{dark matter halo number density}
    \end{subfigure}\vspace{0.5em}
    \begin{subfigure}[b]{\textwidth}
        \includegraphics[width=\textwidth]{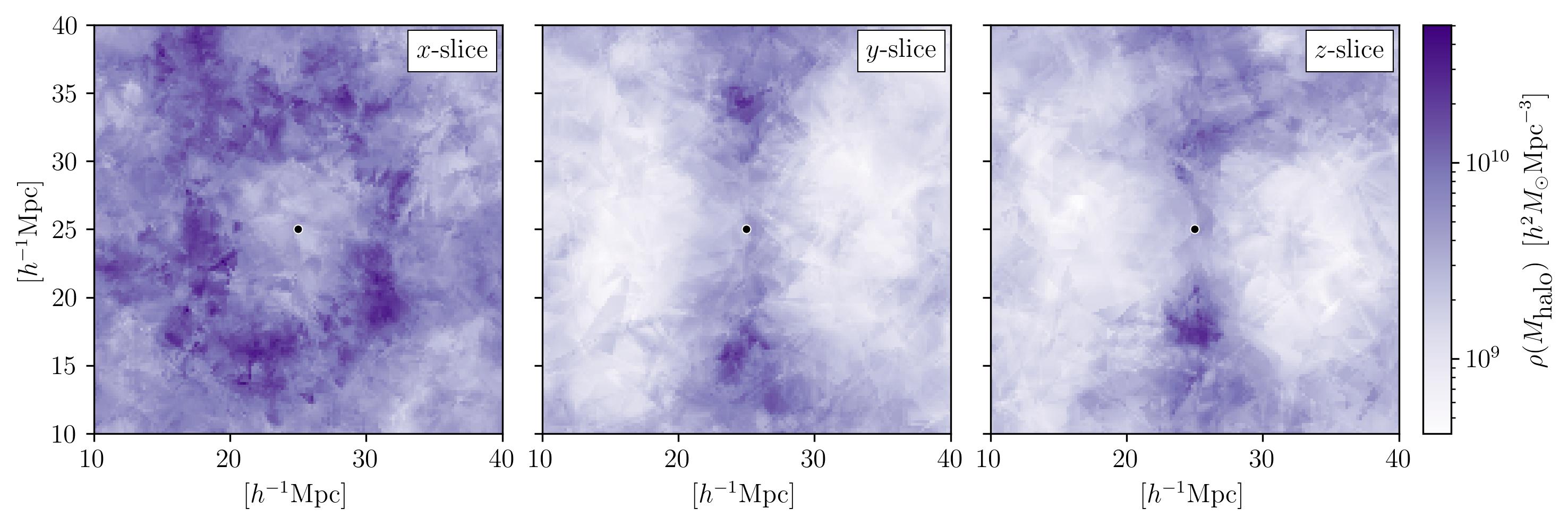}
        \caption{dark matter halo mass density}
    \end{subfigure}
    \caption{Median halo densities inferred from the suite of 20 high-resolution constraint simulations of cosmic wall formation. As before, shown are the $x$-, $y$- and $z$-slice (left, middle, right respectively) through the constraint point. The upper panels show the median halo number density, whereas the lower panels display median density obtained from the halo masses.}
    \label{fig:halo_densities}
\end{figure}

In recent years, the dependency of halo bias on the cosmic web \cite{Yang2017} has received increasing attention. While \cite{AlonsoEardleyPeacock2015} argued that the halo populations in the voids, walls, filaments and clusters can be can be accurately modelled from the dark matter density contrast alone, the environment-dependent halo bias found in \cite{Yang2017} suggests that the tidal field may play a decisive role in the halo assembly in the different cosmic web elements. We leave a detailed investigation of halo bias in constrained simulations to future work, but conclude this section with a brief discussion of the characteristic halo density around the simulated cosmic walls.

To this end, we infer the median halo number and mass densities for the suite of constrained simulations we have discussed above. We obtain these by evaluating the DTFE halo density for each simulation, and subsequently stacking them and taking taking the median, as was done in \cref{subsec:sims-fields} for the dark matter densities. \Cref{fig:halo_densities} shows that both the median halo number and halo mass densities reflect the imposed wall geometry, with the planar halo overdensity separating two voids whose halo occupation is much more sparse. The tessellated appearance of the plots is an artifact of the limited number of simulations and moderate number of haloes in each simulations; the inferred fields are therefore sensitive to the underlying Delaunay tessellations. Nevertheless, it is clear that the voids manifest themselves as regions near-empty of haloes, whereas the cosmic wall is about as dense in haloes as the full simulation volume, for which we infer the mean halo number density $\bar{\rho}(N_{\textrm{halo}}) \approx 0.02 \,h^{3}\textrm{Mpc}^{-3}$ and the mean halo mass density $\bar{\rho}(M_{\textrm{halo}}) \approx 1.4 \, h^2 M_{\odot}\textrm{Mpc}^{-3}$ respectively. The walls are about an order of magnitude more dense in haloes than the voids. Note, however, that the halo densities are somewhat arbitrary, as e.g. the identification of FoF groups rather than \verb|Subfind| haloes would be expected to result in fewer and in turn more massive haloes. Nonetheless, the identified \verb|Subfind| halo densities are in clear agreement with the discussion above: The distribution of haloes in the plane of the wall is not uniform, but peaks in a ring-like structure surrounding the sheet interior. Considering the random wall realisations, it is clear that these halo overdensities are due to the large-scale filaments bounding the cosmic walls, along which the numerous and typically more massive haloes are found. To illustrate this further, \cref{fig:halo_densities_cut} plots the wall-face median halo number densities restricted to low-, intermediate and high-mass haloes respectively. It is apparent that the most massive haloes lie far away from the centre of the wall, whereas low-mass haloes, while still typically lying in small-scale filaments, are found more within the wall sheet.

\begin{figure}
    \centering
    \includegraphics[width=\textwidth]{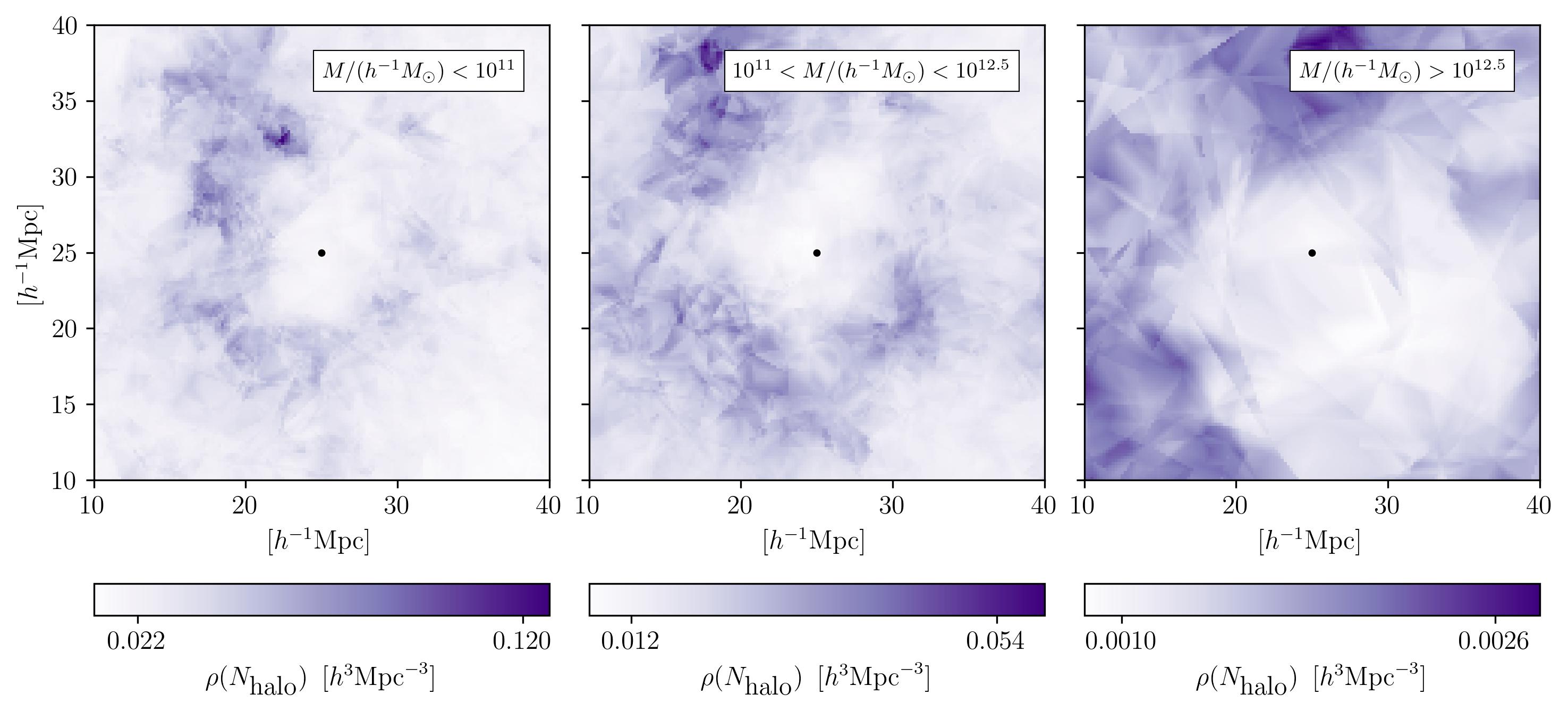}
    \caption{Median halo number densities for different halo mass ranges (low, intermediate and high mass haloes in the left, middle and right panel respectively). Shown are median densities in the wall plane ($x$-slice, see \cref{fig:halo_densities}) and restricted to those haloes that lie within the chosen mass range.}
    \label{fig:halo_densities_cut}
\end{figure}

Altogether, our studies confirm the observations of numerous previous studies \cite{Hahn+2007, Cautun+2014, Metuki+2015, AlonsoEardleyPeacock2015, MetukiLibeskindHoffman2016, Libeskind+2017} finding that haloes within cosmic walls are less massive than those found within the highly dense large-scale filaments and clusters. Moreover, with the consideration of small-scale caustics within the cusp sheet, we provide for the first time a natural explanation of the non-uniform, filamentary distribution of low-mass haloes within the cosmic walls. We will further investigate these findings in the context of Local Universe reconstructions in forthcoming studies.


\section{Constraint simulations of conventional saddle point conditions}
\label{sec:alternatives}

Ever since the seminal study of Bardeen et al. (1986) \cite{Bardeen1986}, the peaks and troughs of the primordial density fields have been of central interest to understanding the emergence of the cosmic web. Conventionally, the correspondence between the early-time fields and the late-time structures has been formulated from two distinct points of view. The first emphasises the role of the morphology of the primordial density perturbation $\delta$ as a proxy of the Eulerian density contrast in the linear regime and subsequent non-linear collapse. The second acknowledges the central role of the tidal fields in shaping the intricate geometry of the large-scale structure. In this dynamical formulation, the emergence of voids, walls, filaments and clusters is attributed to the tidal tensor $\mathcal{H}\phi$, and thus the morphology of the primordial potential perturbation $\phi$. In both cases, the critical points of the respective points have been proposed as the progenitors of the structural elements of the cosmic web. In this section, we compare our proposed proto-wall constraint derived from the caustic skeleton formalism to the conventional constraints based on the morphology of $\phi$ and $\delta$. The aim of this section is not to improve on the conventional constraints, but instead to provide a fair comparison of the simulations resulting from our novel constraint to those resulting from the constraints as have been studied in the literature. Despite the wide use of the saddle point constraints in cosmic web analyses, to our knowledge, the present study provides the first detailed exploration of respective constraint simulations, and is thus an important check even aside from the comparison with the caustic skeleton model.

The relevance of the tidal tensor $\mathcal{H}\phi$ and its eigenvalue fields has been appreciated early on since the seminal works by Zel'dovich (1970) \cite{Zeldovich1970} and Doroshkevich (1970) \cite{Doroshkevich1970}. Moreover, the T-web classification has been a widely used standard in the community, differentiating cosmic web elements through their signature in the tidal tensor. Preceding the modern identification methods, two studies by Haarlem \& van de Weygaert (1993) \cite{HaarlemWeygaert1993}  and van de Weygaert \& Bertschinger (1996) \cite{WeygaertBertschinger1996} proposed the correspondence of the progenitors of the cosmic web elements to the critical points of primordial potential perturbation $\phi$. Concretely, the late-time voids, walls, filaments and clusters were assumed to correspond to the $\phi^{(---)}$ maxima,  $\phi^{(+--)}$ saddles,  $\phi^{(++-)}$ saddles and  $\phi^{(+++)}$ minima respectively. In \cite{WeygaertBertschinger1996}, the authors implemented the respective constraints in random fields using the Bertschinger-Hoffman-Ribak algorithm \cite{Bertschinger1987,HoffmanRibak1991} to infer field realisations and statistical properties of the constraints. The choice of the $\phi$ critical points as constraints for the cosmic web progenitors is rather natural, as it directly reflects the role of the eigenvalue fields underlying the Zel'dovich theory \cite{Zeldovich1970} of structure formation. However, the downside of this approach is that the relation between the primordial potential perturbation $\phi(\bm{q})$ and the Eulerian density field $\rho(\bm{x})$ is not manifestly clear. We suspect that this is the reason why, to our knowledge, the critical points of $\phi$ have not been further employed in analytical investigations of cosmic web properties from the primordial conditions.

Instead, a series of papers \cite{NovikovColombiDore2006, Sousbie+2008, Pogosyan+2009} pioneered by Novikov et al. (2006) \cite{NovikovColombiDore2006} investigated the identification of the cosmic web from the critical points and associated critical lines of the density field. These studies can be considered a continuation of early works such as Bond et al. (1996) \cite{BondKofmanPogosyan1996}, and were further extended by the development of the topological morphology filter \verb|DisPerSE| \cite{Sousbie2011, Sousbie2011b}. The \verb|DisPerSE| method identifies the cosmic web elements from the Morse-Smale complex of the Eulerian density field, and has been widely applied by the community in recent years, see e.g. \cite{Libeskind+2017}. The underlying \textit{skeleton model} \cite{NovikovColombiDore2006, Pogosyan+2009} of the cosmic web attributes a central role to the saddle points of the density field $\delta$ and classifies the voids, walls, filaments and clusters with the $\delta^{(+++)}$ minima, $\delta^{(++-)}$ saddles, $\delta^{(+--)}$ saddles and $\delta^{(---)}$ maxima respectively; see also \cite{Sousbie+2008, CodisPichonPogosyan2015, CodisPogosyanPichon2018}. Note that \cite{NovikovColombiDore2006, Pogosyan+2009} investigate primarily the Gaussian random field case. It is unclear from these works whether the skeleton formalism is meant to be applied to the Gaussian primordial density perturbation or the highly non-Gaussian late-time density field. However, numerous comments in these and related works \cite{Sousbie+2008, CodisPichonPogosyan2015, CodisPogosyanPichon2018} suggest that the identification from the high-redshift near-Gaussian configuration holds for the late universe, as was brought forward early on e.g. in Bond et al. (1996) \cite{BondKofmanPogosyan1996}. The linear clustering phase indeed preserves the morphology inherited from the primordial fields. However, under the late-time non-linear gravitational evolution, the Eulerian density field is expected to become increasingly oblivious to the small-scale morphology of the primordial density perturbation $\delta$. This sheds doubt on the applicability of the respective saddle point conditions as the progenitors of the cosmic web elements. Nevertheless, to our understanding, the Lagrangian identification was made explicit in a series of studies that used the saddle points of  the primordial potential perturbation $\delta$ to infer various properties of the late-time cosmic web. In particular, \cite{CodisPichonPogosyan2015} calculated spin acquisition near filamentary-type saddle points, \cite{CodisPogosyanPichon2018} inferred the connectivity of filaments and clusters from the morphology of $\delta$, and \cite{Cadiou+2020} investigated the assembly of the cosmic web elements from the scale-space merger of the critical points. Despite the wide use of the $\delta$ saddle points in these quantitative analyses, it is unclear whether these primordial conditions are actually able to reproduce the observed cosmic web elements, as systematic constrained simulations have so far been absent from the literature. In the present article, we therefore commence at filling this gap and investigate the formation of cosmic walls in constrained simulation of the wall-type $\delta^{(++-)}$ saddle point.

\subsection{Saddle points of primordial displacement potential}
\label{subsec:alternatives-phi}

\begin{figure*}
    \includegraphics[width=\textwidth]{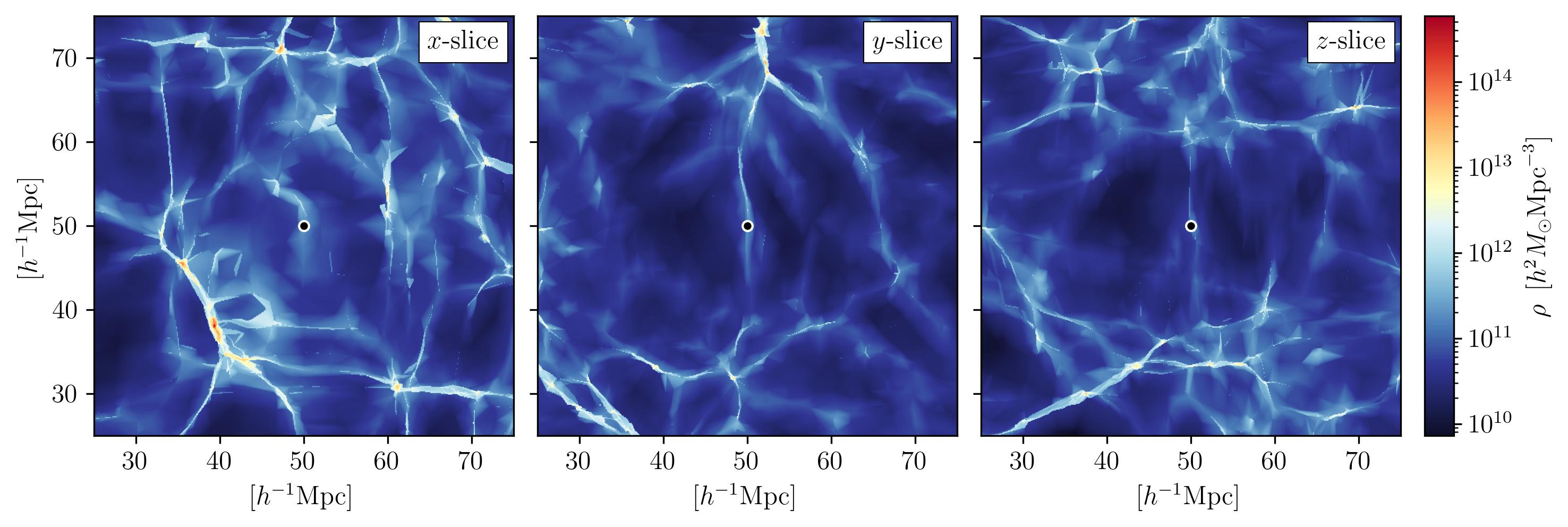}
    \includegraphics[width=\textwidth]{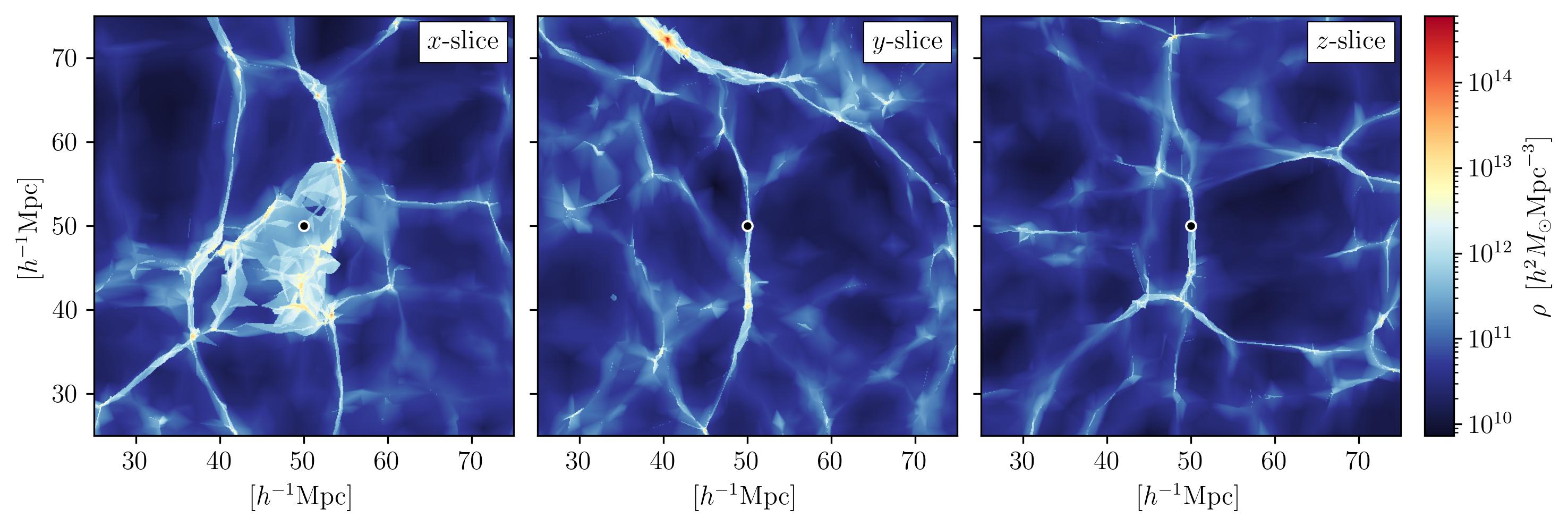}
    \includegraphics[width=\textwidth]{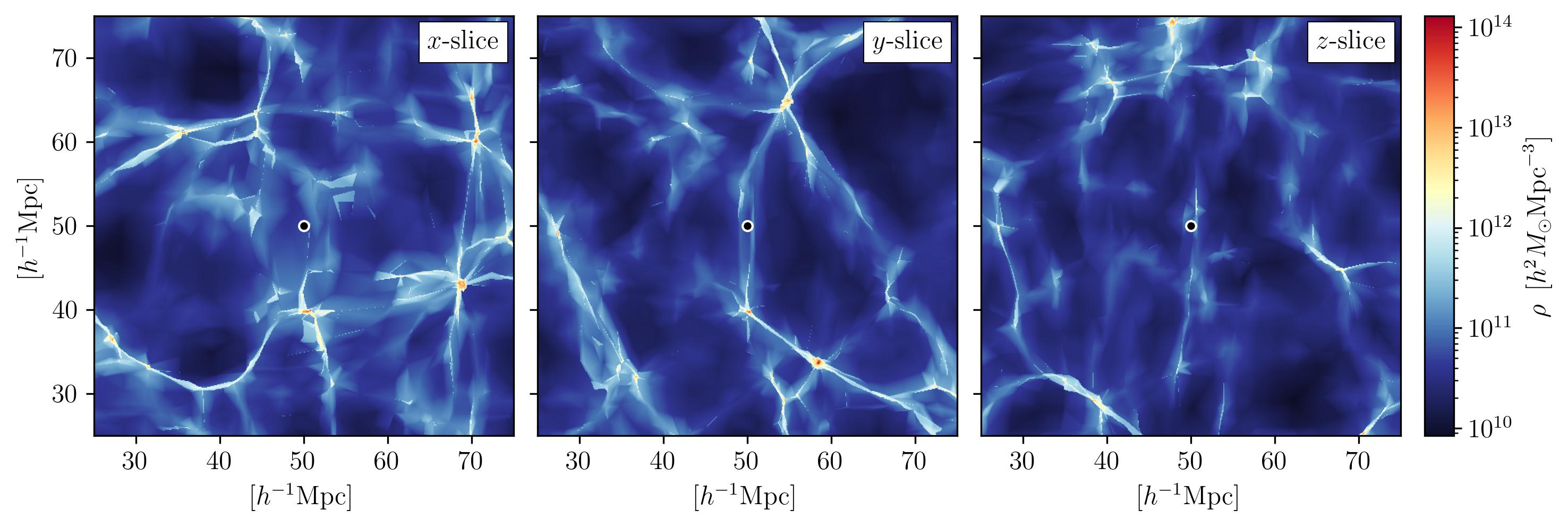}
    \caption{Exemplary density field realisations of the $\Psi^{(+--)}$ constraint for $\sigma = 2.0\, h^{-1}\textrm{Mpc}$.}
    \label{fig:phi_constraint_realisations}
\end{figure*}

\begin{figure*}
    \includegraphics[width=\textwidth]{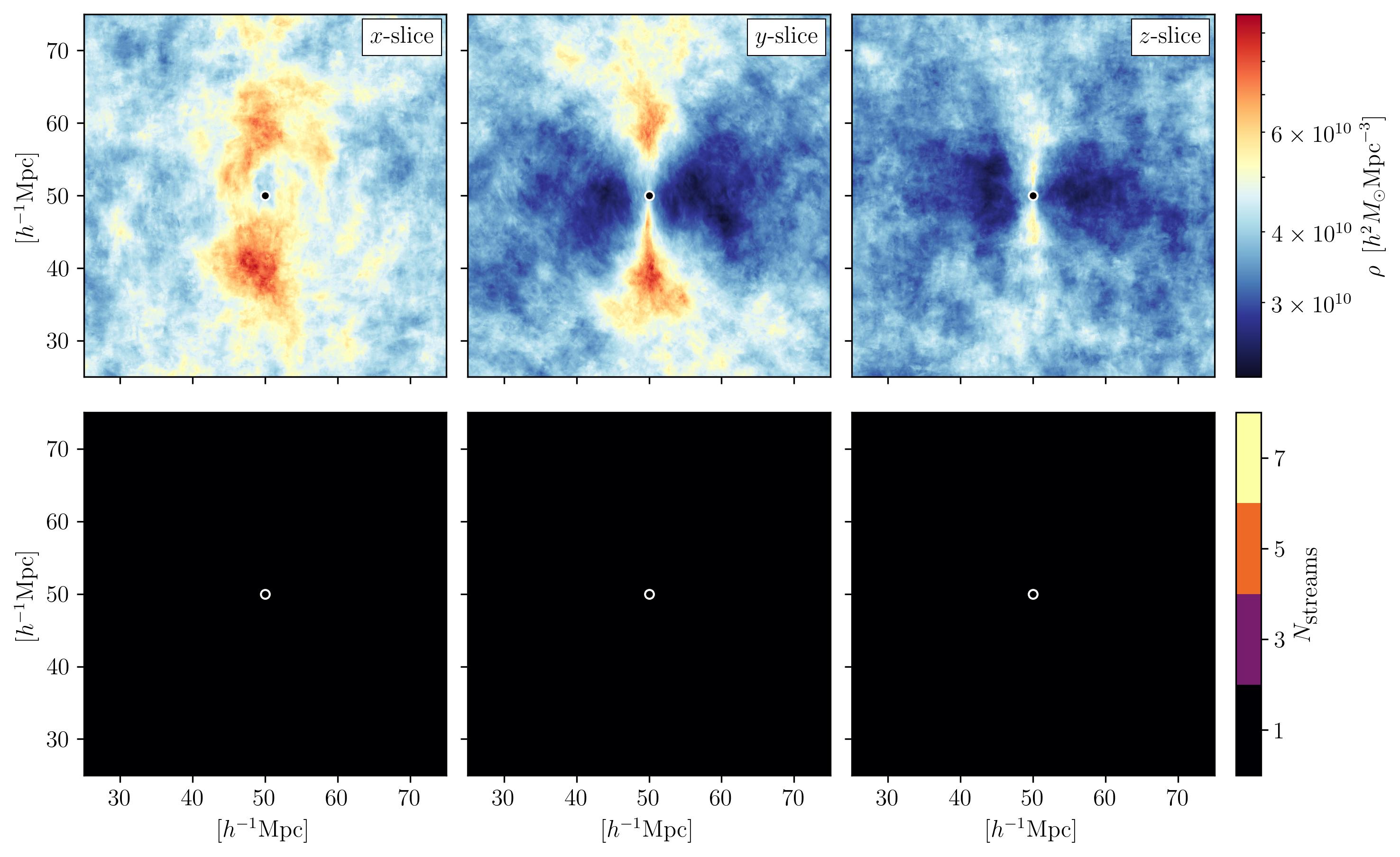}
    \caption{Median fields of the $\Psi^{(+--)}$ constraint for $\sigma = 2.0\, h^{-1}\textrm{Mpc}$. The upper and lower row show the median density and number of streams fields respectively.}
    \label{fig:phi_constraint_mean_fields}
\end{figure*}

\paragraph{Constraint}
In the discussion above, we followed \cite{HaarlemWeygaert1993, WeygaertBertschinger1996} in motivating the role of the critical points in the primordial potential perturbation $\phi$. To stay consistent with the convention of the preceding sections, however, we now choose to work instead with the primordial displacement potential $\Psi \propto \phi$. As the former is related to the latter by a linear proportionality, the discussion is remains unchanged and can be trivially translated.

The condition for the occurrence of a saddle point in the primordial displacement potential $\Psi$ is given by
\begin{equation}
    \nabla \Psi = \bm{0} = (t_1, t_2,t_3) \,,
\end{equation}
where on the right-hand side we evaluated the constraint in the first-order Cartesian field derivatives. To assert that the saddle point is of type $(+--)$, we evaluate the eigenvalues of the Hessian $\mathcal{H} \Psi$ and restrict the constraint sampling to configurations with one positive and two negative eigenvalues. We therefore sample constraint realisations
\begin{equation}
    \bm{C} = (0,\, 0,\, 0,\, t_{11},\, t_{12},\, t_{13},\, t_{22},\, t_{23},\, t_{33})
\end{equation}
in the space of first- and second-order field derivatives from the respective multivariate distribution (see \cref{subsec:recipe-grf_theory}).

Note that $\Psi$ (or equivalently $\phi$) is an arbitrary gauge field, and the critical value of the saddle point is not physical. As proposed in \cite{HaarlemWeygaert1993, WeygaertBertschinger1996}, the $\Psi^{(+--)}$ saddle point constraint therefore has only a single physical parameter, namely the smoothing scale $\sigma$ at which the field is evaluated. Comparing this to our proposed $A_3$ wall centre constraint, the $\Psi^{(+--)}$ constraint by construction cannot impose the collapse time of the hypothetical wall object. Nevertheless, we use the condition in this form as given in the literature \cite{HaarlemWeygaert1993, WeygaertBertschinger1996} for the fair comparison with the novel caustics-based analysis.

\paragraph{Constraint orientation}
We infer the orientation of a $\Psi^{(+--)}$ constraint realisation by evaluating the eigenvectors $\bm{v}_i$ of $\mathcal{H} \Psi$ from the sampled second-order derivatives $t_{ij}$. Choosing again the ordering $\lambda_1 > \lambda_2 > \lambda_3$ and assuming a right-handed eigenvector system, we calculate the Euler angles $(\alpha, \beta, \gamma)$ that rotate the eigenvectors $\{\bm{v}_i\}$ into $\{ \bm{\hat{x}},  \bm{\hat{y}},  \bm{\hat{z}}\}$. The Lagrangian orientation is then performed by imposing the rotated constraint 
\begin{equation}
    \bm{\tilde{C}} = R^{(1,2)}(\alpha, \beta, \gamma) \bm{C} \,,
\end{equation}
where the rotation matrix $R^{(1,2)}$ of the first- through second-order derivatives is constructed analogously to \cref{eq:rot_derivs_mat}. For the Eulerian orientation, note that $\bm{v}_{2, E}$ and $\bm{v}_{3, E}$ span the hypothetical wall plane. Following the arguments of \cref{subsec:recipe-orientation}, we therefore use the tangent vectors $\bm{v}_{2, E}$ and $\bm{\tilde{v}}_{3, E} = \bm{v}_{3, E} - \text{proj}_{\bm{v}_{2, E}} \bm{v}_{3, E}$, construct the normal vector $\bm{\tilde{v}}_{1, E} =\bm{v}_{2, E} \times \bm{\tilde{v}}_{3, E}$ and calculate the Euler angles that rotate $\{ \bm{\tilde{v}}_{1, E}, \bm{v}_{2, E}, \bm{\tilde{v}}_{3, E}\}$ into $\{ \bm{\hat{x}}, \bm{\hat{y}},  \bm{\hat{z}}\}$. By shifting the constraint into the centre of the box and rotating the Eulerian coordinates by the calculated Euler angles about the same point, the hypothesised wall resulting from the simulation is oriented in consistency with the simulations of \cref{subsec:sims-fields}.

\paragraph{Field realisations}
We run a suite of 100 simulations of the $\Psi^{(+--)}$ saddle point constraint with $64^3$ particles in a box of side length $100 \,h^{-1} \textrm{Mpc}$. \Cref{fig:phi_constraint_realisations} shows three exemplary constraint realisations, with the median density and number of streams fields of the simulation suite being shown in \cref{fig:phi_constraint_mean_fields}. We find that the $\Psi$ saddle constraint is successful at simulating planar overdensities that can be understood as progenitors to cosmic walls. This is clearly seen in the median density fields, where an extended overdense volume separates two proto-voids. However, compared to the density fields from our proposed $A_3$ wall centre constraint, \cref{fig:A3_constraint_mean_field}, we find that the overdensities induced by the $\Psi$ saddle points are significantly lower. This is most prominently seen in the upper and lower row of \cref{fig:phi_constraint_realisations}, where only a very weak planar overdensity (if any) is apparent. The reason for this is that while the $\Psi$ saddle points determine the local geometry of the gravitational potential, the first eigenvalue $\lambda_1$ remains unspecified, such that shell-crossing from the fold caustic $\lambda_1 = b_c^{-1}$ is not imposed. For fairness, we have included one shell-crossed realisation in the second-row of \cref{fig:phi_constraint_realisations}. However, with the unspecified level of the first eigenvalue being sampled from the distribution of the second derivatives $t_{ij}$, we find that the majority of $\Psi^{(+--)}$ constraint realisations do not shell-cross into multistreaming cosmic walls. This is confirmed in the median number of streams field, revealing that the constrained volume typically remains single-streaming. This also explains the quadrupolar structure of the density field: The constraint point is very weakly overdense, and the median density field is most sensitive to the higher overdensities from the filaments forming within or around the planar geometry.

Altogether, our simulations demonstrate that the $\Psi$ constraint results in a planar structure in the late-time density field, but is not able to induce typical overdense and shell-crossed cosmic walls without further constraints. It is at this point unclear to what extent the  $\Psi^{(+--)}$ saddle constraint together with fold caustic constraint $\lambda_1 = b_c^{-1}$ would succeed at simulating cosmic walls. In particular, compared to our proposed constraint based on the cusp geometry, it not clear that an extended $\Psi^{(+--)}$ saddle would correspond to a centre of a cosmic wall, well away from filaments that would unavoidably form within the planar overdensity from higher singularities.  In this spirit, it may also be interesting to investigate the correlation between the proposed wall centre constraint, \cref{eq:A3_centre_constraint}, and the $\Psi^{(+--)}$ saddle points. We leave such investigations to further studies.

\subsection{Saddle points of primordial density perturbation}
\label{subsec:alternatives-density}

\paragraph{Constraint and statistics}

\begin{figure}
    \centering
    \begin{subfigure}[b]{0.49\textwidth}
        \includegraphics[width=\textwidth]{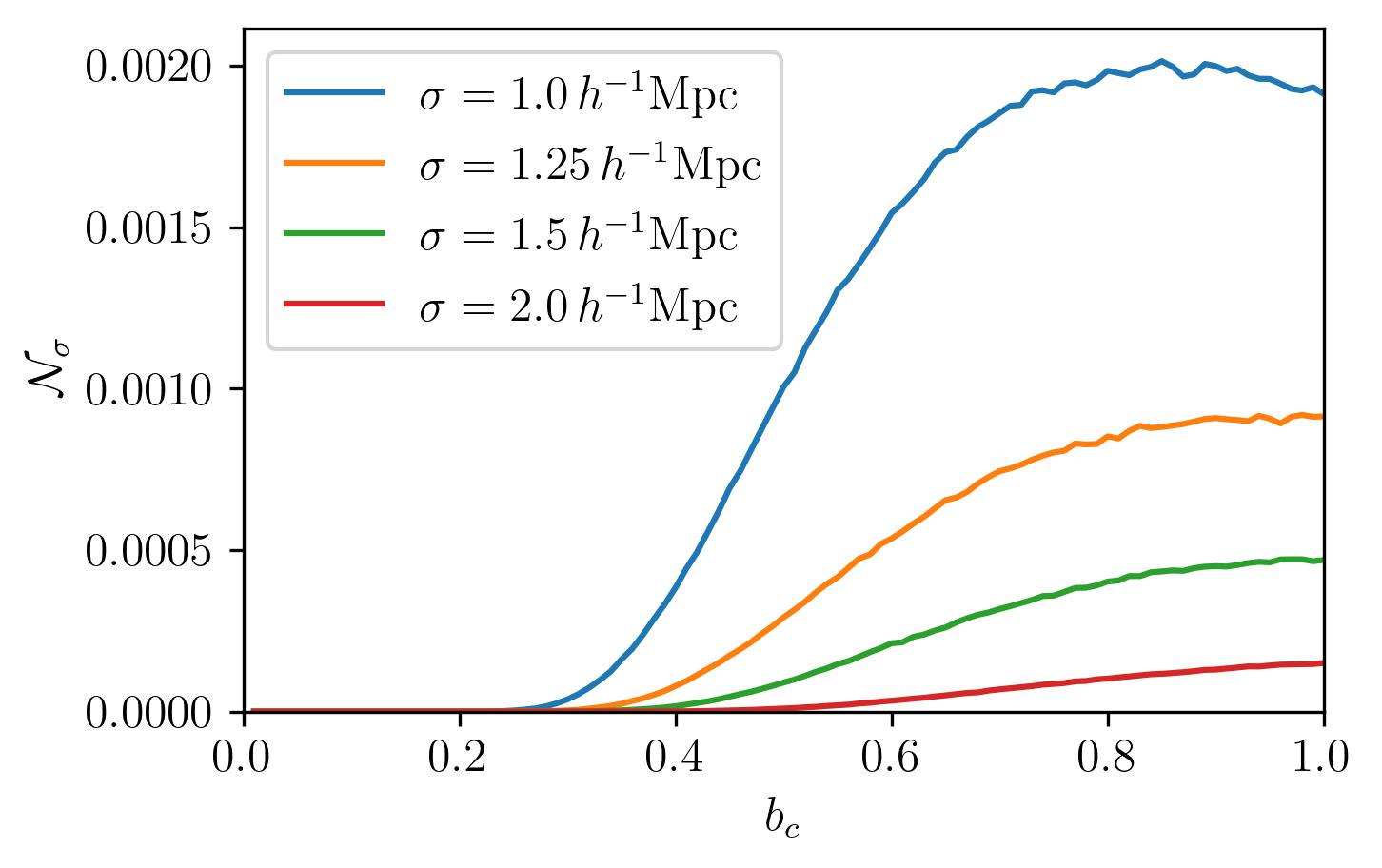}
        \caption{$\delta^{(++-)}$ saddle constraint}
    \end{subfigure}
    \begin{subfigure}[b]{0.49\textwidth}
    \includegraphics[width=\textwidth]{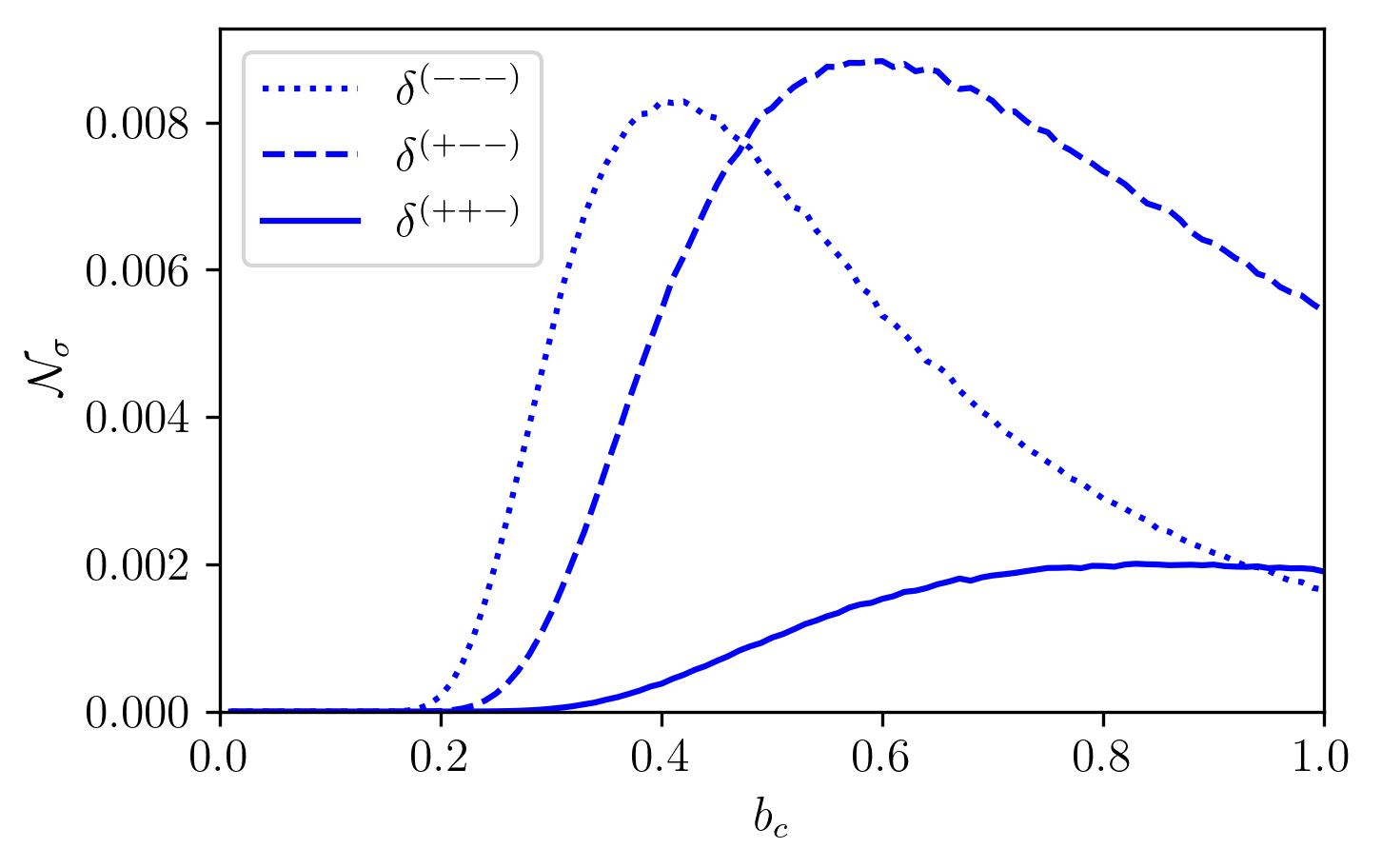}
        \caption{$\delta$ critical point constraints}
    \end{subfigure}
    \caption{Number density of the $\delta^{(++-)}$ saddle points as a function of $b_c$ for varying $\sigma$, evaluated using Rice's formula.}
    \label{fig:delta_constraint_number_density}
\end{figure}

The primordial density perturbation is given by the Laplacian of the displacement potential, $\delta = \nabla^2 \Psi$. In terms of the Cartesian field derivatives, we write
\begin{equation}
    \delta = t_{11} + t_{22} + t_{33} \,,
\end{equation}
so that the condition for a critical point in $\delta$ becomes
\begin{equation}
    \nabla \delta = \bm{0} = (T_{111} + T_{122} + T_{133}, \, T_{112} + T_{222} + T_{233},\, T_{133} + T_{223} + T_{333})
\end{equation}
We sample $\delta$-saddle points of type $(++-)$ by evaluating the Hessian $\mathcal{H} \delta = 
\mathcal{H} (T_{11} + T_{22} + T_{33})$ and restricting the sampling to those realisations for which the Hessian has two positive eigenvalues.

Contrary to the primordial potential perturbation, the level of the density perturbation is not an arbitrary gauge parameter and the critical value $\nu$ of the $\delta$ saddle point has a physical meaning. To make the comparison with the $A_3$ constraint as fair as possible, we hence consider the constraint in the two-dimensional parameter space $(\sigma, \nu)$.
In addition to the gradient $\nabla \delta$ vanishing, we impose the critical value $\delta = T_{11} + T_{22} + T_{33} = \nu$ and sample constraint realisations $\{\bm{C}_i\}$ on the space of second- through third-order derivatives from the constraint statistic
\begin{equation}
    \begin{split}
    \bm{Y} = (&t_{11},\, t_{12},\, t_{13},\, t_{22},\, t_{23},\,\nu - t_{11} - t_{22},\\
    &t_{111},\, t_{112},\, t_{113},\, t_{122},\, t_{123},\, -t_{111} - t_{122},\, t_{222},\, t_{223},\, -t_{112} - t_{222},\, -t_{133} - t_{223} )
    \end{split}
    \label{eq:delta_constraint_statistic}
\end{equation}

While the covariance of the field derivatives $t_{i\ldots k}$ unambiguously defines the distribution of the critical value of the saddle point for a given $\sigma$, the physical meaning of $\nu$ is not evident. To be consistent with the $A_3$ centre analysis, we wish to relate this value to the predicted formation time of the hypothetical cosmic wall object.
However, the $\delta$ saddle constraint lacks a rigorous mathematical foundation, and so there is at there is at present no universally accepted prescription for doing so within the community. Exemplarily, \cite{Cadiou+2020} conceptually replaces the formation time by the scale-dependence of formation through the excursion set formalism of the Morse-Smale complex (and so the saddle points) of $\delta$. Similarly, the spin alignment calculations of \cite{CodisPichonPogosyan2015} only consider the height of the assumed filamentary-type saddle point without any discussion of the implied cosmological formation time scales. To improve on these shortcomings and enable a fair comparison with our proposed caustic constraint, we adopt here a simple analytical model for the collapse time from the critical value $\nu$. We choose the spherical collapse model \cite{PressSchechter1974, Peebles1994} which predicts collapse to occur at the growing mode $b_c$ when the linearly extrapolated density crosses the threshold value
\begin{equation}
    b_c \nu = \delta_{c} \approx 1.69
    \label{eq:spherical_collapse}
\end{equation}
Note that the spherical collapse model is conceptually different from the ZA, as the threshold value $\delta_{c} \approx 1.69$ is derived from an Eulerian rather than Lagrangian treatment, see \cite{PressSchechter1974, Peebles1994}.

We remark here that that the spherical collapse model is not strictly applicable to the anisotropic situation we are considering here, as the value $\delta_{c} \approx 1.69$ is derived for the spherically symmetric case \cite{PressSchechter1974, Peebles1994}. Anisotropies are known to induce earlier collapse \cite{ShethMoTormen2001}, and one therefore might speculate about minor deviations in the threshold value $\delta_c$ for cosmic walls, filaments and clusters respectively. However, these quantitative considerations are negligible compared to the qualitative shortcomings of the $\delta^{(++-)}$ constraint. A small variation of the order-unity parameter $\delta_{c}$ does not modify our analysis on a qualitative level.

Before running simulations, we proceed as in \cref{subsubsec:sims-stochastic_geometry-N} and investigate the statistics of the $\delta^{(++-)}$ constraint in the parameter space $(\sigma, b_c)$. Again, we use Rice's formula to calculate the number density of the constraint as a function of the growing mode $b_c$ for fixed $\sigma$ to infer physically realistic parameter configurations. We insert \cref{eq:spherical_collapse} into the constraint statistic \cref{eq:delta_constraint_statistic}, and evaluate the number density
\begin{equation}
    \mathcal{N}_{\sigma}(b_c) = \frac{\delta_c}{b_c^2}\left\langle | \det \mathcal{H} \delta | \delta_D^{(1)}(\delta-\delta_c b_c^{-1}) \delta_D^{(3)}(\nabla \delta) \delta_D^{(1)}(\gamma(\mathcal{H}\delta)-2) \right\rangle
    \label{eq:delta_++-_number_density}
\end{equation}
in the second- through fourth-order Cartesian field derivatives, where the second- and third-order derivatives are given by \cref{eq:delta_constraint_statistic}. $\gamma(\mathcal{H}\delta)$ denotes the Morse index of the saddle point, restricting the sampling to realisations with two positive eigenvalues of $\mathcal{H}\delta$.

The left panel of \cref{fig:delta_constraint_number_density} shows the number density of the $\delta^{(++-)}$ constraint for a set of smoothing scales from $\sigma= 1.0\, h^{-1}\textrm{Mpc}$ to $\sigma= 2.0\, h^{-1}\textrm{Mpc}$. Comparing with \cref{fig:A3_number_density}, is apparent that the $\delta^{(++-)}$ constraint formation times peak much later than the $A_3$ centre constraints at all smoothing scales. Consistent with the hierarchical formation of structure, large-scale saddle points are found to peak later than the small-scale saddle points. However, Rice's formula suggests that only constraints at $\sigma=1.0\,h^{-1}\textrm{Mpc}$ peak significantly before the current Universe ($b_c = 1.0$). The late formation time peaks at the considered smoothing scales do not necessarily contradict the model of hierarchical structure formation. However, an interesting observation can be made by comparing the number densities of the supposed wall-type $\delta^{(++-)}$ constraints with those obtained for the filament-type $\delta^{(+--)}$ \cite{NovikovColombiDore2006, CodisPichonPogosyan2015, CodisPogosyanPichon2018} or cluster-type $\delta^{(---)}$ \cite{NovikovColombiDore2006, CodisPichonPogosyan2015, CodisPogosyanPichon2018} constraints. The right panel of   \cref{fig:delta_constraint_number_density} shows the respective number densities evaluated for $\sigma=2.0\,h^{-1}\textrm{Mpc}$. The calculation clearly shows the number density of the cluster- and filament-type saddle point peak significantly before the wall-type saddle points. Clearly, this in contradiction to the Zel'dovich theory of large-scale structure formation which predicts that the first objects to form are the flattened Zel'dovich pancakes or cosmic walls \cite{Zeldovich1970}. The plots of \cref{fig:delta_constraint_number_density} reverse this order by predicting that the first objects to form are clusters and that the walls form only after the filaments. This is also in stark contrast to the geometric discussion of \cref{subsubsec:theory-caustics}, where the caustic skeleton model revealed how the cusp sheets form the embedding for the formation of higher caustics corresponding to cosmic filaments in clusters. In the follow-up paper, we will analyse this in more detail and compare the number densities to the formation times predicted by the respective caustic skeleton constraints. For now, we do not further comment on these statistical inconsistencies, but continue by identifying the formation peaks $(\sigma, b_c)$ of the $\delta^{(++-)}$ constraint to sample physically realistic constraint configurations.

\paragraph{Constraint orientation}

Having sampled a $\delta^{(++-)}$ saddle constraint realisation $\bm{C}$, we proceed as before and orient the constraint in both Lagrangian and Eulerian space. For the Lagrangian orientation, we determine the eigenvectors $\{\bm{v}_i\}$ of $\mathcal{H} \delta = \mathcal{H} (T_{11} + T_{22} + T_{33})$ and calculate the Euler angles $(\alpha, \beta, \gamma)$ to rotate these into the basis $\{ \bm{\hat{y}}, \bm{\hat{z}}, \bm{\hat{x}}\}$, asserting a right-handed basis system. With the ordering from the two positive eigenvalues of $\mathcal{H} \delta$ --- rather than one positive eigenvalue as for the $A_3$ centre and $\phi^{(+--)}$ constraints --- this choice of orientation imposes the hypothetical planar collapse in the $yz$-plane, as we want for comparison with the simulations of \cref{subsec:sims-fields}. The constraint is rotated into 
\begin{equation}
    \bm{\tilde{C}} = R^{(2,3)}(\alpha, \beta, \gamma) \bm{C}
\end{equation}
and implemented into a random field realisation. After the $N$-body evolution, the vectors $\bm{v}_{1, E}$ and $\bm{v}_{2, E}$ span the hypothetical Eulerian wall plane. Similarly to \cref{subsec:recipe-orientation} and \cref{subsec:alternatives-phi}, we reject $\bm{\tilde{v}}_{2, E} = \bm{v}_{2, E} - \text{proj}_{\bm{v}_{1, E}} \bm{v}_{2, E}$, construct the normal vector $\bm{\tilde{v}}_{3, E} =\bm{v}_{1, E} \times \bm{\tilde{v}}_{2, E}$ and use the system $\{\bm{\tilde{v}}_{3, E}, \bm{v}_{1, E}, \bm{\tilde{v}}_{2, E} \}$ to calculate the Euler angles for rotation into $\{ \bm{\hat{x}}, \bm{\hat{y}}, \bm{\hat{z}}\}$. After shifting and rotating the Eulerian coordinates accordingly, the hypothesised wall object is oriented in accordance with the simulations discussed in \cref{subsec:sims-fields}.

\paragraph{Random field realisations and discussion}

As in the case of the $A_3$ centre constraint, we infer physically realistic parameter configurations for the $\delta^{(++-)}$ constraint from \cref{fig:delta_constraint_number_density}. The number densities peak at late times for all smoothing scales and we choose the parameter pairs $(\sigma, b_c) = \{ (1.0\,h^{-1}\textrm{Mpc}, \,0.8), (1.5\,h^{-1}\textrm{Mpc}, \,0.9), (2.0\,h^{-1}\textrm{Mpc}, \,0.95) \}$ to run suites of 50, 50 and 100 $N$-body simulations respectively. We make the field realisations available at the additional materials webpage.

\begin{figure*}
    \includegraphics[width=\textwidth]{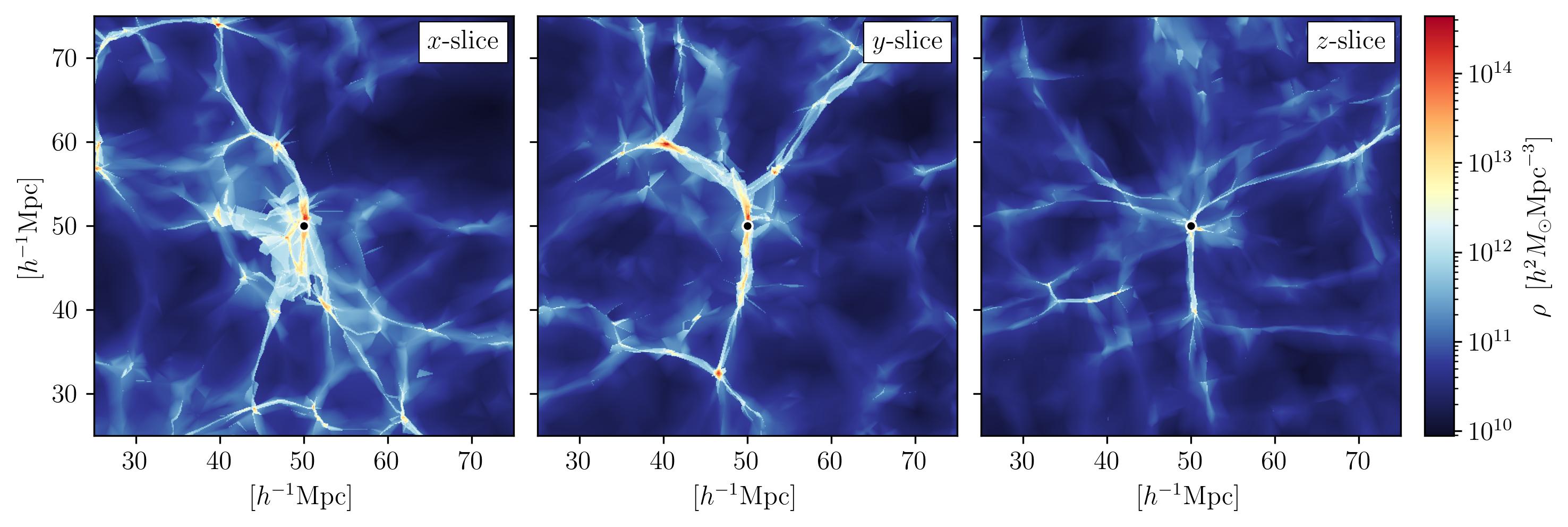}
    \includegraphics[width=\textwidth]{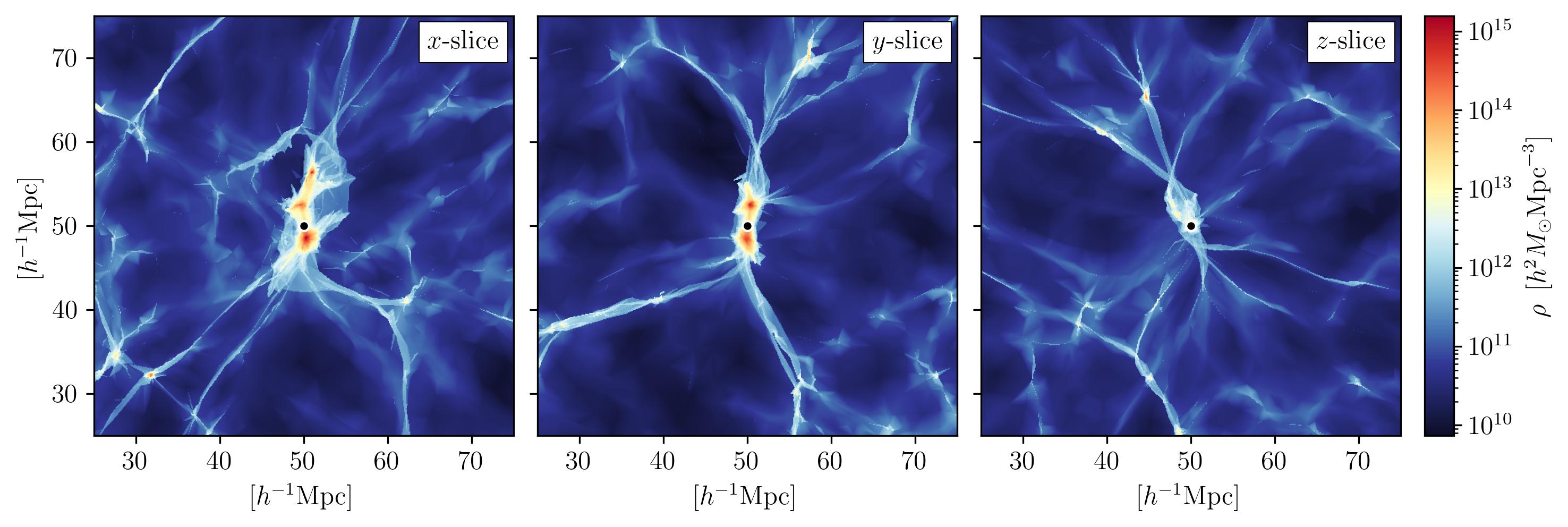}
    \includegraphics[width=\textwidth]{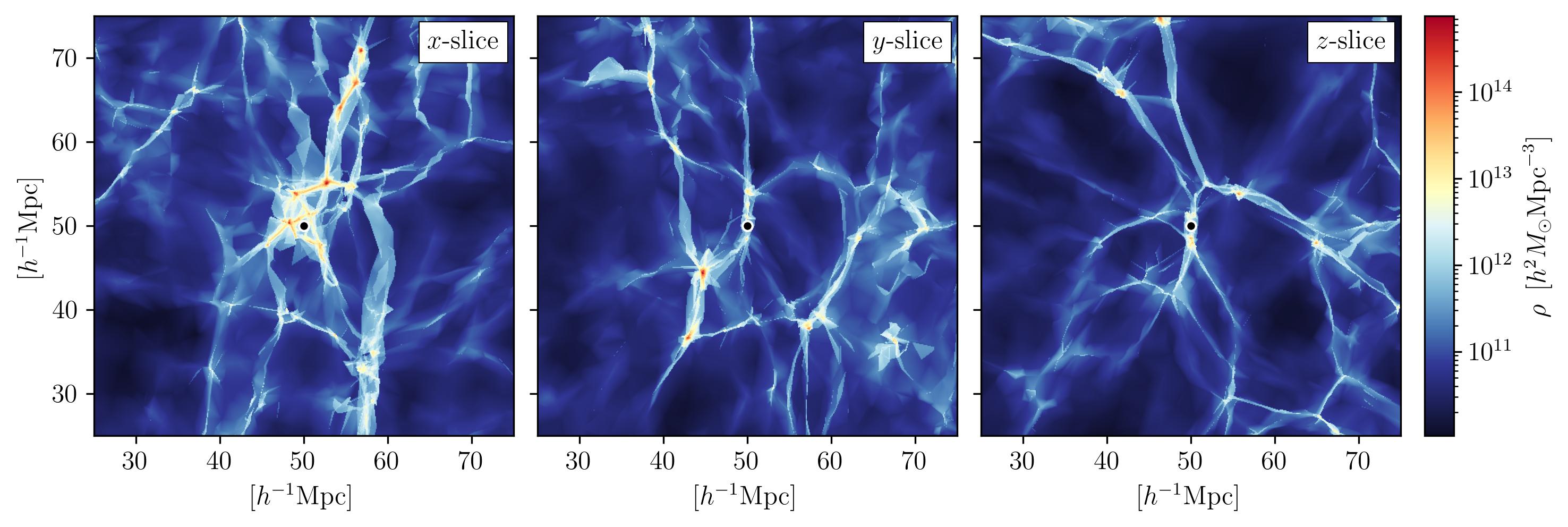}
    \caption{Exemplary density field realisations of the $\delta^{(++-)}$ constraint for $(\sigma, b_c) = (2.0\, h^{-1}\textrm{Mpc}, \,0.95)$.}
    \label{fig:delta_constraint_realisations}
\end{figure*}

In \cref{fig:delta_constraint_realisations}, we show the density fields of three exemplary random realisations for $(\sigma, b_c) = (2.0\,h^{-1}\textrm{Mpc},\, 0.9)$. It is apparent that the  $\delta^{(++-)}$ constraint results in collapsed structures that are oriented roughly vertically in $yx$- and $zx$-planes, as wanted. However, despite the large smoothing scaled and the late formation time, we find that the resulting structures are spatially compact with significant overdensities over the cosmic background. The objects do not appear as flattened, sheet-like overdensities, but rather as extended clusters with line-like trunks or a small-scale multistreaming surfaces (``walls'') attached. Importantly, the constraint point does not form the centre of the hypothesised wall, but typically resides in or very close to a very-high-density cluster within the moderately extended object. We observe the same for the smaller-scale realisations with $(\sigma, b_c) = (1.0\,h^{-1}\textrm{Mpc},\, 0.8)$ and $(\sigma, b_c) = (1.5\,h^{-1}\textrm{Mpc},\, 0.9)$.

\begin{figure*}
    \includegraphics[width=\textwidth]{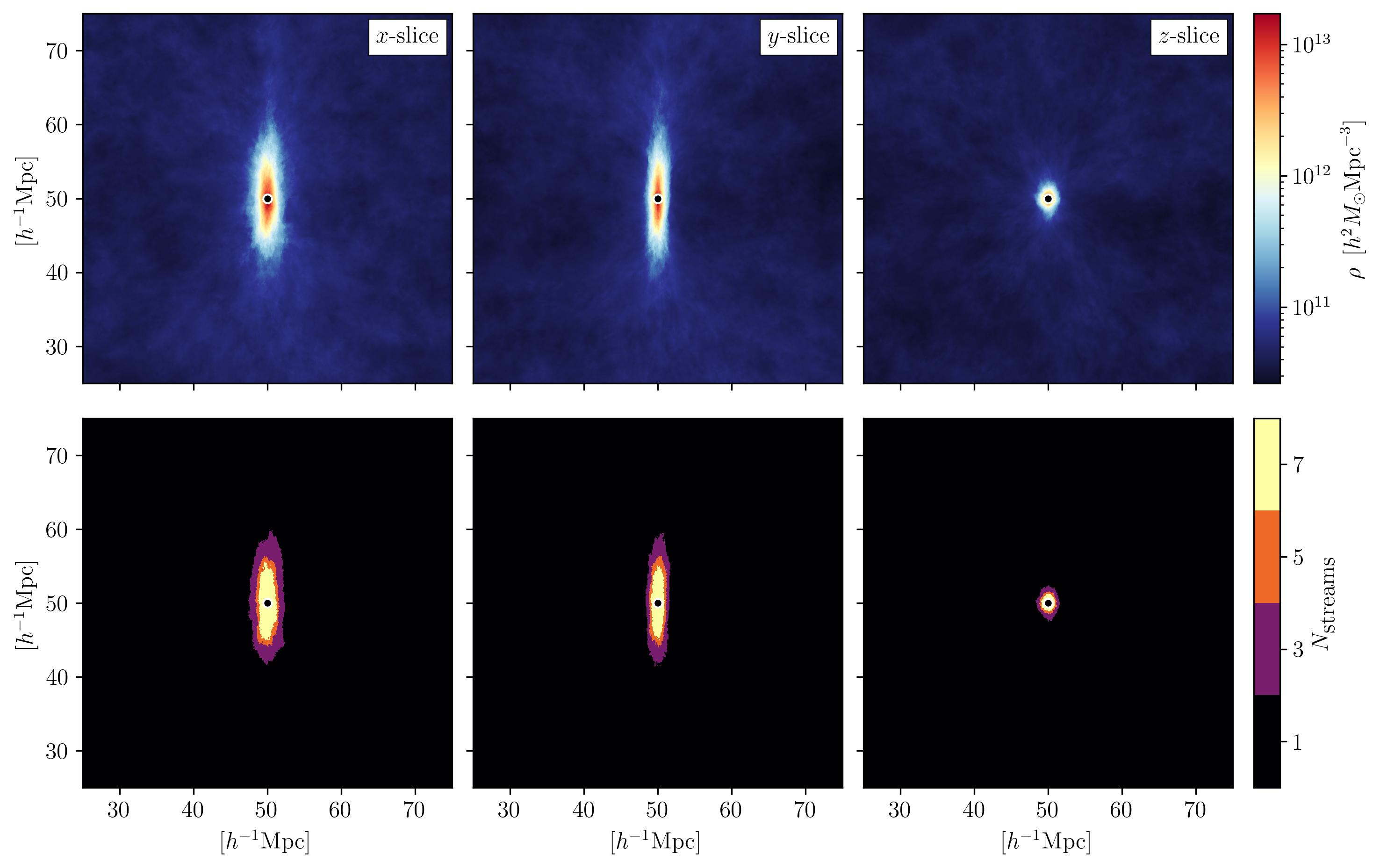}
    \caption{Median fields of the $\delta^{(++-)}$ constraint for $(\sigma, b_c) = (2.0\, h^{-1}\textrm{Mpc}, \,0.95)$. The upper and lower row show the median density and number of streams fields respectively.}
    \label{fig:delta_constraint_mean_field}
\end{figure*}

The resulting median density field for is shown in the upper panel of \cref{fig:delta_constraint_mean_field}. Evidently, the median density field is not sheet-like. Instead, to our surprise, we find that dominant structure of the median field is a short line-like trunk. This is a consequence of the consistent and accurate orientation of the simulated object that was described above. However, closer inspection of the random fields reveals that the variation of the individual density fields is large (see  \cref{fig:delta_constraint_realisations}), and the median density field is therefore less informative than in the $A_3$ case. Nonetheless, the median number of streams in the lower panel of \cref{fig:delta_constraint_mean_field} reveals a crucial property of the $\delta^{(++-)}$ constraint. The median number of streams is much larger than for the $A_3$ constraint. We have truncated the plotted colour scale at $N_{\textrm{streams}}=11$ here for clarity, but find that phase-mixing in clusters can contribute up to $N_{\textrm{streams}}=10^2 \textrm{--}10^3$ and thus increase the median field to the significantly higher value $N_{\textrm{streams}}=35$ (not shown in the figure). Clearly, the constraint point does not reside at the centre of a sheet-like three-streaming region. This is consistent with the high median density value at the constraint point, which is about an order of magnitude larger than the densities inferred from the $A_3$ simulations, and so about 100 times as dense as the cosmic background. Clearly, the high density contrast is inconsistent with the literature on the density of cosmic walls \cite{Forero-Romero+2009, AragonCalvo+2010,  ShandarinSalmanHeitmann2012, Hoffman+2012, Cautun+2014, Libeskind+2017}.

The theoretical reason for the compact nature of the resulting objects is that the power spectrum of primordial density field $\delta = \nabla^2 \Psi$ is given by $P_{\delta}(k) = k^{-4} P(k)$. The density perturbation therefore has a smaller correlation length than $\Psi$, and saddle point constraints in $\delta$ are expected to only influence their close cosmic vicinity. One might attempt to alleviate the situation by imposing the $\delta^{(++-)}$ constraint at a larger smoothing scale $\sigma$. While we speculate that this will indeed result in more extended structures, the constraint still lacks the geometric structure to seed a flat, sheet-like wall, well away from cosmic filaments. In particular, we expect that the constraint point would still result in a cluster-like overdensity within an amorphous extended object. Moreover, Rice's formula predicts near-flat formation time curves for constraints with $\sigma \gtrsim 2.0 \,h^{-1}\textrm{Mpc}$. Hypothetical walls from $\delta^{(++-)}$ saddle points at larger $\sigma$ would therefore not have a characteristic formation peak in the past Universe. This is in stark contrast to our measurements of the emerging wall areas and the number density of the wall centre points from the caustic skeleton model, see \cref{fig:A3_area} and \cref{fig:A3_number_density} respectively.

From our discussion, we find that the $\delta^{(++-)}$ constraint is not adequate for simulations of realistic cosmic walls. While the peak picture in the Eulerian density field is successful at predicting halo populations \cite{PressSchechter1974, Peebles1994}, this ability cannot be extrapolated to study the cosmic web from the primordial density field $\delta$. As will be further illustrated in our follow-up paper, the saddle point constraints in $\delta$ lack the geometric capacity to seed the different cosmic environments, and generally result but in amorphous extended clusters.

To summarise, the conventional ideas of primordial constraints for cosmic web formation are typically formulated either statistically in the morphology of the primordial density perturbation or dynamically in the primordial tidal force field. We have investigated both using constrained simulations and conclude that the conventional constraints are not as successful at creating physically realistic cosmic walls as our novel caustics-based wall centre constraint. The singularity morphogenesis underlying the wall formation is more complex than the structure of critical points in the primordial fields and requires a rigorous treatment based on the phase-space dynamics. These are imprinted in the higher derivatives $T_{ij\ldots k}$ determining the non-trivial configuration of the eigenvalue and eigenvector fields that give rise to the formation of extended cusp sheets connected to and interspersed by cosmic filaments.


\section{Conclusion}
\label{sec:conclusion}

The cosmic web the largest known structural pattern in our Universe. Its intricately connected and multiscale geometry is woven by the non-linear gravitational collapse of the dark and baryonic matter. The caustic skeleton model \cite{ArnoldShandarinZeldovich1982, Feldbrugge+2018} provides a mathematically rigorous and parameter-free classification of the geometric backbone of the emerging cosmic web. In a series of two articles, we follow up on the preceding studies \cite{FeldbruggeWeygaert2023, FeldbruggeWeygaert2024} and, for the first time, investigate the different cosmic web environments in the realistic three-dimensional Universe using constraint simulations from the caustic skeleton model. For the present article, we have studied the physical properties of the cosmic walls.

The walls form an integral part of the cosmic web. They are the first structures to form \cite{Zeldovich1970} and separate the near-empty void regions through multistreaming membranes \cite{ArnoldShandarinZeldovich1982, ShandarinZeldovich1989, Hidding+2013, Feldbrugge+2018}. Although less widely recognised, the walls also constitute the cradle for the formation of the prominent cosmic filaments. Despite their elusive nature in cosmic surveys — due to their low surface density seeding low-mass haloes and faint galaxies \cite{Metuki+2015} — the walls play a crucial role in the geometry and topology of the large-scale structure \cite{WeygaertBond2008}. Using caustic skeleton theory \cite{ArnoldShandarinZeldovich1982, Feldbrugge+2018, FeldbruggeWeygaert2023}, we for the first time investigate the cosmic walls associated with the cusp caustics as a function of their formation time and length scale. To this end, we construct the definition of the proto-wall in the primordial gravitational potential, and, by combining Hamiltonian Monte Carlo \cite{Duane+1987} techniques with the Bertschinger-Hoffman-Ribak method \cite{Bertschinger1987,HoffmanRibak1991,WeygaertBertschinger1996, FeldbruggeWeygaert2023}, develop a rigorous prescription for constrained simulations of wall formation. Applying these methods, we analyse their morphology in the present-day Universe. Alongside the constrained simulations, we extend the Rice formula \cite{Rice1945, AzaïsWschebor2009} to evaluate the statistical properties of the cosmic walls in terms of the primordial power spectrum from first principles; see also \cite{FeldbruggeYanWeygaert2023, FeldbruggeWeygaert2024}.
The constrained simulations show that the caustic definition of the progenitor of the cosmic wall outperforms previous proposals in terms of the saddle points of the primordial gravitational potential \cite{HaarlemWeygaert1993, WeygaertBertschinger1996} or density perturbation \cite{NovikovColombiDore2006, Sousbie+2008, Pogosyan+2009}. We find that the density in the walls exceeds the mean cosmic density by a factor $\rho/\bar{\rho} =2\textrm{--}10$,  in accordance with the existing literature \cite{Forero-Romero+2009, AragonCalvo+2010,  ShandarinSalmanHeitmann2012, Hoffman+2012, Cautun+2014, Libeskind+2017}.  Moreover, we show that the haloes in the walls are typically lighter than those found in filaments and clusters, which is again in agreement with preceding studies \cite{Hahn+2007, Cautun+2014, Metuki+2015, AlonsoEardleyPeacock2015, MetukiLibeskindHoffman2016, Libeskind+2017}. Crucially, however, our simulations reveal that the haloes are not uniformly distributed in the wall sheets, but follow  filamentary patterns that are traced by the scale-space caustic skeleton. This observation reflects the inherent connectivity of the emerging caustic network, and extends conventional notions of cosmic web environments as identified e.g. by \verb|NEXUS(+)| \cite{AragonCalvo+2007, Cautun+2012} or \verb|DisPerSE| \cite{Sousbie2011, Sousbie2011b}.

There are numerous future avenues for the methods and results developed within this article. Firstly, while we have focused on the theoretical foundations of the constraint simulations, our pancake recipe may be directly applied to detailed quantitative analyses of the large-scale dark matter fields that characterise the cosmic walls. This is particularly relevant for understanding mass transport in the cosmic web \cite{Cautun+2014} and the fate of the walls over cosmic time. Beyond the dark matter fields, constrained simulations also offer a well-suited and controlled framework for investigating halo populations across the different web environments, particularly with regard to the recently debated web-dependency of the halo bias \cite{AlonsoEardleyPeacock2015, Yang2017}.
Constrained simulations may be instrumental in distinguishing the impact of the density contrast and the tidal field on the halo bias in cosmic walls. Moreover, a long-standing prospect of the caustic skeleton has been its dressing with baryons: By including baryonic physics in very-high-resolution constrained simulations, we plan to systematically examine the properties of the notoriously elusive wall galaxies, specifically their mass function \cite{Metuki+2015}, metallicity \cite{Domíngues-Gómez+2023} and intrinsic alignments \cite{Codis+2015, Ganeshaiah+2018, Codis+2018}.

The present article has demonstrates that the caustic skeleton is key to understanding the observational reality of the physical cosmic web. The haloes identified in our constrained simulations of cosmic wall formation directly reflect the structures observed in recent galaxy surveys, notably the Coma Wall \cite{GellerHuchra1989}, the Sloan Great Wall \cite{Gott+2005, Einasto+2011} and the BOSS Great Wall \cite{Lietzen+2016, Einasto+2017}. Aside from the constraint simulations, in future studies, we plan to build on these results and will further examine the tracing of the dark matter haloes by the scale-space caustic skeleton down to even lower length and mass scales. This will be of particular relevance with regard to our Local Universe, which is currently being probed in ever-increasing detail by ongoing cosmic surveys. By applying the recently developed reconstruction techniques \cite{JascheWandelt2013, JascheLavaux2019, Valade+2022, Valade+2024, McAlpine+2025}, we aim to identify the dominant singularities that underlie the formation of the large-scale features in our cosmic neighbourhood. We anticipate that the the caustic skeleton will offer new perspectives and valuable insights into the dynamics of the nearby dark matter distribution and the properties of its embedded galaxies.


\acknowledgments

RvdW and JF dedicate this study to the memory of Sergei Shandarin. We owe the inspiration  and initiative of the Caustic Skeleton project to the great insights and enthusiasm of and numerous motivating discussions with Sergei. His Groningen months have become a dear memory in our lives, and a milestone in our career.

The authors thank Erwin Platen for his permission to re-produce the SDSS maps of \cref{fig:SDSS_maps}, and Yan-Chuan Chai for his input contracting walls. This work relied heavily on the computational facilities of the \textit{Eddie} supercomputer of the University of Edinburgh, the computing cluster of the University of Edinburgh's School of Physics and Astromomy and the \textit{Hábrók} supercomputer of the University of Groningen. The authors express their gratitude towards their respective technical support staff members, without whose help this work would not have been possible.  BH is supported by a Science and Technology Funding Council (STFC) PhD studentship, JF is supported by the STFC Consolidated Grant ‘Particle Physics at the Higgs Centre,’ and by a Higgs Fellowship, RvdW acknowledges funding from EU Hoirzon Europe (EXCOSM, grant nr. 101159513). For the purpose of open access, the authors have applied a Creative Commons Attribution (CC BY) license to any Author Accepted Manuscript version arising from this submission.


\bibliographystyle{plain}
\bibliography{bibliography.bib}

\appendix
\section{Caustic skeleton formalism}
\label{app:caustic_skeleton}

We here review the caustic skeleton formalism first for a generic Hamiltonian cosmological flow and subsequently for the Zel'dovich approximation. In doing so, we summarise the caustic conditions presented in Feldbrugge et al. (2018) \cite{Feldbrugge+2018}.

\subsection{Caustic conditions in the general formalism}
\label{app:caustic_skeleton-general}

The first in the list of caustics is the \textit{fold caustic} $A_2$, given by the condition
\begin{equation}
    A_2(t) :\quad 1 + \mu_1(\bm{q},t)=0 \,.
    \label{eq:A2}
\end{equation}
The fold caustic characterises the particle mesh folding over, and thus bounds a multistream region. This is illustrated by the light blue lines in \cref{fig:caustics}. In Arndol'd's terminology \cite{Arnold1986, Arnold1992, ArnoldGuseinZadeVarchenko2012}, the instantaneous foldings (light blue lines) are the \textit{small} fold caustic, while the \textit{big} fold caustic,
\begin{equation}
    A_2^{(b)}(t) = \bigcup_{t^{\prime} \leq t}  A_2(t') \,,
    \label{eq:big_caustic}
\end{equation}
traces out small caustics up to time $t$, thus giving the multistreaming volume at time $t$ (the volume enclosed in the light blue lines). The distinction between the small and big caustics is important here and analogous for the following catastrophes. In \cref{fig:caustics} and the visualisations of \cite{Hidding+2013, Feldbrugge+2018, FeldbruggeWeygaert2023}, the fold caustic is generally illustrated through the small (instantaneous) caustic, whereas the higher catastrophes are shown through the big caustics that trace out the structural elements of the cosmic web over time.

The next in the list of catastrophes is the \textit{cusp caustic} $A_3$, which occurs on the fold caustic $A_2$ when
\begin{equation}
    A_3(t) :\quad \bm{v}_1(\bm{q},t) \cdot \nabla \mu_1(\bm{q},t)=0  \,.
    \label{eq:A3}
\end{equation}
In the first panel of \cref{fig:cusp_sketch}, the instantaneous cusp caustic is given by the cuspy edges of the multistream regions, giving the multistream region its characteristic inside-out silhouette. In the three-dimensional picture, the cusp points are to be understood as an ellipsoidal ring bounding the three-dimensional multistreaming pancake. Over time, the multistream region grows (see \cref{fig:cusp_sketch}), and the cusp points trace out the cusp sheet corresponding to a cosmic wall, which is illustrated by the red lines in the slices of \cref{fig:caustics}.

The \textit{swallowtail caustic} $A_4$ forms on the cusp sheet when
\begin{equation}
    A_4(t) :\quad \bm{v}_1(\bm{q},t) \cdot \nabla\left(\bm{v}_1(\bm{q},t) \cdot \nabla \mu_1(\bm{q},t) \right)=0  \,.
    \label{eq:A4}
\end{equation}
Geometrically, swallowtail caustics corresponds to the cusp sheet folding onto itself, as illustrated by blue dot in the second panel of \cref{fig:caustics}, with the outgoing walls exhibiting the characteristic swallowtail-like geometry. The dot in \cref{fig:caustics} illustrates a slice through the line-like structure that is traced out by the instantaneous $A_4(t)$ points over time. In the cosmic web, this line-like structure manifests itself as an overdense cosmic filament, and multiple slices through swallowtail filaments may be identified in the density and particle-mesh plots of \cref{fig:sim_256_triple_plot} and \cref{fig:sim_256_caustics}.

The last in the $A$-family of singularities is the \textit{butterfly catastrophe} $A_5$, which forms on the $A_4$  swallowtail filament when
\begin{equation}
    A_5(t) :\quad   \bm{v}_1(\bm{q},t) \cdot \nabla\left(\bm{v}_1(\bm{q},t) \cdot \nabla\left(\bm{v}_1(\bm{q},t) \cdot \nabla \mu_1(\bm{q},t) \right) \right)=0  \,.
    \label{eq:A5}
\end{equation}
The butterfly catastrophe occurs as the swallowtail folds onto itself, and thus forms a cluster in the cosmic web. For more details, we refer to \cite{Feldbrugge+2018}.

The $A$ family of caustics discussed above is characterised by the divergence of the Eulerian density, \cref{eq:Eulerian_density}, due to the first eigenvalue field attaining the value $\lambda_1 = b_c^{-1}$. Catastrophe theory reveals that there exists a second family of caustics, known as the $D$ family, for which the divergence occurs due to both the first and second eigenvalue field\footnote{For completeness, we note that catastrophe theory further reveals (see e.g. \cite{Saunders1980}) that in the three-dimensional Universe, there exist no stable singularities with the three eigenvalues coinciding as $\lambda_1 = \lambda_2 =\lambda_3 = b_c^{-1}$. This configuration can only occur as a unstable transient event, and so cannot contribute to the structural elements of the cosmic web.}. The first in this family is the \textit{umbilic caustic} $D_4$, given by
\begin{equation}
    D_4(t) :\quad 1+\mu_1(\bm{q},t)= 1+\mu_2(\bm{q},t) = 0  \,.
    \label{eq:D4}
\end{equation}
The geometry of the umbilic caustic is visualised in \cref{fig:caustics} by the green dot with the characteristic threefold symmetry of the outgoing walls. Over time, the $D_4$ points trace out a line-like structure that constitutes a second and geometrically distinct type of filament. A few umbilic filaments may be identified in \cref{fig:sim_256_caustics}. For completeness, we mention here that the $D_4$ caustic encompasses two mathematically distinct catastrophes, the $D_4^+$ \textit{hyperbolic} and the $D_4^-$ \textit{elliptic umbilic caustic} \cite{Feldbrugge+2018}, with different geometries. These are not of further relevance article, but  will be investigated in the follow-up paper on cosmic filament formation.

The last in the list of caustics is the \textit{parabolic umbilic caustic} $D_5$, which forms on the $D_4$ filament when
\begin{equation}
    D_5(t):\quad \bm{v}_1(\bm{q},t) \cdot \nabla\left(\mu_1(\bm{q},t) - \mu_2(\bm{q},t) \right) = \bm{v}_2(\bm{q},t) \cdot \nabla\left(\mu_1(\bm{q},t) - \mu_2(\bm{q},t) \right) = 0  \,.
\end{equation}
The parabolic umbilic corresponds to a second type of clusters in the cosmic web, and we refer to \cite{Feldbrugge+2018} for a more detailed discussion. This concludes the list of stable catastrophes making up the structural elements of the cosmic web in the three-dimensional Universe.

\subsection{Caustic conditions in the Zel'dovich approximation}
\label{app:caustic_skeleton-ZA}

The condition for the $A_2$ fold caustic to appear at growing mode $b_c$ is given by
\begin{equation}
    A_2(t) : \quad \lambda_1(\bm{q})=\frac{1}{b_c}
    \label{eq:A2_ZA}
\end{equation}
and the Lagrangian-space multistreaming volume at growing mode $b_c$ is defined by the condition $\lambda_1(\bm{q}) \geq b_c^{-1}$.
The $A_3$ cusp condition on the $A_2$ caustic becomes
\begin{equation}
    A_3(t) : \quad \bm{v}_1(\bm{q}) \cdot \nabla \lambda_1(\bm{q}) = 0  \,.
    \label{eq:A3_ZA}
\end{equation}
The $A_4$ swallowtail condition on the $A_3$ caustic becomes
\begin{equation}
    A_4(t) :
    \quad \bm{v}_1(\bm{q}) \cdot \nabla\left(\bm{v}_1(\bm{q}) \cdot \nabla \lambda_1(\bm{q}) \right)  =0  \,.
    \label{eq:A4_ZA}
\end{equation}
The $A_5$ butterfly condition on the $A_4$ caustic becomes
\begin{equation}
    A_5(t) :\quad \bm{v}_1(\bm{q}) \cdot \nabla \left(\bm{v}_1(\bm{q}) \cdot \nabla\left(\bm{v}_1(\bm{q}) \cdot \nabla \lambda_1(\bm{q}) \right) \right)=0  \,.
    \label{eq:A5_ZA}
\end{equation}
The $D_4$ umbilic condition becomes
\begin{equation}
    D_4(t) : \quad \lambda_1(\bm{q})=\lambda_2(\bm{q})=\frac{1}{b_c}
    \label{eq:D4_ZA}
\end{equation}
and, finally, the $D_5$ parabolic umbilic condition on the $D_4$ caustic becomes
\begin{equation}
    D_5(t):\quad \bm{v}_1 \cdot \nabla\left(\lambda_1(\bm{q}) - \lambda_2(\bm{q}) \right) = \bm{v}_2 \cdot \nabla\left(\lambda_1(\bm{q}) - \lambda_2(\bm{q}) \right) = 0
    \label{eq:D5_ZA}  \,.
\end{equation}
\Cref{tab:caustic_skeleton} summarises the relevant conditions of the ZA caustics along with their correspondence to the structural elements of the cosmic web.

At this point, it is instructive to briefly comment on the algebraic structure of the ZA caustic conditions. The fold condition, \cref{eq:A2_ZA}, acts as an isocontour on the eigenvalue field $\lambda_1(\bm{q})$, with the multistream regions given by the filter $\lambda_1(\bm{q}) \geq b_c^{-1}$. With increasing time $b(t)$, the isocontour is lowered and the multistream regions grow. Panel b) and c) of \cref{fig:eigenvalue_visualisation} visualise several isocontours of $\lambda_1(\bm{q})$ to illustrate the growth of cosmic structure. The cusp condition, \cref{eq:A3_ZA}, is a single constraint equation and geometrically defines a two-dimensional structure on the primordial potential perturbation. This is consistent with the identification of cosmic walls. The swallowtail condition, \cref{eq:A4_ZA}, acts as a second constraint equation on the cusp sheet, and thus defines a one-dimensional structure, consistent with the identification of the cosmic filaments. While the situation is more subtle for the umbilic constraint, \cref{eq:D4_ZA}, we will argue in the follow-up paper that the same is true for the umbilic filament family.

\section{Eigenframe identities}
\label{app:eigenframe}

We now list the eigenframe solution derived in \cref{sec:eigenframe} as the solution of the system of equations \crefrange{eq:invariants_eigenframe}{eq:eigenframe_system_3}. From the given construction, the solution constitutes a list of identities in up to fourth-order field derivatives $T_{ij \ldots k}$, as are relevant to the present study.

The eigenvalues are given by
\begin{equation}
    \begin{split}
        \lambda_1 &= T_{11} \\
        \lambda_2 &= T_{22} \\
        \lambda_3 &= T_{33} \\
    \end{split}
\end{equation}
with the corresponding eigenvectors
\begin{equation}
    \bm{v}_1 =
    \begin{pmatrix}
        1 \\ 0 \\ 0
    \end{pmatrix} \qquad
    \bm{v}_2 =
    \begin{pmatrix}
        0 \\ 1 \\ 0
    \end{pmatrix} \qquad
    \bm{v}_3 =
    \begin{pmatrix}
        0 \\ 0 \\ 1
    \end{pmatrix}
\end{equation}
The first-order derivatives of the eigenvalue fields are not corrected by the eigenframe, such that
\begin{equation}
    \nabla \lambda_1 =
    \begin{pmatrix}
        T_{111} \\ T_{112} \\ T_{113}
    \end{pmatrix} \qquad
    \nabla \lambda_2 =
    \begin{pmatrix}
        T_{122} \\ T_{222} \\ T_{223}
    \end{pmatrix} \qquad
    \nabla \lambda_3 =
    \begin{pmatrix}
        T_{133} \\ T_{233} \\ T_{333}
    \end{pmatrix}
\end{equation}
and so e.g. the cusp condition is $\bm{v}_1 \cdot \nabla \lambda_1 = T_{111}$, as was extensively used above.

The first-order derivatives of the eigenvector fields are
\begin{equation}
    \begin{split}
        \partial_1 \bm{v}_1 =
        \begin{pmatrix}
            0 \\ \frac{T_{112}}{T_{11}-T_{22}} \\ \frac{T_{113}}{T_{11}-T_{33}}
        \end{pmatrix} \qquad
        \partial_2 \bm{v}_1 =
        \begin{pmatrix}
            0 \\ \frac{T_{122}}{T_{11}-T_{22}} \\ \frac{T_{123}}{T_{11}-T_{33}}
        \end{pmatrix} \qquad
        \partial_3 \bm{v}_1 =
        \begin{pmatrix}
            0 \\ \frac{T_{123}}{T_{11}-T_{22}} \\ \frac{T_{133}}{T_{11}-T_{33}}
        \end{pmatrix} \\
        \partial_1 \bm{v}_2 =
        \begin{pmatrix}
            \frac{T_{112}}{T_{22}-T_{11}} \\ 0 \\ \frac{T_{123}}{T_{22}-T_{33}}
        \end{pmatrix} \qquad
        \partial_2 \bm{v}_2 =
        \begin{pmatrix}
            \frac{T_{122}}{T_{22}-T_{11}} \\ 0 \\ \frac{T_{223}}{T_{22}-T_{33}}
        \end{pmatrix} \qquad
        \partial_3 \bm{v}_2 =
        \begin{pmatrix}
            \frac{T_{123}}{T_{22}-T_{11}} \\ 0 \\ \frac{T_{233}}{T_{22}-T_{33}}
        \end{pmatrix} \\
        \partial_1 \bm{v}_3 =
        \begin{pmatrix}
            \frac{T_{113}}{T_{33}-T_{11}} \\ \frac{T_{123}}{T_{33}-T_{22}} \\ 0
        \end{pmatrix} \qquad
        \partial_2 \bm{v}_3 =
        \begin{pmatrix}
            \frac{T_{123}}{T_{33}-T_{11}} \\ \frac{T_{223}}{T_{33}-T_{22}} \\ 0
        \end{pmatrix} \qquad
        \partial_3 \bm{v}_3 =
        \begin{pmatrix}
            \frac{T_{133}}{T_{33}-T_{11}} \\ \frac{T_{233}}{T_{33}-T_{22}} \\ 0
        \end{pmatrix}
    \end{split}
\end{equation}

The second-order derivatives are given by
\begin{equation}
    \begin{split}
        \partial_{x}^2\lambda_1 &= \frac{2 T_{112}^2}{T_{11}-T_{22}}+\frac{2 T_{113}^2}{T_{11}-T_{33}}+T_{1111} \\
        \partial_{x}\partial_{y}\lambda_1 &= \frac{2 T_{112} T_{122}}{T_{11}-T_{22}}+\frac{2 T_{113} T_{123}}{T_{11}-T_{33}}+T_{1112}\\
        \partial_{x}\partial_{z}\lambda_1 &= \frac{2 T_{112} T_{123}}{T_{11}-T_{22}}+\frac{2 T_{113} T_{133}}{T_{11}-T_{33}}+T_{1113} \\
        \partial_{y}^2\lambda_1 &= \frac{2 T_{122}^2}{T_{11}-T_{22}}+\frac{2 T_{123}^2}{T_{11}-T_{33}}+T_{1122} \\
        \partial_{y}\partial_{z}\lambda_1 &= 2 T_{123} \left(\frac{T_{122}}{T_{11}-T_{22}}+\frac{T_{133}}{T_{11}-T_{33}}\right)+T_{1123} \\
        \partial_{z}^2 \lambda_1 &= \frac{2 T_{123}^2}{T_{11}-T_{22}}+\frac{2 T_{133}^2}{T_{11}-T_{33}}+T_{1133}
    \end{split}
\end{equation}
\begin{equation}
    \begin{split}
        \partial_{x}^2 \lambda_2 &=\frac{2 T_{112}^2}{T_{22}-T_{11}}+\frac{2 T_{123}^2}{T_{22}-T_{33}}+T_{1122} \\
        \partial_{x}\partial_{y} \lambda_2 &= \frac{2 T_{112} T_{122}}{T_{22}-T_{11}}+\frac{2 T_{123} T_{223}}{T_{22}-T_{33}}+T_{1222} \\
        \partial_{x}\partial_{z} \lambda_2 &= 2 T_{123} \left(\frac{T_{112}}{T_{22}-T_{11}}+\frac{T_{233}}{T_{22}-T_{33}}\right)+T_{1223} \\
        \partial_{y}^2 \lambda_2 &= \frac{2 T_{122}^2}{T_{22}-T_{11}}+\frac{2 T_{223}^2}{T_{22}-T_{33}}+T_{2222} \\
        \partial_{y}\partial_{z} \lambda_2 &=\frac{2 T_{122} T_{123}}{T_{22}-T_{11}}+\frac{2 T_{223} T_{233}}{T_{22}-T_{33}}+T_{2223} \\
        \partial_{z}^2 \lambda_2 &= \frac{2 T_{123}^2}{T_{22}-T_{11}}+\frac{2 T_{233}^2}{T_{22}-T_{33}}+T_{2233} \\
    \end{split}
\end{equation}
\begin{equation}
    \begin{split}
        \partial_{x}^2 \lambda_3 &= \frac{2 T_{113}^2}{T_{33}-T_{11}}+\frac{2 T_{123}^2}{T_{33}-T_{22}}+T_{1133} \\
        \partial_{x} \partial_{y} \lambda_3 &=2 T_{123} \left(\frac{T_{113}}{T_{33}-T_{11}}+\frac{T_{223}}{T_{33}-T_{22}}\right)+T_{1233} \\
        \partial_{x} \partial_{z} \lambda_3 &= \frac{2 T_{113} T_{133}}{T_{33}-T_{11}}+\frac{2 T_{123} T_{233}}{T_{33}-T_{22}}+T_{1333}\\
        \partial_{y}^2 \lambda_3 &= \frac{2 T_{123}^2}{T_{33}-T_{11}}+\frac{2 T_{223}^2}{T_{33}-T_{22}}+T_{2233} \\
        \partial_{y}  \partial_{z} \lambda_3 &= \frac{2 T_{123} T_{133}}{T_{33}-T_{11}}+\frac{2 T_{223} T_{233}}{T_{33}-T_{22}}+T_{2333} \\
        \partial_{z}^2 \lambda_3 &= \frac{2 T_{133}^2}{T_{33}-T_{11}}+\frac{2 T_{233}^2}{T_{33}-T_{22}}+T_{3333}
    \end{split}
\end{equation}
To evaluate the normal of the cusp condition and the expectation values of \cref{subsec:sims-formation_time}, we furthermore need the expressions
\begin{equation}
    \begin{split}
    \nabla \left(\bm{v}_1 \cdot \nabla \lambda_1 \right) &=
    \begin{pmatrix}
        \frac{2 T_{113}^2}{T_{33}-T_{11}}+\frac{2 T_{123}^2}{T_{33}-T_{22}}+T_{1133} \\
        2 T_{123} \left(\frac{T_{113}}{T_{33}-T_{11}}+\frac{T_{223}}{T_{33}-T_{22}}\right)+T_{1233} \\
        \frac{2 T_{113} T_{133}}{T_{33}-T_{11}}+\frac{2 T_{123} T_{233}}{T_{33}-T_{22}}+T_{1333} 
    \end{pmatrix} \\
    \nabla \left(\bm{v}_2 \cdot \nabla \lambda_2 \right) &= 
    \begin{pmatrix}
    \frac{3 T_{112} T_{122}}{T_{22}-T_{11}}+\frac{3 T_{123}T_{223}}{T_{22}-T_{33}}+T_{1222} \\
    \frac{3 T_{122}^2}{T_{22}-T_{11}}+\frac{3T_{223}^2}{T_{22}-T_{33}}+T_{2222}\\
    \frac{3 T_{122} T_{123}}{T_{22}-T_{11}}+\frac{3T_{223} T_{233}}{T_{22}-T_{33}}+T_{2223}
    \end{pmatrix} \\
    \nabla \left(\bm{v}_2 \cdot \nabla \lambda_3 \right) &= 
    \begin{pmatrix}
    \frac{T_{112} T_{133}}{T_{22}-T_{11}}+2 T_{123}\left(\frac{T_{113}}{T_{33}-T_{11}}+\frac{T_{223}}{T_{33}-T_{22}}\right)+\frac{T_{123}T_{333}}{T_{22}-T_{33}}+T_{1233}\\
    \frac{2 T_{123}^2}{T_{33}-T_{11}}+\frac{2T_{223}^2}{T_{33}-T_{22}}+\frac{T_{122} T_{133}}{T_{22}-T_{11}}+\frac{T_{223}T_{333}}{T_{22}-T_{33}}+T_{2233}\\
    \frac{T_{123} T_{133}}{T_{22}-T_{11}}+\frac{2 T_{123}T_{133}}{T_{33}-T_{11}}+\frac{2 T_{223} T_{233}}{T_{33}-T_{22}}+\frac{T_{233}T_{333}}{T_{22}-T_{33}}+T_{2333}
    \end{pmatrix} \\
    \nabla \left(\bm{v}_3 \cdot \nabla \lambda_2 \right) &= 
    \begin{pmatrix}
    \frac{T_{113} T_{122}}{T_{33}-T_{11}}+\frac{T_{123} T_{222}}{T_{33}-T_{22}}+2T_{123}\left(\frac{T_{112}}{T_{22}-T_{11}}+\frac{T_{233}}{T_{22}-T_{33}}\right)+T_{1223}\\
    \frac{2 T_{122} T_{123}}{T_{22}-T_{11}}+\frac{T_{122}T_{123}}{T_{33}-T_{11}}+\frac{T_{222} T_{223}}{T_{33}-T_{22}}+\frac{2 T_{223}T_{233}}{T_{22}-T_{33}}+T_{2223}\\
    \frac{2 T_{123}^2}{T_{22}-T_{11}}+\frac{2T_{233}^2}{T_{22}-T_{33}}+\frac{T_{122} T_{133}}{T_{33}-T_{11}}+\frac{T_{222}T_{233}}{T_{33}-T_{22}}+T_{2233}
    \end{pmatrix} \\
    \nabla \left(\bm{v}_3 \cdot \nabla \lambda_3 \right) &= 
    \begin{pmatrix}
    \frac{3 T_{113} T_{133}}{T_{33}-T_{11}}+\frac{3 T_{123}T_{233}}{T_{33}-T_{22}}+T_{1333}\\
    \frac{3 T_{123} T_{133}}{T_{33}-T_{11}}+\frac{3T_{223} T_{233}}{T_{33}-T_{22}}+T_{2333}\\
    \frac{3 T_{133}^2}{T_{33}-T_{11}}+\frac{3T_{233}^2}{T_{33}-T_{22}}+T_{3333}
    \end{pmatrix} 
    \end{split}
\end{equation}

\section{Derivative rotations}
\label{app:derivs_rotation}

Throughout this article, we work in the $ZYZ$-convention of Euler angles and express a spatial rotation by the Euler matrix $R(\alpha, \beta, \gamma)$ with Euler angles $(\alpha, \beta, \gamma)$,
\begin{equation}
    R(\alpha, \beta, \gamma) = \begin{pmatrix}
        c_{\alpha } c_{\beta } c_{\gamma }-s_{\alpha } s_{\gamma } & -c_{\alpha } c_{\beta } s_{\gamma
   }-c_{\gamma } s_{\alpha } & c_{\alpha } s_{\beta } \\
 c_{\beta } c_{\gamma } s_{\alpha }+c_{\alpha } s_{\gamma } & c_{\alpha } c_{\gamma }-c_{\beta }
   s_{\alpha } s_{\gamma } & s_{\alpha } s_{\beta } \\
 -c_{\gamma } s_{\beta } & s_{\beta } s_{\gamma } & c_{\beta } \\
    \end{pmatrix}
    \label{eq:rot_matrix_ZYZ} \,,
\end{equation}
where we have used the shorthand notation $s_\alpha = \sin \alpha$ and $c_\alpha = \cos \alpha$ and similarly for $\beta, \gamma$.

The matrix in \cref{eq:rot_matrix_ZYZ} is a representation of the $SO(3)$ group, under which the spatial coordinates $x_i$ transform contravariantly with
\begin{equation}
    x_i^{\prime} = R_{ij} x_j \,.
    \label{eq:trf_contraviant}
\end{equation}
In matrix form, this is succinctly written as 
\begin{equation}
    \bm{x}^{\prime} = R \bm{x} \,.
\end{equation}
Under the rotation of the spatial coordinates, the derivatives $\partial_{i} = \frac{\partial}{\partial x_i}$ transform covariantly with
\begin{equation}
    \partial_{i}^{\prime} = (R^{-1})_{ij} \partial_{i} = (R^{\mathrm{T}})_{ij} \partial_{i} \,,
    \label{eq:trf_covariant}
\end{equation}
where in the second equality we have used the defining property $R^{\mathrm{T}}R = I$ of $R \in SO(3)$.
The transformation law of the higher-rank tensors $\partial_i \partial_j \ldots \partial_k$ is consequently given by
\begin{equation}
    \partial_i^{\prime} \partial_j^{\prime} \ldots \partial_k^{\prime} =
    (R^{\mathrm{T}})_{i l} (R^{\mathrm{T}})_{jm}  \ldots (R^{\mathrm{T}})_{kn}  \partial_i \partial_j \ldots \partial_k  \,.
    \label{eq:trf_covariant_tensor}
\end{equation}
Explicitly, the second-order derivatives transform with
\begin{equation}
    \partial_i^{\prime} \partial_j^{\prime} =  (R^{\mathrm{T}})_{ik } (R^{\mathrm{T}})_{jl}  \partial_k \partial_l 
\end{equation}
and the transformed Cartesian field derivatives $t_{ij}^{\prime}$ are given by
\begin{equation}
    t_{ij}^{\prime} =  (R^{\mathrm{T}})_{ik } (R^{\mathrm{T}})_{jl}  \partial_k \partial_l \Psi = (R^{\mathrm{T}})_{ik } (R^{\mathrm{T}})_{jl} t_{kl} = (R^{\mathrm{T}})_{ik } t_{kl}  R_{lj} \,.
    \label{eq:trf_2nd_order_field_derivs}
\end{equation}
Denoting the matrix of second-order derivatives $t_{ij}$ as $H$, \cref{eq:trf_2nd_order_field_derivs} is written in matrix form as
\begin{equation}
    H^{\prime} =  R^{\mathrm{T}} H R
\end{equation}
and we have recovered the standard matrix rotation formula used in the eigendecomposition of \cref{eq:eigendecomposition}.

Similarly, the third and fourth-order derivatives transform as rank-3 and rank-4 tensors respectively with
\begin{equation}
    t_{ijk}^{\prime} = (R^{\mathrm{T}})_{i l} (R^{\mathrm{T}})_{jm}  (R^{\mathrm{T}})_{kn} t_{lmn} \qquad t_{ijkl}^{\prime} = (R^{\mathrm{T}})_{i m} (R^{\mathrm{T}})_{jn}  (R^{\mathrm{T}})_{ko} n(R^{\mathrm{T}})_{kp} t_{mnop} \,.
    \label{eq:trf_3rd_4th_order_field_derivs}
\end{equation}

In principle, one can evaluate the transformation of the relevant field derivatives $t_{ij\ldots k}$ under the spatial rotation directly from \cref{eq:trf_2nd_order_field_derivs} and  \cref{eq:trf_3rd_4th_order_field_derivs}. However, given the commutativity $\partial_i \partial_j = \partial_j \partial_i$, it is more feasible to work only with the independent elements of the tensors $t_{ij}$, $t_{ijk}$ and $t_{ijkl}$. Using the convention of \cref{eq:t_statistic}, these are given by
\begin{equation}
    \bm{t} = \{ t_{11}, t_{12}, t_{13}, t_{22}, \ldots, t_{3333}\}
\end{equation}
We now wish to derive a another representation of the rotation group $SO(3)$ that acts on the 31-dimensional vector $\bm{t}$ in matrix form as
\begin{equation}
    \bm{t}^{\prime} = R^{(2,3,4)} \bm{t}
\end{equation}
The entries of the $31 \times 31$ matrix  $R^{(2,3,4)}$ are found by evaluating the relevant Einstein summations in \cref{eq:trf_covariant_tensor}. As expected on consistency grounds, the rotations of the derivatives of different orders decouple, and the generalised rotation matrix takes block-diagonal form, so that
\begin{equation}
    \bm{t}^{\prime} = \begin{pmatrix} \bm{t}^{(2)\prime} \\ \bm{t}^{(3)\prime} \\ \bm{t}^{(4)\prime} \end{pmatrix}
    = R^{(2,3,4)} \begin{pmatrix} \bm{t}^{(2)} \\ \bm{t}^{(3)} \\ \bm{t}^{(4)} \end{pmatrix}
    = \begin{pmatrix}
        R^{(2)} & 0 & 0 \\
        0 & R^{(3)} & 0 \\
        0 & 0 & R^{(4)}
    \end{pmatrix} \begin{pmatrix} \bm{t}^{(2)} \\ \bm{t}^{(3)} \\ \bm{t}^{(4)} \end{pmatrix} \,.
    \label{eq:rot_derivs_mat}
\end{equation}
The entries of the generalised rotation block matrices $R^{(2)}$, $R^{(3)}$ and $R^{(4)}$ are lengthy and therefore omitted here. We present the full matrices in the \verb|julia| file \verb|derivs_rotation.jl| attached as supplementary material to this article.

\section{Constructing Euler angles from vector bases}
\label{app:euler_angles}

In order to orient the imposed caustic constraints along a given coordinate system, we wish to find the Euler angles $(\alpha, \beta, \gamma)$ that rotate orthonormal set of three basis vectors $\bm{v}_i$ into another given set of orthonormal basis vectors $\bm{\tilde{v}}_i$ of the same handedness through the $ZYZ$-convention rotation matrix \cref{eq:rot_matrix_ZYZ}. While rigid-body rotations are extensively studied in the mathematical literature, engineering applications typically use other conventions for the Euler angles, and we were not able to find the relevant equations for our given problem in the $ZYZ$ convention in the literature. To perform the rotation our desired convention, we therefore solve the problem by a mathematical trick: We first construct the rotation quaternion $\bm{q} = (q_r, q_x, q_y, q_z)$ that rotates $\bm{v}_i$ into $\bm{\tilde{v}}_i$ according to the conjugation composition
\begin{equation}
    (0, \bm{\tilde{v}}) = \bm{q} \cdot (0, \bm{v}) \cdot \bm{q}^{-1} \,.
    \label{eq:quat_rotation}
\end{equation}
Here, we denote by $(0, \bm{v})$ a quaternion with vector part $\bm{v}$, and the conjugation operation can be found it in numerous mathematics text books, e.g. \cite{Hanson2005, Kanatani2020}. From the rotation quaternion $\bm{q}$, we then construct the rotation matrix $R(\alpha, \beta, \gamma)$ and solve \cref{eq:rot_matrix_ZYZ} for the $ZYZ$-covention Euler angles $(\alpha, \beta, \gamma)$.

We follow \cite{BeslMcKay1992} for the construction of the rotation quaternion $\bm{q}$ that rotates  $\{\bm{v}_i\}$ into $\{\bm{\tilde{v}}_i\}$. In their study, the authors solve the general problem of finding the rotation that minimises the distance of a rotated set of vectors to a reference set of vectors, with the rotation being expressed through a quaternion $\bm{q}$. The two sets vectors are here of arbitrary but equal size. For two sets $(\{\bm{v}_i\}, \{\bm{\tilde{v}}_i\})$ of three orthogonal vectors of the same handedness, the minimisation problem identically reduces to the construction of the quaternion rotating the basis system $\{\bm{v}_i\}$ into the target system $\{\bm{\tilde{v}}_i\}$. We refer the reader to \cite{BeslMcKay1992} for the detailed derivation, and briefly summarise here the algorithm.

One starts by constructing the $3 \times 3$ matrix given by the sum of the outer products of the basis system and target system vectors,
\begin{equation}
    M = \sum_{i = 1,2,3} \bm{v}_i \otimes \bm{\tilde{v}}_i = \sum_{i = 1,2,3} \bm{v}_i (\bm{\tilde{v}}_i)^{\mathrm{T}} \,,
\end{equation}
where $\bm{v}_i$ denotes the $i$th vector in the basis, and we explicitly spelled out the Einstein summation for clarity.

From $M$, one then constructs the symmetric $4 \times 4$-matrix $N$ given by

\begin{equation}
    N = \begin{pmatrix}
        M_{11} + M_{22} + M_{33} & M_{23} - M_{32} & M_{31} - M_{13} & M_{12} - M_{21} \\
        M_{23} - M_{32} & M_{11} - M_{22} - M_{33} & M_{1 2} + M_{21} & M_{31} + M_{13} \\
        M_{31} - M_{13} & M_{12} + M_{21} &  -M_{11} + M_{22} - M_{33} &  M_{23} + M_{32} \\
        M_{12} - M_{21} & M_{31} + M_{13} & M_{23} + M_{32} &  -M_{11} - M_{22} + M_{33} 
    \end{pmatrix}
\end{equation}
The quaternion $\bm{q} =(q_r, q_x, q_y, q_z)$ yielding the desired rotation, \cref{eq:quat_rotation}, is given by the 4-dimensional eigenvector $\bm{w}_1$ of $N$ with maximal eigenvalue $\nu_1$, defined by the characteristic equation
\begin{equation}
    N \bm{w}_1 = \nu_1 \bm{w}_1 
\end{equation}

Having obtained the rotation quaternion $\bm{q} = (q_r, q_x, q_y, q_z)$, the corresponding rotation matrix is given by the well-known expression \cite{BeslMcKay1992, Hanson2005, Kanatani2020}
\begin{equation}
    R(\bm{q}) = \begin{pmatrix}
        1 - 2(q_y^2 + q_z^2) & 2(q_x q_y - q_r q_z) &  2(q_x q_z - q_r q_y) \\
        2(q_x q_y - q_r q_z) & 1 - 2(q_x^2 + q_z^2) & 2(q_y q_z - q_r q_x) \\
        2(q_x q_z - q_r q_y) & 2(q_y q_z - q_r q_x)  & 1 - 2(q_x^2 + q_y^2) 
    \end{pmatrix}
    \label{eq:rot_matrix_q}
\end{equation}
From the constructed rotation matrix $R = R(\bm{q})$, the Euler angles $(\alpha, \beta, \gamma)$ of a chosen convention are obtained by equating \cref{eq:rot_matrix_q} to the Euler matrix $R(\alpha, \beta, \gamma)$ of the same convention and solving for $(\alpha, \beta, \gamma)$. We are interested in the $ZYZ$-convention Euler angles of \cref{eq:rot_matrix_q}, and use the standard result of found in numerous engineering textbooks, e.g. \cite{SpongHutchinsonVidyasagar2020},
\begin{equation}
    \begin{split}
        \alpha &= \arctan \left( R_{23},  R_{13} \right) \\ 
        \beta &=  \arctan \left( \sqrt{R_{13}^2+ R_{23}^2}, R_{33} \right) \\
        \gamma &= \arctan \left( R_{32},  -R_{31} \right) \,.
    \end{split}
\end{equation}
Here, the $\arctan$-function denotes the 2-argument arctangent.

This concludes the construction of the Euler angles for the rotation of vector basis system into a target system.

\begin{figure*}
    \includegraphics[width=\textwidth]{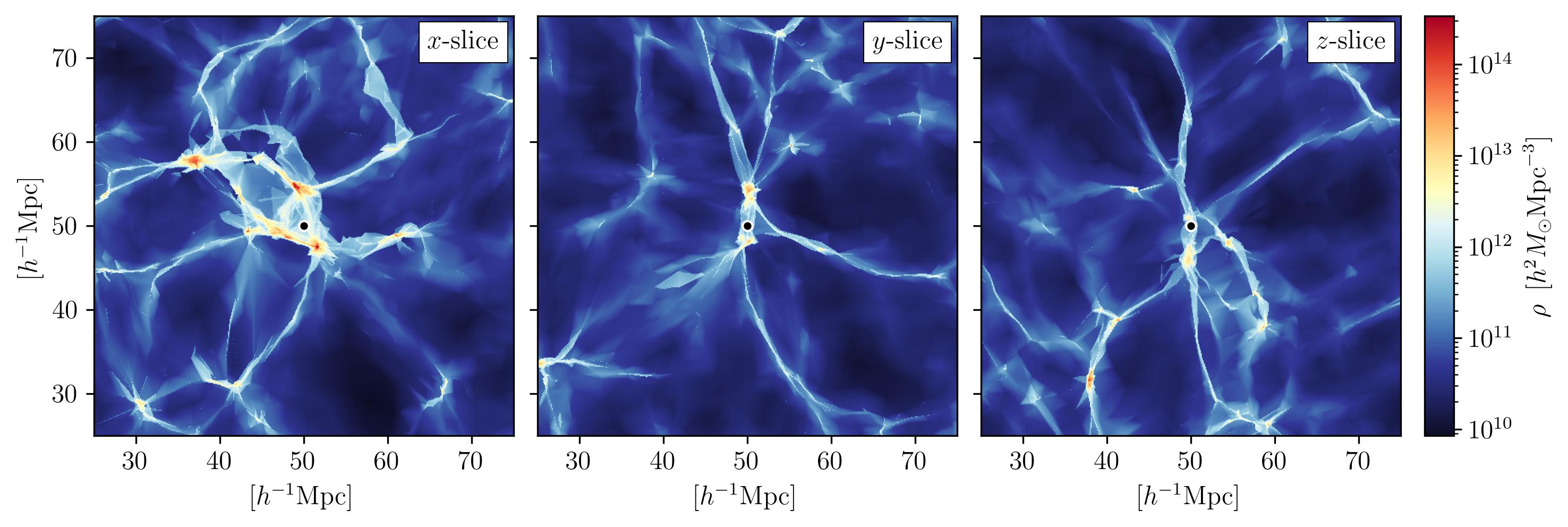}
    \includegraphics[width=\textwidth]{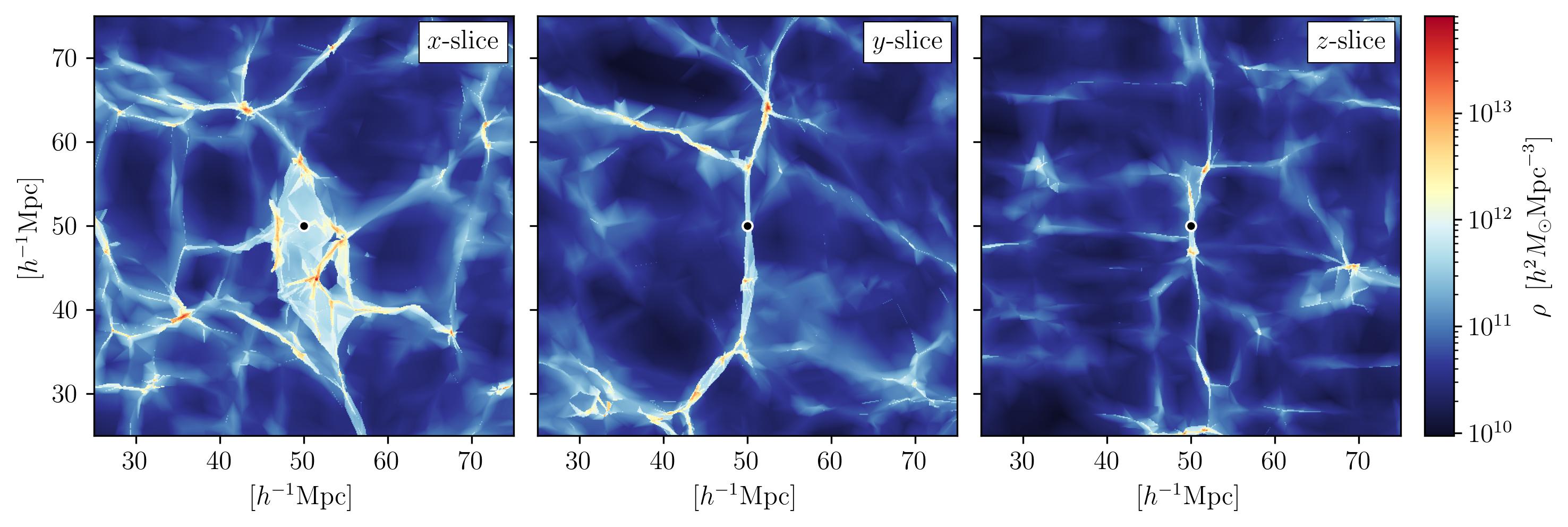}
    \includegraphics[width=\textwidth]{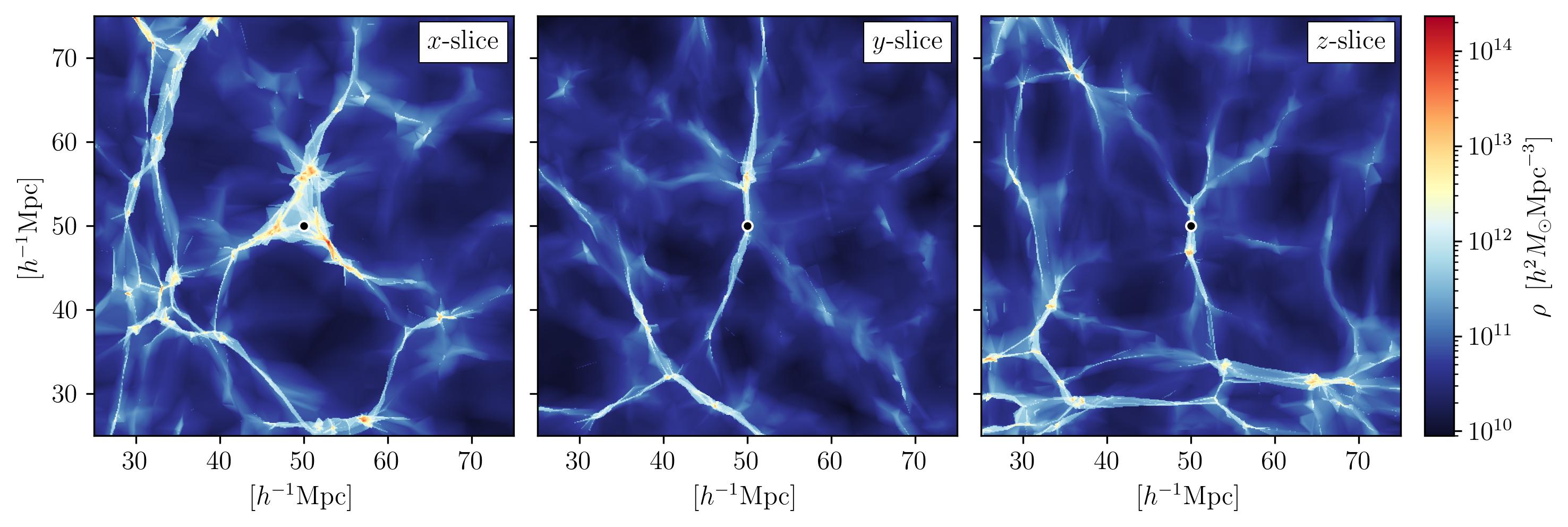}
    \caption{Exemplary density field realisations for the contracting $A_3$ wall centre constraint for $(\sigma, b_c)=(2.0\, h^{-1}\textrm{Mpc}, \, 0.8)$.}
    \label{fig:A3_contract_realisations}
\end{figure*}

\section{Contracting walls}
\label{app:contracting_walls}

In \cref{sec:A3_centre_constraint}, we proposed to simulate cosmic walls by imposing the novel cusp centre constraint, \cref{eq:A3_centre_constraint}, with the additional condition $\lambda_2 < 0$ for linear expansion of the wall sheet. The reason for doing so is that within this work, we are primarily interested in simulating extended objects to infer physical observables and halo properties within walls. However, walls in the observed cosmic web need not necessarily expand. In fact, the caustic skeleton formalism naturally suggests that collapsing walls seed higher catastrophes, which manifest themselves in the higher-density filaments and clusters. By \cref{eq:kappa}, the linear expansion within the cusp sheet is given by the summation of the second and third eigenvalue. Their marginal distribution is therefore a quantitative measure of the probability of walls expanding or contracting within the cosmic web. In particular, the second eigenvalue $\lambda_2$ is an approximate measure of contraction, with $\lambda_2 < 0$ corresponding to expansion and the limiting case $\lambda_2 = \lambda_1$ recovering the umbilic filament condition, corresponding to a fully collapsed wall. As seen in \cref{fig:sample_histogram}, a substantial part of the marginal distribution lives in the space  $\lambda_2 > 0$.

\begin{figure*}
    \includegraphics[width=\textwidth]{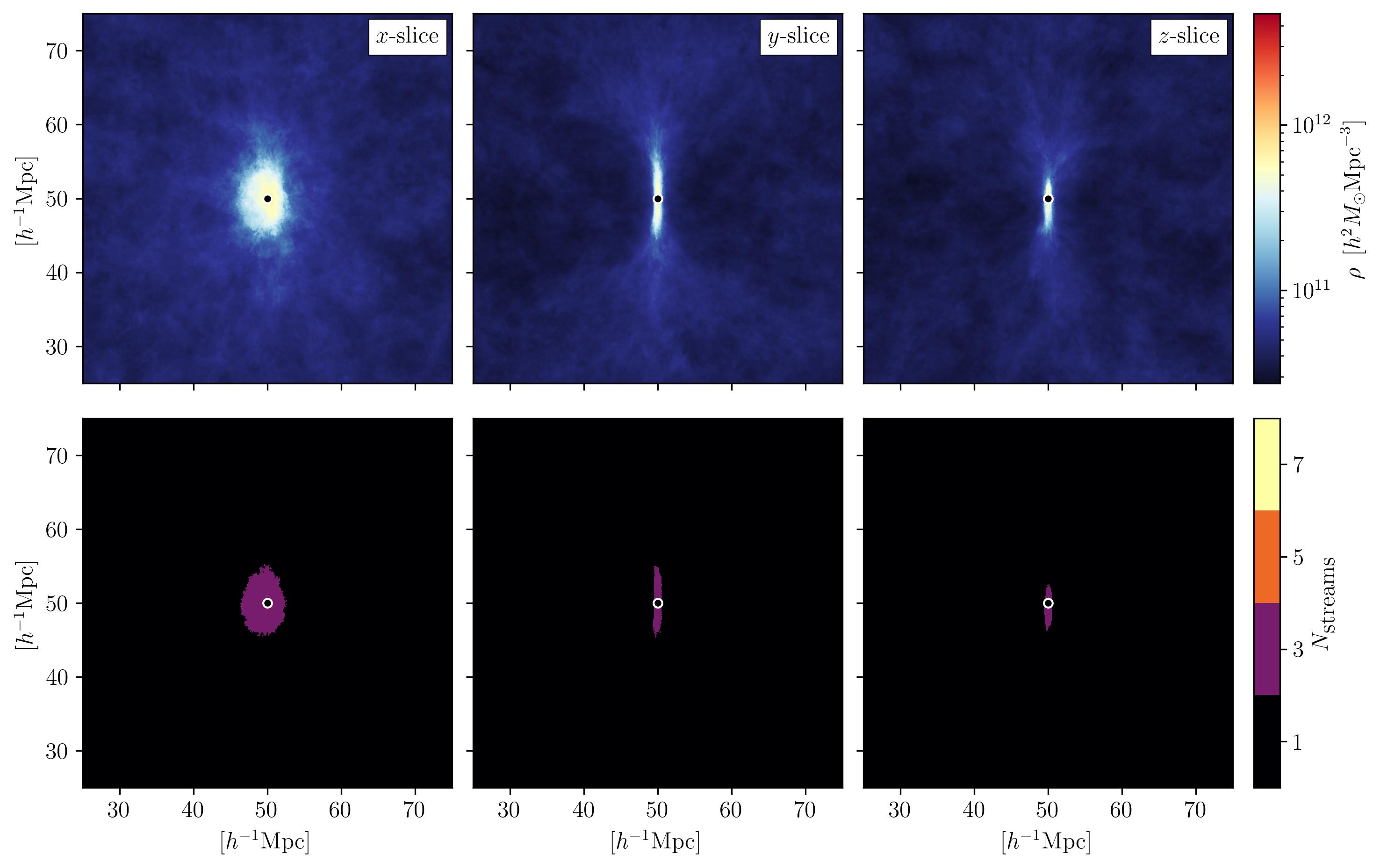}
    \caption{Median density and number of streams fields for the contracting $A_3$ wall centre constraint for $(\sigma, b_c)=(2.0\, h^{-1}\textrm{Mpc}, \, 0.8)$.}
    \label{fig:A3_contract_mean_fields}
\end{figure*}

\Cref{fig:A3_contract_realisations} shows random realisations of the wall centre constraint with the sampling restricted to the explicit contraction condition $\lambda_2 > 0$. Clearly, the shown density fields are qualitatively similar to those of \cref{fig:A3_realisations}, where we imposed the wall centre constraint with the explicit expansion condition $\lambda_2 <0$. We find, however, that the wall areas seen in the $x$-slices of \cref{fig:A3_contract_realisations} are typically less extended, as is expected from the contraction in the linear regime. Within our suite of 100 constraint realisations, we observe that the linearly contracting walls still form extended and accurately oriented planar overdensities. However, the overdensities are now both higher and more spatially compact. This is suitably illustrated in the median density field shown in \cref{fig:A3_contract_mean_fields}. While the spatial extent of the overdensity is somewhat less than in \cref{fig:A3_constraint_mean_field} ($\approx 8 \,h^{-1}\textrm{Mpc}$ rather than $\gtrsim 10 \,h^{-1}\textrm{Mpc}$), the overdensity increases by about a factor of two, reflecting the contraction of the wall into a smaller cosmological volume. Moreover, the thickness of the wall may be increased in some cases. Overall, we find that the density contrast of the contracting walls in still in good agreement with the values reported in the general literature on cosmic walls \cite{Forero-Romero+2009, AragonCalvo+2010,  ShandarinSalmanHeitmann2012, Hoffman+2012, Cautun+2014, Libeskind+2017}, keeping in mind the limitations of previous analyses (see \cref{subsec:sims-fields}).

When running a large suite of simulations on ergodically sampled constraint realisations, it is expected that a few constraints have a large second eigenvalue $\lambda_2 \approx \lambda_1$. Cosmic walls from these constraint realisations may contract beyond the typical sheet-like morphology of a cosmic wall, and appear either as a filamentary trunk or an extended cluster. Indeed, we find this to be case for a small number of simulations in our suite. This is not a numerical shortcoming of the pancake recipe, but a consequence of the realistic random sampling from the smooth constraint potential, i.e. the log-likelihood of the constrained eigenframe derivatives, given in \cref{eq:sampling_potential}. \Cref{fig:A3_contract_mean_fields} demonstrates that the impact of these realisation is still subdominant to the sheet-like nature of the $A_3$ centre constraint. Within the present article, we have nevertheless focused the positive stretch condition $\lambda_2 < 0$ to ensure that each simulation results in a well-extended cosmic wall, as was particularly relevant for the identification of wall haloes in the smaller suite of high-resolution simulations in \cref{sec:haloes}. Future studies on galaxy formation is cosmic walls may wish to relax the condition $\lambda_2$ in order to study cosmic walls in all their possible morphologies occurring in our Universe.

\end{document}